\documentclass[onecolumn,aps,prx,longbibliography,superscriptaddress,floatfix,10pt]{revtex4-2}

\usepackage{expressivity-of-analog}
\usepackage{optidef}
\usepackage{booktabs}

\begin{document}

\title{Universal Dynamics with Globally Controlled Analog Quantum Simulators}

\begin{abstract}

Analog quantum simulators with global control fields have emerged as powerful platforms for exploring complex quantum phenomena. Recent breakthroughs, such as the coherent control of thousands of atoms, highlight the growing potential for quantum applications at scale. Despite these advances, a fundamental theoretical question remains unresolved: to what extent can such systems realize universal quantum dynamics under global control? Here we establish a necessary and sufficient condition for universal quantum computation using only global pulse control, proving that a broad class of analog quantum simulators is, in fact, universal. We further extend this framework to fermionic and bosonic systems, including modern platforms such as ultracold atoms in optical superlattices.
Moreover, we observe that analog simulators driven by random global pulses exhibit information scrambling comparable to random unitary circuits. In a dual-species neutral-atom array setup, the measurement outcomes anti-concentrate on a $\log N$ timescale despite the presence of only temporal randomness, opening opportunities for efficient randomness generation. To bridge theoretical possibility with experimental reality, we introduce \emph{direct quantum optimal control}, a control framework that enables the synthesis of complex effective Hamiltonians while incorporating realistic hardware constraints. 
Using this approach, we experimentally engineer three-body interactions outside the blockade regime and demonstrate topological dynamics on a Rydberg-atom array. Experimental measurements reveal dynamical signatures of symmetry-protected-topological edge modes, confirming both the expressivity and feasibility of our method. Our work opens a new avenue for quantum simulation beyond native hardware Hamiltonians, enabling the engineering of effective multi-body interactions and advancing the frontier of quantum information processing with globally-controlled analog platforms.

\end{abstract}
\date{\today}

\author{Hong-Ye Hu}
\altaffiliation{These authors contributed equally to this work.}
\thanks{Corresponding author: \href{mailto:hongyehu@fas.harvard.edu}{hongyehu@fas.harvard.edu}}
\affiliation{Department of Physics, Harvard University, Cambridge, MA 02138, USA}
\author{Abigail McClain Gomez}
\altaffiliation{These authors contributed equally to this work.}
\affiliation{Department of Physics, Harvard University, Cambridge, MA 02138, USA}
\author{Liyuan Chen}
\altaffiliation{These authors contributed equally to this work.}
\affiliation{Department of Physics, Harvard University, Cambridge, MA 02138, USA}
\affiliation{School of Engineering and Applied Sciences, Harvard University, Allston, MA 02134, USA}
\author{Aaron Trowbridge}
\affiliation{Robotics Institute, Carnegie Mellon University, Pittsburgh, PA 15213, USA}
\affiliation{Harmoniqs, Inc.}
\author{Andy J. Goldschmidt}
\affiliation{Department of Computer Science,University of Chicago,
Chicago, IL 60637, USA}
\author{Zachary Manchester}
\affiliation{Robotics Institute, Carnegie Mellon University, Pittsburgh, PA 15213, USA}
\author{Frederic T. Chong}
\affiliation{Department of Computer Science,University of Chicago,
Chicago, IL 60637, USA}
\author{Arthur Jaffe}
\affiliation{Department of Physics, Harvard University, Cambridge, MA 02138, USA}
\affiliation{Department of Mathematics, Harvard University, Cambridge, Massachusetts 02138, USA}
\author{Susanne F. Yelin}
\thanks{Corresponding author: \href{mailto:syelin@g.harvard.edu}{syelin@g.harvard.edu}}
\affiliation{Department of Physics, Harvard University, Cambridge, MA 02138, USA}

\maketitle

\section{Introduction}

Precise analog pulse control lies at the foundation of all modern quantum technology. Building on this physical principle, two main paradigms of programmable quantum platforms have emerged: digital quantum computers and analog quantum simulators. In recent years, digital quantum computers have often drawn the spotlight. Remarkable progress in these devices has enabled fault-tolerant quantum error correction that surpasses the break-even point of noise \cite{google_breakeven,2022arXiv220801863R,harvard_breakeven,bosonic_code,real_time_qec}. At the hardware level, digital quantum computers rely on continuous-time analog dynamics governed by carefully shaped control pulses, which implement the local, few-qubit gates necessary for computation \cite{995,PhysRevLett.123.170503,iongates}. 
In parallel, there has been rapid progress in the development of analog quantum simulators with global control fields.
Recent experiments
have demonstrated coherent manipulation of hundreds to thousands of atoms and ions \cite{3000qubits,6000atoms,500ions,spinliquid,nist_ions,France_atoms,2024arXiv240702553K}.
These globally-controlled
analog platforms have emerged as leading candidates for scaling to large system sizes compared to digital architectures that require fully local control.

Although there have been many compelling demonstrations of analog quantum simulation in the study of quantum many-body physics, such as the exploration of complex out-of-equilibrium phenomena \cite{harvard_dynamics,fermigas,google_digital_analog,thermalization,fragmentation,YouLi_information_scrambling}, superconductivity \cite{low_temperature,AFM,quantum_dots,Mott,ustc_3D,Bloch_review}, topological phases \cite{haldane_phase,YouLi_topology,PhysRevA.110.023318,Antoine_ssh}, and lattice gauge theories \cite{quera_string_breaking,Z2}, analog simulators are still often perceived as being limited to emulating Hamiltonians closely tied to their native physical interactions. To bridge the gap between analog and digital quantum platforms, a fundamental theoretical question must be addressed: to what extent can analog systems realize universal quantum dynamics under global control?

In this work, we address this question with a surprisingly simple and general condition for achieving universal quantum computation using only global pulse control, which is both necessary and sufficient. Our approach relies on the tool of dynamical Lie algebras \cite{dAlessandro2021} and matrix representation theory \cite{KHANEJA200111}. In contrast to previous studies \cite{Simon1,Benjamin_PRL,qaoa_universal,Hannes_dual_species}, which analyze specific systems on a case-by-case basis, our condition identifies the fundamental obstruction of universality implicitly broken in these constructions, independent of the specific physical Hamiltonians.
A striking implication is that a broad class of analog quantum simulators are, in fact, capable of universal quantum computation when equipped with appropriate global control pulses. We further extend this framework to establish universal simulation in fermionic and bosonic systems, including modern platforms such as ultracold atoms in optical superlattices \cite{superlattice_chemistry,superlattice_review,PhysRevLett.134.053402}.
One important application of globally driven analog systems without fine-tuned control pulses is the study of quantum information scrambling and randomness, which play central roles in modern quantum science, including quantum supremacy proposals \cite{google_supremacy,PhysRevLett.134.090601}, randomized benchmarking \cite{Joonhee_23,PhysRevLett.131.110601}, and cryptography \cite{2023arXiv230301625A,Liu_25,CHAMON2022}. Motivated by recent advances in dual-species neutral-atom arrays \cite{dualspecies,dualspecies2,dualspecies3}, we numerically investigate information scrambling in the measurement outcomes of such systems. Remarkably, we find that the measurement outcomes anti-concentrate on a $\log N$ timescale despite the presence of only temporal randomness. This result highlights the potential of globally driven analog platforms for efficient randomness generation with broad applications.

$$
\boxed{
\underset{\text{\normalsize (Theory)}}{\text{Dynamical~Lie~Algebra}} +
\underset{\text{\normalsize (Experiment)}}{\text{Direct Quantum Optimal Control}}
\;\Longrightarrow\;
\text{Useful Universal Analog Simulator}
}
$$

While the theory of dynamical Lie algebras tells us what is achievable with global pulse engineering, it is not constructive. As suggested by the main thematic equation above, realizing a useful and universal analog quantum simulator also requires a practical tool. To this end, we introduce a new quantum control technique into the experimental workflow: direct quantum optimal control.
Originally developed in the context of aerospace and robotics \cite{Zac_Direct_Method,goldschmidt2022model,trowbridge2023direct}, this method treats both the analog control pulses and the quantum states (or unitaries) at each time point as trainable parameters, with the Schrödinger equation imposed as a constraint. In contrast to conventional pulse control methods such as Floquet engineering \cite{FloquetPRX,Bloch_floquet,RydbergFloquet,floquet_spin,2024arXiv240802741U} and GRadient Ascent Pulse Engineering (GRAPE) \cite{grape} (categorized as indirect methods), direct quantum optimal control allows the optimization algorithm to initialize in and explore ``nonphysical'' regions --- these are points in the joint state and control trajectory space that are \textit{infeasible}, i.e., they do not satisfy the dynamics constraints --- during intermediate steps before convergence to a feasible solution. This greater flexibility in finding effective solutions makes direct methods particularly well suited for realistic experimental settings involving multiple hardware constraints, where traditional techniques often struggle with local minima or fail to accommodate system limitations.

In this work, we combine our new developments in both theory and experiment to demonstrate the potential of globally controlled analog quantum simulators. To substantiate the power of our approach, we focus on Rydberg atom arrays with van der Waals interactions. Qubits are encoded in the ground and Rydberg states, which are also known as the analog encoding. Although the native Hamiltonian contains only long-range two-body interactions, we demonstrate the engineering of effective \emph{three-body} interactions using global pulse control. As a concrete example, we focus on the synthesis of a ZXZ Hamiltonian, which lies in a symmetry-protected topological (SPT) phase \cite{ruben_spt}.

In Rydberg analog quantum simulation, previous experiments have typically operated within the Rydberg blockade regime, where strong native two-body interactions enable high-precision control \cite{995,semeghini_spinliquid,quantum_scar,Fast_Rydberg_gate_2000,quantum_ising_blockade}. However, using the theoretical tools we developed, we show that the ZXZ Hamiltonian cannot be realized within the typical blockade regime but becomes possible outside it. To achieve this goal, we overcome several key experimental challenges. The first is decoherence arising from the residual atomic temperature, since atoms are untrapped during evolution in the analog encoding. We develop a phenomenological model that captures this process well. In addition to thermal motion, we must also satisfy strict constraints on the laser control parameters, further increasing the experimental complexity. Using our control framework, we identify smooth, short-duration pulses that achieve high-fidelity dynamics. Experimental measurements reveal dynamical signatures of symmetry-protected topological edge modes, confirming both the expressivity and experimental feasibility of our approach.

Our work paves a new path for analog quantum simulation through advanced pulse control techniques that go beyond the native Hamiltonians of the simulators. It significantly expands the expressive power of analog platforms and lays the foundation for simulating exotic quantum many-body phenomena through synthesized multi-body interactions, as well as for large-scale quantum information processing using analog quantum devices.

\begin{figure*}[htbp]
    \centering
    \includegraphics[width=0.98\linewidth]{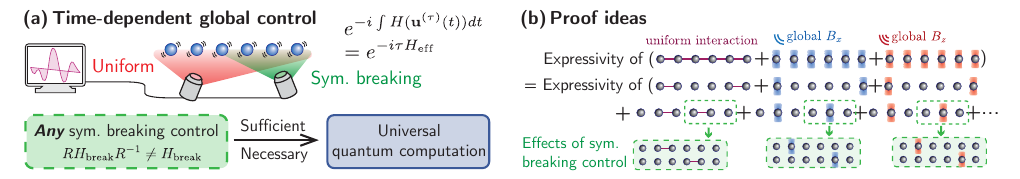}
    \caption{\textbf{Expressivity of globally controlled analog quantum simulators.} (a) An analog quantum system driven by uniform global control fields (red) together with an additional symmetry-breaking field (green) that breaks reflection symmetry. Any global field that violates this symmetry suffices. Here we illustrate this with a green field applied only to the right half of the system. We prove that the presence of such a symmetry-breaking control is both \emph{necessary} and \emph{sufficient} for universal quantum computation. The symmetry-breaking Hamiltonian $H_{\mathrm{break}}$ only needs to satisfy $R H_{\mathrm{break}} R^{-1} \neq H_{\mathrm{break}}$, where $R$ denotes the reflection operator.
(b) Schematic illustration of the proof strategy based on group representation theory. Uniform global controls enable near-independent control of local fields and interactions, up to a residual reflection symmetry that couples qubits at mirrored positions (e.g., the first and last sites). Using representation theory, we show that \emph{any} additional symmetry-breaking control suffices to lift this constraint and complete the controllable space.}
    \label{fig:fig1}
\end{figure*}

\section{Universal quantum computation under global controls}

\subsection{Minimal requirement for universal quantum computation}

The Hilbert space of an $N$-qubit system is the $N$-fold tensor product of $\mathbb{C}^2$, i.e., $\mathcal{H}_q=(\mathds{C}^2)^{\otimes N}$. Given a set of time-dependent control pulses $\mathbf{u}(t)$, the dynamics of an analog system are governed by a time-dependent Hamiltonian of the form
$H(t)=\sum_{\alpha}u_{\alpha}(t)H_{\alpha}$.
Here, $H_{\alpha} \in \{H_1,H_2,\cdots,H_l\}$ is a set of basic control Hamiltonians, and $u_{\alpha}(t)$ is the corresponding component of the control pulse vector $\mathbf{u}(t)$.
More concretely, the unitary dynamics generated by $H(t)$ is given by:
\eqs{\label{eq:time_dependent_Hamiltonian}
U=\mathcal{T}\left[e^{-i\int H(t)dt}\right]=\mathcal{T}\left[e^{-i\int \sum_{\alpha} u_{\alpha}(t)H_{\alpha}dt}\right]\;,
}
where $\mathcal{T}$ denotes time-ordering. 
A system is said to be universal or to realize universal quantum computation (UQC) if we can approximate any target unitary $V$ on $\mathcal{H}_q$ by $U$.
To be precise, given any target unitary $V \in \mathrm{SU}(2^N)$ and any desired precision $\epsilon$, there exists a control sequence $\mathbf{u}(t)$ such that $\norm{U-V}_{\diamond}\leq \epsilon$, where $\norm{\cdot}_{\diamond}$ is the diamond norm between quantum channels \cite{seth_science_UQC,Seth_almost_any_gate_universal,PhysRevLett.89.247902}.

In this work, we investigate UQC in globally-controlled systems.
We denote the support of a basic control Hamiltonian $H_\alpha$, i.e., the qubits it acts nontrivially on, by $\mathrm{Supp}(H_\alpha)$.
If all control fields have extensive support, namely $|\mathrm{Supp}(H_{\alpha})| = \mathcal{O}(N)$ for all $H_\alpha$, we say the system is globally-controlled. In such systems, only $\mathcal{O}(1)$ independent control fields are required, significantly simplifying experimental implementation.
By contrast, traditional approaches to UQC rely on universal local gate sets, such as the Clifford plus T gates~\cite{Nielsen_Chuang_2010}, where each gate acts  on only a small constant number of qubits, i.e. $|\mathrm{Supp}(H_{\alpha})| = \mathcal{O}(1)$ for each $H_\alpha$. These local control schemes require full spatial addressability of individual qubits, with the number of controls scaling as $\mathcal{O}(N)$. This demand introduces considerable experimental overhead compared to global control. Therefore, it naturally raises the question: is an extensive amount of local control necessary for UQC, or can one achieve universality using only a small number of global control fields, independent of system size?

Surprisingly, the latter is possible. In a pioneering work, Benjamin \cite{Benjamin_PRL} showed that a globally-controlled quantum dot system with alternating on-off interactions can realize universal computation. This theoretical framework has since inspired the design of contemporary globally controlled hardware platforms \cite{Menta1,Menta2}. Furthermore, it has been shown that this question is also related to the expressivity of variational quantum algorithms, where the control pulses $\mathbf{u}(t)$ serve as trainable parameters \cite{Lie_barrenplateau,classification_lie,Shengyu_DLA}. In this context, Lloyd \cite{qaoa_universal} demonstrated that a globally controlled variant of the Quantum Approximate Optimization Algorithm (QAOA) is universal. However, previous studies of universality under global control have primarily followed a case-by-case analysis for a specific set of control Hamiltonians $\{H_{\alpha}\}$. 
To approach the problem from a fundamental perspective, we move beyond such examples and establish a necessary and sufficient condition for UQC under global control.

To this end, we first identify a natural framework for globally controlled analog quantum simulators and analyze the conditions for universality within it. 
Without loss of generality, we focus on a one-dimensional chain of qubits; extensions to higher-dimensional lattices are straightforward (see \Cref{appendix:Universality}). 
Each two-level system functions as a qubit provided it supports rotations about the \(X\)- and \(Z\)-axes via suitable control pulses. Consequently, a natural extension for a globally controlled system is to have global \(X\)- and \(Z\)-rotations, i.e. $\sum_{i}X_i$ and $\sum_{i}Z_i$. In addition, the system needs interactions and entanglements to be universal. Here, we choose the uniform single-Pauli type interaction, i.e. $\sum_{i}P_{i}\otimes P_{i+1}$, as a canonical representation of the interaction. Most current physical platforms, such as Rydberg atom arrays and trapped ions \cite{spinliquid, nist_ions}, meet this description. Without loss of generality, we choose the Ising-type $ZZ$ interaction. In Appendix~\ref{appendix:Universality}, we show the detail of the single-Pauli interaction is not important and also discuss some sufficient requirements for arbitrary homogeneous nearest-neighbor couplings.

In this canonical setting, the system we analyze is described by
\begin{equation}
H_q(t)
= u_X(t) H_X + u_Z(t) H_Z + u_{ZZ}(t) H_{ZZ}
= u_X(t)\sum_{i} X_i + u_Z(t)\sum_{i} Z_i + u_{ZZ}(t)\sum_{\langle i,j\rangle} Z_i Z_j\;.
\label{eq:uniform_global}
\end{equation}
As a one-dimensional open chain, this time-dependent Hamiltonian possesses only lattice reflection symmetry about its center. 
We identify this geometrical symmetry as the sole obstruction to universality. 
In particular, introducing \emph{any} additional control term $H_{\mathrm{break}}$ that breaks reflection symmetry renders the system universal, as formalized in \Cref{qubit:UQC_chain} and shown in \Cref{fig:fig1} (b).

\begin{theorem}[Minimal requirement for universal quantum computation on a qubit chain]
\label{qubit:UQC_chain}
Consider an open chain of qubits with a homogeneous nearest-neighbor single-Pauli interaction of the form $\sum_{i} P_i \otimes P_{i+1}$. 
Assume the system is equipped with tunable global $X$- and $Z$-fields,
$
H_{X} = \sum_j X_j, 
H_{Z} = \sum_j Z_j .
$
Then the system is capable of universal quantum computation if and only if there exists at least one additional control field $H_{\mathrm{break}}$ that breaks lattice reflection symmetry, i.e.
$
R H_{\mathrm{break}} R^{-1} \neq H_{\mathrm{break}},
$
where $R$ denotes the matrix representation of lattice reflection.
\end{theorem}

\Cref{qubit:UQC_chain} is general in various senses.
First, the theorem reveals a fundamental connection between geometric symmetry breaking and universal quantum computation. In particular, it does not depend on the specific form of $H_{\mathrm{break}}$, provided that $R H_{\mathrm{break}} R^{-1} \neq H_{\mathrm{break}}$.
The one-dimensional proof is generalizable to any higher-dimensional lattice, or even to a general graph (see discussions in \Cref{appendix:Universality}), and we leave the detailed proof for future study.
In addition, our theorem implies the results of previous studies as special cases.
For example, existing proofs of universality for globally controlled quantum dots~\cite{Benjamin_PRL} and QAOA~\cite{qaoa_universal} require qubits arranged in an even–odd alternating pattern of the form $ABAB\cdots$. In the language of dynamical Lie algebras, their control schemes are equivalent to the following set of Hamiltonians, which together generate universal quantum computation:
\eqs{
H_1 = \sum_j X_j,~ H_A=\sum_j Z_{2j},~H_B=\sum_j Z_{2j+1},~H_{AB}=\sum_{j}Z_{2j}Z_{2j+1},~H_{BA}=\sum_{j}Z_{2j+1}Z_{2j+2}.
\label{eq:Seth}
}
Compared to \Cref{eq:uniform_global}, we find the equivalence $H_X = H_1$, $H_Z = H_A + H_B$, and $H_{ZZ} = H_{AB}+H_{BA}$. Therefore, the control structure in \Cref{eq:Seth} is included in \Cref{qubit:UQC_chain}, with two specific additional symmetry-breaking terms $H_{\text{break1}} = \sum_j Z_{2j}$ and $H_{\text{break2}} = \sum_j Z_{2j} Z_{2j+1}$.
Notably, this structure turns out to be convenient in current dual-species neutral atom platforms \cite{dualspecies,dualspecies2,dualspecies3}, which further demonstrates the applicability of our theorem. 

We also extend the above framework to fermionic and bosonic systems, with an emphasis on ultracold atoms in optical superlattices \cite{superlattice_chemistry,PhysRevLett.134.053402,programming_optical_lattice}. In the Methods and \Cref{appendix:NNN_hopping}, we also show that such systems are surprisingly powerful, supporting universal fermionic and bosonic quantum simulation. One open question is the complexity of universal analog quantum systems and their associated compilation efficiency. While a comprehensive investigation lies beyond the scope of this work, we provide a detailed analysis of pulse optimization complexity in \Cref{app:compilation_efficiency}. Building on \Cref{prop:size_QOC_problem} in the Appendix, we show that the complexity of the quantum optimal control problem scales polynomially with the dimension $D_{\mathcal{U}^+}$ of the manifold of unitaries reachable in time polynomial to $D_{\mathcal{U}^+}$.

\subsection{Information scrambling and anti-concentration in globally driven analog quantum systems}
\begin{figure*}[htbp]
    \centering
    \includegraphics[width=0.99\linewidth]{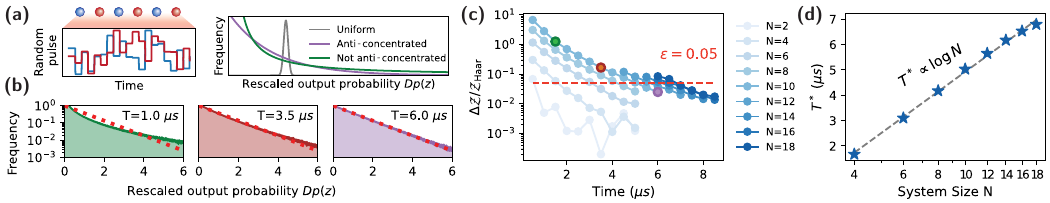}
    \caption{\textbf{Information scrambling in globally driven quantum systems.} (a) Dual-species neutral-atom array driven by spatially uniform and temporally random global pulses, and illustration of the rescaled output probability distribution $D p(z)$, where $D$ is the Hilbert space dimension. A uniform distribution yields a delta peak at $D p(z)=1$, while random Haar states follow the Porter–Thomas distribution, corresponding to an anti-concentrated output. In contrast, non-anti-concentrated distributions exhibit long tails, indicating enhanced probabilities for specific outcomes.
(b) Histograms of rescaled output probabilities for a dual-species atom array with interatomic spacing $d = 8.9~\mu\mathrm{m}$ under random global pulses. As the evolution time increases, the distribution transitions from non-anti-concentrated to anti-concentrated. See \Cref{app:scrambling} for details.
(c) Relative error of the averaged collision probability $\mathcal{Z}$ with respect to Haar-random states for different system sizes.
(d) Defining the characteristic anti-concentration time $T^*$ by a relative error threshold $\epsilon = 5\%$, we observe a logarithmic scaling $T^* \propto \log N$.}
    \label{fig:scrambling}
\end{figure*}

Randomness is a cornerstone of modern quantum science, underpinning our understanding of complexity in random quantum circuits (RUCs) \cite{google_supremacy,PhysRevLett.134.090601} and black hole dynamics \cite{Patrick_Hayden_2007,Haferkamp}, while enabling protocols for benchmarking \cite{Joonhee_23,PhysRevLett.131.110601,shadow,PhysRevResearch.5.023027,PhysRevResearch.6.043118,PhysRevResearch.4.013054} and cryptography \cite{2023arXiv230301625A,Liu_25,CHAMON2022}. While RUCs provide a paradigmatic model of scrambling via local gates, the emergence of randomness in globally driven Hamiltonian systems remains less understood. Although fully controllable systems eventually converge to Haar-random ensembles \cite{PhysRevX.14.041059}, the efficiency of this process under global control is an open question. Here, we demonstrate that a globally controlled system, inspired by dual-species Rydberg atom arrays, exhibits fast information scrambling in the measurement basis \footnote{We emphasize that the convergence to the PT distribution in the measurement basis does not indicate that it converges to a random state design. We leave it to future study.}. We show that, despite lacking spatial control compared to RUCs, this system approaches the Porter-Thomas (PT) distribution in $O(\log N)$ time, matching the scrambling rate of RUCs with spatio-temporal randomness \cite{Brandao2022}. We implement this architecture using a dual-species neutral atom array \cite{dualspecies,dualspecies2,dualspecies3}, arranged as an interwoven one-dimensional chain [\Cref{fig:scrambling} (a)]. The system features static van der Waals interactions driven by species-selective global Rabi controls $\Omega_{A/B}(t)$ and detuning $\Delta_{A/B}(t)$, which together enable universal quantum dynamics (see \Cref{app:scrambling} for details). Under global time-dependent random pulses, the system evolves from an initial state $\ket{0}^{\small \otimes N}$ to a complex state $\ket{\psi}$ at late times. If one measures all the qubits, instead of having a perfect uniform probability of measuring all bitstrings $z$, it exhibits a speckle pattern: over different $z$’s, $p(z)$ fluctuates about $1/D$ due to random interference in coherent quantum dynamics, with $D$ being the Hilbert space dimension (see \Cref{fig:scrambling} (a)). Its statistical properties are universal: the fraction of $p(z)$’s in a given interval $p(z)\in [x+dx]$ is given by the Porter-Thomas (PT) distribution, i.e. $P[p(z)=x]dx=\mu^{-1}\exp(-x/\mu)dx$ with a mean $\mu = 1/D$. This phenomenon is also called ``anti-concentration", which is the cornerstone for quantum supremacy benchmarking \cite{google_supremacy}.

As shown in \Cref{fig:scrambling} (a), we use random piecewise constant global pulses with a constant time resolution $\Delta t=0.05\mu\mathrm{s}$. Each constant value is uniformly sampled from an allowed range, with details of the parameters in \Cref{app:scrambling}. We showed the frequency of rescaled output probability $Dp(z)$ for different evolution times in \Cref{fig:scrambling} (b) for a ten-atom system. It is clear that when the quench time is short, such as $T=1.0\mu s$, the frequency deviates from the ideal PT distribution and has longer tails, which indicates that some bitstring outcomes are more likely to occur. However, this frequency quickly converges to the PT distribution at $T=6.0\mu s$, indicating the information scrambling. To better understand how quickly such a globally driven system converges to the PT distribution, we calculate the collision probability over the ensemble of random states generated $\mathcal{Z}_{\mathcal{E}}=\mathbb{E}_{\ket{\psi}\in \mathcal{E}}\left[\sum_{\textbf{z}}|\langle \textbf{z}|\psi\rangle|^{4}\right]$. For the Haar random state ensemble, the ideal value is $\mathcal{Z}_{\mathrm{Haar}}=2/(D+1)$. In \Cref{fig:scrambling} (c), we show the decay of relative error $\varepsilon_r=|\mathcal{Z}_{\mathcal{E}}-\mathcal{Z}_{\mathrm{Haar}}|/\mathcal{Z}_{\mathrm{Haar}}$. Suppose we choose a threshold of relative error $\varepsilon_r=5\%$ for the characteristic time $T^{*}$ to reach the PT distribution or anti-concentration, we clearly see a logarithmic scaling $T^{*}\propto \log N$ in \Cref{fig:scrambling} (d). This indicates fast scrambling in the measurement outcomes at the same rate as one-dimensional RUC, despite only temporal randomness. This result would further inspire new applications in randomized benchmarking and cryptography with globally driven analog quantum simulators. Lastly, in \Cref{app:additional_large_scale}, we further demonstrate that even in analog quantum simulators with always-on interactions and global control pulses, high-fidelity subregion entangling gates can be engineered despite the system’s natural tendency to generate global entanglement.

\section{Engineered three-body interaction and topological dynamics\label{sec:SPT}}

Building on our theoretical framework, we demonstrate that globally-controlled analog simulators can realize effective Hamiltonians far beyond their native physical interactions through pulse engineering [\Cref{fig:fig1}(b)]. In this section, we showcase this potential by experimentally engineering a three-body interacting Hamiltonian using Rydberg atom arrays, as illustrated in \Cref{fig:fig1}(c). To our knowledge, this is the first realization of effective three-body interactions with Rydberg atoms outside the blockade regime, going beyond the native two-body interactions. Crucially, we introduce a new control technique, dubbed \emph{direct quantum optimal control}, that bridges the gap between theoretical possibility and experimental reality. This method enables us to overcome key experimental challenges, including hardware constraints and atom position fluctuations outside the blockade regime, by identifying smooth, short-duration pulses that achieve high-fidelity dynamics. In the following, we introduce the physics model of interest, which has a topological phase, and describe its implementation using Rydberg atom arrays.

\subsection{Example application: symmetry-protected topological phases}

Topology has emerged as a cornerstone of modern quantum many-body physics. Topological phases are typically characterized by non-local order parameters and topological edge modes localized at the system boundaries \cite{RevModPhys.89.041004, RevModPhys.83.1057, ruben_spt}. In one-dimensional systems, such a topological phase is well defined in the presence of symmetries, giving rise to symmetry-protected topological (SPT) phases \cite{ruben_spt,PhysRevX.8.011055, GZC2D,disorder_spt}. 
With the development of programmable quantum devices, researchers have begun experimentally investigating exotic topological phases of matter, ranging from preparing topological quantum states \cite{semeghini_spinliquid, Hannes_state_prep} to exploring topological phases and transitions \cite{Antoine_ssh,spt_transition_ssh, YuJie_phase_transition}, Floquet topological dynamics \cite{Dongling_fspt, FSPT_IONS}, and anyon statistics \cite{ion_anyon,anyon2,liyuan}. Most of these demonstrations and proposals have relied on digital quantum circuits. 

\begin{figure*}[htbp]
    \centering
    \includegraphics[width=1\linewidth]{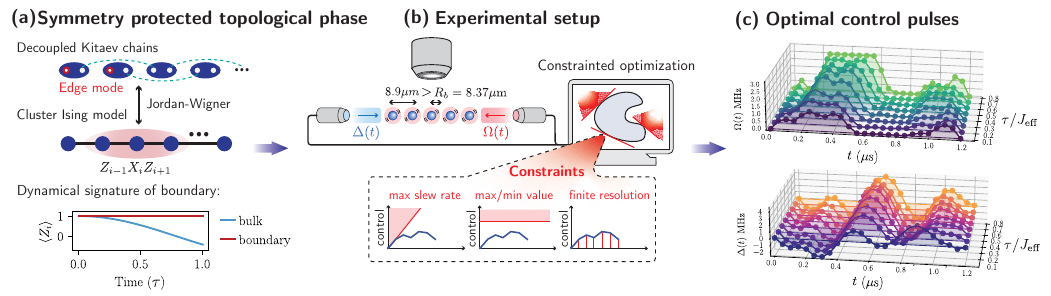}
    \caption{\textbf{Experimental realization of symmetry-protected topological dynamics using optimal control}.
(a) Symmetry-protected topological Hamiltonian. A pair of decoupled Kitaev chains, which host topological edge modes, can be mapped to a qubit model with three-body interactions known as the cluster-Ising model. Due to experimental constraints, only measurements in the Z-basis are available. A key dynamical signature distinguishing boundary from bulk qubits is that the Z expectation values of the boundary qubits remain unchanged under evolution with the ZXZ Hamiltonian.
(b) Schematic of the experimental setup. Rydberg atom arrays are globally driven by time-dependent Rabi frequency $\Omega(t)$ and detuning $\Delta(t)$. Atoms are spaced beyond the blockade radius $R_b$ \cite{blockade}, resulting in position fluctuations due to residual atomic temperature. Optimal quantum control is employed to design global pulses that mitigate errors and also satisfies machine constraints.
(c) Overview of optimal control pulses used in the experiments. Time-dependent control waveforms $\Omega(t)$ and $\Delta(t)$ are engineered to simulate the effective SPT Hamiltonian. Constraints such as maximum slew rate, amplitude bounds, and finite time resolution are incorporated in the control optimization.}
    \label{fig:fig2}
\end{figure*}

In parallel, there has been significant progress in analog quantum devices and digital-analog hybrid platforms. By avoiding the discretization errors introduced by Trotterization, analog and hybrid devices offer improved simulation quality, especially for quantum dynamics \cite{thermalization, stability_analog, practical_advantage}. Notably, Ref.~\cite{Antoine_ssh} demonstrated a bosonic version of the SPT model using Rydberg atom arrays. However, including this demonstration, most analog quantum simulations still relied heavily on the native Hamiltonians of the device and primarily focused on engineered two-body interactions.
Building on the theoretical framework we developed, we show that the expressive power of analog quantum devices is significantly greater. These systems can simulate broad classes of exotic Hamiltonians that go well beyond their native two-body interactions. For example, with global pulse engineering, one can, in principle, engineer multi-body interactions. In this work, we demonstrate the quantum dynamics driven by an effective ZXZ Hamiltonian (also called cluster-Ising model):
\eqs{
H_{\text{ZXZ}}=J_{\text{eff}}\sum_{j}Z_{j-1}X_{j}Z_{j+1}\;,\label{eq:zxz}
}
which also serves as a pedagogical example of a system realizing a symmetry-protected topological (SPT) phase. Through the Jordan-Wigner transformation, this model can be mapped to a pair of decoupled Kitaev chains \cite{KitaevChain, ruben_spt}, supporting a topological edge mode, as illustrated in \Cref{fig:fig2}(a) (see \Cref{app:spt} for details). Although both the cluster-Ising model and the bosonic Su–Schrieffer–Heeger model in Ref.\cite{Antoine_ssh} are SPT Hamiltonians, our work highlights the potential to engineer \emph{multi-body interactions} beyond the native interactions of the hardware.

In particular, we use the platform of Rydberg atoms trapped in a tweezer array, where a qubit is encoded as the ground and Rydberg states of the atom, which is commonly referred to as the analog mode of Rydberg atom arrays. In the analog mode, atoms are initially trapped in their ground states but become untrapped once global control pulses excite them to Rydberg states. The residual thermal motion causes fluctuations in atomic positions, which in turn lead to noise during quantum simulation \cite{Ceren}. As a result, most analog quantum simulations with Rydberg atom arrays have been performed within the blockade regime \cite{semeghini_spinliquid,quantum_scar}, where neighboring atoms are separated by less than the blockade radius $R_b = (C_6 / \Omega_{\mathrm{max}})^{1/6}$ \cite{blockade,Fast_Rydberg_gate_2000,quantum_ising_blockade}. In the blockade regime, the nearest-neighbor interactions are much stronger than the Rabi frequency. Therefore, simultaneous excitations of neighboring atoms are suppressed, and nearest-neighbor interactions are effectively eliminated. While this configuration offers greater robustness against atomic position fluctuations, entanglement has also been experimentally demonstrated in the weak coupling regime using carefully timed coherent control \cite{weakly_coupling}.

Using the theoretical tools we developed, we show that the three-body ZXZ Hamiltonian cannot be engineered within the typical blockade regime. Instead, it is only possible to realize the effective Hamiltonian $H_{ZXZ}$ outside it (see \Cref{appendix:limitation_blockade}).  In our experiment, the maximum Rabi frequency is $\Omega_{\mathrm{max}} = 2.4~\mathrm{MHz}$, yielding a blockade radius of $R_b = 8.37~\mu\mathrm{m}$. As illustrated in \Cref{fig:fig2}(b), we arrange atoms in a one-dimensional chain with a spacing of $d = 8.9~\mu\mathrm{m}$, just outside the blockade radius \cite{blockade}. The atoms are globally driven by a time-dependent Rabi frequency $\Omega(t)$ and detuning $\Delta(t)$, which serves the role of $u_{\alpha}(t)$ in \Cref{eq:time_dependent_Hamiltonian}. The goal is to design experimentally feasible control pulses $u_{\alpha}^{(\tau)}(t)$ such that
\eqs{
\mathcal{T}\left[e^{-i\int dt \sum_{\alpha}u_{\alpha}^{(\tau)}(t)H_{\alpha} }\right]=e^{-i \tau H_{\mathrm{ZXZ}}},
\label{eq:goal}
}
thereby realizing the target unitary evolution governed by the cluster–Ising Hamiltonian. Note that the control pulse $\mathbf{u}^{(\tau)}(t)$ depends on the effective simulation time $\tau$. Further details of the experimental setup are provided in the Methods.

Several key experimental challenges must be overcome to achieve this goal. As illustrated in \Cref{fig:fig2}(b), the laser control parameters $\Omega(t)$ and $\Delta(t)$ are constrained by maximum slew rates, bounded amplitudes, and finite time discretization. These limitations make the optimization problem non-convex and challenging for traditional quantum control methods. Moreover, since the atoms are not trapped during evolution (see \Cref{subsec:experimental_results}), long pulse durations inevitably lead to deviations from ideal dynamics. We emphasize that although there exists a circuit decomposition $e^{-i \tau H_\mathrm{ZXZ}} = U_\mathrm{ZZ}(\frac{\tau}{2}) e^{-i \tau H^{\mathrm{bulk}}_{X}} U_{ZZ}(\frac{-\tau}{2})$ with $U_\mathrm{ZZ}(\theta)=\prod_{i}^{N-1}\mathrm{exp}(-i\theta Z_i Z_{i+1})$ and $H^{\mathrm{bulk}}_{X} = \sum_{j=2}^{N-1} X_j$, this circuit decomposition is not directly accessible on the analog quantum hardware with always-on interactions, which is the focus of this work. These features highlight the potential challenges that the control problem faces.  Thus, it is essential to design smooth, short-duration control pulses. To overcome these obstacles, we introduce a new quantum control method adapted from a cutting-edge robotics technique called direct trajectory optimization and benchmark it against standard approaches. This method successfully identifies a family of smooth pulses, as shown in \Cref{fig:fig2}(c), that achieve the experimental target defined in \Cref{eq:goal}, where traditional methods fail.

\subsection{Direct Quantum Optimal Control}
\begin{figure}[htbp]
    \centering
    \includegraphics[width=0.99\linewidth]{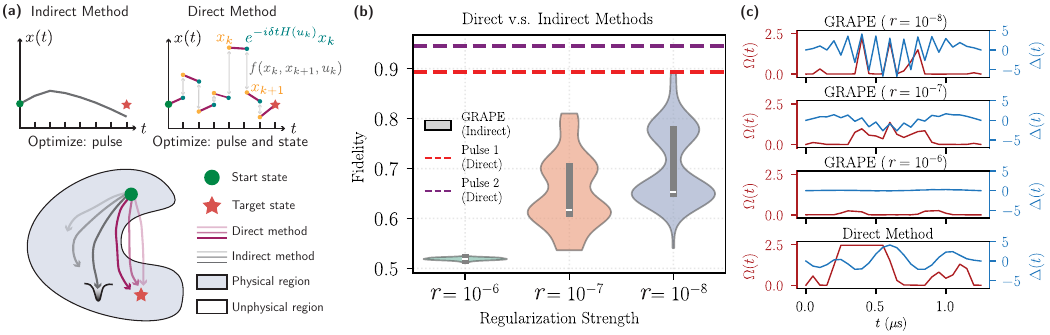}
    \caption{\textbf{Comparison between direct and indirect quantum optimal control methods.} (a) Schematic illustration of indirect v.s. direct quantum optimal control. The indirect method (left) optimizes only the control pulses $\mathbf{u}_k$, with intermediate states $x_k$ determined solely by the Schr\"odinger equation. In contrast, the direct method (right) simultaneously optimizes both the control pulses and intermediate states, treating the Schr\"odinger equation as a constraint — allowing traversal of unphysical regions during optimization (visualized as discrete jumps in state space). Therefore, they explore different loss landscape during optimization as shown in the lower panel. 
(b) Comparison of final unitary fidelities achieved by GRAPE (indirect) and the direct method, targeting a effective evolution time $\tau/J_{\text{eff}}=0.8$ under realistic machine constraints. GRAPE results are shown for three regularization strengths $r$ that penalize the second derivative to impose pulse smoothness. For each $r$, 100 GRAPE runs with random initializations are shown as violin plots; the thick bar marks the interquartile range and the white marker indicates the median. Two smooth pulse control obtained from the direct method (Pulse 1 with $T=1.2\mu\mathrm{s}$ and Pulse 2 with $T=3.6\mu \mathrm{s}$) are marked for comparison.
(c) Typical examples of optimized pulses for Rabi frequency $\Omega(t)$ (MHz) and detuning $\Delta(t)$ (MHz) using GRAPE and the direct method (Pulse 1).}
    \label{fig:fig3}
\end{figure}

To tackle these difficulties, we introduce a new quantum optimal control approach, dubbed \emph{direct quantum optimal control} (direct method hereafter), inspired by the direct trajectory optimization technique from robotics \cite{Zac_Direct_Method}. As we will demonstrate later in this section, this method explores a fundamentally different loss landscape compared to conventional quantum control techniques, making it particularly effective in the presence of stringent hardware constraints where other methods often fail.

In quantum optimal control (QOC), we start with the unitary propagator $U(t=0)=I_{2^N}$. It evolves according to the Sch\"odinger equation $\frac{\partial U}{\partial t}=-iH[\textbf{u}(t)]U,\label{eq:schodinger_dynamics}$
where $\textbf{u}(t)$ represents the control pulse trajectory. The goal of QOC is to find an optimal $\textbf{u}(t),~t\in[0,T]$ such that $U(T)$ can approximate a target unitary $U_{\mathrm{target}}$ while satisfying hardware constraints. We discretize the time interval $[0, T]$ into $n$ steps with duration $\delta t$ and denote the unitary propagator and control parameters at the $k$-th time step by $U_k$ and $\mathbf{u}_k$, respectively. Then, the QOC problem can be expressed as the following optimization problem: 
\begin{maxi!}
  {\{\mathbf{u}_k\}_{k=1}^n}
  {\mathcal{F}(U_n, U_{\mathrm{target}})}
  {\label{eq:rollout}}
  {}
  \addConstraint{U_n}{= \prod_{k=1}^{n} \exp\left[-i H(\mathbf{u}_k)\delta t\right]}
\end{maxi!}
where the objective function is the unitary fidelity, denoted $\mathcal{F}$, between $U_n$ and $U_{\text{target}}$. A widely used approach for solving this type of optimization problem is the GRadient Ascent Pulse Engineering (GRAPE) method \cite{grape}. For an initial guess of the control pulse $\mathbf{u}(t)$, one can evaluate the gradient $\partial \mathcal{F}(U_n,U_{\mathrm{target}})/\partial \textbf{u}_k$ with respect to pulse control parameters and then use the gradient ascent method to refine it iteratively. This workflow
is illustrated on the left side and lower panel of \Cref{fig:fig3}(a). In the optimal control literature, this strategy is commonly known as \emph{indirect shooting method} \cite{Betts_survey}.

While GRAPE performs well in relatively simple settings without control constraints, we find that it struggles when hardware constraints are introduced and is often prone to getting trapped in local minima. In contrast, aerospace and robotics researchers have developed alternative techniques known as \emph{direct methods} \cite{direct_collocation_intro,OC_book}, which reformulate the above optimization in a fundamentally different way: Instead of using the Schr\"odinger dynamics in \Cref{eq:schodinger_dynamics} to compute the unitary evolution $U_n$ across time steps $[t_1, t_2, \dots, t_n]$, the Schr\"odinger equation is treated as an additional constraint, and the intermediate unitary operators $U_k$ at each time step are augmented as trainable parameters. Specifically, the above optimization problem is rewritten as the following constrained nonlinear program:

\begin{maxi!}
  {\{\mathbf{u}_k, U_k\}_{k=1}^n}
  {\mathcal{F}(U_n, U_{\mathrm{target}})}
  {\label{eq:optimal_control}}
  {}
  \addConstraint{U_{k+1}}{= \exp\left[-i H(\mathbf{u}_k)\delta t\right] U_k,}{\quad \forall k \in \{1,\dots,n\}.}
\end{maxi!}

Although it may appear to be a simple reformulation, the direct method explores a fundamentally different loss landscape compared to the indirect approach. In practice, the Schrödinger equation is enforced via Lagrange multipliers. The optimization can begin with an infeasible initial guess, allowing violations of the Schrödinger dynamics that manifest as discontinuities in $U(t)$, as shown in the right panel of \Cref{fig:fig3}(a). This flexibility enables the optimizer to explore unphysical regions (e.g., unphysical jumps in $U(t)$), which can help in escaping local minima. This key distinction is further illustrated in the lower panel of \Cref{fig:fig3}(a). Another advantage of the direct method is that the framework naturally handles any constraints on its trainable parameters, which include both states and controls. This feature is particularly useful for hardware experiments, where laser controls are subject to numerous technical limitations and durations should be kept as short as possible to limit decoherence.

We compared the performance of the direct method and GRAPE in synthesizing the unitary evolution generated by the ZXZ Hamiltonian for an effective time of $\tau/J_{\mathrm{eff}} = 0.8$. 
Using the direct method, we explored two total pulse durations: $T = 1.2~\mu\mathrm{s}$ and $T = 3.6~\mu\mathrm{s}$ (we will justify these choices later). The shorter pulse, Pulse 1 (shown in \Cref{fig:fig3}(b)), achieves a unitary fidelity of 89.4\% on a three-atom chain, while the longer Pulse 2 yields a higher fidelity of 94.5\% \footnote{Here, we emphasize that unitary fidelity may be an overly stringent metric, as analog simulations typically focus on local observables and signals, which are often more robust and perform better than global measures like unitary fidelity \cite{stability_analog}.}. For GRAPE, we imposed the same hardware constraints using Lagrange multipliers and introduced an additional penalty term $r$ on the second derivative of the control pulses to enforce smoothness. For each value of $r = 10^{-6}, 10^{-7}, 10^{-8}$, we performed 100 random initializations of the controls. \Cref{fig:fig3}(b) shows the distribution of resulting unitary fidelities as violin plots, where the width reflects the density of solutions. Compared to GRAPE, the direct method achieves significantly higher fidelities, and GRAPE’s performance is highly sensitive to the choice of $r$. \Cref{fig:fig3}(c) shows typical control pulses found by both methods. For small $r$, GRAPE yields irregular control solutions, while slightly larger $r$ jeopardizes the optimization, resulting in nearly flat pulses. In contrast, the direct method finds smooth, high-fidelity pulses that meet all experimental constraints. Further details are provided in \Cref{app:direct_and_indirect_data}. In the next section, we implement the global control pulses (Pulse 1 and Pulse 2) in the experiment. As we will see, although Pulse 2 achieves higher unitary fidelity, its longer duration leads to increased decoherence, since the atoms are not trapped. In contrast, despite having lower unitary fidelity, the shorter Pulse 1 produces clearer topological edge signatures. This highlights that unitary fidelity can be an overly stringent metric, as analog simulations primarily focus on local observables, which tend to be more robust and better preserved than global quantities like unitary fidelity \cite{stability_analog}.

\section{Experimental Results \label{subsec:experimental_results}}

\begin{figure*}[htbp]
    \centering
    \includegraphics[width=1\linewidth]{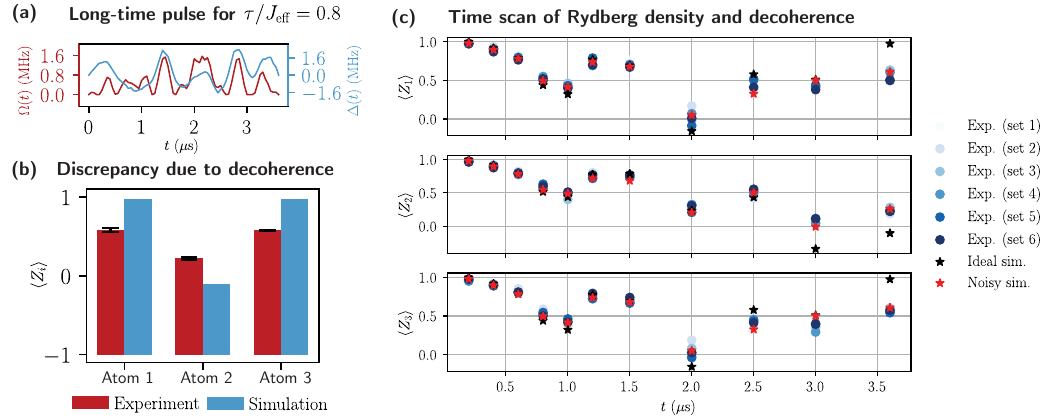}
    \caption{\textbf{Characterizing noise in experiment using time-resolved Rydberg density}.
(a) Control pulses ($\Omega(t)$ and $\Delta (t)$) obtained from the direct method (Pulse 2) targeting effective evolution time $\tau/J_{\text{eff}}=0.8$.
(b) $\langle Z_i \rangle = 1 - 2\langle n_i \rangle$ measured at the end of the pulse sequence for a three-atom chain initially prepared in the ground state. While the ideal unitary evolution predicts two peaks at atom 1 and atom 3, experimental results (red) deviate from the ideal simulation (blue), indicating the presence of noise.
(c) Time-resolved scan of ($\langle Z_i \rangle $) for all three atoms during pulse application. Six separated three-atom clusters were measured in experiments in parallel (light to dark blue). Black stars indicate noiseless simulation results, and red stars represent noisy simulations incorporating Rydberg decay (incoherent error) and drifts in $\Omega(t)$ and $\Delta(t)$ (coherent error). The apparent deviation for $t>1.5\mu \mathrm{s}$ from ideal dynamics highlights the impact of realistic experimental imperfections. Each experiment consists of 1000 measurement shots, and the size of the experimental error bars (standard deviation) is smaller than that of the markers.}
    \label{fig:fig4}
\end{figure*}

In the experiments, we apply global laser controls $\Omega(t)$ and $\Delta(t)$ to a chain of atoms separated by $d = 8.9~\mu\mathrm{m}$, slightly larger than the Rydberg blockade radius. The atoms are initially prepared in their ground states, trapped, and sorted into the desired configuration using optical tweezers. We use fluorescence imaging to preselect valid initial configurations and discard experimental runs in which atoms are not successfully imaged. In the current experimental setup \cite{aquila}, measurements are restricted to the Pauli-$Z$ basis of each atom, corresponding to imaging the Rydberg density. With this constraint, our key observable is the stability of the boundary operators $Z_L$ and $ Z_R$, which correspond to the two edge qubits. Theoretically, they satisfy the commutation relations:
\eqs{
[Z_\mathrm{L}, H_{\text{ZXZ}}] = [Z_\mathrm{R}, H_{\text{ZXZ}}] = 0,\\
[Z_\mathrm{L} Z_\mathrm{R}, H_{\text{ZXZ}}] = 0.
}
 Moreover, the persistence of stable correlations between the edge qubits, captured by the stability of $\langle Z_L Z_R\rangle$, serves as a dynamical signature of topological edge modes. We acknowledge this $Z$ basis measurement constraint is a limitation in the current platform, which can be mitigated by future hardware improvements. For example, the digital-analog Rydberg array platform \cite{analog_digital_lu,PhysRevA.107.042602} may enable measurement in arbitrary bases through single-qubit rotations prior to readout.

Although longer pulse durations offer greater control flexibility and can generally achieve higher fidelities, they are not ideal in real experiments due to decoherence from Rydberg decay. In the previous section, we found two global control pulses (Pulse 1 and Pulse 2) that  can synthesize the unitary evolution of the ZXZ Hamiltonian for an effective time of $\tau/J_{\mathrm{eff}}=0.8$. While Pulse 2 (shown in \Cref{fig:fig4}(a)) achieves higher unitary fidelity, its experimental performance deviates from the ideal pulse simulation. As illustrated in \Cref{fig:fig4}(b), ideal simulations on a three-atom chain predict high values of $\langle Z_i \rangle=1-2\langle n_i\rangle$ for the boundary atoms, whereas the experimental data show a clear decay from these ideal values. To better understand this discrepancy, we apply the global control pulse and measure $\langle Z_i\rangle$ of each atom at various intermediate time points. These results, shown as blue dots in \Cref{fig:fig4}(c), reveal increasing deviation from the ideal pulse simulation (black stars) as time progresses.

\begin{figure*}[htbp]
    \centering
    \includegraphics[width=1\linewidth]{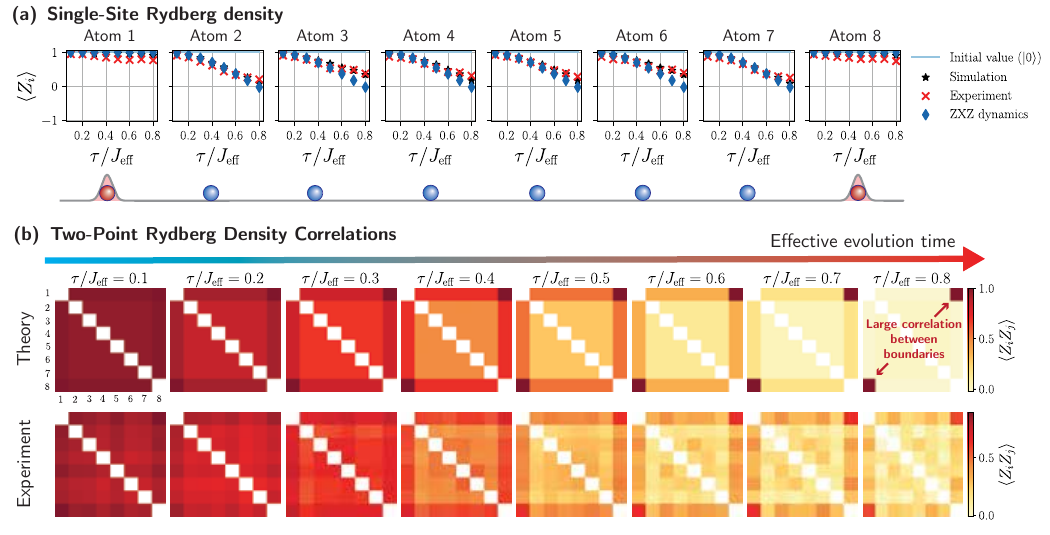}
    \caption{\textbf{Experimental signatures of topological dynamics.}
(a) Expectation values $\langle Z_i \rangle$ for each atom after applying the optimized global pulse [see \Cref{fig:fig2}(c)] to an eight-atom chain initialized in the ground state. Boundary atoms (1 and 8) retain large expectation values, while bulk atoms exhibit decay with increasing $\tau/J_{\mathrm{eff}}$, consistent with the edge-mode structure of the ZXZ Hamiltonian (blue diamonds). Experimental data (red crosses) closely agree with ideal simulations without noise (black stars), enabled by short-duration control pulses. Residual deviations at atoms 3 and 6 arise from control imperfections. See \Cref{app:additional_large_scale} for further discussion. (b) We measure connected two-point correlations $\langle Z_i Z_j\rangle$ between all pairs of atoms at different $\tau/J_{\mathrm{eff}}$. The top row shows theoretical predictions under ZXZ dynamics, which feature strong and persistent correlations between the boundary atoms. The bottom row presents experimental results, revealing similar edge correlations that reflect the presence of topological edge modes in the realized dynamics. Each experiment consists of 1000 measurement shots, and the size of the experimental error bars (standard deviation) is smaller than that of the markers.}
    \label{fig:fig5}
\end{figure*}

This discrepancy between the experimental data and the ideal simulation for the long-duration pulse can be well explained by the following Lindblad error model:
\eqs{
\partial \rho/\partial t=-\frac{i}{\hbar}[H(t),\rho]+\sum_{l}\gamma_l \left(\sigma_{l}^{-}\rho\sigma_{l}^{+}-\frac{1}{2}\left\{\sigma_{l}^{+}\sigma_{l}^{-},\rho\right\}\right)
}
where the Lindblad operators $\sigma_{l}^{-}$ capture Rydberg-state decay for each atom (incoherent error). To account for coherent errors, we incorporate calibration offsets in the Rabi frequency $\Omega(t)$ and detuning $\Delta(t)$ of the Hamiltonian $H(t)$ (see \Cref{eq:Rydberg_Hamiltonian}):
\eqs{
\Delta_{\mathrm{exp}}(t)=\Delta_{\mathrm{input}}(t)+\delta \Delta,\quad \Omega_{\mathrm{exp}}(t)=\Omega_{\mathrm{input}}(t)+(\delta \Omega+k\Omega_{\mathrm{input}}(t)),
}
where $\Delta_{\mathrm{exp}}(t)$ and $\Omega_{\mathrm{exp}}(t)$ are the experimentally realized values of detuning and Rabi frequency respectively; $\Delta_{\mathrm{input}}(t)$ and $\Omega_{\mathrm{input}}(t)$ are input control parameters; and $\delta \Delta$ and $\delta \Omega$ are constant shifts \footnote{We include a multiplicative shift term $k\Omega(t)$ in the model for $\Omega(t)$ because calibration experiments, including our own and those described in Ref.~\cite{Ceren}, indicate that the control error varies with the input values.}. Using this error model, we find the best fit with $\gamma_l=0.049$, $\delta \Delta = -0.049~\mathrm{MHz}$, $\delta \Omega = -0.032~\mathrm{MHz}$, and $k = -0.05$. The resulting noisy simulation (red stars in \Cref{fig:fig4}(c)) shows good agreement with the experimental data and is consistent with our calibration experiments.

While the coherent errors in the model can be compensated by adjusting the input parameters, the most effective strategy to mitigate incoherent errors without modifying the hardware is to employ short-duration control pulses. Indeed, \Cref{fig:fig4}(c) shows that the experimental results closely match the ideal simulations when the pulse duration is shorter than $1.5~\mu\mathrm{s}$. Based on this, we use the direct method to design a smooth control pulse with a total duration of $T = 1.2~\mu\mathrm{s}$ (Pulse 1) for $\tau/J_{\mathrm{eff}} = 0.8$. Starting from this pulse, we gradually construct smooth control pulses for $\tau/J_{\mathrm{eff}} = 0.7, 0.6, \dots, 0.1$. This family of global control pulses is visualized in \Cref{fig:fig2}(c). 

We apply this family of global laser controls to a Rydberg atom array consisting of eight atoms, initialized in their ground states and arranged into the desired spatial configuration. After the system evolves under the control pulses, we measure the Rydberg density $n_i$ (proportional to the Pauli-$Z$ operator) and two-point correlations between all pairs of atoms using fluorescence imaging. As shown in \Cref{fig:fig5}(a), the quantity $\langle Z_i \rangle$ exhibits behavior consistent with the theoretical predictions of topological SPT dynamics: the expectation values for the two boundary atoms remain stabilized, while those for the bulk atoms decay as the effective evolution time $\tau/J_{\mathrm{eff}}$ increases. As for correlations, since $[Z_\mathrm{L} Z_\mathrm{R},H_{\mathrm{ZXZ}}]=0$, the correlation between the two boundary atoms is expected to be stabilized. Since the system starts in a product state of ground states, strong correlations between the boundaries are expected to emerge and persist during the evolution under the ZXZ Hamiltonian. Our experimental results confirm this behavior, as shown in \Cref{fig:fig5}(b), where the boundary correlations follow the signature given by theoretical predictions. 

The deviations between the experimental results and ideal ZXZ dynamics are primarily attributed to the constraints of short-duration control pulses, a limitation imposed by current hardware coherence times. In \Cref{app:additional_large_scale}, we provide supplementary numerical comparisons between ideal ZXZ dynamics, short-time experimental pulses (both with and without Lindbladian noise), and a high-fidelity long-duration pulse. These results validate our interpretation. Furthermore, to demonstrate the scalability of our approach and the ``train-on-small, deploy-on-large'' paradigm, we numerically evaluate the experimental pulse performance on a 50-atom chain. While the finite trapping window of our current hardware precludes experimental realization at this scale, our simulations (detailed in \Cref{app:additional_large_scale}) reveal that the 50-atom system exhibits the same characteristic dynamical signatures observed in the 8-atom experiments. This confirms that the globally optimized control fields remain effective when extrapolated to the many-body regime.

We acknowledge that, due to current hardware limitations, measurements are restricted to the Pauli-$Z$ basis. Nonetheless, the above results provide strong evidence for the realization of the ZXZ Hamiltonian. In \Cref{app:additional_experimental}, we present additional experimental data where an initial global control pulse is applied to drive the system from the ground state to a different initial state, followed by our main control pulse. These results also show good agreement with theoretical predictions. However, a drawback of this approach is that prepending an extra pulse increases the total experimental time and inevitably introduces additional noise due to atomic motion, which is not ideal. In the future, we look forward to the expansion of hardware capabilities to enable measurements in arbitrary bases, allowing for full ansatz-free Hamiltonian learning \cite{HL_Gu,robust_learning,rss,ansatz_free_learning} and direct verification \cite{higher_correlations} of the effective Hamiltonian. 
Despite current limitations, our experiment represents a significant first step toward synthesizing \emph{three-body} interactions beyond the native hardware Hamiltonian on analog quantum simulators using quantum optimal control. We anticipate that this work will inspire a wide range of future applications.

\section{Discussion and Outlook}

Our work provides both a rigorous theoretical foundation and an experimental pathway for universal dynamics in globally controlled analog quantum simulators. By proving a necessary and sufficient condition for universality and introducing direct quantum optimal control, we demonstrate the engineering of effective multi-body Hamiltonians beyond native interactions, including topological dynamics in Rydberg arrays. This synergy highlights the expressive power of analog platforms under global control and establishes a versatile framework for quantum simulation. In addition, our matrix representation tools offer a systematic approach to classifying dynamical Lie algebras \cite{classification_lie}, while the numerical control techniques developed here have broad applicability across diverse experimental platforms \cite{classification_lie,global_superconducting_architecture,belt_superconducting,tunable_superconducting}.

To fully realize the potential of programmable quantum simulators, we identify three key directions for future exploration: (1) efficient pulse-level calibration and noise learning, (2) large-scale quantum optimal control, and (3) robust analog quantum simulation. In experimental settings, control noise and pulse imperfections are unavoidable. While our current method accounts for known error models during optimization, unknown or drifting errors may still degrade performance. It is therefore important to develop tools that use our approach as a warm start for pulse design, while incorporating real-time experimental feedback to calibrate and refine control sequences on the fly. Such an online optimization strategy could benefit from randomized measurement protocols \cite{classical_shadow,randomized_measurement,rss}, which offer scalable quantum-classical data interfaces. These measurement results can also be used to infer effective Hamiltonians, enabling both improved pulse calibration and direct verification of synthesized interactions \cite{ansatz_free_learning,robust_learning,HL_Gu}.

A second direction is the extension of our approach to large-scale quantum systems beyond classical simulation limits. For example, it would be valuable to engineer exotic Hamiltonians with multi-body interactions in two-dimensional systems. Several strategies could support this goal. Recent developments in classical simulation tools, such as Pauli propagation \cite{pauli_prop,j1gg-s6zb,PhysRevLett.133.120603} and tensor network methods \cite{PhysRevResearch.6.013326,PRXQuantum.5.010308,bp_TN,2025arXiv250711424R}, enable efficient estimation of local observables in noisy, high-dimensional systems. Incorporating these tools into quantum optimal control frameworks could produce high-quality control pulses offline. Then, one can further refine them with online optimal control. Alternatively, machine learning models could be trained to learn the smooth extrapolation of optimal global pulses found on different small system sizes to larger system sizes \cite{PhysRevX.12.011059}. Another promising approach is to adopt a divide-and-conquer strategy, in which control protocols optimized for small subsystems are combined to enable large-scale simulation \cite{10.1145/3618260.3649722}.

Finally, although analog quantum simulation lacks the full protection of fault-tolerant architectures, it demonstrates inherent robustness for local observables \cite{stability_analog}. Understanding the sensitivity of the new effective analog simulation scheme to noise remains an important direction. A natural question is whether ideas from quantum error correction can be incorporated into analog protocols to further enhance their resilience. Recent advances have shown that logical subspaces engineered via error detecting codes can suppress coherent errors by introducing a spectral gap between the analog subspace and higher-lying states \cite{2024arXiv241207764C}. As the gap increases, the system becomes more robust to coherent control errors, which resembles the threshold behavior of quantum error correction. Extending such ideas to pulse-level control may offer a promising path toward scalable and noise-resilient analog quantum technologies.

From a quantum optimal control perspective, our approach to effective analog quantum simulation via variational global pulse control also raises fundamental theoretical questions. One key question is how to define and quantify the “distance” between a machine-native Hamiltonian and a desired target Hamiltonian \cite{Geometric_speed,QSL,geometric_lower_bound}. While different native Hamiltonians may be capable of realizing the same target dynamics, the difficulty or overhead required can vary significantly. This notion of distance could help classify control Hamiltonians into equivalence classes under polynomial overhead. It also prompts the broader question of whether all universal Hamiltonians are equivalent in this sense. We believe these new results will stimulate further discussions in quantum many-body physics, quantum optimal control, and quantum simulation.

\section{Methods}
\subsection{Proof ideas for the main theorem}

In this section, we outline the tools, main ideas and essential lemmas for proving the main theorems with details can be found in \Cref{appendix:Universality}. To establish this result, we employ the framework of the dynamical Lie algebra (DLA), a tool originating in quantum control theory \cite{dAlessandro2021,2008arXiv0803.1193D}, and more recently applied to variational quantum algorithms \cite{Lie_barrenplateau,classification_lie,Shengyu_DLA,analogQML} as well as matrix representation theory \cite{KHANEJA200111}. The dynamical Lie algebra (DLA) is formally defined in \Cref{def:DLA_method}.

\begin{definition}[Dynamical Lie Algebra] \label{def:DLA_method}
   Given a control system with generators $\mathcal{G} = \{iH_1,iH_2,\dots,iH_l\}$, the
Dynamical Lie Algebra (DLA) $\mathfrak{g}$ is the subalgebra of
$\mathfrak{su}(d)$ spanned by the repeated nested commutators of
the elements in $\mathcal{G}$, i.e.
$$\mathfrak{g}=\mathrm{span}_{\mathbb{R}}\langle iH_1,iH_2,\dots,iH_l\rangle_{\mathrm{Lie}}\subseteq \mathfrak{su}(d)\;,$$
where $\mathrm{span}_{\mathbb{R}}\langle iH_1,iH_2,\dots,iH_l\rangle_{\mathrm{Lie}}$ denotes the Lie closure under nested commutators, and $d$ is the dimension of the Hilbert space.
\end{definition}

Intuitively, given a set of control Hamiltonians $\{H_1,H_2,\dots,H_{l}\}$, one can construct the linear space $\mathcal{V} = \mathrm{span}_{\mathbb{R}}\langle \cup_{j=1}^l i H_j\rangle_{\mathrm{Lie}}$, defined as the span of all possible nested commutators,
$[\cdots[[iH_i, iH_j], \dots], iH_k]$, over the real numbers. It can be shown that for any anti-Hermitian operator $L \in \mathcal{V}$, the corresponding unitary $e^{-L}$ can be approximated to arbitrary precision $\epsilon$ by designing an appropriate control sequence $\mathbf{u}(t)$ (see \Cref{lemma:Linear_Combo_Generators,lemma:Commutator_Generators} in \Cref{appendix:Notations_Preliminaries}). Therefore, a dynamical system governed by \Cref{eq:time_dependent_Hamiltonian} is said to be universal if $\mathrm{span}_{\mathbb{R}}\langle \cup_{j=1}^l i H_j\rangle_{\mathrm{Lie}} = \mathrm{span}_{\mathbb{R}}\langle i\mathbb{P}_N\backslash\mathbb{I}\rangle_{\mathrm{Lie}}$, where $\mathbb{P}_N\backslash \mathbb{I}$ denotes the set of all $N$-qubit Pauli operators $\mathbb{P}_N$ excluding the identity operator $\mathbb{I}$. This tool enables us to study the expressivity of dynamical systems without considering the details of $\mathbf{u}(t)$ and instead focusing on the important control Hamiltonians $H_\alpha$.

A well-known universal gate set is generated by single-qubit $X$- and $Z$-rotations plus arbitrary nearest-neighboring-qubit interactions~\cite{Barenco_1995_UQC}.
Formally speaking, the generators $X_j,Z_j$ and $Z_jZ_{j+1}$ for all site $j$ span the whole $\mathfrak{su}(2^N)$ Lie algebra, i.e., $\mathrm{span}_{\mathbb{R}}\langle \cup_j iX_j, \cup_j iZ_j,\cup_j iZ_jZ_{j+1}\rangle_{\mathrm{Lie}} = \mathrm{span}_{\mathbb{R}}\langle i\mathbb{P}_N\backslash\mathbb{I}\rangle_{\mathrm{Lie}}$.
Surprisingly, we found that the uniform controls, as in \Cref{eq:uniform_global}, allow for almost individual controllability of local fields and interactions, except for a remaining reflection symmetry that couples qubits at mirrored positions (e.g., the first and last qubits).
Given the reflection operation $\mathcal{R}$ yielding $\mathcal{R}(X_j/Z_j) = X_{(N-j)+1}/Z_{(N-j)+1}$, the uniform control Hamiltonians generate the following DLA:
\eqs{ \label{eqn:universal_symm_algebra}
\mathrm{span}_{\mathbb{R}}\langle iH_{X},iH_{Z},iH_{ZZ}\rangle_{\mathrm{Lie}}
=\mathrm{span}_{\mathbb{R}}\langle \cup_{j} i\left( X_j+\mathcal{R}(X_j)\right),\cup_{j} i\left( Z_j + \mathcal{R}(Z_j)\right),
&\cup_j i\left(Z_{j}Z_{j+1} + \mathcal{R}(Z_{j})\mathcal{R}(Z_{j+1}) \right)\rangle_{\mathrm{Lie}}\;,
}
where the reflection-symmetric controls are illustrated in \Cref{fig:fig1} (a) without the green control fields. 
Next, we introduce an additional control term $H_{\mathrm{break}}$ that breaks the lattice reflection symmetry, i.e. $RH_{\mathrm{break}}R^{-1}\neq H_{\mathrm{break}}$ where $R$ denotes the matrix representation of lattice reflection. The specific form of $H_{\text{break}}$ is not important—it could represent interactions, local rotations, or other types of controls. Then we use matrix representation theory to show $H_{\text{break}}$ makes the whole system universal. First, we show the following observation (\Cref{prop:fix-span})
\eqs{
\mathfrak{l}\coloneqq \{X\in\mathfrak{su}(2^N):[X,R]=0\}=\mathrm{span}_{\mathbb{R}}\langle iH_{X},iH_{Z},iH_{ZZ}\rangle_{\mathrm{Lie}}.
}
Decompose \(\mathcal H=\mathcal H_+\oplus_v\mathcal H_-\) into \(R\)–eigenspaces with dimensions \(d_\pm\).
An operator commutes with \(R\) iff it is block–diagonal in this basis (\Cref{lem:block-diag}).
Imposing skew–Hermiticity and tracelessness yields the decomposition of the representation matrix.
\eqs{
\mathfrak l\cong\mathfrak{su}(d_+)\oplus_m \mathfrak{su}(d_-) \oplus_m \mathfrak u(1)_{\mathrm{rel}},
}
with the relative phase $\mathfrak{u}(1)_{\mathrm{rel}}$ central in \(\mathfrak l\) (\Cref{prop:l-structure}). In the above, we use $\oplus_v$ to denote the direct sum of the vector space (Hilbert space) and $\oplus_m$ for the direct sum of the matrix space. For any $X\in \mathfrak{su}(2^N)$, its matrix representation in the $R$-diagonal basis can be written as $X=X_{+}+X_{-}$, where $RX_{+}R^{-1}=X_{+}$ and $RX_{-}R^{-1}=X_{-}$. And all $X_-$ span the space $\mathfrak{m}\coloneqq \{X: RXR^{-1}=-X\}$, which indicates $\mathfrak{su}(2^N)=\mathfrak{l}\oplus_{m}\mathfrak{m}$ (\Cref{lem:theta-proj}). Therefore, $H_{\text{break}}=(H_{\text{break}})_{+}+(H_{\text{break}})_{-}$, with $(H_{\text{break}})_{+}$ is block-diagonal and $(H_{\text{break}})_{-}$ is block-off-diagonal. It is crucial that with uniform global control, we have universality in $\mathfrak{l}$. Then in \Cref{prop:irreducible-m}, we show for any nonzero $M\in \mathfrak{m}$, 
\eqs{
\mathrm{span}_{\mathbb R}\Big\langle\,\ad_{K_t}\cdots\ad_{K_1}(M):\ t\ge0,\ K_s\in\mathfrak l\,\Big\rangle=\mathfrak m,
}
where $\ad_{K_i}(M)=[K_i,M]$. It means that there exists one way to do commutation on $M$, by elements in $\mathfrak{l}$, and linearly combine the resulting operators, to generate the full $\mathfrak{m}$. Then for any $H_{\text{break}}$ that satisfies $RH_{\text{break}}R^{-1}\neq H_{\text{break}}$, we can show $(H_{\text{break}})_{-}\in\mathfrak{m}\neq 0$. With the above tools, we finally show the universality in the reflection-symmetric subalgebra (\Cref{eqn:universal_symm_algebra}) enables the addition of \textit{any} $H_{\mathrm{break}}$ to generate the full universal gate set as:
\eqs{
\mathrm{span}_{\mathbb{R}}\langle iH_{X},iH_{Z},iH_{ZZ},iH_{\mathrm{break}}\rangle_{\mathrm{Lie}}
=\mathfrak{su}(2^N)\;,
}
which demonstrates the sufficiency part in \Cref{qubit:UQC_chain}. 
The necessity part is straightforward, as the reflection symmetry will decompose the Hilbert space into different sectors, which prohibits universal quantum computation.

\begin{remark} \label{remark:generalization_thm_1}
Although Theorem~\ref{qubit:UQC_chain} is stated for homogeneous nearest-neighbor single-Pauli interactions of the form $\sum_{i} P_i \otimes P_{i+1}$, the result extends to arbitrary nearest-neighbor interactions. 
In particular, universality holds provided the interaction coefficients do not lie in degenerate regimes that induce additional symmetries, as formalized in \Cref{corollary:generalization_interaction,corollary:generalization_interaction_2} in Appendix~\ref{appendix:Universality}.
\end{remark}

\subsection{Universality for fermionic and bosonic systems}

For fermionic and bosonic simulations, a set of physical controls is termed universal if it can generate the entire unitary group acting on the corresponding Hilbert space. 
Experimentally, all relevant physical operations preserve particle number, so we can restrict our analysis to fixed-particle-number Hilbert spaces of fermions and bosons, denoted by $\mathcal{H}_f$ and $\mathcal{H}_b$, respectively.
By Theorems 2 and 3 of Ref.~\cite{Oszmaniec_2017}, free fermion/boson operations plus a uniform Hubbard-type interaction generate all unitaries in $\mathcal{H}_f/\mathcal{H}_b$, so it suffices to consider the realization of those controls.

For simplicity, we focus on a one-dimensional optical superlattice with spinless fermions or bosons, and the generalization to higher-dimensional superlattices and spinful particles is straightforward (see discussions in \Cref{appendix:Universal_fermion}).
We use $c_i^\dagger$ and $c_i$ to denote the fermionic/bosonic creation and annihilation operators at site $i$, respectively, and the corresponding particle number operator is denoted by $n_i = c_i^\dagger c_i$.
The uniform Hubbard interaction is given by $H_U = \sum_j n_j n_{j+1}$.
Moreover, the double periodicity in the optical trapping potential enables us to control the hopping and chemical potential terms globally in an even-odd alternating pattern.
This yields the following basic control Hamiltonians: $H_{\mathrm{even}}^{(\mathrm{hop})}=\sum_{j}\left(c_{2j}^{\dagger}c_{2j+1}+\mathrm{h.c.}\right)$, $H_{\mathrm{odd}}^{(\mathrm{hop})}=\sum_{j}\left(c_{2j+1}^{\dagger}c_{2j+2}+\mathrm{h.c.}\right)$, $H_{\mathrm{even}}^{(\mu)}=\sum_{j}n_{2j}$, and $H_{\mathrm{odd}}^{(\mu)}=\sum_{j}n_{2j+1}$.
Therefore, the system is described by:
\eqs{
H(t)=-J_{\mathrm{even}}(t)H_{\mathrm{even}}^{\mathrm{(hop)}}-J_{\mathrm{odd}}(t)H_{\mathrm{odd}}^{(\mathrm{hop})}+\mu_{\mathrm{even}}(t)H_{\mathrm{even}}^{(\mu)}+\mu_{\mathrm{odd}}(t)H_{\mathrm{odd}}^{(\mu)}+U H_{U}\;,\label{eq:superlattice}
}
where $J_{\mathrm{even/odd}}(t)$ and $\mu_{\mathrm{even/odd}}(t)$ are the corresponding time-dependent control pulses.

As discussed earlier, to prove universality, provided the uniform Hubbard interaction $H_U$, it suffices to demonstrate that the other control Hamiltonians in \Cref{eq:superlattice} can generate all free fermion/boson operations.
To achieve that, as detailed in \Cref{appendix:Universal_fermion}, we first show that they realize the independent control over nearest-neighbor hopping and on-site chemical potential terms, as given by: 
\eqs{
\mathrm{span}_{\mathbb{R}}\langle iH_{\mathrm{even}}^{(\mathrm{hop})},iH_{\mathrm{odd}}^{(\mathrm{hop})},iH_{\mathrm{even}}^{(\mu)},iH_{\mathrm{odd}}^{(\mu)}\rangle_{\mathrm{Lie}}=
\mathrm{span}_{\mathbb{R}}\langle \cup_j i(c_j^\dagger c_{j+1} + \mathrm{h.c.}), \cup_j i n_j
\rangle_{\mathrm{Lie}}\;.
\label{eq:free_fermion_equivalence}
}
Then, using \Cref{lemma:fermion_H_free_generators} in \Cref{appendix:Universal_fermion}, we equate the DLA on the right-hand side of \Cref{eq:free_fermion_equivalence} to that of free operations, which completes the proof (see \Cref{thm:universal_1d_fermion,thm:universal_1d_boson}). 

For spinful particles, the spin degree of freedom can be treated as two internal modes. 
Using a similar method (see~\Cref{thm:universal_1d_spinful_fermion}), we show that two additional control fields are required to achieve universality: (1) a uniform global spin-$X$ magnetic field and (2) a global spin-$Z$ magnetic field with a linear gradient. 
These extend the previously established universal control set for two-site systems to the full optical superlattice \cite{dwave}. We anticipate that our framework sets the stage for ultracold atom platforms to probe richer quantum phenomena in both fundamental and applied contexts.


\subsection{Experimental setup and Rydberg atom arrays}

The physical platform we use is a neutral atom array trapped in optical tweezers designed by QuEra Computing \cite{aquila}. Qubits are encoded in the electronic ground state $|g\rangle=|5S_{1/2}\rangle$ as $|0\rangle$ and in the Rydberg excited state $|r\rangle=|70S_{1/2}\rangle$ as $|1\rangle$ of the $^{87}\text{Rb}$ atom. Transitions between these states are driven by a two-photon process using laser beams at 420 nm and 1013 nm, while the atoms interact via van der Waals interactions between Rydberg states. The full control Hamiltonian is given by
\eqs{ \label{eq:Rydberg_Hamiltonian}
H(t)/\hbar = &\frac{\Omega(t)}{2} \sum_{l} \left( |g_l\rangle \langle r_l| + |r_l\rangle \langle g_l| \right) - \Delta(t) \sum_{l} |r_l\rangle \langle r_l| \\
&+ \sum_{j<l} V_{jl} |r_j\rangle \langle r_j| \otimes |r_l\rangle \langle r_l| ,
}
where $\Omega(t)$ is the global Rabi frequency, $\Delta(t)$ is the global detuning, $V_{jl} = \mathrm{C}_6 / |\vec{x}_j - \vec{x}_l|^6$ describes the van der Waals interaction between atoms $j$ and $l$, with $\mathrm{C}_6 = 862{,}690 \times 2\pi$ MHz·$\mu$m$^6$ and $\hbar$ is the reduced Planck's constant. With the qubit encoding described above, this Hamiltonian closely resembles a transverse-field Ising model with long-range, algebraically decaying interactions.

This qubit encoding is commonly referred to as the analog mode of Rydberg atom arrays, in contrast to the digital mode, which encodes qubits in two hyperfine ground states \cite{995,rydberg_chemistry,2025arXiv250118554E}. In the analog mode, atoms are initially trapped in their ground states but become untrapped once global control pulses excite them to Rydberg states. The residual thermal motion causes fluctuations in atomic positions, which in turn lead to decoherence through Rydberg state decay. As a result, most analog quantum simulations with Rydberg atom arrays have been performed within the blockade regime \cite{semeghini_spinliquid,quantum_scar}, where neighboring atoms are separated by less than the blockade radius $R_b = (C_6 / \Omega_{\mathrm{max}})^{1/6}$ \cite{blockade,Fast_Rydberg_gate_2000,quantum_ising_blockade}. In the blockade regime, the nearest-neighbor interactions are much stronger than the Rabi frequency. Therefore, simultaneous excitations of neighboring atoms are suppressed and nearest-neighbor interactions are effectively eliminated. This configuration is more robust to atom position fluctuations.

However, using the theoretical tools we developed, we show that although the three-body $ZXZ$ Hamiltonian cannot be engineered within the typical blockade regime, $H_{ZXZ}$ can be realized outside it (see \Cref{appendix:limitation_blockade}).  In our experiment, the maximum Rabi frequency is $\Omega_{\mathrm{max}} = 2.4~\mathrm{MHz}$, yielding a blockade radius of $R_b = 8.37~\mu\mathrm{m}$. We arrange atoms in a one-dimensional chain with a spacing of $d = 8.9~\mu\mathrm{m}$, just outside the blockade radius. 

Achieving this goal requires overcoming several key experimental challenges. The first is the decoherence error from the finite residual temperature of the atoms. As we will discuss in later sections, this noise source can be effectively mitigated by designing short-duration control pulses, without modifying the hardware. In addition to thermal motion, we must also respect strict constraints on the laser control parameters, as detailed in \Cref{app:additional_experimental}. These combined challenges call for the design of global control pulses that satisfy all experimental constraints, are smooth, and remain short in duration to minimize decoherence. Together, these factors highlight the experimental complexity and the need for fine-tuned quantum optimal control.

\section{Acknowledgements}
We thank Ceren B. Dag for the inspiring discussions on noise modeling of the Rydberg atom arrays, benchmarking and experimental setups. We also thank Hsin-Yuan (Robert) Huang, Nik O. Gjonbalaj, Henning Schloemer, Muqing Xu, Majd Hamdan, Francisco Machado, Shengtao Wang, Jun Yang, Zhengwei Liu, and Richard Allen for many helpful discussions. We are also grateful to Pengfei Zhang for pointing out an important caveat in the theorems after the first version of the preprint, and to Soonwon Choi for inspiring discussions that helped strengthen the main theorem. S.F.Y. and H.Y.H. thank DARPA through their IMPAQT Program (HR0011-23-3-0023) and DOE through the QUACQ program (DE-SC0025572). A.M.G. acknowledges support from the NSF through the Graduate Research Fellowships Program (DGE 2140743) and from the Theodore H. Ashford Fellowships in the Sciences. 
L.C. and A.J. were supported in part by the Army Research Office MURI Grant W911NF-20-1-0082.
A.J.G. was supported by an appointment to the Intelligence Community Postdoctoral Research Fellowship Program at University of Chicago administered by Oak Ridge Institute for Science and Education (ORISE) through an interagency agreement between the U.S. Department of Energy and the Office of the Director of National Intelligence (ODNI). F.T.C. is funded in part by the STAQ project under award NSF Phy-232580 and in part by the US Department of Energy Office of Advanced Scientific Computing Research, Accelerated Research for Quantum Computing Program. 
A.T. and Z.M. are supported, in part, by the National Science Foundation (NSF), and this material is based upon work conducted within the Center for Quantum Computing and Information Technologies (QCiT), which is an industry-university collaborative research center at Carnegie Mellon University under
NSF Award No. 2310949.
A.T. is the CEO and founder of Harmoniqs, Inc., which provided advice and support for the control design part of this work. F.T.C. is the Chief Scientist for Quantum Software at Infleqtion and an advisor to Quantum Circuits, Inc.

Author contributions: H.Y.H. and S.F.Y. conceived the idea and designed the project. H.Y.H. and L.C. developed the theory, and L.C. proved the main theorems. H.Y.H. proposed the information scrambling application and conducted the analysis. H.Y.H. and A.M.G. set up and carried out the experiments, and A.M.G. processed the experimental data. A.T. and A.J.G. developed the pulse optimization tools. Z.M., F.T.C., A.J., and S.F.Y. supervised the project. All authors contributed equally to the writing of the manuscript.

\section{Data Availability}
Source data are available for this paper. All other data supporting the plots within this paper and other study findings are available from the corresponding author upon reasonable request.

\section{Code Availability}
The code used in this study is available from the corresponding author upon request.

\bibliographystyle{apsrev4-2} 
\bibliography{refs}

\begin{thebibliography}{174}%
\makeatletter
\providecommand \@ifxundefined [1]{%
 \@ifx{#1\undefined}
}%
\providecommand \@ifnum [1]{%
 \ifnum #1\expandafter \@firstoftwo
 \else \expandafter \@secondoftwo
 \fi
}%
\providecommand \@ifx [1]{%
 \ifx #1\expandafter \@firstoftwo
 \else \expandafter \@secondoftwo
 \fi
}%
\providecommand \natexlab [1]{#1}%
\providecommand \enquote  [1]{``#1''}%
\providecommand \bibnamefont  [1]{#1}%
\providecommand \bibfnamefont [1]{#1}%
\providecommand \citenamefont [1]{#1}%
\providecommand \href@noop [0]{\@secondoftwo}%
\providecommand \href [0]{\begingroup \@sanitize@url \@href}%
\providecommand \@href[1]{\@@startlink{#1}\@@href}%
\providecommand \@@href[1]{\endgroup#1\@@endlink}%
\providecommand \@sanitize@url [0]{\catcode `\\12\catcode `\$12\catcode `\&12\catcode `\#12\catcode `\^12\catcode `\_12\catcode `\%12\relax}%
\providecommand \@@startlink[1]{}%
\providecommand \@@endlink[0]{}%
\providecommand \url  [0]{\begingroup\@sanitize@url \@url }%
\providecommand \@url [1]{\endgroup\@href {#1}{\urlprefix }}%
\providecommand \urlprefix  [0]{URL }%
\providecommand \Eprint [0]{\href }%
\providecommand \doibase [0]{https://doi.org/}%
\providecommand \selectlanguage [0]{\@gobble}%
\providecommand \bibinfo  [0]{\@secondoftwo}%
\providecommand \bibfield  [0]{\@secondoftwo}%
\providecommand \translation [1]{[#1]}%
\providecommand \BibitemOpen [0]{}%
\providecommand \bibitemStop [0]{}%
\providecommand \bibitemNoStop [0]{.\EOS\space}%
\providecommand \EOS [0]{\spacefactor3000\relax}%
\providecommand \BibitemShut  [1]{\csname bibitem#1\endcsname}%
\let\auto@bib@innerbib\@empty
\bibitem [{\citenamefont {Acharya}\ \emph {et~al.}(2025)\citenamefont {Acharya}, \citenamefont {Abanin}, \citenamefont {Aghababaie-Beni}, \citenamefont {Aleiner}, \citenamefont {Andersen}, \citenamefont {Ansmann}, \citenamefont {Arute}, \citenamefont {Arya}, \citenamefont {Asfaw}, \citenamefont {Astrakhantsev}, \citenamefont {Atalaya}, \citenamefont {Babbush}, \citenamefont {Bacon}, \citenamefont {Ballard}, \citenamefont {Bardin}, \citenamefont {Bausch}, \citenamefont {Bengtsson}, \citenamefont {Bilmes}, \citenamefont {Blackwell}, \citenamefont {Boixo}, \citenamefont {Bortoli}, \citenamefont {Bourassa}, \citenamefont {Bovaird}, \citenamefont {Brill}, \citenamefont {Broughton}, \citenamefont {Browne}, \citenamefont {Buchea}, \citenamefont {Buckley}, \citenamefont {Buell}, \citenamefont {Burger}, \citenamefont {Burkett}, \citenamefont {Bushnell}, \citenamefont {Cabrera}, \citenamefont {Campero}, \citenamefont {Chang}, \citenamefont {Chen}, \citenamefont {Chen}, \citenamefont {Chiaro}, \citenamefont {Chik},
  \citenamefont {Chou}, \citenamefont {Claes}, \citenamefont {Cleland}, \citenamefont {Cogan}, \citenamefont {Collins}, \citenamefont {Conner}, \citenamefont {Courtney}, \citenamefont {Crook}, \citenamefont {Curtin}, \citenamefont {Das}, \citenamefont {Davies}, \citenamefont {De~Lorenzo}, \citenamefont {Debroy}, \citenamefont {Demura}, \citenamefont {Devoret}, \citenamefont {Di~Paolo}, \citenamefont {Donohoe}, \citenamefont {Drozdov}, \citenamefont {Dunsworth}, \citenamefont {Earle}, \citenamefont {Edlich}, \citenamefont {Eickbusch}, \citenamefont {Elbag}, \citenamefont {Elzouka}, \citenamefont {Erickson}, \citenamefont {Faoro}, \citenamefont {Farhi}, \citenamefont {Ferreira}, \citenamefont {Burgos}, \citenamefont {Forati}, \citenamefont {Fowler}, \citenamefont {Foxen}, \citenamefont {Ganjam}, \citenamefont {Garcia}, \citenamefont {Gasca}, \citenamefont {Genois}, \citenamefont {Giang}, \citenamefont {Gidney}, \citenamefont {Gilboa}, \citenamefont {Gosula}, \citenamefont {Dau}, \citenamefont {Graumann},
  \citenamefont {Greene}, \citenamefont {Gross}, \citenamefont {Habegger}, \citenamefont {Hall}, \citenamefont {Hamilton}, \citenamefont {Hansen}, \citenamefont {Harrigan}, \citenamefont {Harrington}, \citenamefont {Heras}, \citenamefont {Heslin}, \citenamefont {Heu}, \citenamefont {Higgott}, \citenamefont {Hill}, \citenamefont {Hilton}, \citenamefont {Holland}, \citenamefont {Hong}, \citenamefont {Huang}, \citenamefont {Huff}, \citenamefont {Huggins}, \citenamefont {Ioffe}, \citenamefont {Isakov}, \citenamefont {Iveland}, \citenamefont {Jeffrey}, \citenamefont {Jiang}, \citenamefont {Jones}, \citenamefont {Jordan}, \citenamefont {Joshi}, \citenamefont {Juhas}, \citenamefont {Kafri}, \citenamefont {Kang}, \citenamefont {Karamlou}, \citenamefont {Kechedzhi}, \citenamefont {Kelly}, \citenamefont {Khaire}, \citenamefont {Khattar}, \citenamefont {Khezri}, \citenamefont {Kim}, \citenamefont {Klimov}, \citenamefont {Klots}, \citenamefont {Kobrin}, \citenamefont {Kohli}, \citenamefont {Korotkov}, \citenamefont
  {Kostritsa}, \citenamefont {Kothari}, \citenamefont {Kozlovskii}, \citenamefont {Kreikebaum}, \citenamefont {Kurilovich}, \citenamefont {Lacroix}, \citenamefont {Landhuis}, \citenamefont {Lange-Dei}, \citenamefont {Langley}, \citenamefont {Laptev}, \citenamefont {Lau}, \citenamefont {Le~Guevel}, \citenamefont {Ledford}, \citenamefont {Lee}, \citenamefont {Lee}, \citenamefont {Lensky}, \citenamefont {Leon}, \citenamefont {Lester}, \citenamefont {Li}, \citenamefont {Li}, \citenamefont {Lill}, \citenamefont {Liu}, \citenamefont {Livingston}, \citenamefont {Locharla}, \citenamefont {Lucero}, \citenamefont {Lundahl}, \citenamefont {Lunt}, \citenamefont {Madhuk}, \citenamefont {Malone}, \citenamefont {Maloney}, \citenamefont {Mandr{\`a}}, \citenamefont {Manyika}, \citenamefont {Martin}, \citenamefont {Martin}, \citenamefont {Martin}, \citenamefont {Maxfield}, \citenamefont {McClean}, \citenamefont {McEwen}, \citenamefont {Meeks}, \citenamefont {Megrant}, \citenamefont {Mi}, \citenamefont {Miao}, \citenamefont
  {Mieszala}, \citenamefont {Molavi}, \citenamefont {Molina}, \citenamefont {Montazeri}, \citenamefont {Morvan}, \citenamefont {Movassagh}, \citenamefont {Mruczkiewicz}, \citenamefont {Naaman}, \citenamefont {Neeley}, \citenamefont {Neill}, \citenamefont {Nersisyan}, \citenamefont {Neven}, \citenamefont {Newman}, \citenamefont {Ng}, \citenamefont {Nguyen}, \citenamefont {Nguyen}, \citenamefont {Ni}, \citenamefont {Niu}, \citenamefont {O'Brien}, \citenamefont {Oliver}, \citenamefont {Opremcak}, \citenamefont {Ottosson}, \citenamefont {Petukhov}, \citenamefont {Pizzuto}, \citenamefont {Platt}, \citenamefont {Potter}, \citenamefont {Pritchard}, \citenamefont {Pryadko}, \citenamefont {Quintana}, \citenamefont {Ramachandran}, \citenamefont {Reagor}, \citenamefont {Redding}, \citenamefont {Rhodes}, \citenamefont {Roberts}, \citenamefont {Rosenberg}, \citenamefont {Rosenfeld}, \citenamefont {Roushan}, \citenamefont {Rubin}, \citenamefont {Saei}, \citenamefont {Sank}, \citenamefont {Sankaragomathi}, \citenamefont
  {Satzinger}, \citenamefont {Schurkus}, \citenamefont {Schuster}, \citenamefont {Senior}, \citenamefont {Shearn}, \citenamefont {Shorter}, \citenamefont {Shutty}, \citenamefont {Shvarts}, \citenamefont {Singh}, \citenamefont {Sivak}, \citenamefont {Skruzny}, \citenamefont {Small}, \citenamefont {Smelyanskiy}, \citenamefont {Smith}, \citenamefont {Somma}, \citenamefont {Springer}, \citenamefont {Sterling}, \citenamefont {Strain}, \citenamefont {Suchard}, \citenamefont {Szasz}, \citenamefont {Sztein}, \citenamefont {Thor}, \citenamefont {Torres}, \citenamefont {Torunbalci}, \citenamefont {Vaishnav}, \citenamefont {Vargas}, \citenamefont {Vdovichev}, \citenamefont {Vidal}, \citenamefont {Villalonga}, \citenamefont {Heidweiller}, \citenamefont {Waltman}, \citenamefont {Wang}, \citenamefont {Ware}, \citenamefont {Weber}, \citenamefont {Weidel}, \citenamefont {White}, \citenamefont {Wong}, \citenamefont {Woo}, \citenamefont {Xing}, \citenamefont {Yao}, \citenamefont {Yeh}, \citenamefont {Ying}, \citenamefont
  {Yoo}, \citenamefont {Yosri}, \citenamefont {Young}, \citenamefont {Zalcman}, \citenamefont {Zhang}, \citenamefont {Zhu}, \citenamefont {Zobrist}, \citenamefont {AI},\ and\ \citenamefont {Collaborators}}]{google_breakeven}%
  \BibitemOpen
  \bibfield  {author} {\bibinfo {author} {\bibfnamefont {R.}~\bibnamefont {Acharya}}, \bibinfo {author} {\bibfnamefont {D.~A.}\ \bibnamefont {Abanin}}, \bibinfo {author} {\bibfnamefont {L.}~\bibnamefont {Aghababaie-Beni}}, \bibinfo {author} {\bibfnamefont {I.}~\bibnamefont {Aleiner}}, \bibinfo {author} {\bibfnamefont {T.~I.}\ \bibnamefont {Andersen}}, \bibinfo {author} {\bibfnamefont {M.}~\bibnamefont {Ansmann}}, \bibinfo {author} {\bibfnamefont {F.}~\bibnamefont {Arute}}, \bibinfo {author} {\bibfnamefont {K.}~\bibnamefont {Arya}}, \bibinfo {author} {\bibfnamefont {A.}~\bibnamefont {Asfaw}}, \bibinfo {author} {\bibfnamefont {N.}~\bibnamefont {Astrakhantsev}}, \bibinfo {author} {\bibfnamefont {J.}~\bibnamefont {Atalaya}}, \bibinfo {author} {\bibfnamefont {R.}~\bibnamefont {Babbush}}, \bibinfo {author} {\bibfnamefont {D.}~\bibnamefont {Bacon}}, \bibinfo {author} {\bibfnamefont {B.}~\bibnamefont {Ballard}}, \bibinfo {author} {\bibfnamefont {J.~C.}\ \bibnamefont {Bardin}}, \bibinfo {author} {\bibfnamefont
  {J.}~\bibnamefont {Bausch}}, \bibinfo {author} {\bibfnamefont {A.}~\bibnamefont {Bengtsson}}, \bibinfo {author} {\bibfnamefont {A.}~\bibnamefont {Bilmes}}, \bibinfo {author} {\bibfnamefont {S.}~\bibnamefont {Blackwell}}, \bibinfo {author} {\bibfnamefont {S.}~\bibnamefont {Boixo}}, \bibinfo {author} {\bibfnamefont {G.}~\bibnamefont {Bortoli}}, \bibinfo {author} {\bibfnamefont {A.}~\bibnamefont {Bourassa}}, \bibinfo {author} {\bibfnamefont {J.}~\bibnamefont {Bovaird}}, \bibinfo {author} {\bibfnamefont {L.}~\bibnamefont {Brill}}, \bibinfo {author} {\bibfnamefont {M.}~\bibnamefont {Broughton}}, \bibinfo {author} {\bibfnamefont {D.~A.}\ \bibnamefont {Browne}}, \bibinfo {author} {\bibfnamefont {B.}~\bibnamefont {Buchea}}, \bibinfo {author} {\bibfnamefont {B.~B.}\ \bibnamefont {Buckley}}, \bibinfo {author} {\bibfnamefont {D.~A.}\ \bibnamefont {Buell}}, \bibinfo {author} {\bibfnamefont {T.}~\bibnamefont {Burger}}, \bibinfo {author} {\bibfnamefont {B.}~\bibnamefont {Burkett}}, \bibinfo {author} {\bibfnamefont
  {N.}~\bibnamefont {Bushnell}}, \bibinfo {author} {\bibfnamefont {A.}~\bibnamefont {Cabrera}}, \bibinfo {author} {\bibfnamefont {J.}~\bibnamefont {Campero}}, \bibinfo {author} {\bibfnamefont {H.-S.}\ \bibnamefont {Chang}}, \bibinfo {author} {\bibfnamefont {Y.}~\bibnamefont {Chen}}, \bibinfo {author} {\bibfnamefont {Z.}~\bibnamefont {Chen}}, \bibinfo {author} {\bibfnamefont {B.}~\bibnamefont {Chiaro}}, \bibinfo {author} {\bibfnamefont {D.}~\bibnamefont {Chik}}, \bibinfo {author} {\bibfnamefont {C.}~\bibnamefont {Chou}}, \bibinfo {author} {\bibfnamefont {J.}~\bibnamefont {Claes}}, \bibinfo {author} {\bibfnamefont {A.~Y.}\ \bibnamefont {Cleland}}, \bibinfo {author} {\bibfnamefont {J.}~\bibnamefont {Cogan}}, \bibinfo {author} {\bibfnamefont {R.}~\bibnamefont {Collins}}, \bibinfo {author} {\bibfnamefont {P.}~\bibnamefont {Conner}}, \bibinfo {author} {\bibfnamefont {W.}~\bibnamefont {Courtney}}, \bibinfo {author} {\bibfnamefont {A.~L.}\ \bibnamefont {Crook}}, \bibinfo {author} {\bibfnamefont {B.}~\bibnamefont
  {Curtin}}, \bibinfo {author} {\bibfnamefont {S.}~\bibnamefont {Das}}, \bibinfo {author} {\bibfnamefont {A.}~\bibnamefont {Davies}}, \bibinfo {author} {\bibfnamefont {L.}~\bibnamefont {De~Lorenzo}}, \bibinfo {author} {\bibfnamefont {D.~M.}\ \bibnamefont {Debroy}}, \bibinfo {author} {\bibfnamefont {S.}~\bibnamefont {Demura}}, \bibinfo {author} {\bibfnamefont {M.}~\bibnamefont {Devoret}}, \bibinfo {author} {\bibfnamefont {A.}~\bibnamefont {Di~Paolo}}, \bibinfo {author} {\bibfnamefont {P.}~\bibnamefont {Donohoe}}, \bibinfo {author} {\bibfnamefont {I.}~\bibnamefont {Drozdov}}, \bibinfo {author} {\bibfnamefont {A.}~\bibnamefont {Dunsworth}}, \bibinfo {author} {\bibfnamefont {C.}~\bibnamefont {Earle}}, \bibinfo {author} {\bibfnamefont {T.}~\bibnamefont {Edlich}}, \bibinfo {author} {\bibfnamefont {A.}~\bibnamefont {Eickbusch}}, \bibinfo {author} {\bibfnamefont {A.~M.}\ \bibnamefont {Elbag}}, \bibinfo {author} {\bibfnamefont {M.}~\bibnamefont {Elzouka}}, \bibinfo {author} {\bibfnamefont {C.}~\bibnamefont
  {Erickson}}, \bibinfo {author} {\bibfnamefont {L.}~\bibnamefont {Faoro}}, \bibinfo {author} {\bibfnamefont {E.}~\bibnamefont {Farhi}}, \bibinfo {author} {\bibfnamefont {V.~S.}\ \bibnamefont {Ferreira}}, \bibinfo {author} {\bibfnamefont {L.~F.}\ \bibnamefont {Burgos}}, \bibinfo {author} {\bibfnamefont {E.}~\bibnamefont {Forati}}, \bibinfo {author} {\bibfnamefont {A.~G.}\ \bibnamefont {Fowler}}, \bibinfo {author} {\bibfnamefont {B.}~\bibnamefont {Foxen}}, \bibinfo {author} {\bibfnamefont {S.}~\bibnamefont {Ganjam}}, \bibinfo {author} {\bibfnamefont {G.}~\bibnamefont {Garcia}}, \bibinfo {author} {\bibfnamefont {R.}~\bibnamefont {Gasca}}, \bibinfo {author} {\bibfnamefont {{\'E}.}~\bibnamefont {Genois}}, \bibinfo {author} {\bibfnamefont {W.}~\bibnamefont {Giang}}, \bibinfo {author} {\bibfnamefont {C.}~\bibnamefont {Gidney}}, \bibinfo {author} {\bibfnamefont {D.}~\bibnamefont {Gilboa}}, \bibinfo {author} {\bibfnamefont {R.}~\bibnamefont {Gosula}}, \bibinfo {author} {\bibfnamefont {A.~G.}\ \bibnamefont {Dau}},
  \bibinfo {author} {\bibfnamefont {D.}~\bibnamefont {Graumann}}, \bibinfo {author} {\bibfnamefont {A.}~\bibnamefont {Greene}}, \bibinfo {author} {\bibfnamefont {J.~A.}\ \bibnamefont {Gross}}, \bibinfo {author} {\bibfnamefont {S.}~\bibnamefont {Habegger}}, \bibinfo {author} {\bibfnamefont {J.}~\bibnamefont {Hall}}, \bibinfo {author} {\bibfnamefont {M.~C.}\ \bibnamefont {Hamilton}}, \bibinfo {author} {\bibfnamefont {M.}~\bibnamefont {Hansen}}, \bibinfo {author} {\bibfnamefont {M.~P.}\ \bibnamefont {Harrigan}}, \bibinfo {author} {\bibfnamefont {S.~D.}\ \bibnamefont {Harrington}}, \bibinfo {author} {\bibfnamefont {F.~J.~H.}\ \bibnamefont {Heras}}, \bibinfo {author} {\bibfnamefont {S.}~\bibnamefont {Heslin}}, \bibinfo {author} {\bibfnamefont {P.}~\bibnamefont {Heu}}, \bibinfo {author} {\bibfnamefont {O.}~\bibnamefont {Higgott}}, \bibinfo {author} {\bibfnamefont {G.}~\bibnamefont {Hill}}, \bibinfo {author} {\bibfnamefont {J.}~\bibnamefont {Hilton}}, \bibinfo {author} {\bibfnamefont {G.}~\bibnamefont {Holland}},
  \bibinfo {author} {\bibfnamefont {S.}~\bibnamefont {Hong}}, \bibinfo {author} {\bibfnamefont {H.-Y.}\ \bibnamefont {Huang}}, \bibinfo {author} {\bibfnamefont {A.}~\bibnamefont {Huff}}, \bibinfo {author} {\bibfnamefont {W.~J.}\ \bibnamefont {Huggins}}, \bibinfo {author} {\bibfnamefont {L.~B.}\ \bibnamefont {Ioffe}}, \bibinfo {author} {\bibfnamefont {S.~V.}\ \bibnamefont {Isakov}}, \bibinfo {author} {\bibfnamefont {J.}~\bibnamefont {Iveland}}, \bibinfo {author} {\bibfnamefont {E.}~\bibnamefont {Jeffrey}}, \bibinfo {author} {\bibfnamefont {Z.}~\bibnamefont {Jiang}}, \bibinfo {author} {\bibfnamefont {C.}~\bibnamefont {Jones}}, \bibinfo {author} {\bibfnamefont {S.}~\bibnamefont {Jordan}}, \bibinfo {author} {\bibfnamefont {C.}~\bibnamefont {Joshi}}, \bibinfo {author} {\bibfnamefont {P.}~\bibnamefont {Juhas}}, \bibinfo {author} {\bibfnamefont {D.}~\bibnamefont {Kafri}}, \bibinfo {author} {\bibfnamefont {H.}~\bibnamefont {Kang}}, \bibinfo {author} {\bibfnamefont {A.~H.}\ \bibnamefont {Karamlou}}, \bibinfo {author}
  {\bibfnamefont {K.}~\bibnamefont {Kechedzhi}}, \bibinfo {author} {\bibfnamefont {J.}~\bibnamefont {Kelly}}, \bibinfo {author} {\bibfnamefont {T.}~\bibnamefont {Khaire}}, \bibinfo {author} {\bibfnamefont {T.}~\bibnamefont {Khattar}}, \bibinfo {author} {\bibfnamefont {M.}~\bibnamefont {Khezri}}, \bibinfo {author} {\bibfnamefont {S.}~\bibnamefont {Kim}}, \bibinfo {author} {\bibfnamefont {P.~V.}\ \bibnamefont {Klimov}}, \bibinfo {author} {\bibfnamefont {A.~R.}\ \bibnamefont {Klots}}, \bibinfo {author} {\bibfnamefont {B.}~\bibnamefont {Kobrin}}, \bibinfo {author} {\bibfnamefont {P.}~\bibnamefont {Kohli}}, \bibinfo {author} {\bibfnamefont {A.~N.}\ \bibnamefont {Korotkov}}, \bibinfo {author} {\bibfnamefont {F.}~\bibnamefont {Kostritsa}}, \bibinfo {author} {\bibfnamefont {R.}~\bibnamefont {Kothari}}, \bibinfo {author} {\bibfnamefont {B.}~\bibnamefont {Kozlovskii}}, \bibinfo {author} {\bibfnamefont {J.~M.}\ \bibnamefont {Kreikebaum}}, \bibinfo {author} {\bibfnamefont {V.~D.}\ \bibnamefont {Kurilovich}}, \bibinfo
  {author} {\bibfnamefont {N.}~\bibnamefont {Lacroix}}, \bibinfo {author} {\bibfnamefont {D.}~\bibnamefont {Landhuis}}, \bibinfo {author} {\bibfnamefont {T.}~\bibnamefont {Lange-Dei}}, \bibinfo {author} {\bibfnamefont {B.~W.}\ \bibnamefont {Langley}}, \bibinfo {author} {\bibfnamefont {P.}~\bibnamefont {Laptev}}, \bibinfo {author} {\bibfnamefont {K.-M.}\ \bibnamefont {Lau}}, \bibinfo {author} {\bibfnamefont {L.}~\bibnamefont {Le~Guevel}}, \bibinfo {author} {\bibfnamefont {J.}~\bibnamefont {Ledford}}, \bibinfo {author} {\bibfnamefont {J.}~\bibnamefont {Lee}}, \bibinfo {author} {\bibfnamefont {K.}~\bibnamefont {Lee}}, \bibinfo {author} {\bibfnamefont {Y.~D.}\ \bibnamefont {Lensky}}, \bibinfo {author} {\bibfnamefont {S.}~\bibnamefont {Leon}}, \bibinfo {author} {\bibfnamefont {B.~J.}\ \bibnamefont {Lester}}, \bibinfo {author} {\bibfnamefont {W.~Y.}\ \bibnamefont {Li}}, \bibinfo {author} {\bibfnamefont {Y.}~\bibnamefont {Li}}, \bibinfo {author} {\bibfnamefont {A.~T.}\ \bibnamefont {Lill}}, \bibinfo {author}
  {\bibfnamefont {W.}~\bibnamefont {Liu}}, \bibinfo {author} {\bibfnamefont {W.~P.}\ \bibnamefont {Livingston}}, \bibinfo {author} {\bibfnamefont {A.}~\bibnamefont {Locharla}}, \bibinfo {author} {\bibfnamefont {E.}~\bibnamefont {Lucero}}, \bibinfo {author} {\bibfnamefont {D.}~\bibnamefont {Lundahl}}, \bibinfo {author} {\bibfnamefont {A.}~\bibnamefont {Lunt}}, \bibinfo {author} {\bibfnamefont {S.}~\bibnamefont {Madhuk}}, \bibinfo {author} {\bibfnamefont {F.~D.}\ \bibnamefont {Malone}}, \bibinfo {author} {\bibfnamefont {A.}~\bibnamefont {Maloney}}, \bibinfo {author} {\bibfnamefont {S.}~\bibnamefont {Mandr{\`a}}}, \bibinfo {author} {\bibfnamefont {J.}~\bibnamefont {Manyika}}, \bibinfo {author} {\bibfnamefont {L.~S.}\ \bibnamefont {Martin}}, \bibinfo {author} {\bibfnamefont {O.}~\bibnamefont {Martin}}, \bibinfo {author} {\bibfnamefont {S.}~\bibnamefont {Martin}}, \bibinfo {author} {\bibfnamefont {C.}~\bibnamefont {Maxfield}}, \bibinfo {author} {\bibfnamefont {J.~R.}\ \bibnamefont {McClean}}, \bibinfo {author}
  {\bibfnamefont {M.}~\bibnamefont {McEwen}}, \bibinfo {author} {\bibfnamefont {S.}~\bibnamefont {Meeks}}, \bibinfo {author} {\bibfnamefont {A.}~\bibnamefont {Megrant}}, \bibinfo {author} {\bibfnamefont {X.}~\bibnamefont {Mi}}, \bibinfo {author} {\bibfnamefont {K.~C.}\ \bibnamefont {Miao}}, \bibinfo {author} {\bibfnamefont {A.}~\bibnamefont {Mieszala}}, \bibinfo {author} {\bibfnamefont {R.}~\bibnamefont {Molavi}}, \bibinfo {author} {\bibfnamefont {S.}~\bibnamefont {Molina}}, \bibinfo {author} {\bibfnamefont {S.}~\bibnamefont {Montazeri}}, \bibinfo {author} {\bibfnamefont {A.}~\bibnamefont {Morvan}}, \bibinfo {author} {\bibfnamefont {R.}~\bibnamefont {Movassagh}}, \bibinfo {author} {\bibfnamefont {W.}~\bibnamefont {Mruczkiewicz}}, \bibinfo {author} {\bibfnamefont {O.}~\bibnamefont {Naaman}}, \bibinfo {author} {\bibfnamefont {M.}~\bibnamefont {Neeley}}, \bibinfo {author} {\bibfnamefont {C.}~\bibnamefont {Neill}}, \bibinfo {author} {\bibfnamefont {A.}~\bibnamefont {Nersisyan}}, \bibinfo {author} {\bibfnamefont
  {H.}~\bibnamefont {Neven}}, \bibinfo {author} {\bibfnamefont {M.}~\bibnamefont {Newman}}, \bibinfo {author} {\bibfnamefont {J.~H.}\ \bibnamefont {Ng}}, \bibinfo {author} {\bibfnamefont {A.}~\bibnamefont {Nguyen}}, \bibinfo {author} {\bibfnamefont {M.}~\bibnamefont {Nguyen}}, \bibinfo {author} {\bibfnamefont {C.-H.}\ \bibnamefont {Ni}}, \bibinfo {author} {\bibfnamefont {M.~Y.}\ \bibnamefont {Niu}}, \bibinfo {author} {\bibfnamefont {T.~E.}\ \bibnamefont {O'Brien}}, \bibinfo {author} {\bibfnamefont {W.~D.}\ \bibnamefont {Oliver}}, \bibinfo {author} {\bibfnamefont {A.}~\bibnamefont {Opremcak}}, \bibinfo {author} {\bibfnamefont {K.}~\bibnamefont {Ottosson}}, \bibinfo {author} {\bibfnamefont {A.}~\bibnamefont {Petukhov}}, \bibinfo {author} {\bibfnamefont {A.}~\bibnamefont {Pizzuto}}, \bibinfo {author} {\bibfnamefont {J.}~\bibnamefont {Platt}}, \bibinfo {author} {\bibfnamefont {R.}~\bibnamefont {Potter}}, \bibinfo {author} {\bibfnamefont {O.}~\bibnamefont {Pritchard}}, \bibinfo {author} {\bibfnamefont {L.~P.}\
  \bibnamefont {Pryadko}}, \bibinfo {author} {\bibfnamefont {C.}~\bibnamefont {Quintana}}, \bibinfo {author} {\bibfnamefont {G.}~\bibnamefont {Ramachandran}}, \bibinfo {author} {\bibfnamefont {M.~J.}\ \bibnamefont {Reagor}}, \bibinfo {author} {\bibfnamefont {J.}~\bibnamefont {Redding}}, \bibinfo {author} {\bibfnamefont {D.~M.}\ \bibnamefont {Rhodes}}, \bibinfo {author} {\bibfnamefont {G.}~\bibnamefont {Roberts}}, \bibinfo {author} {\bibfnamefont {E.}~\bibnamefont {Rosenberg}}, \bibinfo {author} {\bibfnamefont {E.}~\bibnamefont {Rosenfeld}}, \bibinfo {author} {\bibfnamefont {P.}~\bibnamefont {Roushan}}, \bibinfo {author} {\bibfnamefont {N.~C.}\ \bibnamefont {Rubin}}, \bibinfo {author} {\bibfnamefont {N.}~\bibnamefont {Saei}}, \bibinfo {author} {\bibfnamefont {D.}~\bibnamefont {Sank}}, \bibinfo {author} {\bibfnamefont {K.}~\bibnamefont {Sankaragomathi}}, \bibinfo {author} {\bibfnamefont {K.~J.}\ \bibnamefont {Satzinger}}, \bibinfo {author} {\bibfnamefont {H.~F.}\ \bibnamefont {Schurkus}}, \bibinfo {author}
  {\bibfnamefont {C.}~\bibnamefont {Schuster}}, \bibinfo {author} {\bibfnamefont {A.~W.}\ \bibnamefont {Senior}}, \bibinfo {author} {\bibfnamefont {M.~J.}\ \bibnamefont {Shearn}}, \bibinfo {author} {\bibfnamefont {A.}~\bibnamefont {Shorter}}, \bibinfo {author} {\bibfnamefont {N.}~\bibnamefont {Shutty}}, \bibinfo {author} {\bibfnamefont {V.}~\bibnamefont {Shvarts}}, \bibinfo {author} {\bibfnamefont {S.}~\bibnamefont {Singh}}, \bibinfo {author} {\bibfnamefont {V.}~\bibnamefont {Sivak}}, \bibinfo {author} {\bibfnamefont {J.}~\bibnamefont {Skruzny}}, \bibinfo {author} {\bibfnamefont {S.}~\bibnamefont {Small}}, \bibinfo {author} {\bibfnamefont {V.}~\bibnamefont {Smelyanskiy}}, \bibinfo {author} {\bibfnamefont {W.~C.}\ \bibnamefont {Smith}}, \bibinfo {author} {\bibfnamefont {R.~D.}\ \bibnamefont {Somma}}, \bibinfo {author} {\bibfnamefont {S.}~\bibnamefont {Springer}}, \bibinfo {author} {\bibfnamefont {G.}~\bibnamefont {Sterling}}, \bibinfo {author} {\bibfnamefont {D.}~\bibnamefont {Strain}}, \bibinfo {author}
  {\bibfnamefont {J.}~\bibnamefont {Suchard}}, \bibinfo {author} {\bibfnamefont {A.}~\bibnamefont {Szasz}}, \bibinfo {author} {\bibfnamefont {A.}~\bibnamefont {Sztein}}, \bibinfo {author} {\bibfnamefont {D.}~\bibnamefont {Thor}}, \bibinfo {author} {\bibfnamefont {A.}~\bibnamefont {Torres}}, \bibinfo {author} {\bibfnamefont {M.~M.}\ \bibnamefont {Torunbalci}}, \bibinfo {author} {\bibfnamefont {A.}~\bibnamefont {Vaishnav}}, \bibinfo {author} {\bibfnamefont {J.}~\bibnamefont {Vargas}}, \bibinfo {author} {\bibfnamefont {S.}~\bibnamefont {Vdovichev}}, \bibinfo {author} {\bibfnamefont {G.}~\bibnamefont {Vidal}}, \bibinfo {author} {\bibfnamefont {B.}~\bibnamefont {Villalonga}}, \bibinfo {author} {\bibfnamefont {C.~V.}\ \bibnamefont {Heidweiller}}, \bibinfo {author} {\bibfnamefont {S.}~\bibnamefont {Waltman}}, \bibinfo {author} {\bibfnamefont {S.~X.}\ \bibnamefont {Wang}}, \bibinfo {author} {\bibfnamefont {B.}~\bibnamefont {Ware}}, \bibinfo {author} {\bibfnamefont {K.}~\bibnamefont {Weber}}, \bibinfo {author}
  {\bibfnamefont {T.}~\bibnamefont {Weidel}}, \bibinfo {author} {\bibfnamefont {T.}~\bibnamefont {White}}, \bibinfo {author} {\bibfnamefont {K.}~\bibnamefont {Wong}}, \bibinfo {author} {\bibfnamefont {B.~W.~K.}\ \bibnamefont {Woo}}, \bibinfo {author} {\bibfnamefont {C.}~\bibnamefont {Xing}}, \bibinfo {author} {\bibfnamefont {Z.~J.}\ \bibnamefont {Yao}}, \bibinfo {author} {\bibfnamefont {P.}~\bibnamefont {Yeh}}, \bibinfo {author} {\bibfnamefont {B.}~\bibnamefont {Ying}}, \bibinfo {author} {\bibfnamefont {J.}~\bibnamefont {Yoo}}, \bibinfo {author} {\bibfnamefont {N.}~\bibnamefont {Yosri}}, \bibinfo {author} {\bibfnamefont {G.}~\bibnamefont {Young}}, \bibinfo {author} {\bibfnamefont {A.}~\bibnamefont {Zalcman}}, \bibinfo {author} {\bibfnamefont {Y.}~\bibnamefont {Zhang}}, \bibinfo {author} {\bibfnamefont {N.}~\bibnamefont {Zhu}}, \bibinfo {author} {\bibfnamefont {N.}~\bibnamefont {Zobrist}}, \bibinfo {author} {\bibfnamefont {G.~Q.}\ \bibnamefont {AI}},\ and\ \bibinfo {author} {\bibnamefont {Collaborators}},\
  }\href {https://doi.org/10.1038/s41586-024-08449-y} {\bibfield  {journal} {\bibinfo  {journal} {Nature}\ }\textbf {\bibinfo {volume} {638}},\ \bibinfo {pages} {920} (\bibinfo {year} {2025})}\BibitemShut {NoStop}%
\bibitem [{\citenamefont {{Ryan-Anderson}}\ \emph {et~al.}(2022)\citenamefont {{Ryan-Anderson}}, \citenamefont {{Brown}}, \citenamefont {{Allman}}, \citenamefont {{Arkin}}, \citenamefont {{Asa-Attuah}}, \citenamefont {{Baldwin}}, \citenamefont {{Berg}}, \citenamefont {{Bohnet}}, \citenamefont {{Braxton}}, \citenamefont {{Burdick}}, \citenamefont {{Campora}}, \citenamefont {{Chernoguzov}}, \citenamefont {{Esposito}}, \citenamefont {{Evans}}, \citenamefont {{Francois}}, \citenamefont {{Gaebler}}, \citenamefont {{Gatterman}}, \citenamefont {{Gerber}}, \citenamefont {{Gilmore}}, \citenamefont {{Gresh}}, \citenamefont {{Hall}}, \citenamefont {{Hankin}}, \citenamefont {{Hostetter}}, \citenamefont {{Lucchetti}}, \citenamefont {{Mayer}}, \citenamefont {{Myers}}, \citenamefont {{Neyenhuis}}, \citenamefont {{Santiago}}, \citenamefont {{Sedlacek}}, \citenamefont {{Skripka}}, \citenamefont {{Slattery}}, \citenamefont {{Stutz}}, \citenamefont {{Tait}}, \citenamefont {{Tobey}}, \citenamefont {{Vittorini}}, \citenamefont
  {{Walker}},\ and\ \citenamefont {{Hayes}}}]{2022arXiv220801863R}%
  \BibitemOpen
  \bibfield  {author} {\bibinfo {author} {\bibfnamefont {C.}~\bibnamefont {{Ryan-Anderson}}}, \bibinfo {author} {\bibfnamefont {N.~C.}\ \bibnamefont {{Brown}}}, \bibinfo {author} {\bibfnamefont {M.~S.}\ \bibnamefont {{Allman}}}, \bibinfo {author} {\bibfnamefont {B.}~\bibnamefont {{Arkin}}}, \bibinfo {author} {\bibfnamefont {G.}~\bibnamefont {{Asa-Attuah}}}, \bibinfo {author} {\bibfnamefont {C.}~\bibnamefont {{Baldwin}}}, \bibinfo {author} {\bibfnamefont {J.}~\bibnamefont {{Berg}}}, \bibinfo {author} {\bibfnamefont {J.~G.}\ \bibnamefont {{Bohnet}}}, \bibinfo {author} {\bibfnamefont {S.}~\bibnamefont {{Braxton}}}, \bibinfo {author} {\bibfnamefont {N.}~\bibnamefont {{Burdick}}}, \bibinfo {author} {\bibfnamefont {J.~P.}\ \bibnamefont {{Campora}}}, \bibinfo {author} {\bibfnamefont {A.}~\bibnamefont {{Chernoguzov}}}, \bibinfo {author} {\bibfnamefont {J.}~\bibnamefont {{Esposito}}}, \bibinfo {author} {\bibfnamefont {B.}~\bibnamefont {{Evans}}}, \bibinfo {author} {\bibfnamefont {D.}~\bibnamefont {{Francois}}},
  \bibinfo {author} {\bibfnamefont {J.~P.}\ \bibnamefont {{Gaebler}}}, \bibinfo {author} {\bibfnamefont {T.~M.}\ \bibnamefont {{Gatterman}}}, \bibinfo {author} {\bibfnamefont {J.}~\bibnamefont {{Gerber}}}, \bibinfo {author} {\bibfnamefont {K.}~\bibnamefont {{Gilmore}}}, \bibinfo {author} {\bibfnamefont {D.}~\bibnamefont {{Gresh}}}, \bibinfo {author} {\bibfnamefont {A.}~\bibnamefont {{Hall}}}, \bibinfo {author} {\bibfnamefont {A.}~\bibnamefont {{Hankin}}}, \bibinfo {author} {\bibfnamefont {J.}~\bibnamefont {{Hostetter}}}, \bibinfo {author} {\bibfnamefont {D.}~\bibnamefont {{Lucchetti}}}, \bibinfo {author} {\bibfnamefont {K.}~\bibnamefont {{Mayer}}}, \bibinfo {author} {\bibfnamefont {J.}~\bibnamefont {{Myers}}}, \bibinfo {author} {\bibfnamefont {B.}~\bibnamefont {{Neyenhuis}}}, \bibinfo {author} {\bibfnamefont {J.}~\bibnamefont {{Santiago}}}, \bibinfo {author} {\bibfnamefont {J.}~\bibnamefont {{Sedlacek}}}, \bibinfo {author} {\bibfnamefont {T.}~\bibnamefont {{Skripka}}}, \bibinfo {author} {\bibfnamefont
  {A.}~\bibnamefont {{Slattery}}}, \bibinfo {author} {\bibfnamefont {R.~P.}\ \bibnamefont {{Stutz}}}, \bibinfo {author} {\bibfnamefont {J.}~\bibnamefont {{Tait}}}, \bibinfo {author} {\bibfnamefont {R.}~\bibnamefont {{Tobey}}}, \bibinfo {author} {\bibfnamefont {G.}~\bibnamefont {{Vittorini}}}, \bibinfo {author} {\bibfnamefont {J.}~\bibnamefont {{Walker}}},\ and\ \bibinfo {author} {\bibfnamefont {D.}~\bibnamefont {{Hayes}}},\ }\href {https://doi.org/10.48550/arXiv.2208.01863} {\bibfield  {journal} {\bibinfo  {journal} {arXiv e-prints}\ ,\ \bibinfo {eid} {arXiv:2208.01863}} (\bibinfo {year} {2022})},\ \Eprint {https://arxiv.org/abs/2208.01863} {arXiv:2208.01863 [quant-ph]} \BibitemShut {NoStop}%
\bibitem [{\citenamefont {Bluvstein}\ \emph {et~al.}(2024)\citenamefont {Bluvstein}, \citenamefont {Evered}, \citenamefont {Geim}, \citenamefont {Li}, \citenamefont {Zhou}, \citenamefont {Manovitz}, \citenamefont {Ebadi}, \citenamefont {Cain}, \citenamefont {Kalinowski}, \citenamefont {Hangleiter}, \citenamefont {Bonilla~Ataides}, \citenamefont {Maskara}, \citenamefont {Cong}, \citenamefont {Gao}, \citenamefont {Sales~Rodriguez}, \citenamefont {Karolyshyn}, \citenamefont {Semeghini}, \citenamefont {Gullans}, \citenamefont {Greiner}, \citenamefont {Vuleti{\'c}},\ and\ \citenamefont {Lukin}}]{harvard_breakeven}%
  \BibitemOpen
  \bibfield  {author} {\bibinfo {author} {\bibfnamefont {D.}~\bibnamefont {Bluvstein}}, \bibinfo {author} {\bibfnamefont {S.~J.}\ \bibnamefont {Evered}}, \bibinfo {author} {\bibfnamefont {A.~A.}\ \bibnamefont {Geim}}, \bibinfo {author} {\bibfnamefont {S.~H.}\ \bibnamefont {Li}}, \bibinfo {author} {\bibfnamefont {H.}~\bibnamefont {Zhou}}, \bibinfo {author} {\bibfnamefont {T.}~\bibnamefont {Manovitz}}, \bibinfo {author} {\bibfnamefont {S.}~\bibnamefont {Ebadi}}, \bibinfo {author} {\bibfnamefont {M.}~\bibnamefont {Cain}}, \bibinfo {author} {\bibfnamefont {M.}~\bibnamefont {Kalinowski}}, \bibinfo {author} {\bibfnamefont {D.}~\bibnamefont {Hangleiter}}, \bibinfo {author} {\bibfnamefont {J.~P.}\ \bibnamefont {Bonilla~Ataides}}, \bibinfo {author} {\bibfnamefont {N.}~\bibnamefont {Maskara}}, \bibinfo {author} {\bibfnamefont {I.}~\bibnamefont {Cong}}, \bibinfo {author} {\bibfnamefont {X.}~\bibnamefont {Gao}}, \bibinfo {author} {\bibfnamefont {P.}~\bibnamefont {Sales~Rodriguez}}, \bibinfo {author} {\bibfnamefont
  {T.}~\bibnamefont {Karolyshyn}}, \bibinfo {author} {\bibfnamefont {G.}~\bibnamefont {Semeghini}}, \bibinfo {author} {\bibfnamefont {M.~J.}\ \bibnamefont {Gullans}}, \bibinfo {author} {\bibfnamefont {M.}~\bibnamefont {Greiner}}, \bibinfo {author} {\bibfnamefont {V.}~\bibnamefont {Vuleti{\'c}}},\ and\ \bibinfo {author} {\bibfnamefont {M.~D.}\ \bibnamefont {Lukin}},\ }\href {https://doi.org/10.1038/s41586-023-06927-3} {\bibfield  {journal} {\bibinfo  {journal} {Nature}\ }\textbf {\bibinfo {volume} {626}},\ \bibinfo {pages} {58} (\bibinfo {year} {2024})}\BibitemShut {NoStop}%
\bibitem [{\citenamefont {Putterman}\ \emph {et~al.}(2025)\citenamefont {Putterman}, \citenamefont {Noh}, \citenamefont {Hann}, \citenamefont {MacCabe}, \citenamefont {Aghaeimeibodi}, \citenamefont {Patel}, \citenamefont {Lee}, \citenamefont {Jones}, \citenamefont {Moradinejad}, \citenamefont {Rodriguez}, \citenamefont {Mahuli}, \citenamefont {Rose}, \citenamefont {Owens}, \citenamefont {Levine}, \citenamefont {Rosenfeld}, \citenamefont {Reinhold}, \citenamefont {Moncelsi}, \citenamefont {Alcid}, \citenamefont {Alidoust}, \citenamefont {Arrangoiz-Arriola}, \citenamefont {Barnett}, \citenamefont {Bienias}, \citenamefont {Carson}, \citenamefont {Chen}, \citenamefont {Chen}, \citenamefont {Chinkezian}, \citenamefont {Chisholm}, \citenamefont {Chou}, \citenamefont {Clerk}, \citenamefont {Clifford}, \citenamefont {Cosmic}, \citenamefont {Curiel}, \citenamefont {Davis}, \citenamefont {DeLorenzo}, \citenamefont {D'Ewart}, \citenamefont {Diky}, \citenamefont {D'Souza}, \citenamefont {Dumitrescu}, \citenamefont
  {Eisenmann}, \citenamefont {Elkhouly}, \citenamefont {Evenbly}, \citenamefont {Fang}, \citenamefont {Fang}, \citenamefont {Fling}, \citenamefont {Fon}, \citenamefont {Garcia}, \citenamefont {Gorshkov}, \citenamefont {Grant}, \citenamefont {Gray}, \citenamefont {Grimberg}, \citenamefont {Grimsmo}, \citenamefont {Haim}, \citenamefont {Hand}, \citenamefont {He}, \citenamefont {Hernandez}, \citenamefont {Hover}, \citenamefont {Hung}, \citenamefont {Hunt}, \citenamefont {Iverson}, \citenamefont {Jarrige}, \citenamefont {Jaskula}, \citenamefont {Jiang}, \citenamefont {Kalaee}, \citenamefont {Karabalin}, \citenamefont {Karalekas}, \citenamefont {Keller}, \citenamefont {Khalajhedayati}, \citenamefont {Kubica}, \citenamefont {Lee}, \citenamefont {Leroux}, \citenamefont {Lieu}, \citenamefont {Ly}, \citenamefont {Madrigal}, \citenamefont {Marcaud}, \citenamefont {McCabe}, \citenamefont {Miles}, \citenamefont {Milsted}, \citenamefont {Minguzzi}, \citenamefont {Mishra}, \citenamefont {Mukherjee}, \citenamefont
  {Naghiloo}, \citenamefont {Oblepias}, \citenamefont {Ortuno}, \citenamefont {Pagdilao}, \citenamefont {Pancotti}, \citenamefont {Panduro}, \citenamefont {Paquette}, \citenamefont {Park}, \citenamefont {Peairs}, \citenamefont {Perello}, \citenamefont {Peterson}, \citenamefont {Ponte}, \citenamefont {Preskill}, \citenamefont {Qiao}, \citenamefont {Refael}, \citenamefont {Resnick}, \citenamefont {Retzker}, \citenamefont {Reyna}, \citenamefont {Runyan}, \citenamefont {Ryan}, \citenamefont {Sahmoud}, \citenamefont {Sanchez}, \citenamefont {Sanil}, \citenamefont {Sankar}, \citenamefont {Sato}, \citenamefont {Scaffidi}, \citenamefont {Siavoshi}, \citenamefont {Sivarajah}, \citenamefont {Skogland}, \citenamefont {Su}, \citenamefont {Swenson}, \citenamefont {Teo}, \citenamefont {Tomada}, \citenamefont {Torlai}, \citenamefont {Wollack}, \citenamefont {Ye}, \citenamefont {Zerrudo}, \citenamefont {Zhang}, \citenamefont {Brand{\~a}o}, \citenamefont {Matheny},\ and\ \citenamefont {Painter}}]{bosonic_code}%
  \BibitemOpen
  \bibfield  {author} {\bibinfo {author} {\bibfnamefont {H.}~\bibnamefont {Putterman}}, \bibinfo {author} {\bibfnamefont {K.}~\bibnamefont {Noh}}, \bibinfo {author} {\bibfnamefont {C.~T.}\ \bibnamefont {Hann}}, \bibinfo {author} {\bibfnamefont {G.~S.}\ \bibnamefont {MacCabe}}, \bibinfo {author} {\bibfnamefont {S.}~\bibnamefont {Aghaeimeibodi}}, \bibinfo {author} {\bibfnamefont {R.~N.}\ \bibnamefont {Patel}}, \bibinfo {author} {\bibfnamefont {M.}~\bibnamefont {Lee}}, \bibinfo {author} {\bibfnamefont {W.~M.}\ \bibnamefont {Jones}}, \bibinfo {author} {\bibfnamefont {H.}~\bibnamefont {Moradinejad}}, \bibinfo {author} {\bibfnamefont {R.}~\bibnamefont {Rodriguez}}, \bibinfo {author} {\bibfnamefont {N.}~\bibnamefont {Mahuli}}, \bibinfo {author} {\bibfnamefont {J.}~\bibnamefont {Rose}}, \bibinfo {author} {\bibfnamefont {J.~C.}\ \bibnamefont {Owens}}, \bibinfo {author} {\bibfnamefont {H.}~\bibnamefont {Levine}}, \bibinfo {author} {\bibfnamefont {E.}~\bibnamefont {Rosenfeld}}, \bibinfo {author} {\bibfnamefont
  {P.}~\bibnamefont {Reinhold}}, \bibinfo {author} {\bibfnamefont {L.}~\bibnamefont {Moncelsi}}, \bibinfo {author} {\bibfnamefont {J.~A.}\ \bibnamefont {Alcid}}, \bibinfo {author} {\bibfnamefont {N.}~\bibnamefont {Alidoust}}, \bibinfo {author} {\bibfnamefont {P.}~\bibnamefont {Arrangoiz-Arriola}}, \bibinfo {author} {\bibfnamefont {J.}~\bibnamefont {Barnett}}, \bibinfo {author} {\bibfnamefont {P.}~\bibnamefont {Bienias}}, \bibinfo {author} {\bibfnamefont {H.~A.}\ \bibnamefont {Carson}}, \bibinfo {author} {\bibfnamefont {C.}~\bibnamefont {Chen}}, \bibinfo {author} {\bibfnamefont {L.}~\bibnamefont {Chen}}, \bibinfo {author} {\bibfnamefont {H.}~\bibnamefont {Chinkezian}}, \bibinfo {author} {\bibfnamefont {E.~M.}\ \bibnamefont {Chisholm}}, \bibinfo {author} {\bibfnamefont {M.-H.}\ \bibnamefont {Chou}}, \bibinfo {author} {\bibfnamefont {A.}~\bibnamefont {Clerk}}, \bibinfo {author} {\bibfnamefont {A.}~\bibnamefont {Clifford}}, \bibinfo {author} {\bibfnamefont {R.}~\bibnamefont {Cosmic}}, \bibinfo {author}
  {\bibfnamefont {A.~V.}\ \bibnamefont {Curiel}}, \bibinfo {author} {\bibfnamefont {E.}~\bibnamefont {Davis}}, \bibinfo {author} {\bibfnamefont {L.}~\bibnamefont {DeLorenzo}}, \bibinfo {author} {\bibfnamefont {J.~M.}\ \bibnamefont {D'Ewart}}, \bibinfo {author} {\bibfnamefont {A.}~\bibnamefont {Diky}}, \bibinfo {author} {\bibfnamefont {N.}~\bibnamefont {D'Souza}}, \bibinfo {author} {\bibfnamefont {P.~T.}\ \bibnamefont {Dumitrescu}}, \bibinfo {author} {\bibfnamefont {S.}~\bibnamefont {Eisenmann}}, \bibinfo {author} {\bibfnamefont {E.}~\bibnamefont {Elkhouly}}, \bibinfo {author} {\bibfnamefont {G.}~\bibnamefont {Evenbly}}, \bibinfo {author} {\bibfnamefont {M.~T.}\ \bibnamefont {Fang}}, \bibinfo {author} {\bibfnamefont {Y.}~\bibnamefont {Fang}}, \bibinfo {author} {\bibfnamefont {M.~J.}\ \bibnamefont {Fling}}, \bibinfo {author} {\bibfnamefont {W.}~\bibnamefont {Fon}}, \bibinfo {author} {\bibfnamefont {G.}~\bibnamefont {Garcia}}, \bibinfo {author} {\bibfnamefont {A.~V.}\ \bibnamefont {Gorshkov}}, \bibinfo {author}
  {\bibfnamefont {J.~A.}\ \bibnamefont {Grant}}, \bibinfo {author} {\bibfnamefont {M.~J.}\ \bibnamefont {Gray}}, \bibinfo {author} {\bibfnamefont {S.}~\bibnamefont {Grimberg}}, \bibinfo {author} {\bibfnamefont {A.~L.}\ \bibnamefont {Grimsmo}}, \bibinfo {author} {\bibfnamefont {A.}~\bibnamefont {Haim}}, \bibinfo {author} {\bibfnamefont {J.}~\bibnamefont {Hand}}, \bibinfo {author} {\bibfnamefont {Y.}~\bibnamefont {He}}, \bibinfo {author} {\bibfnamefont {M.}~\bibnamefont {Hernandez}}, \bibinfo {author} {\bibfnamefont {D.}~\bibnamefont {Hover}}, \bibinfo {author} {\bibfnamefont {J.~S.~C.}\ \bibnamefont {Hung}}, \bibinfo {author} {\bibfnamefont {M.}~\bibnamefont {Hunt}}, \bibinfo {author} {\bibfnamefont {J.}~\bibnamefont {Iverson}}, \bibinfo {author} {\bibfnamefont {I.}~\bibnamefont {Jarrige}}, \bibinfo {author} {\bibfnamefont {J.-C.}\ \bibnamefont {Jaskula}}, \bibinfo {author} {\bibfnamefont {L.}~\bibnamefont {Jiang}}, \bibinfo {author} {\bibfnamefont {M.}~\bibnamefont {Kalaee}}, \bibinfo {author} {\bibfnamefont
  {R.}~\bibnamefont {Karabalin}}, \bibinfo {author} {\bibfnamefont {P.~J.}\ \bibnamefont {Karalekas}}, \bibinfo {author} {\bibfnamefont {A.~J.}\ \bibnamefont {Keller}}, \bibinfo {author} {\bibfnamefont {A.}~\bibnamefont {Khalajhedayati}}, \bibinfo {author} {\bibfnamefont {A.}~\bibnamefont {Kubica}}, \bibinfo {author} {\bibfnamefont {H.}~\bibnamefont {Lee}}, \bibinfo {author} {\bibfnamefont {C.}~\bibnamefont {Leroux}}, \bibinfo {author} {\bibfnamefont {S.}~\bibnamefont {Lieu}}, \bibinfo {author} {\bibfnamefont {V.}~\bibnamefont {Ly}}, \bibinfo {author} {\bibfnamefont {K.~V.}\ \bibnamefont {Madrigal}}, \bibinfo {author} {\bibfnamefont {G.}~\bibnamefont {Marcaud}}, \bibinfo {author} {\bibfnamefont {G.}~\bibnamefont {McCabe}}, \bibinfo {author} {\bibfnamefont {C.}~\bibnamefont {Miles}}, \bibinfo {author} {\bibfnamefont {A.}~\bibnamefont {Milsted}}, \bibinfo {author} {\bibfnamefont {J.}~\bibnamefont {Minguzzi}}, \bibinfo {author} {\bibfnamefont {A.}~\bibnamefont {Mishra}}, \bibinfo {author} {\bibfnamefont
  {B.}~\bibnamefont {Mukherjee}}, \bibinfo {author} {\bibfnamefont {M.}~\bibnamefont {Naghiloo}}, \bibinfo {author} {\bibfnamefont {E.}~\bibnamefont {Oblepias}}, \bibinfo {author} {\bibfnamefont {G.}~\bibnamefont {Ortuno}}, \bibinfo {author} {\bibfnamefont {J.}~\bibnamefont {Pagdilao}}, \bibinfo {author} {\bibfnamefont {N.}~\bibnamefont {Pancotti}}, \bibinfo {author} {\bibfnamefont {A.}~\bibnamefont {Panduro}}, \bibinfo {author} {\bibfnamefont {J.}~\bibnamefont {Paquette}}, \bibinfo {author} {\bibfnamefont {M.}~\bibnamefont {Park}}, \bibinfo {author} {\bibfnamefont {G.~A.}\ \bibnamefont {Peairs}}, \bibinfo {author} {\bibfnamefont {D.}~\bibnamefont {Perello}}, \bibinfo {author} {\bibfnamefont {E.~C.}\ \bibnamefont {Peterson}}, \bibinfo {author} {\bibfnamefont {S.}~\bibnamefont {Ponte}}, \bibinfo {author} {\bibfnamefont {J.}~\bibnamefont {Preskill}}, \bibinfo {author} {\bibfnamefont {J.}~\bibnamefont {Qiao}}, \bibinfo {author} {\bibfnamefont {G.}~\bibnamefont {Refael}}, \bibinfo {author} {\bibfnamefont
  {R.}~\bibnamefont {Resnick}}, \bibinfo {author} {\bibfnamefont {A.}~\bibnamefont {Retzker}}, \bibinfo {author} {\bibfnamefont {O.~A.}\ \bibnamefont {Reyna}}, \bibinfo {author} {\bibfnamefont {M.}~\bibnamefont {Runyan}}, \bibinfo {author} {\bibfnamefont {C.~A.}\ \bibnamefont {Ryan}}, \bibinfo {author} {\bibfnamefont {A.}~\bibnamefont {Sahmoud}}, \bibinfo {author} {\bibfnamefont {E.}~\bibnamefont {Sanchez}}, \bibinfo {author} {\bibfnamefont {R.}~\bibnamefont {Sanil}}, \bibinfo {author} {\bibfnamefont {K.}~\bibnamefont {Sankar}}, \bibinfo {author} {\bibfnamefont {Y.}~\bibnamefont {Sato}}, \bibinfo {author} {\bibfnamefont {T.}~\bibnamefont {Scaffidi}}, \bibinfo {author} {\bibfnamefont {S.}~\bibnamefont {Siavoshi}}, \bibinfo {author} {\bibfnamefont {P.}~\bibnamefont {Sivarajah}}, \bibinfo {author} {\bibfnamefont {T.}~\bibnamefont {Skogland}}, \bibinfo {author} {\bibfnamefont {C.-J.}\ \bibnamefont {Su}}, \bibinfo {author} {\bibfnamefont {L.~J.}\ \bibnamefont {Swenson}}, \bibinfo {author} {\bibfnamefont {S.~M.}\
  \bibnamefont {Teo}}, \bibinfo {author} {\bibfnamefont {A.}~\bibnamefont {Tomada}}, \bibinfo {author} {\bibfnamefont {G.}~\bibnamefont {Torlai}}, \bibinfo {author} {\bibfnamefont {E.~A.}\ \bibnamefont {Wollack}}, \bibinfo {author} {\bibfnamefont {Y.}~\bibnamefont {Ye}}, \bibinfo {author} {\bibfnamefont {J.~A.}\ \bibnamefont {Zerrudo}}, \bibinfo {author} {\bibfnamefont {K.}~\bibnamefont {Zhang}}, \bibinfo {author} {\bibfnamefont {F.~G. S.~L.}\ \bibnamefont {Brand{\~a}o}}, \bibinfo {author} {\bibfnamefont {M.~H.}\ \bibnamefont {Matheny}},\ and\ \bibinfo {author} {\bibfnamefont {O.}~\bibnamefont {Painter}},\ }\href {https://doi.org/10.1038/s41586-025-08642-7} {\bibfield  {journal} {\bibinfo  {journal} {Nature}\ }\textbf {\bibinfo {volume} {638}},\ \bibinfo {pages} {927} (\bibinfo {year} {2025})}\BibitemShut {NoStop}%
\bibitem [{\citenamefont {Sivak}\ \emph {et~al.}(2023)\citenamefont {Sivak}, \citenamefont {Eickbusch}, \citenamefont {Royer}, \citenamefont {Singh}, \citenamefont {Tsioutsios}, \citenamefont {Ganjam}, \citenamefont {Miano}, \citenamefont {Brock}, \citenamefont {Ding}, \citenamefont {Frunzio}, \citenamefont {Girvin}, \citenamefont {Schoelkopf},\ and\ \citenamefont {Devoret}}]{real_time_qec}%
  \BibitemOpen
  \bibfield  {author} {\bibinfo {author} {\bibfnamefont {V.~V.}\ \bibnamefont {Sivak}}, \bibinfo {author} {\bibfnamefont {A.}~\bibnamefont {Eickbusch}}, \bibinfo {author} {\bibfnamefont {B.}~\bibnamefont {Royer}}, \bibinfo {author} {\bibfnamefont {S.}~\bibnamefont {Singh}}, \bibinfo {author} {\bibfnamefont {I.}~\bibnamefont {Tsioutsios}}, \bibinfo {author} {\bibfnamefont {S.}~\bibnamefont {Ganjam}}, \bibinfo {author} {\bibfnamefont {A.}~\bibnamefont {Miano}}, \bibinfo {author} {\bibfnamefont {B.~L.}\ \bibnamefont {Brock}}, \bibinfo {author} {\bibfnamefont {A.~Z.}\ \bibnamefont {Ding}}, \bibinfo {author} {\bibfnamefont {L.}~\bibnamefont {Frunzio}}, \bibinfo {author} {\bibfnamefont {S.~M.}\ \bibnamefont {Girvin}}, \bibinfo {author} {\bibfnamefont {R.~J.}\ \bibnamefont {Schoelkopf}},\ and\ \bibinfo {author} {\bibfnamefont {M.~H.}\ \bibnamefont {Devoret}},\ }\href {https://doi.org/10.1038/s41586-023-05782-6} {\bibfield  {journal} {\bibinfo  {journal} {Nature}\ }\textbf {\bibinfo {volume} {616}},\ \bibinfo {pages}
  {50} (\bibinfo {year} {2023})}\BibitemShut {NoStop}%
\bibitem [{\citenamefont {Evered}\ \emph {et~al.}(2023)\citenamefont {Evered}, \citenamefont {Bluvstein}, \citenamefont {Kalinowski}, \citenamefont {Ebadi}, \citenamefont {Manovitz}, \citenamefont {Zhou}, \citenamefont {Li}, \citenamefont {Geim}, \citenamefont {Wang}, \citenamefont {Maskara}, \citenamefont {Levine}, \citenamefont {Semeghini}, \citenamefont {Greiner}, \citenamefont {Vuleti{\'c}},\ and\ \citenamefont {Lukin}}]{995}%
  \BibitemOpen
  \bibfield  {author} {\bibinfo {author} {\bibfnamefont {S.~J.}\ \bibnamefont {Evered}}, \bibinfo {author} {\bibfnamefont {D.}~\bibnamefont {Bluvstein}}, \bibinfo {author} {\bibfnamefont {M.}~\bibnamefont {Kalinowski}}, \bibinfo {author} {\bibfnamefont {S.}~\bibnamefont {Ebadi}}, \bibinfo {author} {\bibfnamefont {T.}~\bibnamefont {Manovitz}}, \bibinfo {author} {\bibfnamefont {H.}~\bibnamefont {Zhou}}, \bibinfo {author} {\bibfnamefont {S.~H.}\ \bibnamefont {Li}}, \bibinfo {author} {\bibfnamefont {A.~A.}\ \bibnamefont {Geim}}, \bibinfo {author} {\bibfnamefont {T.~T.}\ \bibnamefont {Wang}}, \bibinfo {author} {\bibfnamefont {N.}~\bibnamefont {Maskara}}, \bibinfo {author} {\bibfnamefont {H.}~\bibnamefont {Levine}}, \bibinfo {author} {\bibfnamefont {G.}~\bibnamefont {Semeghini}}, \bibinfo {author} {\bibfnamefont {M.}~\bibnamefont {Greiner}}, \bibinfo {author} {\bibfnamefont {V.}~\bibnamefont {Vuleti{\'c}}},\ and\ \bibinfo {author} {\bibfnamefont {M.~D.}\ \bibnamefont {Lukin}},\ }\href
  {https://doi.org/10.1038/s41586-023-06481-y} {\bibfield  {journal} {\bibinfo  {journal} {Nature}\ }\textbf {\bibinfo {volume} {622}},\ \bibinfo {pages} {268} (\bibinfo {year} {2023})}\BibitemShut {NoStop}%
\bibitem [{\citenamefont {Levine}\ \emph {et~al.}(2019)\citenamefont {Levine}, \citenamefont {Keesling}, \citenamefont {Semeghini}, \citenamefont {Omran}, \citenamefont {Wang}, \citenamefont {Ebadi}, \citenamefont {Bernien}, \citenamefont {Greiner}, \citenamefont {Vuleti\ifmmode~\acute{c}\else \'{c}\fi{}}, \citenamefont {Pichler},\ and\ \citenamefont {Lukin}}]{PhysRevLett.123.170503}%
  \BibitemOpen
  \bibfield  {author} {\bibinfo {author} {\bibfnamefont {H.}~\bibnamefont {Levine}}, \bibinfo {author} {\bibfnamefont {A.}~\bibnamefont {Keesling}}, \bibinfo {author} {\bibfnamefont {G.}~\bibnamefont {Semeghini}}, \bibinfo {author} {\bibfnamefont {A.}~\bibnamefont {Omran}}, \bibinfo {author} {\bibfnamefont {T.~T.}\ \bibnamefont {Wang}}, \bibinfo {author} {\bibfnamefont {S.}~\bibnamefont {Ebadi}}, \bibinfo {author} {\bibfnamefont {H.}~\bibnamefont {Bernien}}, \bibinfo {author} {\bibfnamefont {M.}~\bibnamefont {Greiner}}, \bibinfo {author} {\bibfnamefont {V.}~\bibnamefont {Vuleti\ifmmode~\acute{c}\else \'{c}\fi{}}}, \bibinfo {author} {\bibfnamefont {H.}~\bibnamefont {Pichler}},\ and\ \bibinfo {author} {\bibfnamefont {M.~D.}\ \bibnamefont {Lukin}},\ }\href {https://doi.org/10.1103/PhysRevLett.123.170503} {\bibfield  {journal} {\bibinfo  {journal} {Phys. Rev. Lett.}\ }\textbf {\bibinfo {volume} {123}},\ \bibinfo {pages} {170503} (\bibinfo {year} {2019})}\BibitemShut {NoStop}%
\bibitem [{\citenamefont {{L{\"o}schnauer}}\ \emph {et~al.}(2024)\citenamefont {{L{\"o}schnauer}}, \citenamefont {{Mosca Toba}}, \citenamefont {{Hughes}}, \citenamefont {{King}}, \citenamefont {{Weber}}, \citenamefont {{Srinivas}}, \citenamefont {{Matt}}, \citenamefont {{Nourshargh}}, \citenamefont {{Allcock}}, \citenamefont {{Ballance}}, \citenamefont {{Matthiesen}}, \citenamefont {{Malinowski}},\ and\ \citenamefont {{Harty}}}]{iongates}%
  \BibitemOpen
  \bibfield  {author} {\bibinfo {author} {\bibfnamefont {C.~M.}\ \bibnamefont {{L{\"o}schnauer}}}, \bibinfo {author} {\bibfnamefont {J.}~\bibnamefont {{Mosca Toba}}}, \bibinfo {author} {\bibfnamefont {A.~C.}\ \bibnamefont {{Hughes}}}, \bibinfo {author} {\bibfnamefont {S.~A.}\ \bibnamefont {{King}}}, \bibinfo {author} {\bibfnamefont {M.~A.}\ \bibnamefont {{Weber}}}, \bibinfo {author} {\bibfnamefont {R.}~\bibnamefont {{Srinivas}}}, \bibinfo {author} {\bibfnamefont {R.}~\bibnamefont {{Matt}}}, \bibinfo {author} {\bibfnamefont {R.}~\bibnamefont {{Nourshargh}}}, \bibinfo {author} {\bibfnamefont {D.~T.~C.}\ \bibnamefont {{Allcock}}}, \bibinfo {author} {\bibfnamefont {C.~J.}\ \bibnamefont {{Ballance}}}, \bibinfo {author} {\bibfnamefont {C.}~\bibnamefont {{Matthiesen}}}, \bibinfo {author} {\bibfnamefont {M.}~\bibnamefont {{Malinowski}}},\ and\ \bibinfo {author} {\bibfnamefont {T.~P.}\ \bibnamefont {{Harty}}},\ }\href {https://doi.org/10.48550/arXiv.2407.07694} {\bibfield  {journal} {\bibinfo  {journal} {arXiv
  e-prints}\ ,\ \bibinfo {eid} {arXiv:2407.07694}} (\bibinfo {year} {2024})},\ \Eprint {https://arxiv.org/abs/2407.07694} {arXiv:2407.07694 [quant-ph]} \BibitemShut {NoStop}%
\bibitem [{\citenamefont {{Chiu}}\ \emph {et~al.}(2025)\citenamefont {{Chiu}}, \citenamefont {{Trapp}}, \citenamefont {{Guo}}, \citenamefont {{Abobeih}}, \citenamefont {{Stewart}}, \citenamefont {{Hollerith}}, \citenamefont {{Stroganov}}, \citenamefont {{Kalinowski}}, \citenamefont {{Geim}}, \citenamefont {{Evered}}, \citenamefont {{Li}}, \citenamefont {{Peters}}, \citenamefont {{Bluvstein}}, \citenamefont {{Wang}}, \citenamefont {{Greiner}}, \citenamefont {{Vuleti{\'c}}},\ and\ \citenamefont {{Lukin}}}]{3000qubits}%
  \BibitemOpen
  \bibfield  {author} {\bibinfo {author} {\bibfnamefont {N.-C.}\ \bibnamefont {{Chiu}}}, \bibinfo {author} {\bibfnamefont {E.~C.}\ \bibnamefont {{Trapp}}}, \bibinfo {author} {\bibfnamefont {J.}~\bibnamefont {{Guo}}}, \bibinfo {author} {\bibfnamefont {M.~H.}\ \bibnamefont {{Abobeih}}}, \bibinfo {author} {\bibfnamefont {L.~M.}\ \bibnamefont {{Stewart}}}, \bibinfo {author} {\bibfnamefont {S.}~\bibnamefont {{Hollerith}}}, \bibinfo {author} {\bibfnamefont {P.}~\bibnamefont {{Stroganov}}}, \bibinfo {author} {\bibfnamefont {M.}~\bibnamefont {{Kalinowski}}}, \bibinfo {author} {\bibfnamefont {A.~A.}\ \bibnamefont {{Geim}}}, \bibinfo {author} {\bibfnamefont {S.~J.}\ \bibnamefont {{Evered}}}, \bibinfo {author} {\bibfnamefont {S.~H.}\ \bibnamefont {{Li}}}, \bibinfo {author} {\bibfnamefont {L.~M.}\ \bibnamefont {{Peters}}}, \bibinfo {author} {\bibfnamefont {D.}~\bibnamefont {{Bluvstein}}}, \bibinfo {author} {\bibfnamefont {T.~T.}\ \bibnamefont {{Wang}}}, \bibinfo {author} {\bibfnamefont {M.}~\bibnamefont {{Greiner}}},
  \bibinfo {author} {\bibfnamefont {V.}~\bibnamefont {{Vuleti{\'c}}}},\ and\ \bibinfo {author} {\bibfnamefont {M.~D.}\ \bibnamefont {{Lukin}}},\ }\href {https://doi.org/10.48550/arXiv.2506.20660} {\bibfield  {journal} {\bibinfo  {journal} {arXiv e-prints}\ ,\ \bibinfo {eid} {arXiv:2506.20660}} (\bibinfo {year} {2025})},\ \Eprint {https://arxiv.org/abs/2506.20660} {arXiv:2506.20660 [quant-ph]} \BibitemShut {NoStop}%
\bibitem [{\citenamefont {{Manetsch}}\ \emph {et~al.}(2024)\citenamefont {{Manetsch}}, \citenamefont {{Nomura}}, \citenamefont {{Bataille}}, \citenamefont {{Leung}}, \citenamefont {{Lv}},\ and\ \citenamefont {{Endres}}}]{6000atoms}%
  \BibitemOpen
  \bibfield  {author} {\bibinfo {author} {\bibfnamefont {H.~J.}\ \bibnamefont {{Manetsch}}}, \bibinfo {author} {\bibfnamefont {G.}~\bibnamefont {{Nomura}}}, \bibinfo {author} {\bibfnamefont {E.}~\bibnamefont {{Bataille}}}, \bibinfo {author} {\bibfnamefont {K.~H.}\ \bibnamefont {{Leung}}}, \bibinfo {author} {\bibfnamefont {X.}~\bibnamefont {{Lv}}},\ and\ \bibinfo {author} {\bibfnamefont {M.}~\bibnamefont {{Endres}}},\ }\href {https://doi.org/10.48550/arXiv.2403.12021} {\bibfield  {journal} {\bibinfo  {journal} {arXiv e-prints}\ ,\ \bibinfo {eid} {arXiv:2403.12021}} (\bibinfo {year} {2024})},\ \Eprint {https://arxiv.org/abs/2403.12021} {arXiv:2403.12021 [quant-ph]} \BibitemShut {NoStop}%
\bibitem [{\citenamefont {Guo}\ \emph {et~al.}(2024)\citenamefont {Guo}, \citenamefont {Wu}, \citenamefont {Ye}, \citenamefont {Zhang}, \citenamefont {Lian}, \citenamefont {Yao}, \citenamefont {Wang}, \citenamefont {Yan}, \citenamefont {Yi}, \citenamefont {Xu}, \citenamefont {Li}, \citenamefont {Hou}, \citenamefont {Xu}, \citenamefont {Guo}, \citenamefont {Zhang}, \citenamefont {Qi}, \citenamefont {Zhou}, \citenamefont {He},\ and\ \citenamefont {Duan}}]{500ions}%
  \BibitemOpen
  \bibfield  {author} {\bibinfo {author} {\bibfnamefont {S.~A.}\ \bibnamefont {Guo}}, \bibinfo {author} {\bibfnamefont {Y.~K.}\ \bibnamefont {Wu}}, \bibinfo {author} {\bibfnamefont {J.}~\bibnamefont {Ye}}, \bibinfo {author} {\bibfnamefont {L.}~\bibnamefont {Zhang}}, \bibinfo {author} {\bibfnamefont {W.~Q.}\ \bibnamefont {Lian}}, \bibinfo {author} {\bibfnamefont {R.}~\bibnamefont {Yao}}, \bibinfo {author} {\bibfnamefont {Y.}~\bibnamefont {Wang}}, \bibinfo {author} {\bibfnamefont {R.~Y.}\ \bibnamefont {Yan}}, \bibinfo {author} {\bibfnamefont {Y.~J.}\ \bibnamefont {Yi}}, \bibinfo {author} {\bibfnamefont {Y.~L.}\ \bibnamefont {Xu}}, \bibinfo {author} {\bibfnamefont {B.~W.}\ \bibnamefont {Li}}, \bibinfo {author} {\bibfnamefont {Y.~H.}\ \bibnamefont {Hou}}, \bibinfo {author} {\bibfnamefont {Y.~Z.}\ \bibnamefont {Xu}}, \bibinfo {author} {\bibfnamefont {W.~X.}\ \bibnamefont {Guo}}, \bibinfo {author} {\bibfnamefont {C.}~\bibnamefont {Zhang}}, \bibinfo {author} {\bibfnamefont {B.~X.}\ \bibnamefont {Qi}}, \bibinfo
  {author} {\bibfnamefont {Z.~C.}\ \bibnamefont {Zhou}}, \bibinfo {author} {\bibfnamefont {L.}~\bibnamefont {He}},\ and\ \bibinfo {author} {\bibfnamefont {L.~M.}\ \bibnamefont {Duan}},\ }\href {https://doi.org/10.1038/s41586-024-07459-0} {\bibfield  {journal} {\bibinfo  {journal} {Nature}\ }\textbf {\bibinfo {volume} {630}},\ \bibinfo {pages} {613} (\bibinfo {year} {2024})}\BibitemShut {NoStop}%
\bibitem [{\citenamefont {Semeghini}\ \emph {et~al.}(2021{\natexlab{a}})\citenamefont {Semeghini}, \citenamefont {Levine}, \citenamefont {Keesling}, \citenamefont {Ebadi}, \citenamefont {Wang}, \citenamefont {Bluvstein}, \citenamefont {Verresen}, \citenamefont {Pichler}, \citenamefont {Kalinowski}, \citenamefont {Samajdar}, \citenamefont {Omran}, \citenamefont {Sachdev}, \citenamefont {Vishwanath}, \citenamefont {Greiner}, \citenamefont {Vuletić},\ and\ \citenamefont {Lukin}}]{spinliquid}%
  \BibitemOpen
  \bibfield  {author} {\bibinfo {author} {\bibfnamefont {G.}~\bibnamefont {Semeghini}}, \bibinfo {author} {\bibfnamefont {H.}~\bibnamefont {Levine}}, \bibinfo {author} {\bibfnamefont {A.}~\bibnamefont {Keesling}}, \bibinfo {author} {\bibfnamefont {S.}~\bibnamefont {Ebadi}}, \bibinfo {author} {\bibfnamefont {T.~T.}\ \bibnamefont {Wang}}, \bibinfo {author} {\bibfnamefont {D.}~\bibnamefont {Bluvstein}}, \bibinfo {author} {\bibfnamefont {R.}~\bibnamefont {Verresen}}, \bibinfo {author} {\bibfnamefont {H.}~\bibnamefont {Pichler}}, \bibinfo {author} {\bibfnamefont {M.}~\bibnamefont {Kalinowski}}, \bibinfo {author} {\bibfnamefont {R.}~\bibnamefont {Samajdar}}, \bibinfo {author} {\bibfnamefont {A.}~\bibnamefont {Omran}}, \bibinfo {author} {\bibfnamefont {S.}~\bibnamefont {Sachdev}}, \bibinfo {author} {\bibfnamefont {A.}~\bibnamefont {Vishwanath}}, \bibinfo {author} {\bibfnamefont {M.}~\bibnamefont {Greiner}}, \bibinfo {author} {\bibfnamefont {V.}~\bibnamefont {Vuletić}},\ and\ \bibinfo {author} {\bibfnamefont
  {M.~D.}\ \bibnamefont {Lukin}},\ }\href {https://doi.org/10.1126/science.abi8794} {\bibfield  {journal} {\bibinfo  {journal} {Science}\ }\textbf {\bibinfo {volume} {374}},\ \bibinfo {pages} {1242} (\bibinfo {year} {2021}{\natexlab{a}})},\ \Eprint {https://arxiv.org/abs/https://www.science.org/doi/pdf/10.1126/science.abi8794} {https://www.science.org/doi/pdf/10.1126/science.abi8794} \BibitemShut {NoStop}%
\bibitem [{\citenamefont {Bohnet}\ \emph {et~al.}(2016)\citenamefont {Bohnet}, \citenamefont {Sawyer}, \citenamefont {Britton}, \citenamefont {Wall}, \citenamefont {Rey}, \citenamefont {Foss-Feig},\ and\ \citenamefont {Bollinger}}]{nist_ions}%
  \BibitemOpen
  \bibfield  {author} {\bibinfo {author} {\bibfnamefont {J.~G.}\ \bibnamefont {Bohnet}}, \bibinfo {author} {\bibfnamefont {B.~C.}\ \bibnamefont {Sawyer}}, \bibinfo {author} {\bibfnamefont {J.~W.}\ \bibnamefont {Britton}}, \bibinfo {author} {\bibfnamefont {M.~L.}\ \bibnamefont {Wall}}, \bibinfo {author} {\bibfnamefont {A.~M.}\ \bibnamefont {Rey}}, \bibinfo {author} {\bibfnamefont {M.}~\bibnamefont {Foss-Feig}},\ and\ \bibinfo {author} {\bibfnamefont {J.~J.}\ \bibnamefont {Bollinger}},\ }\href {https://doi.org/10.1126/science.aad9958} {\bibfield  {journal} {\bibinfo  {journal} {Science}\ }\textbf {\bibinfo {volume} {352}},\ \bibinfo {pages} {1297} (\bibinfo {year} {2016})},\ \Eprint {https://arxiv.org/abs/https://www.science.org/doi/pdf/10.1126/science.aad9958} {https://www.science.org/doi/pdf/10.1126/science.aad9958} \BibitemShut {NoStop}%
\bibitem [{\citenamefont {Scholl}\ \emph {et~al.}(2021)\citenamefont {Scholl}, \citenamefont {Schuler}, \citenamefont {Williams}, \citenamefont {Eberharter}, \citenamefont {Barredo}, \citenamefont {Schymik}, \citenamefont {Lienhard}, \citenamefont {Henry}, \citenamefont {Lang}, \citenamefont {Lahaye}, \citenamefont {L{\"a}uchli},\ and\ \citenamefont {Browaeys}}]{France_atoms}%
  \BibitemOpen
  \bibfield  {author} {\bibinfo {author} {\bibfnamefont {P.}~\bibnamefont {Scholl}}, \bibinfo {author} {\bibfnamefont {M.}~\bibnamefont {Schuler}}, \bibinfo {author} {\bibfnamefont {H.~J.}\ \bibnamefont {Williams}}, \bibinfo {author} {\bibfnamefont {A.~A.}\ \bibnamefont {Eberharter}}, \bibinfo {author} {\bibfnamefont {D.}~\bibnamefont {Barredo}}, \bibinfo {author} {\bibfnamefont {K.-N.}\ \bibnamefont {Schymik}}, \bibinfo {author} {\bibfnamefont {V.}~\bibnamefont {Lienhard}}, \bibinfo {author} {\bibfnamefont {L.-P.}\ \bibnamefont {Henry}}, \bibinfo {author} {\bibfnamefont {T.~C.}\ \bibnamefont {Lang}}, \bibinfo {author} {\bibfnamefont {T.}~\bibnamefont {Lahaye}}, \bibinfo {author} {\bibfnamefont {A.~M.}\ \bibnamefont {L{\"a}uchli}},\ and\ \bibinfo {author} {\bibfnamefont {A.}~\bibnamefont {Browaeys}},\ }\href {https://doi.org/10.1038/s41586-021-03585-1} {\bibfield  {journal} {\bibinfo  {journal} {Nature}\ }\textbf {\bibinfo {volume} {595}},\ \bibinfo {pages} {233} (\bibinfo {year} {2021})}\BibitemShut
  {NoStop}%
\bibitem [{\citenamefont {{Kornja{\v{c}}a}}\ \emph {et~al.}(2024)\citenamefont {{Kornja{\v{c}}a}}, \citenamefont {{Hu}}, \citenamefont {{Zhao}}, \citenamefont {{Wurtz}}, \citenamefont {{Weinberg}}, \citenamefont {{Hamdan}}, \citenamefont {{Zhdanov}}, \citenamefont {{Cantu}}, \citenamefont {{Zhou}}, \citenamefont {{Araiza Bravo}}, \citenamefont {{Bagnall}}, \citenamefont {{Basham}}, \citenamefont {{Campo}}, \citenamefont {{Choukri}}, \citenamefont {{DeAngelo}}, \citenamefont {{Frederick}}, \citenamefont {{Haines}}, \citenamefont {{Hammett}}, \citenamefont {{Hsu}}, \citenamefont {{Hu}}, \citenamefont {{Huber}}, \citenamefont {{Jepsen}}, \citenamefont {{Jia}}, \citenamefont {{Karolyshyn}}, \citenamefont {{Kwon}}, \citenamefont {{Long}}, \citenamefont {{Lopatin}}, \citenamefont {{Lukin}}, \citenamefont {{Macr{\`\i}}}, \citenamefont {{Markovi{\'c}}}, \citenamefont {{Mart{\'\i}nez-Mart{\'\i}nez}}, \citenamefont {{Meng}}, \citenamefont {{Ostroumov}}, \citenamefont {{Paquette}}, \citenamefont {{Robinson}},
  \citenamefont {{Sales Rodriguez}}, \citenamefont {{Singh}}, \citenamefont {{Sinha}}, \citenamefont {{Thoreen}}, \citenamefont {{Wan}}, \citenamefont {{Waxman-Lenz}}, \citenamefont {{Wong}}, \citenamefont {{Wu}}, \citenamefont {{Lopes}}, \citenamefont {{Boger}}, \citenamefont {{Gemelke}}, \citenamefont {{Kitagawa}}, \citenamefont {{Keesling}}, \citenamefont {{Gao}}, \citenamefont {{Bylinskii}}, \citenamefont {{Yelin}}, \citenamefont {{Liu}},\ and\ \citenamefont {{Wang}}}]{2024arXiv240702553K}%
  \BibitemOpen
  \bibfield  {author} {\bibinfo {author} {\bibfnamefont {M.}~\bibnamefont {{Kornja{\v{c}}a}}}, \bibinfo {author} {\bibfnamefont {H.-Y.}\ \bibnamefont {{Hu}}}, \bibinfo {author} {\bibfnamefont {C.}~\bibnamefont {{Zhao}}}, \bibinfo {author} {\bibfnamefont {J.}~\bibnamefont {{Wurtz}}}, \bibinfo {author} {\bibfnamefont {P.}~\bibnamefont {{Weinberg}}}, \bibinfo {author} {\bibfnamefont {M.}~\bibnamefont {{Hamdan}}}, \bibinfo {author} {\bibfnamefont {A.}~\bibnamefont {{Zhdanov}}}, \bibinfo {author} {\bibfnamefont {S.~H.}\ \bibnamefont {{Cantu}}}, \bibinfo {author} {\bibfnamefont {H.}~\bibnamefont {{Zhou}}}, \bibinfo {author} {\bibfnamefont {R.}~\bibnamefont {{Araiza Bravo}}}, \bibinfo {author} {\bibfnamefont {K.}~\bibnamefont {{Bagnall}}}, \bibinfo {author} {\bibfnamefont {J.~I.}\ \bibnamefont {{Basham}}}, \bibinfo {author} {\bibfnamefont {J.}~\bibnamefont {{Campo}}}, \bibinfo {author} {\bibfnamefont {A.}~\bibnamefont {{Choukri}}}, \bibinfo {author} {\bibfnamefont {R.}~\bibnamefont {{DeAngelo}}}, \bibinfo {author}
  {\bibfnamefont {P.}~\bibnamefont {{Frederick}}}, \bibinfo {author} {\bibfnamefont {D.}~\bibnamefont {{Haines}}}, \bibinfo {author} {\bibfnamefont {J.}~\bibnamefont {{Hammett}}}, \bibinfo {author} {\bibfnamefont {N.}~\bibnamefont {{Hsu}}}, \bibinfo {author} {\bibfnamefont {M.-G.}\ \bibnamefont {{Hu}}}, \bibinfo {author} {\bibfnamefont {F.}~\bibnamefont {{Huber}}}, \bibinfo {author} {\bibfnamefont {P.~N.}\ \bibnamefont {{Jepsen}}}, \bibinfo {author} {\bibfnamefont {N.}~\bibnamefont {{Jia}}}, \bibinfo {author} {\bibfnamefont {T.}~\bibnamefont {{Karolyshyn}}}, \bibinfo {author} {\bibfnamefont {M.}~\bibnamefont {{Kwon}}}, \bibinfo {author} {\bibfnamefont {J.}~\bibnamefont {{Long}}}, \bibinfo {author} {\bibfnamefont {J.}~\bibnamefont {{Lopatin}}}, \bibinfo {author} {\bibfnamefont {A.}~\bibnamefont {{Lukin}}}, \bibinfo {author} {\bibfnamefont {T.}~\bibnamefont {{Macr{\`\i}}}}, \bibinfo {author} {\bibfnamefont {O.}~\bibnamefont {{Markovi{\'c}}}}, \bibinfo {author} {\bibfnamefont {L.~A.}\ \bibnamefont
  {{Mart{\'\i}nez-Mart{\'\i}nez}}}, \bibinfo {author} {\bibfnamefont {X.}~\bibnamefont {{Meng}}}, \bibinfo {author} {\bibfnamefont {E.}~\bibnamefont {{Ostroumov}}}, \bibinfo {author} {\bibfnamefont {D.}~\bibnamefont {{Paquette}}}, \bibinfo {author} {\bibfnamefont {J.}~\bibnamefont {{Robinson}}}, \bibinfo {author} {\bibfnamefont {P.}~\bibnamefont {{Sales Rodriguez}}}, \bibinfo {author} {\bibfnamefont {A.}~\bibnamefont {{Singh}}}, \bibinfo {author} {\bibfnamefont {N.}~\bibnamefont {{Sinha}}}, \bibinfo {author} {\bibfnamefont {H.}~\bibnamefont {{Thoreen}}}, \bibinfo {author} {\bibfnamefont {N.}~\bibnamefont {{Wan}}}, \bibinfo {author} {\bibfnamefont {D.}~\bibnamefont {{Waxman-Lenz}}}, \bibinfo {author} {\bibfnamefont {T.}~\bibnamefont {{Wong}}}, \bibinfo {author} {\bibfnamefont {K.-H.}\ \bibnamefont {{Wu}}}, \bibinfo {author} {\bibfnamefont {P.~L.~S.}\ \bibnamefont {{Lopes}}}, \bibinfo {author} {\bibfnamefont {Y.}~\bibnamefont {{Boger}}}, \bibinfo {author} {\bibfnamefont {N.}~\bibnamefont {{Gemelke}}}, \bibinfo
  {author} {\bibfnamefont {T.}~\bibnamefont {{Kitagawa}}}, \bibinfo {author} {\bibfnamefont {A.}~\bibnamefont {{Keesling}}}, \bibinfo {author} {\bibfnamefont {X.}~\bibnamefont {{Gao}}}, \bibinfo {author} {\bibfnamefont {A.}~\bibnamefont {{Bylinskii}}}, \bibinfo {author} {\bibfnamefont {S.~F.}\ \bibnamefont {{Yelin}}}, \bibinfo {author} {\bibfnamefont {F.}~\bibnamefont {{Liu}}},\ and\ \bibinfo {author} {\bibfnamefont {S.-T.}\ \bibnamefont {{Wang}}},\ }\href {https://doi.org/10.48550/arXiv.2407.02553} {\bibfield  {journal} {\bibinfo  {journal} {arXiv e-prints}\ ,\ \bibinfo {eid} {arXiv:2407.02553}} (\bibinfo {year} {2024})},\ \Eprint {https://arxiv.org/abs/2407.02553} {arXiv:2407.02553 [quant-ph]} \BibitemShut {NoStop}%
\bibitem [{\citenamefont {Manovitz}\ \emph {et~al.}(2025)\citenamefont {Manovitz}, \citenamefont {Li}, \citenamefont {Ebadi}, \citenamefont {Samajdar}, \citenamefont {Geim}, \citenamefont {Evered}, \citenamefont {Bluvstein}, \citenamefont {Zhou}, \citenamefont {Koyluoglu}, \citenamefont {Feldmeier}, \citenamefont {Dolgirev}, \citenamefont {Maskara}, \citenamefont {Kalinowski}, \citenamefont {Sachdev}, \citenamefont {Huse}, \citenamefont {Greiner}, \citenamefont {Vuleti{\'c}},\ and\ \citenamefont {Lukin}}]{harvard_dynamics}%
  \BibitemOpen
  \bibfield  {author} {\bibinfo {author} {\bibfnamefont {T.}~\bibnamefont {Manovitz}}, \bibinfo {author} {\bibfnamefont {S.~H.}\ \bibnamefont {Li}}, \bibinfo {author} {\bibfnamefont {S.}~\bibnamefont {Ebadi}}, \bibinfo {author} {\bibfnamefont {R.}~\bibnamefont {Samajdar}}, \bibinfo {author} {\bibfnamefont {A.~A.}\ \bibnamefont {Geim}}, \bibinfo {author} {\bibfnamefont {S.~J.}\ \bibnamefont {Evered}}, \bibinfo {author} {\bibfnamefont {D.}~\bibnamefont {Bluvstein}}, \bibinfo {author} {\bibfnamefont {H.}~\bibnamefont {Zhou}}, \bibinfo {author} {\bibfnamefont {N.~U.}\ \bibnamefont {Koyluoglu}}, \bibinfo {author} {\bibfnamefont {J.}~\bibnamefont {Feldmeier}}, \bibinfo {author} {\bibfnamefont {P.~E.}\ \bibnamefont {Dolgirev}}, \bibinfo {author} {\bibfnamefont {N.}~\bibnamefont {Maskara}}, \bibinfo {author} {\bibfnamefont {M.}~\bibnamefont {Kalinowski}}, \bibinfo {author} {\bibfnamefont {S.}~\bibnamefont {Sachdev}}, \bibinfo {author} {\bibfnamefont {D.~A.}\ \bibnamefont {Huse}}, \bibinfo {author} {\bibfnamefont
  {M.}~\bibnamefont {Greiner}}, \bibinfo {author} {\bibfnamefont {V.}~\bibnamefont {Vuleti{\'c}}},\ and\ \bibinfo {author} {\bibfnamefont {M.~D.}\ \bibnamefont {Lukin}},\ }\href {https://doi.org/10.1038/s41586-024-08353-5} {\bibfield  {journal} {\bibinfo  {journal} {Nature}\ }\textbf {\bibinfo {volume} {638}},\ \bibinfo {pages} {86} (\bibinfo {year} {2025})}\BibitemShut {NoStop}%
\bibitem [{\citenamefont {Hartke}\ \emph {et~al.}(2023)\citenamefont {Hartke}, \citenamefont {Oreg}, \citenamefont {Turnbaugh}, \citenamefont {Jia},\ and\ \citenamefont {Zwierlein}}]{fermigas}%
  \BibitemOpen
  \bibfield  {author} {\bibinfo {author} {\bibfnamefont {T.}~\bibnamefont {Hartke}}, \bibinfo {author} {\bibfnamefont {B.}~\bibnamefont {Oreg}}, \bibinfo {author} {\bibfnamefont {C.}~\bibnamefont {Turnbaugh}}, \bibinfo {author} {\bibfnamefont {N.}~\bibnamefont {Jia}},\ and\ \bibinfo {author} {\bibfnamefont {M.}~\bibnamefont {Zwierlein}},\ }\href {https://doi.org/10.1126/science.ade4245} {\bibfield  {journal} {\bibinfo  {journal} {Science}\ }\textbf {\bibinfo {volume} {381}},\ \bibinfo {pages} {82} (\bibinfo {year} {2023})},\ \Eprint {https://arxiv.org/abs/https://www.science.org/doi/pdf/10.1126/science.ade4245} {https://www.science.org/doi/pdf/10.1126/science.ade4245} \BibitemShut {NoStop}%
\bibitem [{\citenamefont {Andersen}\ \emph {et~al.}(2025{\natexlab{a}})\citenamefont {Andersen}, \citenamefont {Astrakhantsev}, \citenamefont {Karamlou}, \citenamefont {Berndtsson}, \citenamefont {Motruk}, \citenamefont {Szasz}, \citenamefont {Gross}, \citenamefont {Schuckert}, \citenamefont {Westerhout}, \citenamefont {Zhang}, \citenamefont {Forati}, \citenamefont {Rossi}, \citenamefont {Kobrin}, \citenamefont {Paolo}, \citenamefont {Klots}, \citenamefont {Drozdov}, \citenamefont {Kurilovich}, \citenamefont {Petukhov}, \citenamefont {Ioffe}, \citenamefont {Elben}, \citenamefont {Rath}, \citenamefont {Vitale}, \citenamefont {Vermersch}, \citenamefont {Acharya}, \citenamefont {Beni}, \citenamefont {Anderson}, \citenamefont {Ansmann}, \citenamefont {Arute}, \citenamefont {Arya}, \citenamefont {Asfaw}, \citenamefont {Atalaya}, \citenamefont {Ballard}, \citenamefont {Bardin}, \citenamefont {Bengtsson}, \citenamefont {Bilmes}, \citenamefont {Bortoli}, \citenamefont {Bourassa}, \citenamefont {Bovaird},
  \citenamefont {Brill}, \citenamefont {Broughton}, \citenamefont {Browne}, \citenamefont {Buchea}, \citenamefont {Buckley}, \citenamefont {Buell}, \citenamefont {Burger}, \citenamefont {Burkett}, \citenamefont {Bushnell}, \citenamefont {Cabrera}, \citenamefont {Campero}, \citenamefont {Chang}, \citenamefont {Chen}, \citenamefont {Chiaro}, \citenamefont {Claes}, \citenamefont {Cleland}, \citenamefont {Cogan}, \citenamefont {Collins}, \citenamefont {Conner}, \citenamefont {Courtney}, \citenamefont {Crook}, \citenamefont {Das}, \citenamefont {Debroy}, \citenamefont {Lorenzo}, \citenamefont {Barba}, \citenamefont {Demura}, \citenamefont {Donohoe}, \citenamefont {Dunsworth}, \citenamefont {Earle}, \citenamefont {Eickbusch}, \citenamefont {Elbag}, \citenamefont {Elzouka}, \citenamefont {Erickson}, \citenamefont {Faoro}, \citenamefont {Fatemi}, \citenamefont {Ferreira}, \citenamefont {Burgos}, \citenamefont {Fowler}, \citenamefont {Foxen}, \citenamefont {Ganjam}, \citenamefont {Gasca}, \citenamefont {Giang},
  \citenamefont {Gidney}, \citenamefont {Gilboa}, \citenamefont {Giustina}, \citenamefont {Gosula}, \citenamefont {Dau}, \citenamefont {Graumann}, \citenamefont {Greene}, \citenamefont {Habegger}, \citenamefont {Hamilton}, \citenamefont {Hansen}, \citenamefont {Harrigan}, \citenamefont {Harrington}, \citenamefont {Heslin}, \citenamefont {Heu}, \citenamefont {Hill}, \citenamefont {Hoffmann}, \citenamefont {Huang}, \citenamefont {Huang}, \citenamefont {Huff}, \citenamefont {Huggins}, \citenamefont {Isakov}, \citenamefont {Jeffrey}, \citenamefont {Jiang}, \citenamefont {Jones}, \citenamefont {Jordan}, \citenamefont {Joshi}, \citenamefont {Juhas}, \citenamefont {Kafri}, \citenamefont {Kang}, \citenamefont {Kechedzhi}, \citenamefont {Khaire}, \citenamefont {Khattar}, \citenamefont {Khezri}, \citenamefont {Kieferov{\'a}}, \citenamefont {Kim}, \citenamefont {Kitaev}, \citenamefont {Klimov}, \citenamefont {Korotkov}, \citenamefont {Kostritsa}, \citenamefont {Kreikebaum}, \citenamefont {Landhuis}, \citenamefont
  {Langley}, \citenamefont {Laptev}, \citenamefont {Lau}, \citenamefont {Guevel}, \citenamefont {Ledford}, \citenamefont {Lee}, \citenamefont {Lee}, \citenamefont {Lensky}, \citenamefont {Lester}, \citenamefont {Li}, \citenamefont {Lill}, \citenamefont {Liu}, \citenamefont {Livingston}, \citenamefont {Locharla}, \citenamefont {Lundahl}, \citenamefont {Lunt}, \citenamefont {Madhuk}, \citenamefont {Maloney}, \citenamefont {Mandr{\`a}}, \citenamefont {Martin}, \citenamefont {Martin}, \citenamefont {Martin}, \citenamefont {Maxfield}, \citenamefont {McClean}, \citenamefont {McEwen}, \citenamefont {Meeks}, \citenamefont {Miao}, \citenamefont {Mieszala}, \citenamefont {Molina}, \citenamefont {Montazeri}, \citenamefont {Morvan}, \citenamefont {Movassagh}, \citenamefont {Neill}, \citenamefont {Nersisyan}, \citenamefont {Newman}, \citenamefont {Nguyen}, \citenamefont {Nguyen}, \citenamefont {Ni}, \citenamefont {Niu}, \citenamefont {Oliver}, \citenamefont {Ottosson}, \citenamefont {Pizzuto}, \citenamefont {Potter},
  \citenamefont {Pritchard}, \citenamefont {Pryadko}, \citenamefont {Quintana}, \citenamefont {Reagor}, \citenamefont {Rhodes}, \citenamefont {Roberts}, \citenamefont {Rocque}, \citenamefont {Rosenberg}, \citenamefont {Rubin}, \citenamefont {Saei}, \citenamefont {Sankaragomathi}, \citenamefont {Satzinger}, \citenamefont {Schurkus}, \citenamefont {Schuster}, \citenamefont {Shearn}, \citenamefont {Shorter}, \citenamefont {Shutty}, \citenamefont {Shvarts}, \citenamefont {Sivak}, \citenamefont {Skruzny}, \citenamefont {Small}, \citenamefont {Smith}, \citenamefont {Springer}, \citenamefont {Sterling}, \citenamefont {Suchard}, \citenamefont {Szalay}, \citenamefont {Sztein}, \citenamefont {Thor}, \citenamefont {Torres}, \citenamefont {Torunbalci}, \citenamefont {Vaishnav}, \citenamefont {Vdovichev}, \citenamefont {Villalonga}, \citenamefont {Heidweiller}, \citenamefont {Waltman}, \citenamefont {Wang}, \citenamefont {White}, \citenamefont {Wong}, \citenamefont {Woo}, \citenamefont {Xing}, \citenamefont {Yao},
  \citenamefont {Yeh}, \citenamefont {Ying}, \citenamefont {Yoo}, \citenamefont {Yosri}, \citenamefont {Young}, \citenamefont {Zalcman}, \citenamefont {Zhu}, \citenamefont {Zobrist}, \citenamefont {Neven}, \citenamefont {Babbush}, \citenamefont {Boixo}, \citenamefont {Hilton}, \citenamefont {Lucero}, \citenamefont {Megrant}, \citenamefont {Kelly}, \citenamefont {Chen}, \citenamefont {Smelyanskiy}, \citenamefont {Vidal}, \citenamefont {Roushan}, \citenamefont {L{\"a}uchli}, \citenamefont {Abanin},\ and\ \citenamefont {Mi}}]{google_digital_analog}%
  \BibitemOpen
  \bibfield  {author} {\bibinfo {author} {\bibfnamefont {T.~I.}\ \bibnamefont {Andersen}}, \bibinfo {author} {\bibfnamefont {N.}~\bibnamefont {Astrakhantsev}}, \bibinfo {author} {\bibfnamefont {A.~H.}\ \bibnamefont {Karamlou}}, \bibinfo {author} {\bibfnamefont {J.}~\bibnamefont {Berndtsson}}, \bibinfo {author} {\bibfnamefont {J.}~\bibnamefont {Motruk}}, \bibinfo {author} {\bibfnamefont {A.}~\bibnamefont {Szasz}}, \bibinfo {author} {\bibfnamefont {J.~A.}\ \bibnamefont {Gross}}, \bibinfo {author} {\bibfnamefont {A.}~\bibnamefont {Schuckert}}, \bibinfo {author} {\bibfnamefont {T.}~\bibnamefont {Westerhout}}, \bibinfo {author} {\bibfnamefont {Y.}~\bibnamefont {Zhang}}, \bibinfo {author} {\bibfnamefont {E.}~\bibnamefont {Forati}}, \bibinfo {author} {\bibfnamefont {D.}~\bibnamefont {Rossi}}, \bibinfo {author} {\bibfnamefont {B.}~\bibnamefont {Kobrin}}, \bibinfo {author} {\bibfnamefont {A.~D.}\ \bibnamefont {Paolo}}, \bibinfo {author} {\bibfnamefont {A.~R.}\ \bibnamefont {Klots}}, \bibinfo {author} {\bibfnamefont
  {I.}~\bibnamefont {Drozdov}}, \bibinfo {author} {\bibfnamefont {V.}~\bibnamefont {Kurilovich}}, \bibinfo {author} {\bibfnamefont {A.}~\bibnamefont {Petukhov}}, \bibinfo {author} {\bibfnamefont {L.~B.}\ \bibnamefont {Ioffe}}, \bibinfo {author} {\bibfnamefont {A.}~\bibnamefont {Elben}}, \bibinfo {author} {\bibfnamefont {A.}~\bibnamefont {Rath}}, \bibinfo {author} {\bibfnamefont {V.}~\bibnamefont {Vitale}}, \bibinfo {author} {\bibfnamefont {B.}~\bibnamefont {Vermersch}}, \bibinfo {author} {\bibfnamefont {R.}~\bibnamefont {Acharya}}, \bibinfo {author} {\bibfnamefont {L.~A.}\ \bibnamefont {Beni}}, \bibinfo {author} {\bibfnamefont {K.}~\bibnamefont {Anderson}}, \bibinfo {author} {\bibfnamefont {M.}~\bibnamefont {Ansmann}}, \bibinfo {author} {\bibfnamefont {F.}~\bibnamefont {Arute}}, \bibinfo {author} {\bibfnamefont {K.}~\bibnamefont {Arya}}, \bibinfo {author} {\bibfnamefont {A.}~\bibnamefont {Asfaw}}, \bibinfo {author} {\bibfnamefont {J.}~\bibnamefont {Atalaya}}, \bibinfo {author} {\bibfnamefont {B.}~\bibnamefont
  {Ballard}}, \bibinfo {author} {\bibfnamefont {J.~C.}\ \bibnamefont {Bardin}}, \bibinfo {author} {\bibfnamefont {A.}~\bibnamefont {Bengtsson}}, \bibinfo {author} {\bibfnamefont {A.}~\bibnamefont {Bilmes}}, \bibinfo {author} {\bibfnamefont {G.}~\bibnamefont {Bortoli}}, \bibinfo {author} {\bibfnamefont {A.}~\bibnamefont {Bourassa}}, \bibinfo {author} {\bibfnamefont {J.}~\bibnamefont {Bovaird}}, \bibinfo {author} {\bibfnamefont {L.}~\bibnamefont {Brill}}, \bibinfo {author} {\bibfnamefont {M.}~\bibnamefont {Broughton}}, \bibinfo {author} {\bibfnamefont {D.~A.}\ \bibnamefont {Browne}}, \bibinfo {author} {\bibfnamefont {B.}~\bibnamefont {Buchea}}, \bibinfo {author} {\bibfnamefont {B.~B.}\ \bibnamefont {Buckley}}, \bibinfo {author} {\bibfnamefont {D.~A.}\ \bibnamefont {Buell}}, \bibinfo {author} {\bibfnamefont {T.}~\bibnamefont {Burger}}, \bibinfo {author} {\bibfnamefont {B.}~\bibnamefont {Burkett}}, \bibinfo {author} {\bibfnamefont {N.}~\bibnamefont {Bushnell}}, \bibinfo {author} {\bibfnamefont {A.}~\bibnamefont
  {Cabrera}}, \bibinfo {author} {\bibfnamefont {J.}~\bibnamefont {Campero}}, \bibinfo {author} {\bibfnamefont {H.~S.}\ \bibnamefont {Chang}}, \bibinfo {author} {\bibfnamefont {Z.}~\bibnamefont {Chen}}, \bibinfo {author} {\bibfnamefont {B.}~\bibnamefont {Chiaro}}, \bibinfo {author} {\bibfnamefont {J.}~\bibnamefont {Claes}}, \bibinfo {author} {\bibfnamefont {A.~Y.}\ \bibnamefont {Cleland}}, \bibinfo {author} {\bibfnamefont {J.}~\bibnamefont {Cogan}}, \bibinfo {author} {\bibfnamefont {R.}~\bibnamefont {Collins}}, \bibinfo {author} {\bibfnamefont {P.}~\bibnamefont {Conner}}, \bibinfo {author} {\bibfnamefont {W.}~\bibnamefont {Courtney}}, \bibinfo {author} {\bibfnamefont {A.~L.}\ \bibnamefont {Crook}}, \bibinfo {author} {\bibfnamefont {S.}~\bibnamefont {Das}}, \bibinfo {author} {\bibfnamefont {D.~M.}\ \bibnamefont {Debroy}}, \bibinfo {author} {\bibfnamefont {L.~D.}\ \bibnamefont {Lorenzo}}, \bibinfo {author} {\bibfnamefont {A.~D.~T.}\ \bibnamefont {Barba}}, \bibinfo {author} {\bibfnamefont {S.}~\bibnamefont
  {Demura}}, \bibinfo {author} {\bibfnamefont {P.}~\bibnamefont {Donohoe}}, \bibinfo {author} {\bibfnamefont {A.}~\bibnamefont {Dunsworth}}, \bibinfo {author} {\bibfnamefont {C.}~\bibnamefont {Earle}}, \bibinfo {author} {\bibfnamefont {A.}~\bibnamefont {Eickbusch}}, \bibinfo {author} {\bibfnamefont {A.~M.}\ \bibnamefont {Elbag}}, \bibinfo {author} {\bibfnamefont {M.}~\bibnamefont {Elzouka}}, \bibinfo {author} {\bibfnamefont {C.}~\bibnamefont {Erickson}}, \bibinfo {author} {\bibfnamefont {L.}~\bibnamefont {Faoro}}, \bibinfo {author} {\bibfnamefont {R.}~\bibnamefont {Fatemi}}, \bibinfo {author} {\bibfnamefont {V.~S.}\ \bibnamefont {Ferreira}}, \bibinfo {author} {\bibfnamefont {L.~F.}\ \bibnamefont {Burgos}}, \bibinfo {author} {\bibfnamefont {A.~G.}\ \bibnamefont {Fowler}}, \bibinfo {author} {\bibfnamefont {B.}~\bibnamefont {Foxen}}, \bibinfo {author} {\bibfnamefont {S.}~\bibnamefont {Ganjam}}, \bibinfo {author} {\bibfnamefont {R.}~\bibnamefont {Gasca}}, \bibinfo {author} {\bibfnamefont {W.}~\bibnamefont
  {Giang}}, \bibinfo {author} {\bibfnamefont {C.}~\bibnamefont {Gidney}}, \bibinfo {author} {\bibfnamefont {D.}~\bibnamefont {Gilboa}}, \bibinfo {author} {\bibfnamefont {M.}~\bibnamefont {Giustina}}, \bibinfo {author} {\bibfnamefont {R.}~\bibnamefont {Gosula}}, \bibinfo {author} {\bibfnamefont {A.~G.}\ \bibnamefont {Dau}}, \bibinfo {author} {\bibfnamefont {D.}~\bibnamefont {Graumann}}, \bibinfo {author} {\bibfnamefont {A.}~\bibnamefont {Greene}}, \bibinfo {author} {\bibfnamefont {S.}~\bibnamefont {Habegger}}, \bibinfo {author} {\bibfnamefont {M.~C.}\ \bibnamefont {Hamilton}}, \bibinfo {author} {\bibfnamefont {M.}~\bibnamefont {Hansen}}, \bibinfo {author} {\bibfnamefont {M.~P.}\ \bibnamefont {Harrigan}}, \bibinfo {author} {\bibfnamefont {S.~D.}\ \bibnamefont {Harrington}}, \bibinfo {author} {\bibfnamefont {S.}~\bibnamefont {Heslin}}, \bibinfo {author} {\bibfnamefont {P.}~\bibnamefont {Heu}}, \bibinfo {author} {\bibfnamefont {G.}~\bibnamefont {Hill}}, \bibinfo {author} {\bibfnamefont {M.~R.}\ \bibnamefont
  {Hoffmann}}, \bibinfo {author} {\bibfnamefont {H.~Y.}\ \bibnamefont {Huang}}, \bibinfo {author} {\bibfnamefont {T.}~\bibnamefont {Huang}}, \bibinfo {author} {\bibfnamefont {A.}~\bibnamefont {Huff}}, \bibinfo {author} {\bibfnamefont {W.~J.}\ \bibnamefont {Huggins}}, \bibinfo {author} {\bibfnamefont {S.~V.}\ \bibnamefont {Isakov}}, \bibinfo {author} {\bibfnamefont {E.}~\bibnamefont {Jeffrey}}, \bibinfo {author} {\bibfnamefont {Z.}~\bibnamefont {Jiang}}, \bibinfo {author} {\bibfnamefont {C.}~\bibnamefont {Jones}}, \bibinfo {author} {\bibfnamefont {S.}~\bibnamefont {Jordan}}, \bibinfo {author} {\bibfnamefont {C.}~\bibnamefont {Joshi}}, \bibinfo {author} {\bibfnamefont {P.}~\bibnamefont {Juhas}}, \bibinfo {author} {\bibfnamefont {D.}~\bibnamefont {Kafri}}, \bibinfo {author} {\bibfnamefont {H.}~\bibnamefont {Kang}}, \bibinfo {author} {\bibfnamefont {K.}~\bibnamefont {Kechedzhi}}, \bibinfo {author} {\bibfnamefont {T.}~\bibnamefont {Khaire}}, \bibinfo {author} {\bibfnamefont {T.}~\bibnamefont {Khattar}}, \bibinfo
  {author} {\bibfnamefont {M.}~\bibnamefont {Khezri}}, \bibinfo {author} {\bibfnamefont {M.}~\bibnamefont {Kieferov{\'a}}}, \bibinfo {author} {\bibfnamefont {S.}~\bibnamefont {Kim}}, \bibinfo {author} {\bibfnamefont {A.}~\bibnamefont {Kitaev}}, \bibinfo {author} {\bibfnamefont {P.}~\bibnamefont {Klimov}}, \bibinfo {author} {\bibfnamefont {A.~N.}\ \bibnamefont {Korotkov}}, \bibinfo {author} {\bibfnamefont {F.}~\bibnamefont {Kostritsa}}, \bibinfo {author} {\bibfnamefont {J.~M.}\ \bibnamefont {Kreikebaum}}, \bibinfo {author} {\bibfnamefont {D.}~\bibnamefont {Landhuis}}, \bibinfo {author} {\bibfnamefont {B.~W.}\ \bibnamefont {Langley}}, \bibinfo {author} {\bibfnamefont {P.}~\bibnamefont {Laptev}}, \bibinfo {author} {\bibfnamefont {K.~M.}\ \bibnamefont {Lau}}, \bibinfo {author} {\bibfnamefont {L.~L.}\ \bibnamefont {Guevel}}, \bibinfo {author} {\bibfnamefont {J.}~\bibnamefont {Ledford}}, \bibinfo {author} {\bibfnamefont {J.}~\bibnamefont {Lee}}, \bibinfo {author} {\bibfnamefont {K.~W.}\ \bibnamefont {Lee}},
  \bibinfo {author} {\bibfnamefont {Y.~D.}\ \bibnamefont {Lensky}}, \bibinfo {author} {\bibfnamefont {B.~J.}\ \bibnamefont {Lester}}, \bibinfo {author} {\bibfnamefont {W.~Y.}\ \bibnamefont {Li}}, \bibinfo {author} {\bibfnamefont {A.~T.}\ \bibnamefont {Lill}}, \bibinfo {author} {\bibfnamefont {W.}~\bibnamefont {Liu}}, \bibinfo {author} {\bibfnamefont {W.~P.}\ \bibnamefont {Livingston}}, \bibinfo {author} {\bibfnamefont {A.}~\bibnamefont {Locharla}}, \bibinfo {author} {\bibfnamefont {D.}~\bibnamefont {Lundahl}}, \bibinfo {author} {\bibfnamefont {A.}~\bibnamefont {Lunt}}, \bibinfo {author} {\bibfnamefont {S.}~\bibnamefont {Madhuk}}, \bibinfo {author} {\bibfnamefont {A.}~\bibnamefont {Maloney}}, \bibinfo {author} {\bibfnamefont {S.}~\bibnamefont {Mandr{\`a}}}, \bibinfo {author} {\bibfnamefont {L.~S.}\ \bibnamefont {Martin}}, \bibinfo {author} {\bibfnamefont {O.}~\bibnamefont {Martin}}, \bibinfo {author} {\bibfnamefont {S.}~\bibnamefont {Martin}}, \bibinfo {author} {\bibfnamefont {C.}~\bibnamefont {Maxfield}},
  \bibinfo {author} {\bibfnamefont {J.~R.}\ \bibnamefont {McClean}}, \bibinfo {author} {\bibfnamefont {M.}~\bibnamefont {McEwen}}, \bibinfo {author} {\bibfnamefont {S.}~\bibnamefont {Meeks}}, \bibinfo {author} {\bibfnamefont {K.~C.}\ \bibnamefont {Miao}}, \bibinfo {author} {\bibfnamefont {A.}~\bibnamefont {Mieszala}}, \bibinfo {author} {\bibfnamefont {S.}~\bibnamefont {Molina}}, \bibinfo {author} {\bibfnamefont {S.}~\bibnamefont {Montazeri}}, \bibinfo {author} {\bibfnamefont {A.}~\bibnamefont {Morvan}}, \bibinfo {author} {\bibfnamefont {R.}~\bibnamefont {Movassagh}}, \bibinfo {author} {\bibfnamefont {C.}~\bibnamefont {Neill}}, \bibinfo {author} {\bibfnamefont {A.}~\bibnamefont {Nersisyan}}, \bibinfo {author} {\bibfnamefont {M.}~\bibnamefont {Newman}}, \bibinfo {author} {\bibfnamefont {A.}~\bibnamefont {Nguyen}}, \bibinfo {author} {\bibfnamefont {M.}~\bibnamefont {Nguyen}}, \bibinfo {author} {\bibfnamefont {C.~H.}\ \bibnamefont {Ni}}, \bibinfo {author} {\bibfnamefont {M.~Y.}\ \bibnamefont {Niu}}, \bibinfo
  {author} {\bibfnamefont {W.~D.}\ \bibnamefont {Oliver}}, \bibinfo {author} {\bibfnamefont {K.}~\bibnamefont {Ottosson}}, \bibinfo {author} {\bibfnamefont {A.}~\bibnamefont {Pizzuto}}, \bibinfo {author} {\bibfnamefont {R.}~\bibnamefont {Potter}}, \bibinfo {author} {\bibfnamefont {O.}~\bibnamefont {Pritchard}}, \bibinfo {author} {\bibfnamefont {L.~P.}\ \bibnamefont {Pryadko}}, \bibinfo {author} {\bibfnamefont {C.}~\bibnamefont {Quintana}}, \bibinfo {author} {\bibfnamefont {M.~J.}\ \bibnamefont {Reagor}}, \bibinfo {author} {\bibfnamefont {D.~M.}\ \bibnamefont {Rhodes}}, \bibinfo {author} {\bibfnamefont {G.}~\bibnamefont {Roberts}}, \bibinfo {author} {\bibfnamefont {C.}~\bibnamefont {Rocque}}, \bibinfo {author} {\bibfnamefont {E.}~\bibnamefont {Rosenberg}}, \bibinfo {author} {\bibfnamefont {N.~C.}\ \bibnamefont {Rubin}}, \bibinfo {author} {\bibfnamefont {N.}~\bibnamefont {Saei}}, \bibinfo {author} {\bibfnamefont {K.}~\bibnamefont {Sankaragomathi}}, \bibinfo {author} {\bibfnamefont {K.~J.}\ \bibnamefont
  {Satzinger}}, \bibinfo {author} {\bibfnamefont {H.~F.}\ \bibnamefont {Schurkus}}, \bibinfo {author} {\bibfnamefont {C.}~\bibnamefont {Schuster}}, \bibinfo {author} {\bibfnamefont {M.~J.}\ \bibnamefont {Shearn}}, \bibinfo {author} {\bibfnamefont {A.}~\bibnamefont {Shorter}}, \bibinfo {author} {\bibfnamefont {N.}~\bibnamefont {Shutty}}, \bibinfo {author} {\bibfnamefont {V.}~\bibnamefont {Shvarts}}, \bibinfo {author} {\bibfnamefont {V.}~\bibnamefont {Sivak}}, \bibinfo {author} {\bibfnamefont {J.}~\bibnamefont {Skruzny}}, \bibinfo {author} {\bibfnamefont {S.}~\bibnamefont {Small}}, \bibinfo {author} {\bibfnamefont {W.~C.}\ \bibnamefont {Smith}}, \bibinfo {author} {\bibfnamefont {S.}~\bibnamefont {Springer}}, \bibinfo {author} {\bibfnamefont {G.}~\bibnamefont {Sterling}}, \bibinfo {author} {\bibfnamefont {J.}~\bibnamefont {Suchard}}, \bibinfo {author} {\bibfnamefont {M.}~\bibnamefont {Szalay}}, \bibinfo {author} {\bibfnamefont {A.}~\bibnamefont {Sztein}}, \bibinfo {author} {\bibfnamefont {D.}~\bibnamefont
  {Thor}}, \bibinfo {author} {\bibfnamefont {A.}~\bibnamefont {Torres}}, \bibinfo {author} {\bibfnamefont {M.~M.}\ \bibnamefont {Torunbalci}}, \bibinfo {author} {\bibfnamefont {A.}~\bibnamefont {Vaishnav}}, \bibinfo {author} {\bibfnamefont {S.}~\bibnamefont {Vdovichev}}, \bibinfo {author} {\bibfnamefont {B.}~\bibnamefont {Villalonga}}, \bibinfo {author} {\bibfnamefont {C.~V.}\ \bibnamefont {Heidweiller}}, \bibinfo {author} {\bibfnamefont {S.}~\bibnamefont {Waltman}}, \bibinfo {author} {\bibfnamefont {S.~X.}\ \bibnamefont {Wang}}, \bibinfo {author} {\bibfnamefont {T.}~\bibnamefont {White}}, \bibinfo {author} {\bibfnamefont {K.}~\bibnamefont {Wong}}, \bibinfo {author} {\bibfnamefont {B.~W.~K.}\ \bibnamefont {Woo}}, \bibinfo {author} {\bibfnamefont {C.}~\bibnamefont {Xing}}, \bibinfo {author} {\bibfnamefont {Z.~J.}\ \bibnamefont {Yao}}, \bibinfo {author} {\bibfnamefont {P.}~\bibnamefont {Yeh}}, \bibinfo {author} {\bibfnamefont {B.}~\bibnamefont {Ying}}, \bibinfo {author} {\bibfnamefont {J.}~\bibnamefont {Yoo}},
  \bibinfo {author} {\bibfnamefont {N.}~\bibnamefont {Yosri}}, \bibinfo {author} {\bibfnamefont {G.}~\bibnamefont {Young}}, \bibinfo {author} {\bibfnamefont {A.}~\bibnamefont {Zalcman}}, \bibinfo {author} {\bibfnamefont {N.}~\bibnamefont {Zhu}}, \bibinfo {author} {\bibfnamefont {N.}~\bibnamefont {Zobrist}}, \bibinfo {author} {\bibfnamefont {H.}~\bibnamefont {Neven}}, \bibinfo {author} {\bibfnamefont {R.}~\bibnamefont {Babbush}}, \bibinfo {author} {\bibfnamefont {S.}~\bibnamefont {Boixo}}, \bibinfo {author} {\bibfnamefont {J.}~\bibnamefont {Hilton}}, \bibinfo {author} {\bibfnamefont {E.}~\bibnamefont {Lucero}}, \bibinfo {author} {\bibfnamefont {A.}~\bibnamefont {Megrant}}, \bibinfo {author} {\bibfnamefont {J.}~\bibnamefont {Kelly}}, \bibinfo {author} {\bibfnamefont {Y.}~\bibnamefont {Chen}}, \bibinfo {author} {\bibfnamefont {V.}~\bibnamefont {Smelyanskiy}}, \bibinfo {author} {\bibfnamefont {G.}~\bibnamefont {Vidal}}, \bibinfo {author} {\bibfnamefont {P.}~\bibnamefont {Roushan}}, \bibinfo {author}
  {\bibfnamefont {A.~M.}\ \bibnamefont {L{\"a}uchli}}, \bibinfo {author} {\bibfnamefont {D.~A.}\ \bibnamefont {Abanin}},\ and\ \bibinfo {author} {\bibfnamefont {X.}~\bibnamefont {Mi}},\ }\href {https://doi.org/10.1038/s41586-024-08460-3} {\bibfield  {journal} {\bibinfo  {journal} {Nature}\ }\textbf {\bibinfo {volume} {638}},\ \bibinfo {pages} {79} (\bibinfo {year} {2025}{\natexlab{a}})}\BibitemShut {NoStop}%
\bibitem [{\citenamefont {Andersen}\ \emph {et~al.}(2025{\natexlab{b}})\citenamefont {Andersen}, \citenamefont {Astrakhantsev}, \citenamefont {Karamlou}, \citenamefont {Berndtsson}, \citenamefont {Motruk}, \citenamefont {Szasz}, \citenamefont {Gross}, \citenamefont {Schuckert}, \citenamefont {Westerhout}, \citenamefont {Zhang}, \citenamefont {Forati}, \citenamefont {Rossi}, \citenamefont {Kobrin}, \citenamefont {Paolo}, \citenamefont {Klots}, \citenamefont {Drozdov}, \citenamefont {Kurilovich}, \citenamefont {Petukhov}, \citenamefont {Ioffe}, \citenamefont {Elben}, \citenamefont {Rath}, \citenamefont {Vitale}, \citenamefont {Vermersch}, \citenamefont {Acharya}, \citenamefont {Beni}, \citenamefont {Anderson}, \citenamefont {Ansmann}, \citenamefont {Arute}, \citenamefont {Arya}, \citenamefont {Asfaw}, \citenamefont {Atalaya}, \citenamefont {Ballard}, \citenamefont {Bardin}, \citenamefont {Bengtsson}, \citenamefont {Bilmes}, \citenamefont {Bortoli}, \citenamefont {Bourassa}, \citenamefont {Bovaird},
  \citenamefont {Brill}, \citenamefont {Broughton}, \citenamefont {Browne}, \citenamefont {Buchea}, \citenamefont {Buckley}, \citenamefont {Buell}, \citenamefont {Burger}, \citenamefont {Burkett}, \citenamefont {Bushnell}, \citenamefont {Cabrera}, \citenamefont {Campero}, \citenamefont {Chang}, \citenamefont {Chen}, \citenamefont {Chiaro}, \citenamefont {Claes}, \citenamefont {Cleland}, \citenamefont {Cogan}, \citenamefont {Collins}, \citenamefont {Conner}, \citenamefont {Courtney}, \citenamefont {Crook}, \citenamefont {Das}, \citenamefont {Debroy}, \citenamefont {Lorenzo}, \citenamefont {Barba}, \citenamefont {Demura}, \citenamefont {Donohoe}, \citenamefont {Dunsworth}, \citenamefont {Earle}, \citenamefont {Eickbusch}, \citenamefont {Elbag}, \citenamefont {Elzouka}, \citenamefont {Erickson}, \citenamefont {Faoro}, \citenamefont {Fatemi}, \citenamefont {Ferreira}, \citenamefont {Burgos}, \citenamefont {Fowler}, \citenamefont {Foxen}, \citenamefont {Ganjam}, \citenamefont {Gasca}, \citenamefont {Giang},
  \citenamefont {Gidney}, \citenamefont {Gilboa}, \citenamefont {Giustina}, \citenamefont {Gosula}, \citenamefont {Dau}, \citenamefont {Graumann}, \citenamefont {Greene}, \citenamefont {Habegger}, \citenamefont {Hamilton}, \citenamefont {Hansen}, \citenamefont {Harrigan}, \citenamefont {Harrington}, \citenamefont {Heslin}, \citenamefont {Heu}, \citenamefont {Hill}, \citenamefont {Hoffmann}, \citenamefont {Huang}, \citenamefont {Huang}, \citenamefont {Huff}, \citenamefont {Huggins}, \citenamefont {Isakov}, \citenamefont {Jeffrey}, \citenamefont {Jiang}, \citenamefont {Jones}, \citenamefont {Jordan}, \citenamefont {Joshi}, \citenamefont {Juhas}, \citenamefont {Kafri}, \citenamefont {Kang}, \citenamefont {Kechedzhi}, \citenamefont {Khaire}, \citenamefont {Khattar}, \citenamefont {Khezri}, \citenamefont {Kieferov{\'a}}, \citenamefont {Kim}, \citenamefont {Kitaev}, \citenamefont {Klimov}, \citenamefont {Korotkov}, \citenamefont {Kostritsa}, \citenamefont {Kreikebaum}, \citenamefont {Landhuis}, \citenamefont
  {Langley}, \citenamefont {Laptev}, \citenamefont {Lau}, \citenamefont {Guevel}, \citenamefont {Ledford}, \citenamefont {Lee}, \citenamefont {Lee}, \citenamefont {Lensky}, \citenamefont {Lester}, \citenamefont {Li}, \citenamefont {Lill}, \citenamefont {Liu}, \citenamefont {Livingston}, \citenamefont {Locharla}, \citenamefont {Lundahl}, \citenamefont {Lunt}, \citenamefont {Madhuk}, \citenamefont {Maloney}, \citenamefont {Mandr{\`a}}, \citenamefont {Martin}, \citenamefont {Martin}, \citenamefont {Martin}, \citenamefont {Maxfield}, \citenamefont {McClean}, \citenamefont {McEwen}, \citenamefont {Meeks}, \citenamefont {Miao}, \citenamefont {Mieszala}, \citenamefont {Molina}, \citenamefont {Montazeri}, \citenamefont {Morvan}, \citenamefont {Movassagh}, \citenamefont {Neill}, \citenamefont {Nersisyan}, \citenamefont {Newman}, \citenamefont {Nguyen}, \citenamefont {Nguyen}, \citenamefont {Ni}, \citenamefont {Niu}, \citenamefont {Oliver}, \citenamefont {Ottosson}, \citenamefont {Pizzuto}, \citenamefont {Potter},
  \citenamefont {Pritchard}, \citenamefont {Pryadko}, \citenamefont {Quintana}, \citenamefont {Reagor}, \citenamefont {Rhodes}, \citenamefont {Roberts}, \citenamefont {Rocque}, \citenamefont {Rosenberg}, \citenamefont {Rubin}, \citenamefont {Saei}, \citenamefont {Sankaragomathi}, \citenamefont {Satzinger}, \citenamefont {Schurkus}, \citenamefont {Schuster}, \citenamefont {Shearn}, \citenamefont {Shorter}, \citenamefont {Shutty}, \citenamefont {Shvarts}, \citenamefont {Sivak}, \citenamefont {Skruzny}, \citenamefont {Small}, \citenamefont {Smith}, \citenamefont {Springer}, \citenamefont {Sterling}, \citenamefont {Suchard}, \citenamefont {Szalay}, \citenamefont {Sztein}, \citenamefont {Thor}, \citenamefont {Torres}, \citenamefont {Torunbalci}, \citenamefont {Vaishnav}, \citenamefont {Vdovichev}, \citenamefont {Villalonga}, \citenamefont {Heidweiller}, \citenamefont {Waltman}, \citenamefont {Wang}, \citenamefont {White}, \citenamefont {Wong}, \citenamefont {Woo}, \citenamefont {Xing}, \citenamefont {Yao},
  \citenamefont {Yeh}, \citenamefont {Ying}, \citenamefont {Yoo}, \citenamefont {Yosri}, \citenamefont {Young}, \citenamefont {Zalcman}, \citenamefont {Zhu}, \citenamefont {Zobrist}, \citenamefont {Neven}, \citenamefont {Babbush}, \citenamefont {Boixo}, \citenamefont {Hilton}, \citenamefont {Lucero}, \citenamefont {Megrant}, \citenamefont {Kelly}, \citenamefont {Chen}, \citenamefont {Smelyanskiy}, \citenamefont {Vidal}, \citenamefont {Roushan}, \citenamefont {L{\"a}uchli}, \citenamefont {Abanin},\ and\ \citenamefont {Mi}}]{thermalization}%
  \BibitemOpen
  \bibfield  {author} {\bibinfo {author} {\bibfnamefont {T.~I.}\ \bibnamefont {Andersen}}, \bibinfo {author} {\bibfnamefont {N.}~\bibnamefont {Astrakhantsev}}, \bibinfo {author} {\bibfnamefont {A.~H.}\ \bibnamefont {Karamlou}}, \bibinfo {author} {\bibfnamefont {J.}~\bibnamefont {Berndtsson}}, \bibinfo {author} {\bibfnamefont {J.}~\bibnamefont {Motruk}}, \bibinfo {author} {\bibfnamefont {A.}~\bibnamefont {Szasz}}, \bibinfo {author} {\bibfnamefont {J.~A.}\ \bibnamefont {Gross}}, \bibinfo {author} {\bibfnamefont {A.}~\bibnamefont {Schuckert}}, \bibinfo {author} {\bibfnamefont {T.}~\bibnamefont {Westerhout}}, \bibinfo {author} {\bibfnamefont {Y.}~\bibnamefont {Zhang}}, \bibinfo {author} {\bibfnamefont {E.}~\bibnamefont {Forati}}, \bibinfo {author} {\bibfnamefont {D.}~\bibnamefont {Rossi}}, \bibinfo {author} {\bibfnamefont {B.}~\bibnamefont {Kobrin}}, \bibinfo {author} {\bibfnamefont {A.~D.}\ \bibnamefont {Paolo}}, \bibinfo {author} {\bibfnamefont {A.~R.}\ \bibnamefont {Klots}}, \bibinfo {author} {\bibfnamefont
  {I.}~\bibnamefont {Drozdov}}, \bibinfo {author} {\bibfnamefont {V.}~\bibnamefont {Kurilovich}}, \bibinfo {author} {\bibfnamefont {A.}~\bibnamefont {Petukhov}}, \bibinfo {author} {\bibfnamefont {L.~B.}\ \bibnamefont {Ioffe}}, \bibinfo {author} {\bibfnamefont {A.}~\bibnamefont {Elben}}, \bibinfo {author} {\bibfnamefont {A.}~\bibnamefont {Rath}}, \bibinfo {author} {\bibfnamefont {V.}~\bibnamefont {Vitale}}, \bibinfo {author} {\bibfnamefont {B.}~\bibnamefont {Vermersch}}, \bibinfo {author} {\bibfnamefont {R.}~\bibnamefont {Acharya}}, \bibinfo {author} {\bibfnamefont {L.~A.}\ \bibnamefont {Beni}}, \bibinfo {author} {\bibfnamefont {K.}~\bibnamefont {Anderson}}, \bibinfo {author} {\bibfnamefont {M.}~\bibnamefont {Ansmann}}, \bibinfo {author} {\bibfnamefont {F.}~\bibnamefont {Arute}}, \bibinfo {author} {\bibfnamefont {K.}~\bibnamefont {Arya}}, \bibinfo {author} {\bibfnamefont {A.}~\bibnamefont {Asfaw}}, \bibinfo {author} {\bibfnamefont {J.}~\bibnamefont {Atalaya}}, \bibinfo {author} {\bibfnamefont {B.}~\bibnamefont
  {Ballard}}, \bibinfo {author} {\bibfnamefont {J.~C.}\ \bibnamefont {Bardin}}, \bibinfo {author} {\bibfnamefont {A.}~\bibnamefont {Bengtsson}}, \bibinfo {author} {\bibfnamefont {A.}~\bibnamefont {Bilmes}}, \bibinfo {author} {\bibfnamefont {G.}~\bibnamefont {Bortoli}}, \bibinfo {author} {\bibfnamefont {A.}~\bibnamefont {Bourassa}}, \bibinfo {author} {\bibfnamefont {J.}~\bibnamefont {Bovaird}}, \bibinfo {author} {\bibfnamefont {L.}~\bibnamefont {Brill}}, \bibinfo {author} {\bibfnamefont {M.}~\bibnamefont {Broughton}}, \bibinfo {author} {\bibfnamefont {D.~A.}\ \bibnamefont {Browne}}, \bibinfo {author} {\bibfnamefont {B.}~\bibnamefont {Buchea}}, \bibinfo {author} {\bibfnamefont {B.~B.}\ \bibnamefont {Buckley}}, \bibinfo {author} {\bibfnamefont {D.~A.}\ \bibnamefont {Buell}}, \bibinfo {author} {\bibfnamefont {T.}~\bibnamefont {Burger}}, \bibinfo {author} {\bibfnamefont {B.}~\bibnamefont {Burkett}}, \bibinfo {author} {\bibfnamefont {N.}~\bibnamefont {Bushnell}}, \bibinfo {author} {\bibfnamefont {A.}~\bibnamefont
  {Cabrera}}, \bibinfo {author} {\bibfnamefont {J.}~\bibnamefont {Campero}}, \bibinfo {author} {\bibfnamefont {H.~S.}\ \bibnamefont {Chang}}, \bibinfo {author} {\bibfnamefont {Z.}~\bibnamefont {Chen}}, \bibinfo {author} {\bibfnamefont {B.}~\bibnamefont {Chiaro}}, \bibinfo {author} {\bibfnamefont {J.}~\bibnamefont {Claes}}, \bibinfo {author} {\bibfnamefont {A.~Y.}\ \bibnamefont {Cleland}}, \bibinfo {author} {\bibfnamefont {J.}~\bibnamefont {Cogan}}, \bibinfo {author} {\bibfnamefont {R.}~\bibnamefont {Collins}}, \bibinfo {author} {\bibfnamefont {P.}~\bibnamefont {Conner}}, \bibinfo {author} {\bibfnamefont {W.}~\bibnamefont {Courtney}}, \bibinfo {author} {\bibfnamefont {A.~L.}\ \bibnamefont {Crook}}, \bibinfo {author} {\bibfnamefont {S.}~\bibnamefont {Das}}, \bibinfo {author} {\bibfnamefont {D.~M.}\ \bibnamefont {Debroy}}, \bibinfo {author} {\bibfnamefont {L.~D.}\ \bibnamefont {Lorenzo}}, \bibinfo {author} {\bibfnamefont {A.~D.~T.}\ \bibnamefont {Barba}}, \bibinfo {author} {\bibfnamefont {S.}~\bibnamefont
  {Demura}}, \bibinfo {author} {\bibfnamefont {P.}~\bibnamefont {Donohoe}}, \bibinfo {author} {\bibfnamefont {A.}~\bibnamefont {Dunsworth}}, \bibinfo {author} {\bibfnamefont {C.}~\bibnamefont {Earle}}, \bibinfo {author} {\bibfnamefont {A.}~\bibnamefont {Eickbusch}}, \bibinfo {author} {\bibfnamefont {A.~M.}\ \bibnamefont {Elbag}}, \bibinfo {author} {\bibfnamefont {M.}~\bibnamefont {Elzouka}}, \bibinfo {author} {\bibfnamefont {C.}~\bibnamefont {Erickson}}, \bibinfo {author} {\bibfnamefont {L.}~\bibnamefont {Faoro}}, \bibinfo {author} {\bibfnamefont {R.}~\bibnamefont {Fatemi}}, \bibinfo {author} {\bibfnamefont {V.~S.}\ \bibnamefont {Ferreira}}, \bibinfo {author} {\bibfnamefont {L.~F.}\ \bibnamefont {Burgos}}, \bibinfo {author} {\bibfnamefont {A.~G.}\ \bibnamefont {Fowler}}, \bibinfo {author} {\bibfnamefont {B.}~\bibnamefont {Foxen}}, \bibinfo {author} {\bibfnamefont {S.}~\bibnamefont {Ganjam}}, \bibinfo {author} {\bibfnamefont {R.}~\bibnamefont {Gasca}}, \bibinfo {author} {\bibfnamefont {W.}~\bibnamefont
  {Giang}}, \bibinfo {author} {\bibfnamefont {C.}~\bibnamefont {Gidney}}, \bibinfo {author} {\bibfnamefont {D.}~\bibnamefont {Gilboa}}, \bibinfo {author} {\bibfnamefont {M.}~\bibnamefont {Giustina}}, \bibinfo {author} {\bibfnamefont {R.}~\bibnamefont {Gosula}}, \bibinfo {author} {\bibfnamefont {A.~G.}\ \bibnamefont {Dau}}, \bibinfo {author} {\bibfnamefont {D.}~\bibnamefont {Graumann}}, \bibinfo {author} {\bibfnamefont {A.}~\bibnamefont {Greene}}, \bibinfo {author} {\bibfnamefont {S.}~\bibnamefont {Habegger}}, \bibinfo {author} {\bibfnamefont {M.~C.}\ \bibnamefont {Hamilton}}, \bibinfo {author} {\bibfnamefont {M.}~\bibnamefont {Hansen}}, \bibinfo {author} {\bibfnamefont {M.~P.}\ \bibnamefont {Harrigan}}, \bibinfo {author} {\bibfnamefont {S.~D.}\ \bibnamefont {Harrington}}, \bibinfo {author} {\bibfnamefont {S.}~\bibnamefont {Heslin}}, \bibinfo {author} {\bibfnamefont {P.}~\bibnamefont {Heu}}, \bibinfo {author} {\bibfnamefont {G.}~\bibnamefont {Hill}}, \bibinfo {author} {\bibfnamefont {M.~R.}\ \bibnamefont
  {Hoffmann}}, \bibinfo {author} {\bibfnamefont {H.~Y.}\ \bibnamefont {Huang}}, \bibinfo {author} {\bibfnamefont {T.}~\bibnamefont {Huang}}, \bibinfo {author} {\bibfnamefont {A.}~\bibnamefont {Huff}}, \bibinfo {author} {\bibfnamefont {W.~J.}\ \bibnamefont {Huggins}}, \bibinfo {author} {\bibfnamefont {S.~V.}\ \bibnamefont {Isakov}}, \bibinfo {author} {\bibfnamefont {E.}~\bibnamefont {Jeffrey}}, \bibinfo {author} {\bibfnamefont {Z.}~\bibnamefont {Jiang}}, \bibinfo {author} {\bibfnamefont {C.}~\bibnamefont {Jones}}, \bibinfo {author} {\bibfnamefont {S.}~\bibnamefont {Jordan}}, \bibinfo {author} {\bibfnamefont {C.}~\bibnamefont {Joshi}}, \bibinfo {author} {\bibfnamefont {P.}~\bibnamefont {Juhas}}, \bibinfo {author} {\bibfnamefont {D.}~\bibnamefont {Kafri}}, \bibinfo {author} {\bibfnamefont {H.}~\bibnamefont {Kang}}, \bibinfo {author} {\bibfnamefont {K.}~\bibnamefont {Kechedzhi}}, \bibinfo {author} {\bibfnamefont {T.}~\bibnamefont {Khaire}}, \bibinfo {author} {\bibfnamefont {T.}~\bibnamefont {Khattar}}, \bibinfo
  {author} {\bibfnamefont {M.}~\bibnamefont {Khezri}}, \bibinfo {author} {\bibfnamefont {M.}~\bibnamefont {Kieferov{\'a}}}, \bibinfo {author} {\bibfnamefont {S.}~\bibnamefont {Kim}}, \bibinfo {author} {\bibfnamefont {A.}~\bibnamefont {Kitaev}}, \bibinfo {author} {\bibfnamefont {P.}~\bibnamefont {Klimov}}, \bibinfo {author} {\bibfnamefont {A.~N.}\ \bibnamefont {Korotkov}}, \bibinfo {author} {\bibfnamefont {F.}~\bibnamefont {Kostritsa}}, \bibinfo {author} {\bibfnamefont {J.~M.}\ \bibnamefont {Kreikebaum}}, \bibinfo {author} {\bibfnamefont {D.}~\bibnamefont {Landhuis}}, \bibinfo {author} {\bibfnamefont {B.~W.}\ \bibnamefont {Langley}}, \bibinfo {author} {\bibfnamefont {P.}~\bibnamefont {Laptev}}, \bibinfo {author} {\bibfnamefont {K.~M.}\ \bibnamefont {Lau}}, \bibinfo {author} {\bibfnamefont {L.~L.}\ \bibnamefont {Guevel}}, \bibinfo {author} {\bibfnamefont {J.}~\bibnamefont {Ledford}}, \bibinfo {author} {\bibfnamefont {J.}~\bibnamefont {Lee}}, \bibinfo {author} {\bibfnamefont {K.~W.}\ \bibnamefont {Lee}},
  \bibinfo {author} {\bibfnamefont {Y.~D.}\ \bibnamefont {Lensky}}, \bibinfo {author} {\bibfnamefont {B.~J.}\ \bibnamefont {Lester}}, \bibinfo {author} {\bibfnamefont {W.~Y.}\ \bibnamefont {Li}}, \bibinfo {author} {\bibfnamefont {A.~T.}\ \bibnamefont {Lill}}, \bibinfo {author} {\bibfnamefont {W.}~\bibnamefont {Liu}}, \bibinfo {author} {\bibfnamefont {W.~P.}\ \bibnamefont {Livingston}}, \bibinfo {author} {\bibfnamefont {A.}~\bibnamefont {Locharla}}, \bibinfo {author} {\bibfnamefont {D.}~\bibnamefont {Lundahl}}, \bibinfo {author} {\bibfnamefont {A.}~\bibnamefont {Lunt}}, \bibinfo {author} {\bibfnamefont {S.}~\bibnamefont {Madhuk}}, \bibinfo {author} {\bibfnamefont {A.}~\bibnamefont {Maloney}}, \bibinfo {author} {\bibfnamefont {S.}~\bibnamefont {Mandr{\`a}}}, \bibinfo {author} {\bibfnamefont {L.~S.}\ \bibnamefont {Martin}}, \bibinfo {author} {\bibfnamefont {O.}~\bibnamefont {Martin}}, \bibinfo {author} {\bibfnamefont {S.}~\bibnamefont {Martin}}, \bibinfo {author} {\bibfnamefont {C.}~\bibnamefont {Maxfield}},
  \bibinfo {author} {\bibfnamefont {J.~R.}\ \bibnamefont {McClean}}, \bibinfo {author} {\bibfnamefont {M.}~\bibnamefont {McEwen}}, \bibinfo {author} {\bibfnamefont {S.}~\bibnamefont {Meeks}}, \bibinfo {author} {\bibfnamefont {K.~C.}\ \bibnamefont {Miao}}, \bibinfo {author} {\bibfnamefont {A.}~\bibnamefont {Mieszala}}, \bibinfo {author} {\bibfnamefont {S.}~\bibnamefont {Molina}}, \bibinfo {author} {\bibfnamefont {S.}~\bibnamefont {Montazeri}}, \bibinfo {author} {\bibfnamefont {A.}~\bibnamefont {Morvan}}, \bibinfo {author} {\bibfnamefont {R.}~\bibnamefont {Movassagh}}, \bibinfo {author} {\bibfnamefont {C.}~\bibnamefont {Neill}}, \bibinfo {author} {\bibfnamefont {A.}~\bibnamefont {Nersisyan}}, \bibinfo {author} {\bibfnamefont {M.}~\bibnamefont {Newman}}, \bibinfo {author} {\bibfnamefont {A.}~\bibnamefont {Nguyen}}, \bibinfo {author} {\bibfnamefont {M.}~\bibnamefont {Nguyen}}, \bibinfo {author} {\bibfnamefont {C.~H.}\ \bibnamefont {Ni}}, \bibinfo {author} {\bibfnamefont {M.~Y.}\ \bibnamefont {Niu}}, \bibinfo
  {author} {\bibfnamefont {W.~D.}\ \bibnamefont {Oliver}}, \bibinfo {author} {\bibfnamefont {K.}~\bibnamefont {Ottosson}}, \bibinfo {author} {\bibfnamefont {A.}~\bibnamefont {Pizzuto}}, \bibinfo {author} {\bibfnamefont {R.}~\bibnamefont {Potter}}, \bibinfo {author} {\bibfnamefont {O.}~\bibnamefont {Pritchard}}, \bibinfo {author} {\bibfnamefont {L.~P.}\ \bibnamefont {Pryadko}}, \bibinfo {author} {\bibfnamefont {C.}~\bibnamefont {Quintana}}, \bibinfo {author} {\bibfnamefont {M.~J.}\ \bibnamefont {Reagor}}, \bibinfo {author} {\bibfnamefont {D.~M.}\ \bibnamefont {Rhodes}}, \bibinfo {author} {\bibfnamefont {G.}~\bibnamefont {Roberts}}, \bibinfo {author} {\bibfnamefont {C.}~\bibnamefont {Rocque}}, \bibinfo {author} {\bibfnamefont {E.}~\bibnamefont {Rosenberg}}, \bibinfo {author} {\bibfnamefont {N.~C.}\ \bibnamefont {Rubin}}, \bibinfo {author} {\bibfnamefont {N.}~\bibnamefont {Saei}}, \bibinfo {author} {\bibfnamefont {K.}~\bibnamefont {Sankaragomathi}}, \bibinfo {author} {\bibfnamefont {K.~J.}\ \bibnamefont
  {Satzinger}}, \bibinfo {author} {\bibfnamefont {H.~F.}\ \bibnamefont {Schurkus}}, \bibinfo {author} {\bibfnamefont {C.}~\bibnamefont {Schuster}}, \bibinfo {author} {\bibfnamefont {M.~J.}\ \bibnamefont {Shearn}}, \bibinfo {author} {\bibfnamefont {A.}~\bibnamefont {Shorter}}, \bibinfo {author} {\bibfnamefont {N.}~\bibnamefont {Shutty}}, \bibinfo {author} {\bibfnamefont {V.}~\bibnamefont {Shvarts}}, \bibinfo {author} {\bibfnamefont {V.}~\bibnamefont {Sivak}}, \bibinfo {author} {\bibfnamefont {J.}~\bibnamefont {Skruzny}}, \bibinfo {author} {\bibfnamefont {S.}~\bibnamefont {Small}}, \bibinfo {author} {\bibfnamefont {W.~C.}\ \bibnamefont {Smith}}, \bibinfo {author} {\bibfnamefont {S.}~\bibnamefont {Springer}}, \bibinfo {author} {\bibfnamefont {G.}~\bibnamefont {Sterling}}, \bibinfo {author} {\bibfnamefont {J.}~\bibnamefont {Suchard}}, \bibinfo {author} {\bibfnamefont {M.}~\bibnamefont {Szalay}}, \bibinfo {author} {\bibfnamefont {A.}~\bibnamefont {Sztein}}, \bibinfo {author} {\bibfnamefont {D.}~\bibnamefont
  {Thor}}, \bibinfo {author} {\bibfnamefont {A.}~\bibnamefont {Torres}}, \bibinfo {author} {\bibfnamefont {M.~M.}\ \bibnamefont {Torunbalci}}, \bibinfo {author} {\bibfnamefont {A.}~\bibnamefont {Vaishnav}}, \bibinfo {author} {\bibfnamefont {S.}~\bibnamefont {Vdovichev}}, \bibinfo {author} {\bibfnamefont {B.}~\bibnamefont {Villalonga}}, \bibinfo {author} {\bibfnamefont {C.~V.}\ \bibnamefont {Heidweiller}}, \bibinfo {author} {\bibfnamefont {S.}~\bibnamefont {Waltman}}, \bibinfo {author} {\bibfnamefont {S.~X.}\ \bibnamefont {Wang}}, \bibinfo {author} {\bibfnamefont {T.}~\bibnamefont {White}}, \bibinfo {author} {\bibfnamefont {K.}~\bibnamefont {Wong}}, \bibinfo {author} {\bibfnamefont {B.~W.~K.}\ \bibnamefont {Woo}}, \bibinfo {author} {\bibfnamefont {C.}~\bibnamefont {Xing}}, \bibinfo {author} {\bibfnamefont {Z.~J.}\ \bibnamefont {Yao}}, \bibinfo {author} {\bibfnamefont {P.}~\bibnamefont {Yeh}}, \bibinfo {author} {\bibfnamefont {B.}~\bibnamefont {Ying}}, \bibinfo {author} {\bibfnamefont {J.}~\bibnamefont {Yoo}},
  \bibinfo {author} {\bibfnamefont {N.}~\bibnamefont {Yosri}}, \bibinfo {author} {\bibfnamefont {G.}~\bibnamefont {Young}}, \bibinfo {author} {\bibfnamefont {A.}~\bibnamefont {Zalcman}}, \bibinfo {author} {\bibfnamefont {N.}~\bibnamefont {Zhu}}, \bibinfo {author} {\bibfnamefont {N.}~\bibnamefont {Zobrist}}, \bibinfo {author} {\bibfnamefont {H.}~\bibnamefont {Neven}}, \bibinfo {author} {\bibfnamefont {R.}~\bibnamefont {Babbush}}, \bibinfo {author} {\bibfnamefont {S.}~\bibnamefont {Boixo}}, \bibinfo {author} {\bibfnamefont {J.}~\bibnamefont {Hilton}}, \bibinfo {author} {\bibfnamefont {E.}~\bibnamefont {Lucero}}, \bibinfo {author} {\bibfnamefont {A.}~\bibnamefont {Megrant}}, \bibinfo {author} {\bibfnamefont {J.}~\bibnamefont {Kelly}}, \bibinfo {author} {\bibfnamefont {Y.}~\bibnamefont {Chen}}, \bibinfo {author} {\bibfnamefont {V.}~\bibnamefont {Smelyanskiy}}, \bibinfo {author} {\bibfnamefont {G.}~\bibnamefont {Vidal}}, \bibinfo {author} {\bibfnamefont {P.}~\bibnamefont {Roushan}}, \bibinfo {author}
  {\bibfnamefont {A.~M.}\ \bibnamefont {L{\"a}uchli}}, \bibinfo {author} {\bibfnamefont {D.~A.}\ \bibnamefont {Abanin}},\ and\ \bibinfo {author} {\bibfnamefont {X.}~\bibnamefont {Mi}},\ }\href {https://doi.org/10.1038/s41586-024-08460-3} {\bibfield  {journal} {\bibinfo  {journal} {Nature}\ }\textbf {\bibinfo {volume} {638}},\ \bibinfo {pages} {79} (\bibinfo {year} {2025}{\natexlab{b}})}\BibitemShut {NoStop}%
\bibitem [{\citenamefont {Kohlert}\ \emph {et~al.}(2023)\citenamefont {Kohlert}, \citenamefont {Scherg}, \citenamefont {Sala}, \citenamefont {Pollmann}, \citenamefont {Hebbe~Madhusudhana}, \citenamefont {Bloch},\ and\ \citenamefont {Aidelsburger}}]{fragmentation}%
  \BibitemOpen
  \bibfield  {author} {\bibinfo {author} {\bibfnamefont {T.}~\bibnamefont {Kohlert}}, \bibinfo {author} {\bibfnamefont {S.}~\bibnamefont {Scherg}}, \bibinfo {author} {\bibfnamefont {P.}~\bibnamefont {Sala}}, \bibinfo {author} {\bibfnamefont {F.}~\bibnamefont {Pollmann}}, \bibinfo {author} {\bibfnamefont {B.}~\bibnamefont {Hebbe~Madhusudhana}}, \bibinfo {author} {\bibfnamefont {I.}~\bibnamefont {Bloch}},\ and\ \bibinfo {author} {\bibfnamefont {M.}~\bibnamefont {Aidelsburger}},\ }\href {https://doi.org/10.1103/PhysRevLett.130.010201} {\bibfield  {journal} {\bibinfo  {journal} {Phys. Rev. Lett.}\ }\textbf {\bibinfo {volume} {130}},\ \bibinfo {pages} {010201} (\bibinfo {year} {2023})}\BibitemShut {NoStop}%
\bibitem [{\citenamefont {{Liang}}\ \emph {et~al.}(2024)\citenamefont {{Liang}}, \citenamefont {{Yue}}, \citenamefont {{Chao}}, \citenamefont {{Hua}}, \citenamefont {{Lin}}, \citenamefont {{Khoon Tey}},\ and\ \citenamefont {{You}}}]{YouLi_information_scrambling}%
  \BibitemOpen
  \bibfield  {author} {\bibinfo {author} {\bibfnamefont {X.}~\bibnamefont {{Liang}}}, \bibinfo {author} {\bibfnamefont {Z.}~\bibnamefont {{Yue}}}, \bibinfo {author} {\bibfnamefont {Y.-X.}\ \bibnamefont {{Chao}}}, \bibinfo {author} {\bibfnamefont {Z.-X.}\ \bibnamefont {{Hua}}}, \bibinfo {author} {\bibfnamefont {Y.}~\bibnamefont {{Lin}}}, \bibinfo {author} {\bibfnamefont {M.}~\bibnamefont {{Khoon Tey}}},\ and\ \bibinfo {author} {\bibfnamefont {L.}~\bibnamefont {{You}}},\ }\href {https://doi.org/10.48550/arXiv.2410.16174} {\bibfield  {journal} {\bibinfo  {journal} {arXiv e-prints}\ ,\ \bibinfo {eid} {arXiv:2410.16174}} (\bibinfo {year} {2024})},\ \Eprint {https://arxiv.org/abs/2410.16174} {arXiv:2410.16174 [quant-ph]} \BibitemShut {NoStop}%
\bibitem [{\citenamefont {Xu}\ \emph {et~al.}(2025)\citenamefont {Xu}, \citenamefont {Kendrick}, \citenamefont {Kale}, \citenamefont {Gang}, \citenamefont {Feng}, \citenamefont {Zhang}, \citenamefont {Young}, \citenamefont {Lebrat},\ and\ \citenamefont {Greiner}}]{low_temperature}%
  \BibitemOpen
  \bibfield  {author} {\bibinfo {author} {\bibfnamefont {M.}~\bibnamefont {Xu}}, \bibinfo {author} {\bibfnamefont {L.~H.}\ \bibnamefont {Kendrick}}, \bibinfo {author} {\bibfnamefont {A.}~\bibnamefont {Kale}}, \bibinfo {author} {\bibfnamefont {Y.}~\bibnamefont {Gang}}, \bibinfo {author} {\bibfnamefont {C.}~\bibnamefont {Feng}}, \bibinfo {author} {\bibfnamefont {S.}~\bibnamefont {Zhang}}, \bibinfo {author} {\bibfnamefont {A.~W.}\ \bibnamefont {Young}}, \bibinfo {author} {\bibfnamefont {M.}~\bibnamefont {Lebrat}},\ and\ \bibinfo {author} {\bibfnamefont {M.}~\bibnamefont {Greiner}},\ }\href {https://doi.org/10.1038/s41586-025-09112-w} {\bibfield  {journal} {\bibinfo  {journal} {Nature}\ }\textbf {\bibinfo {volume} {642}},\ \bibinfo {pages} {909} (\bibinfo {year} {2025})}\BibitemShut {NoStop}%
\bibitem [{\citenamefont {Mazurenko}\ \emph {et~al.}(2017)\citenamefont {Mazurenko}, \citenamefont {Chiu}, \citenamefont {Ji}, \citenamefont {Parsons}, \citenamefont {Kan{\'a}sz-Nagy}, \citenamefont {Schmidt}, \citenamefont {Grusdt}, \citenamefont {Demler}, \citenamefont {Greif},\ and\ \citenamefont {Greiner}}]{AFM}%
  \BibitemOpen
  \bibfield  {author} {\bibinfo {author} {\bibfnamefont {A.}~\bibnamefont {Mazurenko}}, \bibinfo {author} {\bibfnamefont {C.~S.}\ \bibnamefont {Chiu}}, \bibinfo {author} {\bibfnamefont {G.}~\bibnamefont {Ji}}, \bibinfo {author} {\bibfnamefont {M.~F.}\ \bibnamefont {Parsons}}, \bibinfo {author} {\bibfnamefont {M.}~\bibnamefont {Kan{\'a}sz-Nagy}}, \bibinfo {author} {\bibfnamefont {R.}~\bibnamefont {Schmidt}}, \bibinfo {author} {\bibfnamefont {F.}~\bibnamefont {Grusdt}}, \bibinfo {author} {\bibfnamefont {E.}~\bibnamefont {Demler}}, \bibinfo {author} {\bibfnamefont {D.}~\bibnamefont {Greif}},\ and\ \bibinfo {author} {\bibfnamefont {M.}~\bibnamefont {Greiner}},\ }\href {https://doi.org/10.1038/nature22362} {\bibfield  {journal} {\bibinfo  {journal} {Nature}\ }\textbf {\bibinfo {volume} {545}},\ \bibinfo {pages} {462} (\bibinfo {year} {2017})}\BibitemShut {NoStop}%
\bibitem [{\citenamefont {Wang}\ \emph {et~al.}(2022)\citenamefont {Wang}, \citenamefont {Khatami}, \citenamefont {Fei}, \citenamefont {Wyrick}, \citenamefont {Namboodiri}, \citenamefont {Kashid}, \citenamefont {Rigosi}, \citenamefont {Bryant},\ and\ \citenamefont {Silver}}]{quantum_dots}%
  \BibitemOpen
  \bibfield  {author} {\bibinfo {author} {\bibfnamefont {X.}~\bibnamefont {Wang}}, \bibinfo {author} {\bibfnamefont {E.}~\bibnamefont {Khatami}}, \bibinfo {author} {\bibfnamefont {F.}~\bibnamefont {Fei}}, \bibinfo {author} {\bibfnamefont {J.}~\bibnamefont {Wyrick}}, \bibinfo {author} {\bibfnamefont {P.}~\bibnamefont {Namboodiri}}, \bibinfo {author} {\bibfnamefont {R.}~\bibnamefont {Kashid}}, \bibinfo {author} {\bibfnamefont {A.~F.}\ \bibnamefont {Rigosi}}, \bibinfo {author} {\bibfnamefont {G.}~\bibnamefont {Bryant}},\ and\ \bibinfo {author} {\bibfnamefont {R.}~\bibnamefont {Silver}},\ }\href {https://doi.org/10.1038/s41467-022-34220-w} {\bibfield  {journal} {\bibinfo  {journal} {Nature Communications}\ }\textbf {\bibinfo {volume} {13}},\ \bibinfo {pages} {6824} (\bibinfo {year} {2022})}\BibitemShut {NoStop}%
\bibitem [{\citenamefont {Hofrichter}\ \emph {et~al.}(2016)\citenamefont {Hofrichter}, \citenamefont {Riegger}, \citenamefont {Scazza}, \citenamefont {H\"ofer}, \citenamefont {Fernandes}, \citenamefont {Bloch},\ and\ \citenamefont {F\"olling}}]{Mott}%
  \BibitemOpen
  \bibfield  {author} {\bibinfo {author} {\bibfnamefont {C.}~\bibnamefont {Hofrichter}}, \bibinfo {author} {\bibfnamefont {L.}~\bibnamefont {Riegger}}, \bibinfo {author} {\bibfnamefont {F.}~\bibnamefont {Scazza}}, \bibinfo {author} {\bibfnamefont {M.}~\bibnamefont {H\"ofer}}, \bibinfo {author} {\bibfnamefont {D.~R.}\ \bibnamefont {Fernandes}}, \bibinfo {author} {\bibfnamefont {I.}~\bibnamefont {Bloch}},\ and\ \bibinfo {author} {\bibfnamefont {S.}~\bibnamefont {F\"olling}},\ }\href {https://doi.org/10.1103/PhysRevX.6.021030} {\bibfield  {journal} {\bibinfo  {journal} {Phys. Rev. X}\ }\textbf {\bibinfo {volume} {6}},\ \bibinfo {pages} {021030} (\bibinfo {year} {2016})}\BibitemShut {NoStop}%
\bibitem [{\citenamefont {Shao}\ \emph {et~al.}(2024{\natexlab{a}})\citenamefont {Shao}, \citenamefont {Wang}, \citenamefont {Zhu}, \citenamefont {Zhu}, \citenamefont {Sun}, \citenamefont {Chen}, \citenamefont {Zhang}, \citenamefont {Fan}, \citenamefont {Deng}, \citenamefont {Yao}, \citenamefont {Chen},\ and\ \citenamefont {Pan}}]{ustc_3D}%
  \BibitemOpen
  \bibfield  {author} {\bibinfo {author} {\bibfnamefont {H.-J.}\ \bibnamefont {Shao}}, \bibinfo {author} {\bibfnamefont {Y.-X.}\ \bibnamefont {Wang}}, \bibinfo {author} {\bibfnamefont {D.-Z.}\ \bibnamefont {Zhu}}, \bibinfo {author} {\bibfnamefont {Y.-S.}\ \bibnamefont {Zhu}}, \bibinfo {author} {\bibfnamefont {H.-N.}\ \bibnamefont {Sun}}, \bibinfo {author} {\bibfnamefont {S.-Y.}\ \bibnamefont {Chen}}, \bibinfo {author} {\bibfnamefont {C.}~\bibnamefont {Zhang}}, \bibinfo {author} {\bibfnamefont {Z.-J.}\ \bibnamefont {Fan}}, \bibinfo {author} {\bibfnamefont {Y.}~\bibnamefont {Deng}}, \bibinfo {author} {\bibfnamefont {X.-C.}\ \bibnamefont {Yao}}, \bibinfo {author} {\bibfnamefont {Y.-A.}\ \bibnamefont {Chen}},\ and\ \bibinfo {author} {\bibfnamefont {J.-W.}\ \bibnamefont {Pan}},\ }\href {https://doi.org/10.1038/s41586-024-07689-2} {\bibfield  {journal} {\bibinfo  {journal} {Nature}\ }\textbf {\bibinfo {volume} {632}},\ \bibinfo {pages} {267} (\bibinfo {year} {2024}{\natexlab{a}})}\BibitemShut {NoStop}%
\bibitem [{\citenamefont {Gross}\ and\ \citenamefont {Bloch}(2017)}]{Bloch_review}%
  \BibitemOpen
  \bibfield  {author} {\bibinfo {author} {\bibfnamefont {C.}~\bibnamefont {Gross}}\ and\ \bibinfo {author} {\bibfnamefont {I.}~\bibnamefont {Bloch}},\ }\href {https://doi.org/10.1126/science.aal3837} {\bibfield  {journal} {\bibinfo  {journal} {Science}\ }\textbf {\bibinfo {volume} {357}},\ \bibinfo {pages} {995} (\bibinfo {year} {2017})},\ \Eprint {https://arxiv.org/abs/https://www.science.org/doi/pdf/10.1126/science.aal3837} {https://www.science.org/doi/pdf/10.1126/science.aal3837} \BibitemShut {NoStop}%
\bibitem [{\citenamefont {Sompet}\ \emph {et~al.}(2022)\citenamefont {Sompet}, \citenamefont {Hirthe}, \citenamefont {Bourgund}, \citenamefont {Chalopin}, \citenamefont {Bibo}, \citenamefont {Koepsell}, \citenamefont {Bojovi{\'c}}, \citenamefont {Verresen}, \citenamefont {Pollmann}, \citenamefont {Salomon}, \citenamefont {Gross}, \citenamefont {Hilker},\ and\ \citenamefont {Bloch}}]{haldane_phase}%
  \BibitemOpen
  \bibfield  {author} {\bibinfo {author} {\bibfnamefont {P.}~\bibnamefont {Sompet}}, \bibinfo {author} {\bibfnamefont {S.}~\bibnamefont {Hirthe}}, \bibinfo {author} {\bibfnamefont {D.}~\bibnamefont {Bourgund}}, \bibinfo {author} {\bibfnamefont {T.}~\bibnamefont {Chalopin}}, \bibinfo {author} {\bibfnamefont {J.}~\bibnamefont {Bibo}}, \bibinfo {author} {\bibfnamefont {J.}~\bibnamefont {Koepsell}}, \bibinfo {author} {\bibfnamefont {P.}~\bibnamefont {Bojovi{\'c}}}, \bibinfo {author} {\bibfnamefont {R.}~\bibnamefont {Verresen}}, \bibinfo {author} {\bibfnamefont {F.}~\bibnamefont {Pollmann}}, \bibinfo {author} {\bibfnamefont {G.}~\bibnamefont {Salomon}}, \bibinfo {author} {\bibfnamefont {C.}~\bibnamefont {Gross}}, \bibinfo {author} {\bibfnamefont {T.~A.}\ \bibnamefont {Hilker}},\ and\ \bibinfo {author} {\bibfnamefont {I.}~\bibnamefont {Bloch}},\ }\href {https://doi.org/10.1038/s41586-022-04688-z} {\bibfield  {journal} {\bibinfo  {journal} {Nature}\ }\textbf {\bibinfo {volume} {606}},\ \bibinfo {pages} {484}
  (\bibinfo {year} {2022})}\BibitemShut {NoStop}%
\bibitem [{\citenamefont {{Yue}}\ \emph {et~al.}(2025)\citenamefont {{Yue}}, \citenamefont {{Mao}}, \citenamefont {{Liang}}, \citenamefont {{Hua}}, \citenamefont {{Ge}}, \citenamefont {{Chao}}, \citenamefont {{Li}}, \citenamefont {{Jia}}, \citenamefont {{Khoon Tey}}, \citenamefont {{Xu}},\ and\ \citenamefont {{You}}}]{YouLi_topology}%
  \BibitemOpen
  \bibfield  {author} {\bibinfo {author} {\bibfnamefont {Z.}~\bibnamefont {{Yue}}}, \bibinfo {author} {\bibfnamefont {Y.-F.}\ \bibnamefont {{Mao}}}, \bibinfo {author} {\bibfnamefont {X.}~\bibnamefont {{Liang}}}, \bibinfo {author} {\bibfnamefont {Z.-X.}\ \bibnamefont {{Hua}}}, \bibinfo {author} {\bibfnamefont {P.}~\bibnamefont {{Ge}}}, \bibinfo {author} {\bibfnamefont {Y.-X.}\ \bibnamefont {{Chao}}}, \bibinfo {author} {\bibfnamefont {K.}~\bibnamefont {{Li}}}, \bibinfo {author} {\bibfnamefont {C.}~\bibnamefont {{Jia}}}, \bibinfo {author} {\bibfnamefont {M.}~\bibnamefont {{Khoon Tey}}}, \bibinfo {author} {\bibfnamefont {Y.}~\bibnamefont {{Xu}}},\ and\ \bibinfo {author} {\bibfnamefont {L.}~\bibnamefont {{You}}},\ }\href {https://doi.org/10.48550/arXiv.2505.06286} {\bibfield  {journal} {\bibinfo  {journal} {arXiv e-prints}\ ,\ \bibinfo {eid} {arXiv:2505.06286}} (\bibinfo {year} {2025})},\ \Eprint {https://arxiv.org/abs/2505.06286} {arXiv:2505.06286 [cond-mat.quant-gas]} \BibitemShut {NoStop}%
\bibitem [{\citenamefont {Lu}\ \emph {et~al.}(2024{\natexlab{a}})\citenamefont {Lu}, \citenamefont {Wang}, \citenamefont {Kanungo}, \citenamefont {Dunning},\ and\ \citenamefont {Killian}}]{PhysRevA.110.023318}%
  \BibitemOpen
  \bibfield  {author} {\bibinfo {author} {\bibfnamefont {Y.}~\bibnamefont {Lu}}, \bibinfo {author} {\bibfnamefont {C.}~\bibnamefont {Wang}}, \bibinfo {author} {\bibfnamefont {S.~K.}\ \bibnamefont {Kanungo}}, \bibinfo {author} {\bibfnamefont {F.~B.}\ \bibnamefont {Dunning}},\ and\ \bibinfo {author} {\bibfnamefont {T.~C.}\ \bibnamefont {Killian}},\ }\href {https://doi.org/10.1103/PhysRevA.110.023318} {\bibfield  {journal} {\bibinfo  {journal} {Phys. Rev. A}\ }\textbf {\bibinfo {volume} {110}},\ \bibinfo {pages} {023318} (\bibinfo {year} {2024}{\natexlab{a}})}\BibitemShut {NoStop}%
\bibitem [{\citenamefont {de~Léséleuc}\ \emph {et~al.}(2019)\citenamefont {de~Léséleuc}, \citenamefont {Lienhard}, \citenamefont {Scholl}, \citenamefont {Barredo}, \citenamefont {Weber}, \citenamefont {Lang}, \citenamefont {Büchler}, \citenamefont {Lahaye},\ and\ \citenamefont {Browaeys}}]{Antoine_ssh}%
  \BibitemOpen
  \bibfield  {author} {\bibinfo {author} {\bibfnamefont {S.}~\bibnamefont {de~Léséleuc}}, \bibinfo {author} {\bibfnamefont {V.}~\bibnamefont {Lienhard}}, \bibinfo {author} {\bibfnamefont {P.}~\bibnamefont {Scholl}}, \bibinfo {author} {\bibfnamefont {D.}~\bibnamefont {Barredo}}, \bibinfo {author} {\bibfnamefont {S.}~\bibnamefont {Weber}}, \bibinfo {author} {\bibfnamefont {N.}~\bibnamefont {Lang}}, \bibinfo {author} {\bibfnamefont {H.~P.}\ \bibnamefont {Büchler}}, \bibinfo {author} {\bibfnamefont {T.}~\bibnamefont {Lahaye}},\ and\ \bibinfo {author} {\bibfnamefont {A.}~\bibnamefont {Browaeys}},\ }\href {https://doi.org/10.1126/science.aav9105} {\bibfield  {journal} {\bibinfo  {journal} {Science}\ }\textbf {\bibinfo {volume} {365}},\ \bibinfo {pages} {775} (\bibinfo {year} {2019})},\ \Eprint {https://arxiv.org/abs/https://www.science.org/doi/pdf/10.1126/science.aav9105} {https://www.science.org/doi/pdf/10.1126/science.aav9105} \BibitemShut {NoStop}%
\bibitem [{\citenamefont {Gonz{\'a}lez-Cuadra}\ \emph {et~al.}(2025)\citenamefont {Gonz{\'a}lez-Cuadra}, \citenamefont {Hamdan}, \citenamefont {Zache}, \citenamefont {Braverman}, \citenamefont {Kornja{\v c}a}, \citenamefont {Lukin}, \citenamefont {Cant{\'u}}, \citenamefont {Liu}, \citenamefont {Wang}, \citenamefont {Keesling}, \citenamefont {Lukin}, \citenamefont {Zoller},\ and\ \citenamefont {Bylinskii}}]{quera_string_breaking}%
  \BibitemOpen
  \bibfield  {author} {\bibinfo {author} {\bibfnamefont {D.}~\bibnamefont {Gonz{\'a}lez-Cuadra}}, \bibinfo {author} {\bibfnamefont {M.}~\bibnamefont {Hamdan}}, \bibinfo {author} {\bibfnamefont {T.~V.}\ \bibnamefont {Zache}}, \bibinfo {author} {\bibfnamefont {B.}~\bibnamefont {Braverman}}, \bibinfo {author} {\bibfnamefont {M.}~\bibnamefont {Kornja{\v c}a}}, \bibinfo {author} {\bibfnamefont {A.}~\bibnamefont {Lukin}}, \bibinfo {author} {\bibfnamefont {S.~H.}\ \bibnamefont {Cant{\'u}}}, \bibinfo {author} {\bibfnamefont {F.}~\bibnamefont {Liu}}, \bibinfo {author} {\bibfnamefont {S.-T.}\ \bibnamefont {Wang}}, \bibinfo {author} {\bibfnamefont {A.}~\bibnamefont {Keesling}}, \bibinfo {author} {\bibfnamefont {M.~D.}\ \bibnamefont {Lukin}}, \bibinfo {author} {\bibfnamefont {P.}~\bibnamefont {Zoller}},\ and\ \bibinfo {author} {\bibfnamefont {A.}~\bibnamefont {Bylinskii}},\ }\href {https://doi.org/10.1038/s41586-025-09051-6} {\bibfield  {journal} {\bibinfo  {journal} {Nature}\ }\textbf {\bibinfo {volume} {642}},\
  \bibinfo {pages} {321} (\bibinfo {year} {2025})}\BibitemShut {NoStop}%
\bibitem [{\citenamefont {Mildenberger}\ \emph {et~al.}(2025)\citenamefont {Mildenberger}, \citenamefont {Mruczkiewicz}, \citenamefont {Halimeh}, \citenamefont {Jiang},\ and\ \citenamefont {Hauke}}]{Z2}%
  \BibitemOpen
  \bibfield  {author} {\bibinfo {author} {\bibfnamefont {J.}~\bibnamefont {Mildenberger}}, \bibinfo {author} {\bibfnamefont {W.}~\bibnamefont {Mruczkiewicz}}, \bibinfo {author} {\bibfnamefont {J.~C.}\ \bibnamefont {Halimeh}}, \bibinfo {author} {\bibfnamefont {Z.}~\bibnamefont {Jiang}},\ and\ \bibinfo {author} {\bibfnamefont {P.}~\bibnamefont {Hauke}},\ }\href {https://doi.org/10.1038/s41567-024-02723-6} {\bibfield  {journal} {\bibinfo  {journal} {Nature Physics}\ }\textbf {\bibinfo {volume} {21}},\ \bibinfo {pages} {312} (\bibinfo {year} {2025})}\BibitemShut {NoStop}%
\bibitem [{\citenamefont {D'Alessandro}(2021)}]{dAlessandro2021}%
  \BibitemOpen
  \bibfield  {author} {\bibinfo {author} {\bibfnamefont {D.}~\bibnamefont {D'Alessandro}},\ }\href {https://doi.org/10.1201/9781003051268} {\emph {\bibinfo {title} {Introduction to Quantum Control and Dynamics}}},\ \bibinfo {edition} {2nd}\ ed.\ (\bibinfo  {publisher} {Chapman and Hall/CRC},\ \bibinfo {year} {2021})\BibitemShut {NoStop}%
\bibitem [{\citenamefont {Khaneja}\ and\ \citenamefont {Glaser}(2001)}]{KHANEJA200111}%
  \BibitemOpen
  \bibfield  {author} {\bibinfo {author} {\bibfnamefont {N.}~\bibnamefont {Khaneja}}\ and\ \bibinfo {author} {\bibfnamefont {S.~J.}\ \bibnamefont {Glaser}},\ }\href {https://doi.org/https://doi.org/10.1016/S0301-0104(01)00318-4} {\bibfield  {journal} {\bibinfo  {journal} {Chemical Physics}\ }\textbf {\bibinfo {volume} {267}},\ \bibinfo {pages} {11} (\bibinfo {year} {2001})}\BibitemShut {NoStop}%
\bibitem [{\citenamefont {Benjamin}(2000)}]{Simon1}%
  \BibitemOpen
  \bibfield  {author} {\bibinfo {author} {\bibfnamefont {S.~C.}\ \bibnamefont {Benjamin}},\ }\href {https://doi.org/10.1103/PhysRevA.61.020301} {\bibfield  {journal} {\bibinfo  {journal} {Phys. Rev. A}\ }\textbf {\bibinfo {volume} {61}},\ \bibinfo {pages} {020301} (\bibinfo {year} {2000})}\BibitemShut {NoStop}%
\bibitem [{\citenamefont {Benjamin}(2001)}]{Benjamin_PRL}%
  \BibitemOpen
  \bibfield  {author} {\bibinfo {author} {\bibfnamefont {S.~C.}\ \bibnamefont {Benjamin}},\ }\href {https://doi.org/10.1103/PhysRevLett.88.017904} {\bibfield  {journal} {\bibinfo  {journal} {Phys. Rev. Lett.}\ }\textbf {\bibinfo {volume} {88}},\ \bibinfo {pages} {017904} (\bibinfo {year} {2001})}\BibitemShut {NoStop}%
\bibitem [{\citenamefont {Lloyd}(2018)}]{qaoa_universal}%
  \BibitemOpen
  \bibfield  {author} {\bibinfo {author} {\bibfnamefont {S.}~\bibnamefont {Lloyd}},\ }\href {https://arxiv.org/abs/1812.11075} {\bibinfo {title} {Quantum approximate optimization is computationally universal}} (\bibinfo {year} {2018}),\ \Eprint {https://arxiv.org/abs/1812.11075} {arXiv:1812.11075 [quant-ph]} \BibitemShut {NoStop}%
\bibitem [{\citenamefont {Cesa}\ and\ \citenamefont {Pichler}(2023)}]{Hannes_dual_species}%
  \BibitemOpen
  \bibfield  {author} {\bibinfo {author} {\bibfnamefont {F.}~\bibnamefont {Cesa}}\ and\ \bibinfo {author} {\bibfnamefont {H.}~\bibnamefont {Pichler}},\ }\href {https://doi.org/10.1103/PhysRevLett.131.170601} {\bibfield  {journal} {\bibinfo  {journal} {Phys. Rev. Lett.}\ }\textbf {\bibinfo {volume} {131}},\ \bibinfo {pages} {170601} (\bibinfo {year} {2023})}\BibitemShut {NoStop}%
\bibitem [{\citenamefont {Gkritsis}\ \emph {et~al.}(2025)\citenamefont {Gkritsis}, \citenamefont {Dux}, \citenamefont {Zhang}, \citenamefont {Jain}, \citenamefont {Gogolin},\ and\ \citenamefont {Preiss}}]{superlattice_chemistry}%
  \BibitemOpen
  \bibfield  {author} {\bibinfo {author} {\bibfnamefont {F.}~\bibnamefont {Gkritsis}}, \bibinfo {author} {\bibfnamefont {D.}~\bibnamefont {Dux}}, \bibinfo {author} {\bibfnamefont {J.}~\bibnamefont {Zhang}}, \bibinfo {author} {\bibfnamefont {N.}~\bibnamefont {Jain}}, \bibinfo {author} {\bibfnamefont {C.}~\bibnamefont {Gogolin}},\ and\ \bibinfo {author} {\bibfnamefont {P.~M.}\ \bibnamefont {Preiss}},\ }\href {https://doi.org/10.1103/PRXQuantum.6.010318} {\bibfield  {journal} {\bibinfo  {journal} {PRX Quantum}\ }\textbf {\bibinfo {volume} {6}},\ \bibinfo {pages} {010318} (\bibinfo {year} {2025})}\BibitemShut {NoStop}%
\bibitem [{\citenamefont {Chalopin}(2021)}]{superlattice_review}%
  \BibitemOpen
  \bibfield  {author} {\bibinfo {author} {\bibfnamefont {T.}~\bibnamefont {Chalopin}},\ }\href {https://doi.org/10.1038/s42254-021-00357-8} {\bibfield  {journal} {\bibinfo  {journal} {Nature Reviews Physics}\ }\textbf {\bibinfo {volume} {3}},\ \bibinfo {pages} {605} (\bibinfo {year} {2021})}\BibitemShut {NoStop}%
\bibitem [{\citenamefont {Chalopin}\ \emph {et~al.}(2025{\natexlab{a}})\citenamefont {Chalopin}, \citenamefont {Bojovi\ifmmode~\acute{c}\else \'{c}\fi{}}, \citenamefont {Bourgund}, \citenamefont {Wang}, \citenamefont {Franz}, \citenamefont {Bloch},\ and\ \citenamefont {Hilker}}]{PhysRevLett.134.053402}%
  \BibitemOpen
  \bibfield  {author} {\bibinfo {author} {\bibfnamefont {T.}~\bibnamefont {Chalopin}}, \bibinfo {author} {\bibfnamefont {P.}~\bibnamefont {Bojovi\ifmmode~\acute{c}\else \'{c}\fi{}}}, \bibinfo {author} {\bibfnamefont {D.}~\bibnamefont {Bourgund}}, \bibinfo {author} {\bibfnamefont {S.}~\bibnamefont {Wang}}, \bibinfo {author} {\bibfnamefont {T.}~\bibnamefont {Franz}}, \bibinfo {author} {\bibfnamefont {I.}~\bibnamefont {Bloch}},\ and\ \bibinfo {author} {\bibfnamefont {T.}~\bibnamefont {Hilker}},\ }\href {https://doi.org/10.1103/PhysRevLett.134.053402} {\bibfield  {journal} {\bibinfo  {journal} {Phys. Rev. Lett.}\ }\textbf {\bibinfo {volume} {134}},\ \bibinfo {pages} {053402} (\bibinfo {year} {2025}{\natexlab{a}})}\BibitemShut {NoStop}%
\bibitem [{\citenamefont {Arute}\ \emph {et~al.}(2019)\citenamefont {Arute}, \citenamefont {Arya}, \citenamefont {Babbush}, \citenamefont {Bacon}, \citenamefont {Bardin}, \citenamefont {Barends}, \citenamefont {Biswas}, \citenamefont {Boixo}, \citenamefont {Brandao}, \citenamefont {Buell}, \citenamefont {Burkett}, \citenamefont {Chen}, \citenamefont {Chen}, \citenamefont {Chiaro}, \citenamefont {Collins}, \citenamefont {Courtney}, \citenamefont {Dunsworth}, \citenamefont {Farhi}, \citenamefont {Foxen}, \citenamefont {Fowler}, \citenamefont {Gidney}, \citenamefont {Giustina}, \citenamefont {Graff}, \citenamefont {Guerin}, \citenamefont {Habegger}, \citenamefont {Harrigan}, \citenamefont {Hartmann}, \citenamefont {Ho}, \citenamefont {Hoffmann}, \citenamefont {Huang}, \citenamefont {Humble}, \citenamefont {Isakov}, \citenamefont {Jeffrey}, \citenamefont {Jiang}, \citenamefont {Kafri}, \citenamefont {Kechedzhi}, \citenamefont {Kelly}, \citenamefont {Klimov}, \citenamefont {Knysh}, \citenamefont {Korotkov},
  \citenamefont {Kostritsa}, \citenamefont {Landhuis}, \citenamefont {Lindmark}, \citenamefont {Lucero}, \citenamefont {Lyakh}, \citenamefont {Mandr{\`a}}, \citenamefont {McClean}, \citenamefont {McEwen}, \citenamefont {Megrant}, \citenamefont {Mi}, \citenamefont {Michielsen}, \citenamefont {Mohseni}, \citenamefont {Mutus}, \citenamefont {Naaman}, \citenamefont {Neeley}, \citenamefont {Neill}, \citenamefont {Niu}, \citenamefont {Ostby}, \citenamefont {Petukhov}, \citenamefont {Platt}, \citenamefont {Quintana}, \citenamefont {Rieffel}, \citenamefont {Roushan}, \citenamefont {Rubin}, \citenamefont {Sank}, \citenamefont {Satzinger}, \citenamefont {Smelyanskiy}, \citenamefont {Sung}, \citenamefont {Trevithick}, \citenamefont {Vainsencher}, \citenamefont {Villalonga}, \citenamefont {White}, \citenamefont {Yao}, \citenamefont {Yeh}, \citenamefont {Zalcman}, \citenamefont {Neven},\ and\ \citenamefont {Martinis}}]{google_supremacy}%
  \BibitemOpen
  \bibfield  {author} {\bibinfo {author} {\bibfnamefont {F.}~\bibnamefont {Arute}}, \bibinfo {author} {\bibfnamefont {K.}~\bibnamefont {Arya}}, \bibinfo {author} {\bibfnamefont {R.}~\bibnamefont {Babbush}}, \bibinfo {author} {\bibfnamefont {D.}~\bibnamefont {Bacon}}, \bibinfo {author} {\bibfnamefont {J.~C.}\ \bibnamefont {Bardin}}, \bibinfo {author} {\bibfnamefont {R.}~\bibnamefont {Barends}}, \bibinfo {author} {\bibfnamefont {R.}~\bibnamefont {Biswas}}, \bibinfo {author} {\bibfnamefont {S.}~\bibnamefont {Boixo}}, \bibinfo {author} {\bibfnamefont {F.~G. S.~L.}\ \bibnamefont {Brandao}}, \bibinfo {author} {\bibfnamefont {D.~A.}\ \bibnamefont {Buell}}, \bibinfo {author} {\bibfnamefont {B.}~\bibnamefont {Burkett}}, \bibinfo {author} {\bibfnamefont {Y.}~\bibnamefont {Chen}}, \bibinfo {author} {\bibfnamefont {Z.}~\bibnamefont {Chen}}, \bibinfo {author} {\bibfnamefont {B.}~\bibnamefont {Chiaro}}, \bibinfo {author} {\bibfnamefont {R.}~\bibnamefont {Collins}}, \bibinfo {author} {\bibfnamefont {W.}~\bibnamefont
  {Courtney}}, \bibinfo {author} {\bibfnamefont {A.}~\bibnamefont {Dunsworth}}, \bibinfo {author} {\bibfnamefont {E.}~\bibnamefont {Farhi}}, \bibinfo {author} {\bibfnamefont {B.}~\bibnamefont {Foxen}}, \bibinfo {author} {\bibfnamefont {A.}~\bibnamefont {Fowler}}, \bibinfo {author} {\bibfnamefont {C.}~\bibnamefont {Gidney}}, \bibinfo {author} {\bibfnamefont {M.}~\bibnamefont {Giustina}}, \bibinfo {author} {\bibfnamefont {R.}~\bibnamefont {Graff}}, \bibinfo {author} {\bibfnamefont {K.}~\bibnamefont {Guerin}}, \bibinfo {author} {\bibfnamefont {S.}~\bibnamefont {Habegger}}, \bibinfo {author} {\bibfnamefont {M.~P.}\ \bibnamefont {Harrigan}}, \bibinfo {author} {\bibfnamefont {M.~J.}\ \bibnamefont {Hartmann}}, \bibinfo {author} {\bibfnamefont {A.}~\bibnamefont {Ho}}, \bibinfo {author} {\bibfnamefont {M.}~\bibnamefont {Hoffmann}}, \bibinfo {author} {\bibfnamefont {T.}~\bibnamefont {Huang}}, \bibinfo {author} {\bibfnamefont {T.~S.}\ \bibnamefont {Humble}}, \bibinfo {author} {\bibfnamefont {S.~V.}\ \bibnamefont
  {Isakov}}, \bibinfo {author} {\bibfnamefont {E.}~\bibnamefont {Jeffrey}}, \bibinfo {author} {\bibfnamefont {Z.}~\bibnamefont {Jiang}}, \bibinfo {author} {\bibfnamefont {D.}~\bibnamefont {Kafri}}, \bibinfo {author} {\bibfnamefont {K.}~\bibnamefont {Kechedzhi}}, \bibinfo {author} {\bibfnamefont {J.}~\bibnamefont {Kelly}}, \bibinfo {author} {\bibfnamefont {P.~V.}\ \bibnamefont {Klimov}}, \bibinfo {author} {\bibfnamefont {S.}~\bibnamefont {Knysh}}, \bibinfo {author} {\bibfnamefont {A.}~\bibnamefont {Korotkov}}, \bibinfo {author} {\bibfnamefont {F.}~\bibnamefont {Kostritsa}}, \bibinfo {author} {\bibfnamefont {D.}~\bibnamefont {Landhuis}}, \bibinfo {author} {\bibfnamefont {M.}~\bibnamefont {Lindmark}}, \bibinfo {author} {\bibfnamefont {E.}~\bibnamefont {Lucero}}, \bibinfo {author} {\bibfnamefont {D.}~\bibnamefont {Lyakh}}, \bibinfo {author} {\bibfnamefont {S.}~\bibnamefont {Mandr{\`a}}}, \bibinfo {author} {\bibfnamefont {J.~R.}\ \bibnamefont {McClean}}, \bibinfo {author} {\bibfnamefont {M.}~\bibnamefont
  {McEwen}}, \bibinfo {author} {\bibfnamefont {A.}~\bibnamefont {Megrant}}, \bibinfo {author} {\bibfnamefont {X.}~\bibnamefont {Mi}}, \bibinfo {author} {\bibfnamefont {K.}~\bibnamefont {Michielsen}}, \bibinfo {author} {\bibfnamefont {M.}~\bibnamefont {Mohseni}}, \bibinfo {author} {\bibfnamefont {J.}~\bibnamefont {Mutus}}, \bibinfo {author} {\bibfnamefont {O.}~\bibnamefont {Naaman}}, \bibinfo {author} {\bibfnamefont {M.}~\bibnamefont {Neeley}}, \bibinfo {author} {\bibfnamefont {C.}~\bibnamefont {Neill}}, \bibinfo {author} {\bibfnamefont {M.~Y.}\ \bibnamefont {Niu}}, \bibinfo {author} {\bibfnamefont {E.}~\bibnamefont {Ostby}}, \bibinfo {author} {\bibfnamefont {A.}~\bibnamefont {Petukhov}}, \bibinfo {author} {\bibfnamefont {J.~C.}\ \bibnamefont {Platt}}, \bibinfo {author} {\bibfnamefont {C.}~\bibnamefont {Quintana}}, \bibinfo {author} {\bibfnamefont {E.~G.}\ \bibnamefont {Rieffel}}, \bibinfo {author} {\bibfnamefont {P.}~\bibnamefont {Roushan}}, \bibinfo {author} {\bibfnamefont {N.~C.}\ \bibnamefont {Rubin}},
  \bibinfo {author} {\bibfnamefont {D.}~\bibnamefont {Sank}}, \bibinfo {author} {\bibfnamefont {K.~J.}\ \bibnamefont {Satzinger}}, \bibinfo {author} {\bibfnamefont {V.}~\bibnamefont {Smelyanskiy}}, \bibinfo {author} {\bibfnamefont {K.~J.}\ \bibnamefont {Sung}}, \bibinfo {author} {\bibfnamefont {M.~D.}\ \bibnamefont {Trevithick}}, \bibinfo {author} {\bibfnamefont {A.}~\bibnamefont {Vainsencher}}, \bibinfo {author} {\bibfnamefont {B.}~\bibnamefont {Villalonga}}, \bibinfo {author} {\bibfnamefont {T.}~\bibnamefont {White}}, \bibinfo {author} {\bibfnamefont {Z.~J.}\ \bibnamefont {Yao}}, \bibinfo {author} {\bibfnamefont {P.}~\bibnamefont {Yeh}}, \bibinfo {author} {\bibfnamefont {A.}~\bibnamefont {Zalcman}}, \bibinfo {author} {\bibfnamefont {H.}~\bibnamefont {Neven}},\ and\ \bibinfo {author} {\bibfnamefont {J.~M.}\ \bibnamefont {Martinis}},\ }\href {https://doi.org/10.1038/s41586-019-1666-5} {\bibfield  {journal} {\bibinfo  {journal} {Nature}\ }\textbf {\bibinfo {volume} {574}},\ \bibinfo {pages} {505} (\bibinfo
  {year} {2019})}\BibitemShut {NoStop}%
\bibitem [{\citenamefont {Gao}\ \emph {et~al.}(2025)\citenamefont {Gao}, \citenamefont {Fan}, \citenamefont {Zha}, \citenamefont {Bei}, \citenamefont {Cai}, \citenamefont {Cai}, \citenamefont {Cao}, \citenamefont {Chen}, \citenamefont {Chen}, \citenamefont {Chen}, \citenamefont {Chen}, \citenamefont {Chen}, \citenamefont {Chen}, \citenamefont {Chen}, \citenamefont {Chen}, \citenamefont {Chu}, \citenamefont {Deng}, \citenamefont {Deng}, \citenamefont {Ding}, \citenamefont {Ding}, \citenamefont {Ding}, \citenamefont {Dong}, \citenamefont {Dong}, \citenamefont {Fan}, \citenamefont {Fu}, \citenamefont {Gao}, \citenamefont {Ge}, \citenamefont {Gong}, \citenamefont {Gui}, \citenamefont {Guo}, \citenamefont {Guo}, \citenamefont {Guo}, \citenamefont {Han}, \citenamefont {He}, \citenamefont {Hong}, \citenamefont {Hu}, \citenamefont {Huang}, \citenamefont {Huo}, \citenamefont {Jiang}, \citenamefont {Jiang}, \citenamefont {Jin}, \citenamefont {Leng}, \citenamefont {Li}, \citenamefont {Li}, \citenamefont {Li},
  \citenamefont {Li}, \citenamefont {Li}, \citenamefont {Li}, \citenamefont {Li}, \citenamefont {Li}, \citenamefont {Li}, \citenamefont {Li}, \citenamefont {Li}, \citenamefont {Li}, \citenamefont {Liang}, \citenamefont {Liang}, \citenamefont {Liao}, \citenamefont {Lin}, \citenamefont {Lin}, \citenamefont {Liu}, \citenamefont {Liu}, \citenamefont {Liu}, \citenamefont {Liu}, \citenamefont {Liu}, \citenamefont {Liu}, \citenamefont {Lou}, \citenamefont {Ma}, \citenamefont {Meng}, \citenamefont {Mou}, \citenamefont {Nan}, \citenamefont {Nie}, \citenamefont {Nie}, \citenamefont {Ning}, \citenamefont {Niu}, \citenamefont {Peng}, \citenamefont {Qian}, \citenamefont {Rong}, \citenamefont {Rong}, \citenamefont {Shen}, \citenamefont {Shen}, \citenamefont {Su}, \citenamefont {Su}, \citenamefont {Sun}, \citenamefont {Sun}, \citenamefont {Sun}, \citenamefont {Sun}, \citenamefont {Tan}, \citenamefont {Tan}, \citenamefont {Tang}, \citenamefont {Tu}, \citenamefont {Wan}, \citenamefont {Wang}, \citenamefont {Wang},
  \citenamefont {Wang}, \citenamefont {Wang}, \citenamefont {Wang}, \citenamefont {Wang}, \citenamefont {Wang}, \citenamefont {Wang}, \citenamefont {Wang}, \citenamefont {Wang}, \citenamefont {Wang}, \citenamefont {Wang}, \citenamefont {Wang}, \citenamefont {Wei}, \citenamefont {Wei}, \citenamefont {Wu}, \citenamefont {Wu}, \citenamefont {Wu}, \citenamefont {Wu}, \citenamefont {Wu}, \citenamefont {Xie}, \citenamefont {Xin}, \citenamefont {Xu}, \citenamefont {Xue}, \citenamefont {Yan}, \citenamefont {Yang}, \citenamefont {Yang}, \citenamefont {Yang}, \citenamefont {Ye}, \citenamefont {Ye}, \citenamefont {Ying}, \citenamefont {Yu}, \citenamefont {Yu}, \citenamefont {Yu}, \citenamefont {Zeng}, \citenamefont {Zhan}, \citenamefont {Zhang}, \citenamefont {Zhang}, \citenamefont {Zhang}, \citenamefont {Zhang}, \citenamefont {Zhang}, \citenamefont {Zhang}, \citenamefont {Zhang}, \citenamefont {Zhang}, \citenamefont {Zhao}, \citenamefont {Zhao}, \citenamefont {Zhao}, \citenamefont {Zhao}, \citenamefont {Zhao},
  \citenamefont {Zhao}, \citenamefont {Zheng}, \citenamefont {Zhou}, \citenamefont {Zhou}, \citenamefont {Zhou}, \citenamefont {Zhou}, \citenamefont {Zhou}, \citenamefont {Zhou}, \citenamefont {Zhou}, \citenamefont {Zhu}, \citenamefont {Zhu}, \citenamefont {Zou}, \citenamefont {Zou}, \citenamefont {Zhang}, \citenamefont {Lu}, \citenamefont {Peng}, \citenamefont {Zhu},\ and\ \citenamefont {Pan}}]{PhysRevLett.134.090601}%
  \BibitemOpen
  \bibfield  {author} {\bibinfo {author} {\bibfnamefont {D.}~\bibnamefont {Gao}}, \bibinfo {author} {\bibfnamefont {D.}~\bibnamefont {Fan}}, \bibinfo {author} {\bibfnamefont {C.}~\bibnamefont {Zha}}, \bibinfo {author} {\bibfnamefont {J.}~\bibnamefont {Bei}}, \bibinfo {author} {\bibfnamefont {G.}~\bibnamefont {Cai}}, \bibinfo {author} {\bibfnamefont {J.}~\bibnamefont {Cai}}, \bibinfo {author} {\bibfnamefont {S.}~\bibnamefont {Cao}}, \bibinfo {author} {\bibfnamefont {F.}~\bibnamefont {Chen}}, \bibinfo {author} {\bibfnamefont {J.}~\bibnamefont {Chen}}, \bibinfo {author} {\bibfnamefont {K.}~\bibnamefont {Chen}}, \bibinfo {author} {\bibfnamefont {X.}~\bibnamefont {Chen}}, \bibinfo {author} {\bibfnamefont {X.}~\bibnamefont {Chen}}, \bibinfo {author} {\bibfnamefont {Z.}~\bibnamefont {Chen}}, \bibinfo {author} {\bibfnamefont {Z.}~\bibnamefont {Chen}}, \bibinfo {author} {\bibfnamefont {Z.}~\bibnamefont {Chen}}, \bibinfo {author} {\bibfnamefont {W.}~\bibnamefont {Chu}}, \bibinfo {author} {\bibfnamefont
  {H.}~\bibnamefont {Deng}}, \bibinfo {author} {\bibfnamefont {Z.}~\bibnamefont {Deng}}, \bibinfo {author} {\bibfnamefont {P.}~\bibnamefont {Ding}}, \bibinfo {author} {\bibfnamefont {X.}~\bibnamefont {Ding}}, \bibinfo {author} {\bibfnamefont {Z.}~\bibnamefont {Ding}}, \bibinfo {author} {\bibfnamefont {S.}~\bibnamefont {Dong}}, \bibinfo {author} {\bibfnamefont {Y.}~\bibnamefont {Dong}}, \bibinfo {author} {\bibfnamefont {B.}~\bibnamefont {Fan}}, \bibinfo {author} {\bibfnamefont {Y.}~\bibnamefont {Fu}}, \bibinfo {author} {\bibfnamefont {S.}~\bibnamefont {Gao}}, \bibinfo {author} {\bibfnamefont {L.}~\bibnamefont {Ge}}, \bibinfo {author} {\bibfnamefont {M.}~\bibnamefont {Gong}}, \bibinfo {author} {\bibfnamefont {J.}~\bibnamefont {Gui}}, \bibinfo {author} {\bibfnamefont {C.}~\bibnamefont {Guo}}, \bibinfo {author} {\bibfnamefont {S.}~\bibnamefont {Guo}}, \bibinfo {author} {\bibfnamefont {X.}~\bibnamefont {Guo}}, \bibinfo {author} {\bibfnamefont {L.}~\bibnamefont {Han}}, \bibinfo {author} {\bibfnamefont
  {T.}~\bibnamefont {He}}, \bibinfo {author} {\bibfnamefont {L.}~\bibnamefont {Hong}}, \bibinfo {author} {\bibfnamefont {Y.}~\bibnamefont {Hu}}, \bibinfo {author} {\bibfnamefont {H.-L.}\ \bibnamefont {Huang}}, \bibinfo {author} {\bibfnamefont {Y.-H.}\ \bibnamefont {Huo}}, \bibinfo {author} {\bibfnamefont {T.}~\bibnamefont {Jiang}}, \bibinfo {author} {\bibfnamefont {Z.}~\bibnamefont {Jiang}}, \bibinfo {author} {\bibfnamefont {H.}~\bibnamefont {Jin}}, \bibinfo {author} {\bibfnamefont {Y.}~\bibnamefont {Leng}}, \bibinfo {author} {\bibfnamefont {D.}~\bibnamefont {Li}}, \bibinfo {author} {\bibfnamefont {D.}~\bibnamefont {Li}}, \bibinfo {author} {\bibfnamefont {F.}~\bibnamefont {Li}}, \bibinfo {author} {\bibfnamefont {J.}~\bibnamefont {Li}}, \bibinfo {author} {\bibfnamefont {J.}~\bibnamefont {Li}}, \bibinfo {author} {\bibfnamefont {J.}~\bibnamefont {Li}}, \bibinfo {author} {\bibfnamefont {J.}~\bibnamefont {Li}}, \bibinfo {author} {\bibfnamefont {N.}~\bibnamefont {Li}}, \bibinfo {author} {\bibfnamefont
  {S.}~\bibnamefont {Li}}, \bibinfo {author} {\bibfnamefont {W.}~\bibnamefont {Li}}, \bibinfo {author} {\bibfnamefont {Y.}~\bibnamefont {Li}}, \bibinfo {author} {\bibfnamefont {Y.}~\bibnamefont {Li}}, \bibinfo {author} {\bibfnamefont {F.}~\bibnamefont {Liang}}, \bibinfo {author} {\bibfnamefont {X.}~\bibnamefont {Liang}}, \bibinfo {author} {\bibfnamefont {N.}~\bibnamefont {Liao}}, \bibinfo {author} {\bibfnamefont {J.}~\bibnamefont {Lin}}, \bibinfo {author} {\bibfnamefont {W.}~\bibnamefont {Lin}}, \bibinfo {author} {\bibfnamefont {D.}~\bibnamefont {Liu}}, \bibinfo {author} {\bibfnamefont {H.}~\bibnamefont {Liu}}, \bibinfo {author} {\bibfnamefont {M.}~\bibnamefont {Liu}}, \bibinfo {author} {\bibfnamefont {X.}~\bibnamefont {Liu}}, \bibinfo {author} {\bibfnamefont {X.}~\bibnamefont {Liu}}, \bibinfo {author} {\bibfnamefont {Y.}~\bibnamefont {Liu}}, \bibinfo {author} {\bibfnamefont {H.}~\bibnamefont {Lou}}, \bibinfo {author} {\bibfnamefont {Y.}~\bibnamefont {Ma}}, \bibinfo {author} {\bibfnamefont {L.}~\bibnamefont
  {Meng}}, \bibinfo {author} {\bibfnamefont {H.}~\bibnamefont {Mou}}, \bibinfo {author} {\bibfnamefont {K.}~\bibnamefont {Nan}}, \bibinfo {author} {\bibfnamefont {B.}~\bibnamefont {Nie}}, \bibinfo {author} {\bibfnamefont {M.}~\bibnamefont {Nie}}, \bibinfo {author} {\bibfnamefont {J.}~\bibnamefont {Ning}}, \bibinfo {author} {\bibfnamefont {L.}~\bibnamefont {Niu}}, \bibinfo {author} {\bibfnamefont {W.}~\bibnamefont {Peng}}, \bibinfo {author} {\bibfnamefont {H.}~\bibnamefont {Qian}}, \bibinfo {author} {\bibfnamefont {H.}~\bibnamefont {Rong}}, \bibinfo {author} {\bibfnamefont {T.}~\bibnamefont {Rong}}, \bibinfo {author} {\bibfnamefont {H.}~\bibnamefont {Shen}}, \bibinfo {author} {\bibfnamefont {Q.}~\bibnamefont {Shen}}, \bibinfo {author} {\bibfnamefont {H.}~\bibnamefont {Su}}, \bibinfo {author} {\bibfnamefont {F.}~\bibnamefont {Su}}, \bibinfo {author} {\bibfnamefont {C.}~\bibnamefont {Sun}}, \bibinfo {author} {\bibfnamefont {L.}~\bibnamefont {Sun}}, \bibinfo {author} {\bibfnamefont {T.}~\bibnamefont {Sun}},
  \bibinfo {author} {\bibfnamefont {Y.}~\bibnamefont {Sun}}, \bibinfo {author} {\bibfnamefont {Y.}~\bibnamefont {Tan}}, \bibinfo {author} {\bibfnamefont {J.}~\bibnamefont {Tan}}, \bibinfo {author} {\bibfnamefont {L.}~\bibnamefont {Tang}}, \bibinfo {author} {\bibfnamefont {W.}~\bibnamefont {Tu}}, \bibinfo {author} {\bibfnamefont {C.}~\bibnamefont {Wan}}, \bibinfo {author} {\bibfnamefont {J.}~\bibnamefont {Wang}}, \bibinfo {author} {\bibfnamefont {B.}~\bibnamefont {Wang}}, \bibinfo {author} {\bibfnamefont {C.}~\bibnamefont {Wang}}, \bibinfo {author} {\bibfnamefont {C.}~\bibnamefont {Wang}}, \bibinfo {author} {\bibfnamefont {C.}~\bibnamefont {Wang}}, \bibinfo {author} {\bibfnamefont {J.}~\bibnamefont {Wang}}, \bibinfo {author} {\bibfnamefont {L.}~\bibnamefont {Wang}}, \bibinfo {author} {\bibfnamefont {R.}~\bibnamefont {Wang}}, \bibinfo {author} {\bibfnamefont {S.}~\bibnamefont {Wang}}, \bibinfo {author} {\bibfnamefont {X.}~\bibnamefont {Wang}}, \bibinfo {author} {\bibfnamefont {X.}~\bibnamefont {Wang}}, \bibinfo
  {author} {\bibfnamefont {X.}~\bibnamefont {Wang}}, \bibinfo {author} {\bibfnamefont {Y.}~\bibnamefont {Wang}}, \bibinfo {author} {\bibfnamefont {Z.}~\bibnamefont {Wei}}, \bibinfo {author} {\bibfnamefont {J.}~\bibnamefont {Wei}}, \bibinfo {author} {\bibfnamefont {D.}~\bibnamefont {Wu}}, \bibinfo {author} {\bibfnamefont {G.}~\bibnamefont {Wu}}, \bibinfo {author} {\bibfnamefont {J.}~\bibnamefont {Wu}}, \bibinfo {author} {\bibfnamefont {S.}~\bibnamefont {Wu}}, \bibinfo {author} {\bibfnamefont {Y.}~\bibnamefont {Wu}}, \bibinfo {author} {\bibfnamefont {S.}~\bibnamefont {Xie}}, \bibinfo {author} {\bibfnamefont {L.}~\bibnamefont {Xin}}, \bibinfo {author} {\bibfnamefont {Y.}~\bibnamefont {Xu}}, \bibinfo {author} {\bibfnamefont {C.}~\bibnamefont {Xue}}, \bibinfo {author} {\bibfnamefont {K.}~\bibnamefont {Yan}}, \bibinfo {author} {\bibfnamefont {W.}~\bibnamefont {Yang}}, \bibinfo {author} {\bibfnamefont {X.}~\bibnamefont {Yang}}, \bibinfo {author} {\bibfnamefont {Y.}~\bibnamefont {Yang}}, \bibinfo {author}
  {\bibfnamefont {Y.}~\bibnamefont {Ye}}, \bibinfo {author} {\bibfnamefont {Z.}~\bibnamefont {Ye}}, \bibinfo {author} {\bibfnamefont {C.}~\bibnamefont {Ying}}, \bibinfo {author} {\bibfnamefont {J.}~\bibnamefont {Yu}}, \bibinfo {author} {\bibfnamefont {Q.}~\bibnamefont {Yu}}, \bibinfo {author} {\bibfnamefont {W.}~\bibnamefont {Yu}}, \bibinfo {author} {\bibfnamefont {X.}~\bibnamefont {Zeng}}, \bibinfo {author} {\bibfnamefont {S.}~\bibnamefont {Zhan}}, \bibinfo {author} {\bibfnamefont {F.}~\bibnamefont {Zhang}}, \bibinfo {author} {\bibfnamefont {H.}~\bibnamefont {Zhang}}, \bibinfo {author} {\bibfnamefont {K.}~\bibnamefont {Zhang}}, \bibinfo {author} {\bibfnamefont {P.}~\bibnamefont {Zhang}}, \bibinfo {author} {\bibfnamefont {W.}~\bibnamefont {Zhang}}, \bibinfo {author} {\bibfnamefont {Y.}~\bibnamefont {Zhang}}, \bibinfo {author} {\bibfnamefont {Y.}~\bibnamefont {Zhang}}, \bibinfo {author} {\bibfnamefont {L.}~\bibnamefont {Zhang}}, \bibinfo {author} {\bibfnamefont {G.}~\bibnamefont {Zhao}}, \bibinfo {author}
  {\bibfnamefont {P.}~\bibnamefont {Zhao}}, \bibinfo {author} {\bibfnamefont {X.}~\bibnamefont {Zhao}}, \bibinfo {author} {\bibfnamefont {X.}~\bibnamefont {Zhao}}, \bibinfo {author} {\bibfnamefont {Y.}~\bibnamefont {Zhao}}, \bibinfo {author} {\bibfnamefont {Z.}~\bibnamefont {Zhao}}, \bibinfo {author} {\bibfnamefont {L.}~\bibnamefont {Zheng}}, \bibinfo {author} {\bibfnamefont {F.}~\bibnamefont {Zhou}}, \bibinfo {author} {\bibfnamefont {L.}~\bibnamefont {Zhou}}, \bibinfo {author} {\bibfnamefont {N.}~\bibnamefont {Zhou}}, \bibinfo {author} {\bibfnamefont {N.}~\bibnamefont {Zhou}}, \bibinfo {author} {\bibfnamefont {S.}~\bibnamefont {Zhou}}, \bibinfo {author} {\bibfnamefont {S.}~\bibnamefont {Zhou}}, \bibinfo {author} {\bibfnamefont {Z.}~\bibnamefont {Zhou}}, \bibinfo {author} {\bibfnamefont {C.}~\bibnamefont {Zhu}}, \bibinfo {author} {\bibfnamefont {Q.}~\bibnamefont {Zhu}}, \bibinfo {author} {\bibfnamefont {G.}~\bibnamefont {Zou}}, \bibinfo {author} {\bibfnamefont {H.}~\bibnamefont {Zou}}, \bibinfo {author}
  {\bibfnamefont {Q.}~\bibnamefont {Zhang}}, \bibinfo {author} {\bibfnamefont {C.-Y.}\ \bibnamefont {Lu}}, \bibinfo {author} {\bibfnamefont {C.-Z.}\ \bibnamefont {Peng}}, \bibinfo {author} {\bibfnamefont {X.}~\bibnamefont {Zhu}},\ and\ \bibinfo {author} {\bibfnamefont {J.-W.}\ \bibnamefont {Pan}},\ }\href {https://doi.org/10.1103/PhysRevLett.134.090601} {\bibfield  {journal} {\bibinfo  {journal} {Phys. Rev. Lett.}\ }\textbf {\bibinfo {volume} {134}},\ \bibinfo {pages} {090601} (\bibinfo {year} {2025})}\BibitemShut {NoStop}%
\bibitem [{\citenamefont {Choi}\ \emph {et~al.}(2023)\citenamefont {Choi}, \citenamefont {Shaw}, \citenamefont {Madjarov}, \citenamefont {Xie}, \citenamefont {Finkelstein}, \citenamefont {Covey}, \citenamefont {Cotler}, \citenamefont {Mark}, \citenamefont {Huang}, \citenamefont {Kale}, \citenamefont {Pichler}, \citenamefont {Brand{\~a}o}, \citenamefont {Choi},\ and\ \citenamefont {Endres}}]{Joonhee_23}%
  \BibitemOpen
  \bibfield  {author} {\bibinfo {author} {\bibfnamefont {J.}~\bibnamefont {Choi}}, \bibinfo {author} {\bibfnamefont {A.~L.}\ \bibnamefont {Shaw}}, \bibinfo {author} {\bibfnamefont {I.~S.}\ \bibnamefont {Madjarov}}, \bibinfo {author} {\bibfnamefont {X.}~\bibnamefont {Xie}}, \bibinfo {author} {\bibfnamefont {R.}~\bibnamefont {Finkelstein}}, \bibinfo {author} {\bibfnamefont {J.~P.}\ \bibnamefont {Covey}}, \bibinfo {author} {\bibfnamefont {J.~S.}\ \bibnamefont {Cotler}}, \bibinfo {author} {\bibfnamefont {D.~K.}\ \bibnamefont {Mark}}, \bibinfo {author} {\bibfnamefont {H.-Y.}\ \bibnamefont {Huang}}, \bibinfo {author} {\bibfnamefont {A.}~\bibnamefont {Kale}}, \bibinfo {author} {\bibfnamefont {H.}~\bibnamefont {Pichler}}, \bibinfo {author} {\bibfnamefont {F.~G. S.~L.}\ \bibnamefont {Brand{\~a}o}}, \bibinfo {author} {\bibfnamefont {S.}~\bibnamefont {Choi}},\ and\ \bibinfo {author} {\bibfnamefont {M.}~\bibnamefont {Endres}},\ }\href {https://doi.org/10.1038/s41586-022-05442-1} {\bibfield  {journal} {\bibinfo  {journal}
  {Nature}\ }\textbf {\bibinfo {volume} {613}},\ \bibinfo {pages} {468} (\bibinfo {year} {2023})}\BibitemShut {NoStop}%
\bibitem [{\citenamefont {Mark}\ \emph {et~al.}(2023)\citenamefont {Mark}, \citenamefont {Choi}, \citenamefont {Shaw}, \citenamefont {Endres},\ and\ \citenamefont {Choi}}]{PhysRevLett.131.110601}%
  \BibitemOpen
  \bibfield  {author} {\bibinfo {author} {\bibfnamefont {D.~K.}\ \bibnamefont {Mark}}, \bibinfo {author} {\bibfnamefont {J.}~\bibnamefont {Choi}}, \bibinfo {author} {\bibfnamefont {A.~L.}\ \bibnamefont {Shaw}}, \bibinfo {author} {\bibfnamefont {M.}~\bibnamefont {Endres}},\ and\ \bibinfo {author} {\bibfnamefont {S.}~\bibnamefont {Choi}},\ }\href {https://doi.org/10.1103/PhysRevLett.131.110601} {\bibfield  {journal} {\bibinfo  {journal} {Phys. Rev. Lett.}\ }\textbf {\bibinfo {volume} {131}},\ \bibinfo {pages} {110601} (\bibinfo {year} {2023})}\BibitemShut {NoStop}%
\bibitem [{\citenamefont {{Aaronson}}\ and\ \citenamefont {{Hung}}(2023)}]{2023arXiv230301625A}%
  \BibitemOpen
  \bibfield  {author} {\bibinfo {author} {\bibfnamefont {S.}~\bibnamefont {{Aaronson}}}\ and\ \bibinfo {author} {\bibfnamefont {S.-H.}\ \bibnamefont {{Hung}}},\ }\href {https://doi.org/10.48550/arXiv.2303.01625} {\bibfield  {journal} {\bibinfo  {journal} {arXiv e-prints}\ ,\ \bibinfo {eid} {arXiv:2303.01625}} (\bibinfo {year} {2023})},\ \Eprint {https://arxiv.org/abs/2303.01625} {arXiv:2303.01625 [quant-ph]} \BibitemShut {NoStop}%
\bibitem [{\citenamefont {Liu}\ \emph {et~al.}(2025)\citenamefont {Liu}, \citenamefont {Shaydulin}, \citenamefont {Niroula}, \citenamefont {DeCross}, \citenamefont {Hung}, \citenamefont {Kon}, \citenamefont {Cervero-Mart{\'\i}n}, \citenamefont {Chakraborty}, \citenamefont {Amer}, \citenamefont {Aaronson}, \citenamefont {Acharya}, \citenamefont {Alexeev}, \citenamefont {Berg}, \citenamefont {Chakrabarti}, \citenamefont {Curchod}, \citenamefont {Dreiling}, \citenamefont {Erickson}, \citenamefont {Foltz}, \citenamefont {Foss-Feig}, \citenamefont {Hayes}, \citenamefont {Humble}, \citenamefont {Kumar}, \citenamefont {Larson}, \citenamefont {Lykov}, \citenamefont {Mills}, \citenamefont {Moses}, \citenamefont {Neyenhuis}, \citenamefont {Eloul}, \citenamefont {Siegfried}, \citenamefont {Walker}, \citenamefont {Lim},\ and\ \citenamefont {Pistoia}}]{Liu_25}%
  \BibitemOpen
  \bibfield  {author} {\bibinfo {author} {\bibfnamefont {M.}~\bibnamefont {Liu}}, \bibinfo {author} {\bibfnamefont {R.}~\bibnamefont {Shaydulin}}, \bibinfo {author} {\bibfnamefont {P.}~\bibnamefont {Niroula}}, \bibinfo {author} {\bibfnamefont {M.}~\bibnamefont {DeCross}}, \bibinfo {author} {\bibfnamefont {S.-H.}\ \bibnamefont {Hung}}, \bibinfo {author} {\bibfnamefont {W.~Y.}\ \bibnamefont {Kon}}, \bibinfo {author} {\bibfnamefont {E.}~\bibnamefont {Cervero-Mart{\'\i}n}}, \bibinfo {author} {\bibfnamefont {K.}~\bibnamefont {Chakraborty}}, \bibinfo {author} {\bibfnamefont {O.}~\bibnamefont {Amer}}, \bibinfo {author} {\bibfnamefont {S.}~\bibnamefont {Aaronson}}, \bibinfo {author} {\bibfnamefont {A.}~\bibnamefont {Acharya}}, \bibinfo {author} {\bibfnamefont {Y.}~\bibnamefont {Alexeev}}, \bibinfo {author} {\bibfnamefont {K.~J.}\ \bibnamefont {Berg}}, \bibinfo {author} {\bibfnamefont {S.}~\bibnamefont {Chakrabarti}}, \bibinfo {author} {\bibfnamefont {F.~J.}\ \bibnamefont {Curchod}}, \bibinfo {author} {\bibfnamefont
  {J.~M.}\ \bibnamefont {Dreiling}}, \bibinfo {author} {\bibfnamefont {N.}~\bibnamefont {Erickson}}, \bibinfo {author} {\bibfnamefont {C.}~\bibnamefont {Foltz}}, \bibinfo {author} {\bibfnamefont {M.}~\bibnamefont {Foss-Feig}}, \bibinfo {author} {\bibfnamefont {D.}~\bibnamefont {Hayes}}, \bibinfo {author} {\bibfnamefont {T.~S.}\ \bibnamefont {Humble}}, \bibinfo {author} {\bibfnamefont {N.}~\bibnamefont {Kumar}}, \bibinfo {author} {\bibfnamefont {J.}~\bibnamefont {Larson}}, \bibinfo {author} {\bibfnamefont {D.}~\bibnamefont {Lykov}}, \bibinfo {author} {\bibfnamefont {M.}~\bibnamefont {Mills}}, \bibinfo {author} {\bibfnamefont {S.~A.}\ \bibnamefont {Moses}}, \bibinfo {author} {\bibfnamefont {B.}~\bibnamefont {Neyenhuis}}, \bibinfo {author} {\bibfnamefont {S.}~\bibnamefont {Eloul}}, \bibinfo {author} {\bibfnamefont {P.}~\bibnamefont {Siegfried}}, \bibinfo {author} {\bibfnamefont {J.}~\bibnamefont {Walker}}, \bibinfo {author} {\bibfnamefont {C.}~\bibnamefont {Lim}},\ and\ \bibinfo {author} {\bibfnamefont
  {M.}~\bibnamefont {Pistoia}},\ }\href {https://doi.org/10.1038/s41586-025-08737-1} {\bibfield  {journal} {\bibinfo  {journal} {Nature}\ }\textbf {\bibinfo {volume} {640}},\ \bibinfo {pages} {343} (\bibinfo {year} {2025})}\BibitemShut {NoStop}%
\bibitem [{\citenamefont {Chamon}\ \emph {et~al.}(2022)\citenamefont {Chamon}, \citenamefont {Mucciolo},\ and\ \citenamefont {Ruckenstein}}]{CHAMON2022}%
  \BibitemOpen
  \bibfield  {author} {\bibinfo {author} {\bibfnamefont {C.}~\bibnamefont {Chamon}}, \bibinfo {author} {\bibfnamefont {E.~R.}\ \bibnamefont {Mucciolo}},\ and\ \bibinfo {author} {\bibfnamefont {A.~E.}\ \bibnamefont {Ruckenstein}},\ }\href {https://doi.org/https://doi.org/10.1016/j.aop.2022.169086} {\bibfield  {journal} {\bibinfo  {journal} {Annals of Physics}\ }\textbf {\bibinfo {volume} {446}},\ \bibinfo {pages} {169086} (\bibinfo {year} {2022})}\BibitemShut {NoStop}%
\bibitem [{\citenamefont {Anand}\ \emph {et~al.}(2024)\citenamefont {Anand}, \citenamefont {Bradley}, \citenamefont {White}, \citenamefont {Ramesh}, \citenamefont {Singh},\ and\ \citenamefont {Bernien}}]{dualspecies}%
  \BibitemOpen
  \bibfield  {author} {\bibinfo {author} {\bibfnamefont {S.}~\bibnamefont {Anand}}, \bibinfo {author} {\bibfnamefont {C.~E.}\ \bibnamefont {Bradley}}, \bibinfo {author} {\bibfnamefont {R.}~\bibnamefont {White}}, \bibinfo {author} {\bibfnamefont {V.}~\bibnamefont {Ramesh}}, \bibinfo {author} {\bibfnamefont {K.}~\bibnamefont {Singh}},\ and\ \bibinfo {author} {\bibfnamefont {H.}~\bibnamefont {Bernien}},\ }\href {https://doi.org/10.1038/s41567-024-02638-2} {\bibfield  {journal} {\bibinfo  {journal} {Nature Physics}\ }\textbf {\bibinfo {volume} {20}},\ \bibinfo {pages} {1744} (\bibinfo {year} {2024})}\BibitemShut {NoStop}%
\bibitem [{\citenamefont {Fang}\ \emph {et~al.}(2025)\citenamefont {Fang}, \citenamefont {Miles}, \citenamefont {Goldwin}, \citenamefont {Lichtman}, \citenamefont {Gillette}, \citenamefont {Bergdolt}, \citenamefont {Deshpande}, \citenamefont {Norrell}, \citenamefont {Huft}, \citenamefont {Kats},\ and\ \citenamefont {Saffman}}]{dualspecies2}%
  \BibitemOpen
  \bibfield  {author} {\bibinfo {author} {\bibfnamefont {C.}~\bibnamefont {Fang}}, \bibinfo {author} {\bibfnamefont {J.}~\bibnamefont {Miles}}, \bibinfo {author} {\bibfnamefont {J.}~\bibnamefont {Goldwin}}, \bibinfo {author} {\bibfnamefont {M.}~\bibnamefont {Lichtman}}, \bibinfo {author} {\bibfnamefont {M.}~\bibnamefont {Gillette}}, \bibinfo {author} {\bibfnamefont {M.}~\bibnamefont {Bergdolt}}, \bibinfo {author} {\bibfnamefont {S.}~\bibnamefont {Deshpande}}, \bibinfo {author} {\bibfnamefont {S.~A.}\ \bibnamefont {Norrell}}, \bibinfo {author} {\bibfnamefont {P.}~\bibnamefont {Huft}}, \bibinfo {author} {\bibfnamefont {M.~A.}\ \bibnamefont {Kats}},\ and\ \bibinfo {author} {\bibfnamefont {M.}~\bibnamefont {Saffman}},\ }\href {https://doi.org/10.1126/sciadv.adw4166} {\bibfield  {journal} {\bibinfo  {journal} {Science Advances}\ }\textbf {\bibinfo {volume} {11}},\ \bibinfo {pages} {eadw4166} (\bibinfo {year} {2025})},\ \Eprint {https://arxiv.org/abs/https://www.science.org/doi/pdf/10.1126/sciadv.adw4166}
  {https://www.science.org/doi/pdf/10.1126/sciadv.adw4166} \BibitemShut {NoStop}%
\bibitem [{\citenamefont {{White}}\ \emph {et~al.}(2026)\citenamefont {{White}}, \citenamefont {{Ramesh}}, \citenamefont {{Impertro}}, \citenamefont {{Anand}}, \citenamefont {{Cesa}}, \citenamefont {{Giudici}}, \citenamefont {{Iadecola}}, \citenamefont {{Pichler}},\ and\ \citenamefont {{Bernien}}}]{dualspecies3}%
  \BibitemOpen
  \bibfield  {author} {\bibinfo {author} {\bibfnamefont {R.}~\bibnamefont {{White}}}, \bibinfo {author} {\bibfnamefont {V.}~\bibnamefont {{Ramesh}}}, \bibinfo {author} {\bibfnamefont {A.}~\bibnamefont {{Impertro}}}, \bibinfo {author} {\bibfnamefont {S.}~\bibnamefont {{Anand}}}, \bibinfo {author} {\bibfnamefont {F.}~\bibnamefont {{Cesa}}}, \bibinfo {author} {\bibfnamefont {G.}~\bibnamefont {{Giudici}}}, \bibinfo {author} {\bibfnamefont {T.}~\bibnamefont {{Iadecola}}}, \bibinfo {author} {\bibfnamefont {H.}~\bibnamefont {{Pichler}}},\ and\ \bibinfo {author} {\bibfnamefont {H.}~\bibnamefont {{Bernien}}},\ }\href@noop {} {\bibfield  {journal} {\bibinfo  {journal} {arXiv e-prints}\ ,\ \bibinfo {eid} {arXiv:2601.16257}} (\bibinfo {year} {2026})},\ \Eprint {https://arxiv.org/abs/2601.16257} {arXiv:2601.16257 [quant-ph]} \BibitemShut {NoStop}%
\bibitem [{\citenamefont {Manchester}\ and\ \citenamefont {Kuindersma}(2019)}]{Zac_Direct_Method}%
  \BibitemOpen
  \bibfield  {author} {\bibinfo {author} {\bibfnamefont {Z.}~\bibnamefont {Manchester}}\ and\ \bibinfo {author} {\bibfnamefont {S.}~\bibnamefont {Kuindersma}},\ }\href {https://doi.org/10.1007/s10514-018-9779-5} {\bibfield  {journal} {\bibinfo  {journal} {Auton. Robots}\ }\textbf {\bibinfo {volume} {43}},\ \bibinfo {pages} {375–387} (\bibinfo {year} {2019})}\BibitemShut {NoStop}%
\bibitem [{\citenamefont {Goldschmidt}\ \emph {et~al.}(2022)\citenamefont {Goldschmidt}, \citenamefont {DuBois}, \citenamefont {Brunton},\ and\ \citenamefont {Kutz}}]{goldschmidt2022model}%
  \BibitemOpen
  \bibfield  {author} {\bibinfo {author} {\bibfnamefont {A.~J.}\ \bibnamefont {Goldschmidt}}, \bibinfo {author} {\bibfnamefont {J.~L.}\ \bibnamefont {DuBois}}, \bibinfo {author} {\bibfnamefont {S.~L.}\ \bibnamefont {Brunton}},\ and\ \bibinfo {author} {\bibfnamefont {J.~N.}\ \bibnamefont {Kutz}},\ }\href {https://doi.org/10.22331/q-2022-10-13-837} {\bibfield  {journal} {\bibinfo  {journal} {{Quantum}}\ }\textbf {\bibinfo {volume} {6}},\ \bibinfo {pages} {837} (\bibinfo {year} {2022})}\BibitemShut {NoStop}%
\bibitem [{\citenamefont {Trowbridge}\ \emph {et~al.}(2023)\citenamefont {Trowbridge}, \citenamefont {Bhardwaj}, \citenamefont {He}, \citenamefont {Schuster},\ and\ \citenamefont {Manchester}}]{trowbridge2023direct}%
  \BibitemOpen
  \bibfield  {author} {\bibinfo {author} {\bibfnamefont {A.}~\bibnamefont {Trowbridge}}, \bibinfo {author} {\bibfnamefont {A.}~\bibnamefont {Bhardwaj}}, \bibinfo {author} {\bibfnamefont {K.}~\bibnamefont {He}}, \bibinfo {author} {\bibfnamefont {D.~I.}\ \bibnamefont {Schuster}},\ and\ \bibinfo {author} {\bibfnamefont {Z.}~\bibnamefont {Manchester}},\ }in\ \href {https://doi.org/10.1109/QCE57702.2023.00144} {\emph {\bibinfo {booktitle} {2023 IEEE International Conference on Quantum Computing and Engineering (QCE)}}},\ Vol.~\bibinfo {volume} {1}\ (\bibinfo {organization} {IEEE},\ \bibinfo {year} {2023})\ pp.\ \bibinfo {pages} {1278--1285}\BibitemShut {NoStop}%
\bibitem [{\citenamefont {Goldman}\ and\ \citenamefont {Dalibard}(2014)}]{FloquetPRX}%
  \BibitemOpen
  \bibfield  {author} {\bibinfo {author} {\bibfnamefont {N.}~\bibnamefont {Goldman}}\ and\ \bibinfo {author} {\bibfnamefont {J.}~\bibnamefont {Dalibard}},\ }\href {https://doi.org/10.1103/PhysRevX.4.031027} {\bibfield  {journal} {\bibinfo  {journal} {Phys. Rev. X}\ }\textbf {\bibinfo {volume} {4}},\ \bibinfo {pages} {031027} (\bibinfo {year} {2014})}\BibitemShut {NoStop}%
\bibitem [{\citenamefont {Bordia}\ \emph {et~al.}(2017)\citenamefont {Bordia}, \citenamefont {L{\"u}schen}, \citenamefont {Schneider}, \citenamefont {Knap},\ and\ \citenamefont {Bloch}}]{Bloch_floquet}%
  \BibitemOpen
  \bibfield  {author} {\bibinfo {author} {\bibfnamefont {P.}~\bibnamefont {Bordia}}, \bibinfo {author} {\bibfnamefont {H.}~\bibnamefont {L{\"u}schen}}, \bibinfo {author} {\bibfnamefont {U.}~\bibnamefont {Schneider}}, \bibinfo {author} {\bibfnamefont {M.}~\bibnamefont {Knap}},\ and\ \bibinfo {author} {\bibfnamefont {I.}~\bibnamefont {Bloch}},\ }\href {https://doi.org/10.1038/nphys4020} {\bibfield  {journal} {\bibinfo  {journal} {Nature Physics}\ }\textbf {\bibinfo {volume} {13}},\ \bibinfo {pages} {460} (\bibinfo {year} {2017})}\BibitemShut {NoStop}%
\bibitem [{\citenamefont {Bluvstein}\ \emph {et~al.}(2021{\natexlab{a}})\citenamefont {Bluvstein}, \citenamefont {Omran}, \citenamefont {Levine}, \citenamefont {Keesling}, \citenamefont {Semeghini}, \citenamefont {Ebadi}, \citenamefont {Wang}, \citenamefont {Michailidis}, \citenamefont {Maskara}, \citenamefont {Ho}, \citenamefont {Choi}, \citenamefont {Serbyn}, \citenamefont {Greiner}, \citenamefont {Vuletić},\ and\ \citenamefont {Lukin}}]{RydbergFloquet}%
  \BibitemOpen
  \bibfield  {author} {\bibinfo {author} {\bibfnamefont {D.}~\bibnamefont {Bluvstein}}, \bibinfo {author} {\bibfnamefont {A.}~\bibnamefont {Omran}}, \bibinfo {author} {\bibfnamefont {H.}~\bibnamefont {Levine}}, \bibinfo {author} {\bibfnamefont {A.}~\bibnamefont {Keesling}}, \bibinfo {author} {\bibfnamefont {G.}~\bibnamefont {Semeghini}}, \bibinfo {author} {\bibfnamefont {S.}~\bibnamefont {Ebadi}}, \bibinfo {author} {\bibfnamefont {T.~T.}\ \bibnamefont {Wang}}, \bibinfo {author} {\bibfnamefont {A.~A.}\ \bibnamefont {Michailidis}}, \bibinfo {author} {\bibfnamefont {N.}~\bibnamefont {Maskara}}, \bibinfo {author} {\bibfnamefont {W.~W.}\ \bibnamefont {Ho}}, \bibinfo {author} {\bibfnamefont {S.}~\bibnamefont {Choi}}, \bibinfo {author} {\bibfnamefont {M.}~\bibnamefont {Serbyn}}, \bibinfo {author} {\bibfnamefont {M.}~\bibnamefont {Greiner}}, \bibinfo {author} {\bibfnamefont {V.}~\bibnamefont {Vuletić}},\ and\ \bibinfo {author} {\bibfnamefont {M.~D.}\ \bibnamefont {Lukin}},\ }\href
  {https://doi.org/10.1126/science.abg2530} {\bibfield  {journal} {\bibinfo  {journal} {Science}\ }\textbf {\bibinfo {volume} {371}},\ \bibinfo {pages} {1355} (\bibinfo {year} {2021}{\natexlab{a}})},\ \Eprint {https://arxiv.org/abs/https://www.science.org/doi/pdf/10.1126/science.abg2530} {https://www.science.org/doi/pdf/10.1126/science.abg2530} \BibitemShut {NoStop}%
\bibitem [{\citenamefont {Geier}\ \emph {et~al.}(2021)\citenamefont {Geier}, \citenamefont {Thaicharoen}, \citenamefont {Hainaut}, \citenamefont {Franz}, \citenamefont {Salzinger}, \citenamefont {Tebben}, \citenamefont {Grimshandl}, \citenamefont {Zürn},\ and\ \citenamefont {Weidemüller}}]{floquet_spin}%
  \BibitemOpen
  \bibfield  {author} {\bibinfo {author} {\bibfnamefont {S.}~\bibnamefont {Geier}}, \bibinfo {author} {\bibfnamefont {N.}~\bibnamefont {Thaicharoen}}, \bibinfo {author} {\bibfnamefont {C.}~\bibnamefont {Hainaut}}, \bibinfo {author} {\bibfnamefont {T.}~\bibnamefont {Franz}}, \bibinfo {author} {\bibfnamefont {A.}~\bibnamefont {Salzinger}}, \bibinfo {author} {\bibfnamefont {A.}~\bibnamefont {Tebben}}, \bibinfo {author} {\bibfnamefont {D.}~\bibnamefont {Grimshandl}}, \bibinfo {author} {\bibfnamefont {G.}~\bibnamefont {Zürn}},\ and\ \bibinfo {author} {\bibfnamefont {M.}~\bibnamefont {Weidemüller}},\ }\href {https://doi.org/10.1126/science.abd9547} {\bibfield  {journal} {\bibinfo  {journal} {Science}\ }\textbf {\bibinfo {volume} {374}},\ \bibinfo {pages} {1149} (\bibinfo {year} {2021})},\ \Eprint {https://arxiv.org/abs/https://www.science.org/doi/pdf/10.1126/science.abd9547} {https://www.science.org/doi/pdf/10.1126/science.abd9547} \BibitemShut {NoStop}%
\bibitem [{\citenamefont {{U{\u{g}}ur K{\"o}yl{\"u}o{\u{g}}lu}}\ \emph {et~al.}(2024)\citenamefont {{U{\u{g}}ur K{\"o}yl{\"u}o{\u{g}}lu}}, \citenamefont {{Maskara}}, \citenamefont {{Feldmeier}},\ and\ \citenamefont {{Lukin}}}]{2024arXiv240802741U}%
  \BibitemOpen
  \bibfield  {author} {\bibinfo {author} {\bibfnamefont {N.}~\bibnamefont {{U{\u{g}}ur K{\"o}yl{\"u}o{\u{g}}lu}}}, \bibinfo {author} {\bibfnamefont {N.}~\bibnamefont {{Maskara}}}, \bibinfo {author} {\bibfnamefont {J.}~\bibnamefont {{Feldmeier}}},\ and\ \bibinfo {author} {\bibfnamefont {M.~D.}\ \bibnamefont {{Lukin}}},\ }\href {https://doi.org/10.48550/arXiv.2408.02741} {\bibfield  {journal} {\bibinfo  {journal} {arXiv e-prints}\ ,\ \bibinfo {eid} {arXiv:2408.02741}} (\bibinfo {year} {2024})},\ \Eprint {https://arxiv.org/abs/2408.02741} {arXiv:2408.02741 [quant-ph]} \BibitemShut {NoStop}%
\bibitem [{\citenamefont {Khaneja}\ \emph {et~al.}(2005{\natexlab{a}})\citenamefont {Khaneja}, \citenamefont {Reiss}, \citenamefont {Kehlet}, \citenamefont {Schulte-Herbrüggen},\ and\ \citenamefont {Glaser}}]{grape}%
  \BibitemOpen
  \bibfield  {author} {\bibinfo {author} {\bibfnamefont {N.}~\bibnamefont {Khaneja}}, \bibinfo {author} {\bibfnamefont {T.}~\bibnamefont {Reiss}}, \bibinfo {author} {\bibfnamefont {C.}~\bibnamefont {Kehlet}}, \bibinfo {author} {\bibfnamefont {T.}~\bibnamefont {Schulte-Herbrüggen}},\ and\ \bibinfo {author} {\bibfnamefont {S.~J.}\ \bibnamefont {Glaser}},\ }\href {https://doi.org/https://doi.org/10.1016/j.jmr.2004.11.004} {\bibfield  {journal} {\bibinfo  {journal} {Journal of Magnetic Resonance}\ }\textbf {\bibinfo {volume} {172}},\ \bibinfo {pages} {296} (\bibinfo {year} {2005}{\natexlab{a}})}\BibitemShut {NoStop}%
\bibitem [{\citenamefont {Verresen}\ \emph {et~al.}(2017)\citenamefont {Verresen}, \citenamefont {Moessner},\ and\ \citenamefont {Pollmann}}]{ruben_spt}%
  \BibitemOpen
  \bibfield  {author} {\bibinfo {author} {\bibfnamefont {R.}~\bibnamefont {Verresen}}, \bibinfo {author} {\bibfnamefont {R.}~\bibnamefont {Moessner}},\ and\ \bibinfo {author} {\bibfnamefont {F.}~\bibnamefont {Pollmann}},\ }\href {https://doi.org/10.1103/PhysRevB.96.165124} {\bibfield  {journal} {\bibinfo  {journal} {Phys. Rev. B}\ }\textbf {\bibinfo {volume} {96}},\ \bibinfo {pages} {165124} (\bibinfo {year} {2017})}\BibitemShut {NoStop}%
\bibitem [{\citenamefont {Semeghini}\ \emph {et~al.}(2021{\natexlab{b}})\citenamefont {Semeghini}, \citenamefont {Levine}, \citenamefont {Keesling}, \citenamefont {Ebadi}, \citenamefont {Wang}, \citenamefont {Bluvstein}, \citenamefont {Verresen}, \citenamefont {Pichler}, \citenamefont {Kalinowski}, \citenamefont {Samajdar}, \citenamefont {Omran}, \citenamefont {Sachdev}, \citenamefont {Vishwanath}, \citenamefont {Greiner}, \citenamefont {Vuletić},\ and\ \citenamefont {Lukin}}]{semeghini_spinliquid}%
  \BibitemOpen
  \bibfield  {author} {\bibinfo {author} {\bibfnamefont {G.}~\bibnamefont {Semeghini}}, \bibinfo {author} {\bibfnamefont {H.}~\bibnamefont {Levine}}, \bibinfo {author} {\bibfnamefont {A.}~\bibnamefont {Keesling}}, \bibinfo {author} {\bibfnamefont {S.}~\bibnamefont {Ebadi}}, \bibinfo {author} {\bibfnamefont {T.~T.}\ \bibnamefont {Wang}}, \bibinfo {author} {\bibfnamefont {D.}~\bibnamefont {Bluvstein}}, \bibinfo {author} {\bibfnamefont {R.}~\bibnamefont {Verresen}}, \bibinfo {author} {\bibfnamefont {H.}~\bibnamefont {Pichler}}, \bibinfo {author} {\bibfnamefont {M.}~\bibnamefont {Kalinowski}}, \bibinfo {author} {\bibfnamefont {R.}~\bibnamefont {Samajdar}}, \bibinfo {author} {\bibfnamefont {A.}~\bibnamefont {Omran}}, \bibinfo {author} {\bibfnamefont {S.}~\bibnamefont {Sachdev}}, \bibinfo {author} {\bibfnamefont {A.}~\bibnamefont {Vishwanath}}, \bibinfo {author} {\bibfnamefont {M.}~\bibnamefont {Greiner}}, \bibinfo {author} {\bibfnamefont {V.}~\bibnamefont {Vuletić}},\ and\ \bibinfo {author} {\bibfnamefont
  {M.~D.}\ \bibnamefont {Lukin}},\ }\href {https://doi.org/10.1126/science.abi8794} {\bibfield  {journal} {\bibinfo  {journal} {Science}\ }\textbf {\bibinfo {volume} {374}},\ \bibinfo {pages} {1242} (\bibinfo {year} {2021}{\natexlab{b}})},\ \Eprint {https://arxiv.org/abs/https://www.science.org/doi/pdf/10.1126/science.abi8794} {https://www.science.org/doi/pdf/10.1126/science.abi8794} \BibitemShut {NoStop}%
\bibitem [{\citenamefont {Bluvstein}\ \emph {et~al.}(2021{\natexlab{b}})\citenamefont {Bluvstein}, \citenamefont {Omran}, \citenamefont {Levine}, \citenamefont {Keesling}, \citenamefont {Semeghini}, \citenamefont {Ebadi}, \citenamefont {Wang}, \citenamefont {Michailidis}, \citenamefont {Maskara}, \citenamefont {Ho}, \citenamefont {Choi}, \citenamefont {Serbyn}, \citenamefont {Greiner}, \citenamefont {Vuletić},\ and\ \citenamefont {Lukin}}]{quantum_scar}%
  \BibitemOpen
  \bibfield  {author} {\bibinfo {author} {\bibfnamefont {D.}~\bibnamefont {Bluvstein}}, \bibinfo {author} {\bibfnamefont {A.}~\bibnamefont {Omran}}, \bibinfo {author} {\bibfnamefont {H.}~\bibnamefont {Levine}}, \bibinfo {author} {\bibfnamefont {A.}~\bibnamefont {Keesling}}, \bibinfo {author} {\bibfnamefont {G.}~\bibnamefont {Semeghini}}, \bibinfo {author} {\bibfnamefont {S.}~\bibnamefont {Ebadi}}, \bibinfo {author} {\bibfnamefont {T.~T.}\ \bibnamefont {Wang}}, \bibinfo {author} {\bibfnamefont {A.~A.}\ \bibnamefont {Michailidis}}, \bibinfo {author} {\bibfnamefont {N.}~\bibnamefont {Maskara}}, \bibinfo {author} {\bibfnamefont {W.~W.}\ \bibnamefont {Ho}}, \bibinfo {author} {\bibfnamefont {S.}~\bibnamefont {Choi}}, \bibinfo {author} {\bibfnamefont {M.}~\bibnamefont {Serbyn}}, \bibinfo {author} {\bibfnamefont {M.}~\bibnamefont {Greiner}}, \bibinfo {author} {\bibfnamefont {V.}~\bibnamefont {Vuletić}},\ and\ \bibinfo {author} {\bibfnamefont {M.~D.}\ \bibnamefont {Lukin}},\ }\href
  {https://doi.org/10.1126/science.abg2530} {\bibfield  {journal} {\bibinfo  {journal} {Science}\ }\textbf {\bibinfo {volume} {371}},\ \bibinfo {pages} {1355} (\bibinfo {year} {2021}{\natexlab{b}})},\ \Eprint {https://arxiv.org/abs/https://www.science.org/doi/pdf/10.1126/science.abg2530} {https://www.science.org/doi/pdf/10.1126/science.abg2530} \BibitemShut {NoStop}%
\bibitem [{\citenamefont {Jaksch}\ \emph {et~al.}(2000)\citenamefont {Jaksch}, \citenamefont {Cirac}, \citenamefont {Zoller}, \citenamefont {Rolston}, \citenamefont {C\^ot\'e},\ and\ \citenamefont {Lukin}}]{Fast_Rydberg_gate_2000}%
  \BibitemOpen
  \bibfield  {author} {\bibinfo {author} {\bibfnamefont {D.}~\bibnamefont {Jaksch}}, \bibinfo {author} {\bibfnamefont {J.~I.}\ \bibnamefont {Cirac}}, \bibinfo {author} {\bibfnamefont {P.}~\bibnamefont {Zoller}}, \bibinfo {author} {\bibfnamefont {S.~L.}\ \bibnamefont {Rolston}}, \bibinfo {author} {\bibfnamefont {R.}~\bibnamefont {C\^ot\'e}},\ and\ \bibinfo {author} {\bibfnamefont {M.~D.}\ \bibnamefont {Lukin}},\ }\href {https://doi.org/10.1103/PhysRevLett.85.2208} {\bibfield  {journal} {\bibinfo  {journal} {Phys. Rev. Lett.}\ }\textbf {\bibinfo {volume} {85}},\ \bibinfo {pages} {2208} (\bibinfo {year} {2000})}\BibitemShut {NoStop}%
\bibitem [{\citenamefont {Labuhn}\ \emph {et~al.}(2016)\citenamefont {Labuhn}, \citenamefont {Barredo}, \citenamefont {Ravets}, \citenamefont {de~L{\'e}s{\'e}leuc}, \citenamefont {Macr{\`\i}}, \citenamefont {Lahaye},\ and\ \citenamefont {Browaeys}}]{quantum_ising_blockade}%
  \BibitemOpen
  \bibfield  {author} {\bibinfo {author} {\bibfnamefont {H.}~\bibnamefont {Labuhn}}, \bibinfo {author} {\bibfnamefont {D.}~\bibnamefont {Barredo}}, \bibinfo {author} {\bibfnamefont {S.}~\bibnamefont {Ravets}}, \bibinfo {author} {\bibfnamefont {S.}~\bibnamefont {de~L{\'e}s{\'e}leuc}}, \bibinfo {author} {\bibfnamefont {T.}~\bibnamefont {Macr{\`\i}}}, \bibinfo {author} {\bibfnamefont {T.}~\bibnamefont {Lahaye}},\ and\ \bibinfo {author} {\bibfnamefont {A.}~\bibnamefont {Browaeys}},\ }\href {https://doi.org/10.1038/nature18274} {\bibfield  {journal} {\bibinfo  {journal} {Nature}\ }\textbf {\bibinfo {volume} {534}},\ \bibinfo {pages} {667} (\bibinfo {year} {2016})}\BibitemShut {NoStop}%
\bibitem [{\citenamefont {Lloyd}(1996)}]{seth_science_UQC}%
  \BibitemOpen
  \bibfield  {author} {\bibinfo {author} {\bibfnamefont {S.}~\bibnamefont {Lloyd}},\ }\href {https://doi.org/10.1126/science.273.5278.1073} {\bibfield  {journal} {\bibinfo  {journal} {Science}\ }\textbf {\bibinfo {volume} {273}},\ \bibinfo {pages} {1073} (\bibinfo {year} {1996})},\ \Eprint {https://arxiv.org/abs/https://www.science.org/doi/pdf/10.1126/science.273.5278.1073} {https://www.science.org/doi/pdf/10.1126/science.273.5278.1073} \BibitemShut {NoStop}%
\bibitem [{\citenamefont {Lloyd}(1995)}]{Seth_almost_any_gate_universal}%
  \BibitemOpen
  \bibfield  {author} {\bibinfo {author} {\bibfnamefont {S.}~\bibnamefont {Lloyd}},\ }\href {https://doi.org/10.1103/PhysRevLett.75.346} {\bibfield  {journal} {\bibinfo  {journal} {Phys. Rev. Lett.}\ }\textbf {\bibinfo {volume} {75}},\ \bibinfo {pages} {346} (\bibinfo {year} {1995})}\BibitemShut {NoStop}%
\bibitem [{\citenamefont {Bremner}\ \emph {et~al.}(2002)\citenamefont {Bremner}, \citenamefont {Dawson}, \citenamefont {Dodd}, \citenamefont {Gilchrist}, \citenamefont {Harrow}, \citenamefont {Mortimer}, \citenamefont {Nielsen},\ and\ \citenamefont {Osborne}}]{PhysRevLett.89.247902}%
  \BibitemOpen
  \bibfield  {author} {\bibinfo {author} {\bibfnamefont {M.~J.}\ \bibnamefont {Bremner}}, \bibinfo {author} {\bibfnamefont {C.~M.}\ \bibnamefont {Dawson}}, \bibinfo {author} {\bibfnamefont {J.~L.}\ \bibnamefont {Dodd}}, \bibinfo {author} {\bibfnamefont {A.}~\bibnamefont {Gilchrist}}, \bibinfo {author} {\bibfnamefont {A.~W.}\ \bibnamefont {Harrow}}, \bibinfo {author} {\bibfnamefont {D.}~\bibnamefont {Mortimer}}, \bibinfo {author} {\bibfnamefont {M.~A.}\ \bibnamefont {Nielsen}},\ and\ \bibinfo {author} {\bibfnamefont {T.~J.}\ \bibnamefont {Osborne}},\ }\href {https://doi.org/10.1103/PhysRevLett.89.247902} {\bibfield  {journal} {\bibinfo  {journal} {Phys. Rev. Lett.}\ }\textbf {\bibinfo {volume} {89}},\ \bibinfo {pages} {247902} (\bibinfo {year} {2002})}\BibitemShut {NoStop}%
\bibitem [{\citenamefont {Nielsen}\ and\ \citenamefont {Chuang}(2010)}]{Nielsen_Chuang_2010}%
  \BibitemOpen
  \bibfield  {author} {\bibinfo {author} {\bibfnamefont {M.~A.}\ \bibnamefont {Nielsen}}\ and\ \bibinfo {author} {\bibfnamefont {I.~L.}\ \bibnamefont {Chuang}},\ }\href@noop {} {\emph {\bibinfo {title} {Quantum Computation and Quantum Information: 10th Anniversary Edition}}}\ (\bibinfo  {publisher} {Cambridge University Press},\ \bibinfo {year} {2010})\BibitemShut {NoStop}%
\bibitem [{\citenamefont {Menta}\ \emph {et~al.}(2026)\citenamefont {Menta}, \citenamefont {Cioni}, \citenamefont {Aiudi}, \citenamefont {Caravelli}, \citenamefont {Polini},\ and\ \citenamefont {Giovannetti}}]{Menta1}%
  \BibitemOpen
  \bibfield  {author} {\bibinfo {author} {\bibfnamefont {R.}~\bibnamefont {Menta}}, \bibinfo {author} {\bibfnamefont {F.}~\bibnamefont {Cioni}}, \bibinfo {author} {\bibfnamefont {R.}~\bibnamefont {Aiudi}}, \bibinfo {author} {\bibfnamefont {F.}~\bibnamefont {Caravelli}}, \bibinfo {author} {\bibfnamefont {M.}~\bibnamefont {Polini}},\ and\ \bibinfo {author} {\bibfnamefont {V.}~\bibnamefont {Giovannetti}},\ }\href {https://doi.org/10.1103/lz5d-lnz2} {\bibfield  {journal} {\bibinfo  {journal} {Phys. Rev. A}\ }\textbf {\bibinfo {volume} {113}},\ \bibinfo {pages} {012614} (\bibinfo {year} {2026})}\BibitemShut {NoStop}%
\bibitem [{\citenamefont {Aiudi}\ \emph {et~al.}(2026)\citenamefont {Aiudi}, \citenamefont {Despres}, \citenamefont {Menta}, \citenamefont {Abedi}, \citenamefont {Menichetti}, \citenamefont {Giovannetti}, \citenamefont {Polini},\ and\ \citenamefont {Caravelli}}]{Menta2}%
  \BibitemOpen
  \bibfield  {author} {\bibinfo {author} {\bibfnamefont {R.}~\bibnamefont {Aiudi}}, \bibinfo {author} {\bibfnamefont {J.}~\bibnamefont {Despres}}, \bibinfo {author} {\bibfnamefont {R.}~\bibnamefont {Menta}}, \bibinfo {author} {\bibfnamefont {A.}~\bibnamefont {Abedi}}, \bibinfo {author} {\bibfnamefont {G.}~\bibnamefont {Menichetti}}, \bibinfo {author} {\bibfnamefont {V.}~\bibnamefont {Giovannetti}}, \bibinfo {author} {\bibfnamefont {M.}~\bibnamefont {Polini}},\ and\ \bibinfo {author} {\bibfnamefont {F.}~\bibnamefont {Caravelli}},\ }\href {https://doi.org/10.1103/zzzc-nqxd} {\bibfield  {journal} {\bibinfo  {journal} {Phys. Rev. A}\ }\textbf {\bibinfo {volume} {113}},\ \bibinfo {pages} {012616} (\bibinfo {year} {2026})}\BibitemShut {NoStop}%
\bibitem [{\citenamefont {Ragone}\ \emph {et~al.}(2024)\citenamefont {Ragone}, \citenamefont {Bakalov}, \citenamefont {Sauvage}, \citenamefont {Kemper}, \citenamefont {Ortiz~Marrero}, \citenamefont {Larocca},\ and\ \citenamefont {Cerezo}}]{Lie_barrenplateau}%
  \BibitemOpen
  \bibfield  {author} {\bibinfo {author} {\bibfnamefont {M.}~\bibnamefont {Ragone}}, \bibinfo {author} {\bibfnamefont {B.~N.}\ \bibnamefont {Bakalov}}, \bibinfo {author} {\bibfnamefont {F.}~\bibnamefont {Sauvage}}, \bibinfo {author} {\bibfnamefont {A.~F.}\ \bibnamefont {Kemper}}, \bibinfo {author} {\bibfnamefont {C.}~\bibnamefont {Ortiz~Marrero}}, \bibinfo {author} {\bibfnamefont {M.}~\bibnamefont {Larocca}},\ and\ \bibinfo {author} {\bibfnamefont {M.}~\bibnamefont {Cerezo}},\ }\href {https://doi.org/10.1038/s41467-024-49909-3} {\bibfield  {journal} {\bibinfo  {journal} {Nature Communications}\ }\textbf {\bibinfo {volume} {15}},\ \bibinfo {pages} {7172} (\bibinfo {year} {2024})}\BibitemShut {NoStop}%
\bibitem [{\citenamefont {Wiersema}\ \emph {et~al.}(2024)\citenamefont {Wiersema}, \citenamefont {K{\"o}kc{\"u}}, \citenamefont {Kemper},\ and\ \citenamefont {Bakalov}}]{classification_lie}%
  \BibitemOpen
  \bibfield  {author} {\bibinfo {author} {\bibfnamefont {R.}~\bibnamefont {Wiersema}}, \bibinfo {author} {\bibfnamefont {E.}~\bibnamefont {K{\"o}kc{\"u}}}, \bibinfo {author} {\bibfnamefont {A.~F.}\ \bibnamefont {Kemper}},\ and\ \bibinfo {author} {\bibfnamefont {B.~N.}\ \bibnamefont {Bakalov}},\ }\href {https://doi.org/10.1038/s41534-024-00900-2} {\bibfield  {journal} {\bibinfo  {journal} {npj Quantum Information}\ }\textbf {\bibinfo {volume} {10}},\ \bibinfo {pages} {110} (\bibinfo {year} {2024})}\BibitemShut {NoStop}%
\bibitem [{\citenamefont {{Allcock}}\ \emph {et~al.}(2024)\citenamefont {{Allcock}}, \citenamefont {{Santha}}, \citenamefont {{Yuan}},\ and\ \citenamefont {{Zhang}}}]{Shengyu_DLA}%
  \BibitemOpen
  \bibfield  {author} {\bibinfo {author} {\bibfnamefont {J.}~\bibnamefont {{Allcock}}}, \bibinfo {author} {\bibfnamefont {M.}~\bibnamefont {{Santha}}}, \bibinfo {author} {\bibfnamefont {P.}~\bibnamefont {{Yuan}}},\ and\ \bibinfo {author} {\bibfnamefont {S.}~\bibnamefont {{Zhang}}},\ }\href {https://doi.org/10.48550/arXiv.2407.12587} {\bibfield  {journal} {\bibinfo  {journal} {arXiv e-prints}\ ,\ \bibinfo {eid} {arXiv:2407.12587}} (\bibinfo {year} {2024})},\ \Eprint {https://arxiv.org/abs/2407.12587} {arXiv:2407.12587 [quant-ph]} \BibitemShut {NoStop}%
\bibitem [{\citenamefont {{Tabares}}\ \emph {et~al.}(2025)\citenamefont {{Tabares}}, \citenamefont {{Kokail}}, \citenamefont {{Zoller}}, \citenamefont {{Gonz{\'a}lez-Cuadra}},\ and\ \citenamefont {{Gonz{\'a}lez-Tudela}}}]{programming_optical_lattice}%
  \BibitemOpen
  \bibfield  {author} {\bibinfo {author} {\bibfnamefont {C.}~\bibnamefont {{Tabares}}}, \bibinfo {author} {\bibfnamefont {C.}~\bibnamefont {{Kokail}}}, \bibinfo {author} {\bibfnamefont {P.}~\bibnamefont {{Zoller}}}, \bibinfo {author} {\bibfnamefont {D.}~\bibnamefont {{Gonz{\'a}lez-Cuadra}}},\ and\ \bibinfo {author} {\bibfnamefont {A.}~\bibnamefont {{Gonz{\'a}lez-Tudela}}},\ }\href {https://doi.org/10.48550/arXiv.2502.05067} {\bibfield  {journal} {\bibinfo  {journal} {arXiv e-prints}\ ,\ \bibinfo {eid} {arXiv:2502.05067}} (\bibinfo {year} {2025})},\ \Eprint {https://arxiv.org/abs/2502.05067} {arXiv:2502.05067 [quant-ph]} \BibitemShut {NoStop}%
\bibitem [{\citenamefont {Hayden}\ and\ \citenamefont {Preskill}(2007)}]{Patrick_Hayden_2007}%
  \BibitemOpen
  \bibfield  {author} {\bibinfo {author} {\bibfnamefont {P.}~\bibnamefont {Hayden}}\ and\ \bibinfo {author} {\bibfnamefont {J.}~\bibnamefont {Preskill}},\ }\href {https://doi.org/10.1088/1126-6708/2007/09/120} {\bibfield  {journal} {\bibinfo  {journal} {Journal of High Energy Physics}\ }\textbf {\bibinfo {volume} {2007}},\ \bibinfo {pages} {120} (\bibinfo {year} {2007})}\BibitemShut {NoStop}%
\bibitem [{\citenamefont {Haferkamp}\ \emph {et~al.}(2022)\citenamefont {Haferkamp}, \citenamefont {Faist}, \citenamefont {Kothakonda}, \citenamefont {Eisert},\ and\ \citenamefont {Yunger~Halpern}}]{Haferkamp}%
  \BibitemOpen
  \bibfield  {author} {\bibinfo {author} {\bibfnamefont {J.}~\bibnamefont {Haferkamp}}, \bibinfo {author} {\bibfnamefont {P.}~\bibnamefont {Faist}}, \bibinfo {author} {\bibfnamefont {N.~B.~T.}\ \bibnamefont {Kothakonda}}, \bibinfo {author} {\bibfnamefont {J.}~\bibnamefont {Eisert}},\ and\ \bibinfo {author} {\bibfnamefont {N.}~\bibnamefont {Yunger~Halpern}},\ }\href {https://doi.org/10.1038/s41567-022-01539-6} {\bibfield  {journal} {\bibinfo  {journal} {Nature Physics}\ }\textbf {\bibinfo {volume} {18}},\ \bibinfo {pages} {528} (\bibinfo {year} {2022})}\BibitemShut {NoStop}%
\bibitem [{\citenamefont {Huang}\ \emph {et~al.}(2020{\natexlab{a}})\citenamefont {Huang}, \citenamefont {Kueng},\ and\ \citenamefont {Preskill}}]{shadow}%
  \BibitemOpen
  \bibfield  {author} {\bibinfo {author} {\bibfnamefont {H.-Y.}\ \bibnamefont {Huang}}, \bibinfo {author} {\bibfnamefont {R.}~\bibnamefont {Kueng}},\ and\ \bibinfo {author} {\bibfnamefont {J.}~\bibnamefont {Preskill}},\ }\href {https://doi.org/10.1038/s41567-020-0932-7} {\bibfield  {journal} {\bibinfo  {journal} {Nature Physics}\ }\textbf {\bibinfo {volume} {16}},\ \bibinfo {pages} {1050} (\bibinfo {year} {2020}{\natexlab{a}})}\BibitemShut {NoStop}%
\bibitem [{\citenamefont {Hu}\ \emph {et~al.}(2023)\citenamefont {Hu}, \citenamefont {Choi},\ and\ \citenamefont {You}}]{PhysRevResearch.5.023027}%
  \BibitemOpen
  \bibfield  {author} {\bibinfo {author} {\bibfnamefont {H.-Y.}\ \bibnamefont {Hu}}, \bibinfo {author} {\bibfnamefont {S.}~\bibnamefont {Choi}},\ and\ \bibinfo {author} {\bibfnamefont {Y.-Z.}\ \bibnamefont {You}},\ }\href {https://doi.org/10.1103/PhysRevResearch.5.023027} {\bibfield  {journal} {\bibinfo  {journal} {Phys. Rev. Res.}\ }\textbf {\bibinfo {volume} {5}},\ \bibinfo {pages} {023027} (\bibinfo {year} {2023})}\BibitemShut {NoStop}%
\bibitem [{\citenamefont {Liu}\ \emph {et~al.}(2024{\natexlab{a}})\citenamefont {Liu}, \citenamefont {Hao},\ and\ \citenamefont {Hu}}]{PhysRevResearch.6.043118}%
  \BibitemOpen
  \bibfield  {author} {\bibinfo {author} {\bibfnamefont {Z.}~\bibnamefont {Liu}}, \bibinfo {author} {\bibfnamefont {Z.}~\bibnamefont {Hao}},\ and\ \bibinfo {author} {\bibfnamefont {H.-Y.}\ \bibnamefont {Hu}},\ }\href {https://doi.org/10.1103/PhysRevResearch.6.043118} {\bibfield  {journal} {\bibinfo  {journal} {Phys. Rev. Res.}\ }\textbf {\bibinfo {volume} {6}},\ \bibinfo {pages} {043118} (\bibinfo {year} {2024}{\natexlab{a}})}\BibitemShut {NoStop}%
\bibitem [{\citenamefont {Hu}\ and\ \citenamefont {You}(2022)}]{PhysRevResearch.4.013054}%
  \BibitemOpen
  \bibfield  {author} {\bibinfo {author} {\bibfnamefont {H.-Y.}\ \bibnamefont {Hu}}\ and\ \bibinfo {author} {\bibfnamefont {Y.-Z.}\ \bibnamefont {You}},\ }\href {https://doi.org/10.1103/PhysRevResearch.4.013054} {\bibfield  {journal} {\bibinfo  {journal} {Phys. Rev. Res.}\ }\textbf {\bibinfo {volume} {4}},\ \bibinfo {pages} {013054} (\bibinfo {year} {2022})}\BibitemShut {NoStop}%
\bibitem [{\citenamefont {Pilatowsky-Cameo}\ \emph {et~al.}(2024)\citenamefont {Pilatowsky-Cameo}, \citenamefont {Marvian}, \citenamefont {Choi},\ and\ \citenamefont {Ho}}]{PhysRevX.14.041059}%
  \BibitemOpen
  \bibfield  {author} {\bibinfo {author} {\bibfnamefont {S.}~\bibnamefont {Pilatowsky-Cameo}}, \bibinfo {author} {\bibfnamefont {I.}~\bibnamefont {Marvian}}, \bibinfo {author} {\bibfnamefont {S.}~\bibnamefont {Choi}},\ and\ \bibinfo {author} {\bibfnamefont {W.~W.}\ \bibnamefont {Ho}},\ }\href {https://doi.org/10.1103/PhysRevX.14.041059} {\bibfield  {journal} {\bibinfo  {journal} {Phys. Rev. X}\ }\textbf {\bibinfo {volume} {14}},\ \bibinfo {pages} {041059} (\bibinfo {year} {2024})}\BibitemShut {NoStop}%
\bibitem [{Note1()}]{Note1}%
  \BibitemOpen
  \bibinfo {note} {We emphasize that the convergence to the PT distribution in the measurement basis does not indicate that it converges to a random state design. We leave it to future study.}\BibitemShut {Stop}%
\bibitem [{\citenamefont {Dalzell}\ \emph {et~al.}(2022)\citenamefont {Dalzell}, \citenamefont {Hunter-Jones},\ and\ \citenamefont {Brand\~ao}}]{Brandao2022}%
  \BibitemOpen
  \bibfield  {author} {\bibinfo {author} {\bibfnamefont {A.~M.}\ \bibnamefont {Dalzell}}, \bibinfo {author} {\bibfnamefont {N.}~\bibnamefont {Hunter-Jones}},\ and\ \bibinfo {author} {\bibfnamefont {F.~G. S.~L.}\ \bibnamefont {Brand\~ao}},\ }\href {https://doi.org/10.1103/PRXQuantum.3.010333} {\bibfield  {journal} {\bibinfo  {journal} {PRX Quantum}\ }\textbf {\bibinfo {volume} {3}},\ \bibinfo {pages} {010333} (\bibinfo {year} {2022})}\BibitemShut {NoStop}%
\bibitem [{\citenamefont {Wen}(2017)}]{RevModPhys.89.041004}%
  \BibitemOpen
  \bibfield  {author} {\bibinfo {author} {\bibfnamefont {X.-G.}\ \bibnamefont {Wen}},\ }\href {https://doi.org/10.1103/RevModPhys.89.041004} {\bibfield  {journal} {\bibinfo  {journal} {Rev. Mod. Phys.}\ }\textbf {\bibinfo {volume} {89}},\ \bibinfo {pages} {041004} (\bibinfo {year} {2017})}\BibitemShut {NoStop}%
\bibitem [{\citenamefont {Qi}\ and\ \citenamefont {Zhang}(2011)}]{RevModPhys.83.1057}%
  \BibitemOpen
  \bibfield  {author} {\bibinfo {author} {\bibfnamefont {X.-L.}\ \bibnamefont {Qi}}\ and\ \bibinfo {author} {\bibfnamefont {S.-C.}\ \bibnamefont {Zhang}},\ }\href {https://doi.org/10.1103/RevModPhys.83.1057} {\bibfield  {journal} {\bibinfo  {journal} {Rev. Mod. Phys.}\ }\textbf {\bibinfo {volume} {83}},\ \bibinfo {pages} {1057} (\bibinfo {year} {2011})}\BibitemShut {NoStop}%
\bibitem [{\citenamefont {Wang}\ and\ \citenamefont {Gu}(2018)}]{PhysRevX.8.011055}%
  \BibitemOpen
  \bibfield  {author} {\bibinfo {author} {\bibfnamefont {Q.-R.}\ \bibnamefont {Wang}}\ and\ \bibinfo {author} {\bibfnamefont {Z.-C.}\ \bibnamefont {Gu}},\ }\href {https://doi.org/10.1103/PhysRevX.8.011055} {\bibfield  {journal} {\bibinfo  {journal} {Phys. Rev. X}\ }\textbf {\bibinfo {volume} {8}},\ \bibinfo {pages} {011055} (\bibinfo {year} {2018})}\BibitemShut {NoStop}%
\bibitem [{\citenamefont {Cheng}\ \emph {et~al.}(2018)\citenamefont {Cheng}, \citenamefont {Bi}, \citenamefont {You},\ and\ \citenamefont {Gu}}]{GZC2D}%
  \BibitemOpen
  \bibfield  {author} {\bibinfo {author} {\bibfnamefont {M.}~\bibnamefont {Cheng}}, \bibinfo {author} {\bibfnamefont {Z.}~\bibnamefont {Bi}}, \bibinfo {author} {\bibfnamefont {Y.-Z.}\ \bibnamefont {You}},\ and\ \bibinfo {author} {\bibfnamefont {Z.-C.}\ \bibnamefont {Gu}},\ }\href {https://doi.org/10.1103/PhysRevB.97.205109} {\bibfield  {journal} {\bibinfo  {journal} {Phys. Rev. B}\ }\textbf {\bibinfo {volume} {97}},\ \bibinfo {pages} {205109} (\bibinfo {year} {2018})}\BibitemShut {NoStop}%
\bibitem [{\citenamefont {Duque}\ \emph {et~al.}(2021)\citenamefont {Duque}, \citenamefont {Hu}, \citenamefont {You}, \citenamefont {Khemani}, \citenamefont {Verresen},\ and\ \citenamefont {Vasseur}}]{disorder_spt}%
  \BibitemOpen
  \bibfield  {author} {\bibinfo {author} {\bibfnamefont {C.~M.}\ \bibnamefont {Duque}}, \bibinfo {author} {\bibfnamefont {H.-Y.}\ \bibnamefont {Hu}}, \bibinfo {author} {\bibfnamefont {Y.-Z.}\ \bibnamefont {You}}, \bibinfo {author} {\bibfnamefont {V.}~\bibnamefont {Khemani}}, \bibinfo {author} {\bibfnamefont {R.}~\bibnamefont {Verresen}},\ and\ \bibinfo {author} {\bibfnamefont {R.}~\bibnamefont {Vasseur}},\ }\href {https://doi.org/10.1103/PhysRevB.103.L100207} {\bibfield  {journal} {\bibinfo  {journal} {Phys. Rev. B}\ }\textbf {\bibinfo {volume} {103}},\ \bibinfo {pages} {L100207} (\bibinfo {year} {2021})}\BibitemShut {NoStop}%
\bibitem [{\citenamefont {Giudici}\ \emph {et~al.}(2022)\citenamefont {Giudici}, \citenamefont {Lukin},\ and\ \citenamefont {Pichler}}]{Hannes_state_prep}%
  \BibitemOpen
  \bibfield  {author} {\bibinfo {author} {\bibfnamefont {G.}~\bibnamefont {Giudici}}, \bibinfo {author} {\bibfnamefont {M.~D.}\ \bibnamefont {Lukin}},\ and\ \bibinfo {author} {\bibfnamefont {H.}~\bibnamefont {Pichler}},\ }\href {https://doi.org/10.1103/PhysRevLett.129.090401} {\bibfield  {journal} {\bibinfo  {journal} {Phys. Rev. Lett.}\ }\textbf {\bibinfo {volume} {129}},\ \bibinfo {pages} {090401} (\bibinfo {year} {2022})}\BibitemShut {NoStop}%
\bibitem [{\citenamefont {Lu}\ \emph {et~al.}(2024{\natexlab{b}})\citenamefont {Lu}, \citenamefont {Wang}, \citenamefont {Kanungo}, \citenamefont {Dunning},\ and\ \citenamefont {Killian}}]{spt_transition_ssh}%
  \BibitemOpen
  \bibfield  {author} {\bibinfo {author} {\bibfnamefont {Y.}~\bibnamefont {Lu}}, \bibinfo {author} {\bibfnamefont {C.}~\bibnamefont {Wang}}, \bibinfo {author} {\bibfnamefont {S.~K.}\ \bibnamefont {Kanungo}}, \bibinfo {author} {\bibfnamefont {F.~B.}\ \bibnamefont {Dunning}},\ and\ \bibinfo {author} {\bibfnamefont {T.~C.}\ \bibnamefont {Killian}},\ }\href {https://doi.org/10.1103/PhysRevA.110.023318} {\bibfield  {journal} {\bibinfo  {journal} {Phys. Rev. A}\ }\textbf {\bibinfo {volume} {110}},\ \bibinfo {pages} {023318} (\bibinfo {year} {2024}{\natexlab{b}})}\BibitemShut {NoStop}%
\bibitem [{\citenamefont {Liu}\ \emph {et~al.}(2024{\natexlab{b}})\citenamefont {Liu}, \citenamefont {Shtengel},\ and\ \citenamefont {Pollmann}}]{YuJie_phase_transition}%
  \BibitemOpen
  \bibfield  {author} {\bibinfo {author} {\bibfnamefont {Y.-J.}\ \bibnamefont {Liu}}, \bibinfo {author} {\bibfnamefont {K.}~\bibnamefont {Shtengel}},\ and\ \bibinfo {author} {\bibfnamefont {F.}~\bibnamefont {Pollmann}},\ }\href {https://doi.org/10.1103/PhysRevResearch.6.043256} {\bibfield  {journal} {\bibinfo  {journal} {Phys. Rev. Res.}\ }\textbf {\bibinfo {volume} {6}},\ \bibinfo {pages} {043256} (\bibinfo {year} {2024}{\natexlab{b}})}\BibitemShut {NoStop}%
\bibitem [{\citenamefont {Zhang}\ \emph {et~al.}(2022)\citenamefont {Zhang}, \citenamefont {Jiang}, \citenamefont {Deng}, \citenamefont {Wang}, \citenamefont {Chen}, \citenamefont {Zhang}, \citenamefont {Ren}, \citenamefont {Dong}, \citenamefont {Xu}, \citenamefont {Gao}, \citenamefont {Jin}, \citenamefont {Zhu}, \citenamefont {Guo}, \citenamefont {Li}, \citenamefont {Song}, \citenamefont {Gorshkov}, \citenamefont {Iadecola}, \citenamefont {Liu}, \citenamefont {Gong}, \citenamefont {Wang}, \citenamefont {Deng},\ and\ \citenamefont {Wang}}]{Dongling_fspt}%
  \BibitemOpen
  \bibfield  {author} {\bibinfo {author} {\bibfnamefont {X.}~\bibnamefont {Zhang}}, \bibinfo {author} {\bibfnamefont {W.}~\bibnamefont {Jiang}}, \bibinfo {author} {\bibfnamefont {J.}~\bibnamefont {Deng}}, \bibinfo {author} {\bibfnamefont {K.}~\bibnamefont {Wang}}, \bibinfo {author} {\bibfnamefont {J.}~\bibnamefont {Chen}}, \bibinfo {author} {\bibfnamefont {P.}~\bibnamefont {Zhang}}, \bibinfo {author} {\bibfnamefont {W.}~\bibnamefont {Ren}}, \bibinfo {author} {\bibfnamefont {H.}~\bibnamefont {Dong}}, \bibinfo {author} {\bibfnamefont {S.}~\bibnamefont {Xu}}, \bibinfo {author} {\bibfnamefont {Y.}~\bibnamefont {Gao}}, \bibinfo {author} {\bibfnamefont {F.}~\bibnamefont {Jin}}, \bibinfo {author} {\bibfnamefont {X.}~\bibnamefont {Zhu}}, \bibinfo {author} {\bibfnamefont {Q.}~\bibnamefont {Guo}}, \bibinfo {author} {\bibfnamefont {H.}~\bibnamefont {Li}}, \bibinfo {author} {\bibfnamefont {C.}~\bibnamefont {Song}}, \bibinfo {author} {\bibfnamefont {A.~V.}\ \bibnamefont {Gorshkov}}, \bibinfo {author} {\bibfnamefont
  {T.}~\bibnamefont {Iadecola}}, \bibinfo {author} {\bibfnamefont {F.}~\bibnamefont {Liu}}, \bibinfo {author} {\bibfnamefont {Z.-X.}\ \bibnamefont {Gong}}, \bibinfo {author} {\bibfnamefont {Z.}~\bibnamefont {Wang}}, \bibinfo {author} {\bibfnamefont {D.-L.}\ \bibnamefont {Deng}},\ and\ \bibinfo {author} {\bibfnamefont {H.}~\bibnamefont {Wang}},\ }\href {https://doi.org/10.1038/s41586-022-04854-3} {\bibfield  {journal} {\bibinfo  {journal} {Nature}\ }\textbf {\bibinfo {volume} {607}},\ \bibinfo {pages} {468} (\bibinfo {year} {2022})}\BibitemShut {NoStop}%
\bibitem [{\citenamefont {Dumitrescu}\ \emph {et~al.}(2022)\citenamefont {Dumitrescu}, \citenamefont {Bohnet}, \citenamefont {Gaebler}, \citenamefont {Hankin}, \citenamefont {Hayes}, \citenamefont {Kumar}, \citenamefont {Neyenhuis}, \citenamefont {Vasseur},\ and\ \citenamefont {Potter}}]{FSPT_IONS}%
  \BibitemOpen
  \bibfield  {author} {\bibinfo {author} {\bibfnamefont {P.~T.}\ \bibnamefont {Dumitrescu}}, \bibinfo {author} {\bibfnamefont {J.~G.}\ \bibnamefont {Bohnet}}, \bibinfo {author} {\bibfnamefont {J.~P.}\ \bibnamefont {Gaebler}}, \bibinfo {author} {\bibfnamefont {A.}~\bibnamefont {Hankin}}, \bibinfo {author} {\bibfnamefont {D.}~\bibnamefont {Hayes}}, \bibinfo {author} {\bibfnamefont {A.}~\bibnamefont {Kumar}}, \bibinfo {author} {\bibfnamefont {B.}~\bibnamefont {Neyenhuis}}, \bibinfo {author} {\bibfnamefont {R.}~\bibnamefont {Vasseur}},\ and\ \bibinfo {author} {\bibfnamefont {A.~C.}\ \bibnamefont {Potter}},\ }\href {https://doi.org/10.1038/s41586-022-04853-4} {\bibfield  {journal} {\bibinfo  {journal} {Nature}\ }\textbf {\bibinfo {volume} {607}},\ \bibinfo {pages} {463} (\bibinfo {year} {2022})}\BibitemShut {NoStop}%
\bibitem [{\citenamefont {Iqbal}\ \emph {et~al.}(2024)\citenamefont {Iqbal}, \citenamefont {Tantivasadakarn}, \citenamefont {Verresen}, \citenamefont {Campbell}, \citenamefont {Dreiling}, \citenamefont {Figgatt}, \citenamefont {Gaebler}, \citenamefont {Johansen}, \citenamefont {Mills}, \citenamefont {Moses}, \citenamefont {Pino}, \citenamefont {Ransford}, \citenamefont {Rowe}, \citenamefont {Siegfried}, \citenamefont {Stutz}, \citenamefont {Foss-Feig}, \citenamefont {Vishwanath},\ and\ \citenamefont {Dreyer}}]{ion_anyon}%
  \BibitemOpen
  \bibfield  {author} {\bibinfo {author} {\bibfnamefont {M.}~\bibnamefont {Iqbal}}, \bibinfo {author} {\bibfnamefont {N.}~\bibnamefont {Tantivasadakarn}}, \bibinfo {author} {\bibfnamefont {R.}~\bibnamefont {Verresen}}, \bibinfo {author} {\bibfnamefont {S.~L.}\ \bibnamefont {Campbell}}, \bibinfo {author} {\bibfnamefont {J.~M.}\ \bibnamefont {Dreiling}}, \bibinfo {author} {\bibfnamefont {C.}~\bibnamefont {Figgatt}}, \bibinfo {author} {\bibfnamefont {J.~P.}\ \bibnamefont {Gaebler}}, \bibinfo {author} {\bibfnamefont {J.}~\bibnamefont {Johansen}}, \bibinfo {author} {\bibfnamefont {M.}~\bibnamefont {Mills}}, \bibinfo {author} {\bibfnamefont {S.~A.}\ \bibnamefont {Moses}}, \bibinfo {author} {\bibfnamefont {J.~M.}\ \bibnamefont {Pino}}, \bibinfo {author} {\bibfnamefont {A.}~\bibnamefont {Ransford}}, \bibinfo {author} {\bibfnamefont {M.}~\bibnamefont {Rowe}}, \bibinfo {author} {\bibfnamefont {P.}~\bibnamefont {Siegfried}}, \bibinfo {author} {\bibfnamefont {R.~P.}\ \bibnamefont {Stutz}}, \bibinfo {author}
  {\bibfnamefont {M.}~\bibnamefont {Foss-Feig}}, \bibinfo {author} {\bibfnamefont {A.}~\bibnamefont {Vishwanath}},\ and\ \bibinfo {author} {\bibfnamefont {H.}~\bibnamefont {Dreyer}},\ }\href {https://doi.org/10.1038/s41586-023-06934-4} {\bibfield  {journal} {\bibinfo  {journal} {Nature}\ }\textbf {\bibinfo {volume} {626}},\ \bibinfo {pages} {505} (\bibinfo {year} {2024})}\BibitemShut {NoStop}%
\bibitem [{\citenamefont {Andersen}\ \emph {et~al.}(2023)\citenamefont {Andersen}, \citenamefont {Lensky}, \citenamefont {Kechedzhi}, \citenamefont {Drozdov}, \citenamefont {Bengtsson}, \citenamefont {Hong}, \citenamefont {Morvan}, \citenamefont {Mi}, \citenamefont {Opremcak}, \citenamefont {Acharya}, \citenamefont {Allen}, \citenamefont {Ansmann}, \citenamefont {Arute}, \citenamefont {Arya}, \citenamefont {Asfaw}, \citenamefont {Atalaya}, \citenamefont {Babbush}, \citenamefont {Bacon}, \citenamefont {Bardin}, \citenamefont {Bortoli}, \citenamefont {Bourassa}, \citenamefont {Bovaird}, \citenamefont {Brill}, \citenamefont {Broughton}, \citenamefont {Buckley}, \citenamefont {Buell}, \citenamefont {Burger}, \citenamefont {Burkett}, \citenamefont {Bushnell}, \citenamefont {Chen}, \citenamefont {Chiaro}, \citenamefont {Chik}, \citenamefont {Chou}, \citenamefont {Cogan}, \citenamefont {Collins}, \citenamefont {Conner}, \citenamefont {Courtney}, \citenamefont {Crook}, \citenamefont {Curtin}, \citenamefont {Debroy},
  \citenamefont {Del Toro~Barba}, \citenamefont {Demura}, \citenamefont {Dunsworth}, \citenamefont {Eppens}, \citenamefont {Erickson}, \citenamefont {Faoro}, \citenamefont {Farhi}, \citenamefont {Fatemi}, \citenamefont {Ferreira}, \citenamefont {Burgos}, \citenamefont {Forati}, \citenamefont {Fowler}, \citenamefont {Foxen}, \citenamefont {Giang}, \citenamefont {Gidney}, \citenamefont {Gilboa}, \citenamefont {Giustina}, \citenamefont {Gosula}, \citenamefont {Dau}, \citenamefont {Gross}, \citenamefont {Habegger}, \citenamefont {Hamilton}, \citenamefont {Hansen}, \citenamefont {Harrigan}, \citenamefont {Harrington}, \citenamefont {Heu}, \citenamefont {Hilton}, \citenamefont {Hoffmann}, \citenamefont {Huang}, \citenamefont {Huff}, \citenamefont {Huggins}, \citenamefont {Ioffe}, \citenamefont {Isakov}, \citenamefont {Iveland}, \citenamefont {Jeffrey}, \citenamefont {Jiang}, \citenamefont {Jones}, \citenamefont {Juhas}, \citenamefont {Kafri}, \citenamefont {Khattar}, \citenamefont {Khezri}, \citenamefont
  {Kieferov{\'a}}, \citenamefont {Kim}, \citenamefont {Kitaev}, \citenamefont {Klimov}, \citenamefont {Klots}, \citenamefont {Korotkov}, \citenamefont {Kostritsa}, \citenamefont {Kreikebaum}, \citenamefont {Landhuis}, \citenamefont {Laptev}, \citenamefont {Lau}, \citenamefont {Laws}, \citenamefont {Lee}, \citenamefont {Lee}, \citenamefont {Lester}, \citenamefont {Lill}, \citenamefont {Liu}, \citenamefont {Locharla}, \citenamefont {Lucero}, \citenamefont {Malone}, \citenamefont {Martin}, \citenamefont {McClean}, \citenamefont {McCourt}, \citenamefont {McEwen}, \citenamefont {Miao}, \citenamefont {Mieszala}, \citenamefont {Mohseni}, \citenamefont {Montazeri}, \citenamefont {Mount}, \citenamefont {Movassagh}, \citenamefont {Mruczkiewicz}, \citenamefont {Naaman}, \citenamefont {Neeley}, \citenamefont {Neill}, \citenamefont {Nersisyan}, \citenamefont {Newman}, \citenamefont {Ng}, \citenamefont {Nguyen}, \citenamefont {Nguyen}, \citenamefont {Niu}, \citenamefont {O'Brien}, \citenamefont {Omonije}, \citenamefont
  {Petukhov}, \citenamefont {Potter}, \citenamefont {Pryadko}, \citenamefont {Quintana}, \citenamefont {Rocque}, \citenamefont {Rubin}, \citenamefont {Saei}, \citenamefont {Sank}, \citenamefont {Sankaragomathi}, \citenamefont {Satzinger}, \citenamefont {Schurkus}, \citenamefont {Schuster}, \citenamefont {Shearn}, \citenamefont {Shorter}, \citenamefont {Shutty}, \citenamefont {Shvarts}, \citenamefont {Skruzny}, \citenamefont {Smith}, \citenamefont {Somma}, \citenamefont {Sterling}, \citenamefont {Strain}, \citenamefont {Szalay}, \citenamefont {Torres}, \citenamefont {Vidal}, \citenamefont {Villalonga}, \citenamefont {Heidweiller}, \citenamefont {White}, \citenamefont {Woo}, \citenamefont {Xing}, \citenamefont {Yao}, \citenamefont {Yeh}, \citenamefont {Yoo}, \citenamefont {Young}, \citenamefont {Zalcman}, \citenamefont {Zhang}, \citenamefont {Zhu}, \citenamefont {Zobrist}, \citenamefont {Neven}, \citenamefont {Boixo}, \citenamefont {Megrant}, \citenamefont {Kelly}, \citenamefont {Chen}, \citenamefont
  {Smelyanskiy}, \citenamefont {Kim}, \citenamefont {Aleiner}, \citenamefont {Roushan}, \citenamefont {AI},\ and\ \citenamefont {Collaborators}}]{anyon2}%
  \BibitemOpen
  \bibfield  {author} {\bibinfo {author} {\bibfnamefont {T.~I.}\ \bibnamefont {Andersen}}, \bibinfo {author} {\bibfnamefont {Y.~D.}\ \bibnamefont {Lensky}}, \bibinfo {author} {\bibfnamefont {K.}~\bibnamefont {Kechedzhi}}, \bibinfo {author} {\bibfnamefont {I.~K.}\ \bibnamefont {Drozdov}}, \bibinfo {author} {\bibfnamefont {A.}~\bibnamefont {Bengtsson}}, \bibinfo {author} {\bibfnamefont {S.}~\bibnamefont {Hong}}, \bibinfo {author} {\bibfnamefont {A.}~\bibnamefont {Morvan}}, \bibinfo {author} {\bibfnamefont {X.}~\bibnamefont {Mi}}, \bibinfo {author} {\bibfnamefont {A.}~\bibnamefont {Opremcak}}, \bibinfo {author} {\bibfnamefont {R.}~\bibnamefont {Acharya}}, \bibinfo {author} {\bibfnamefont {R.}~\bibnamefont {Allen}}, \bibinfo {author} {\bibfnamefont {M.}~\bibnamefont {Ansmann}}, \bibinfo {author} {\bibfnamefont {F.}~\bibnamefont {Arute}}, \bibinfo {author} {\bibfnamefont {K.}~\bibnamefont {Arya}}, \bibinfo {author} {\bibfnamefont {A.}~\bibnamefont {Asfaw}}, \bibinfo {author} {\bibfnamefont {J.}~\bibnamefont
  {Atalaya}}, \bibinfo {author} {\bibfnamefont {R.}~\bibnamefont {Babbush}}, \bibinfo {author} {\bibfnamefont {D.}~\bibnamefont {Bacon}}, \bibinfo {author} {\bibfnamefont {J.~C.}\ \bibnamefont {Bardin}}, \bibinfo {author} {\bibfnamefont {G.}~\bibnamefont {Bortoli}}, \bibinfo {author} {\bibfnamefont {A.}~\bibnamefont {Bourassa}}, \bibinfo {author} {\bibfnamefont {J.}~\bibnamefont {Bovaird}}, \bibinfo {author} {\bibfnamefont {L.}~\bibnamefont {Brill}}, \bibinfo {author} {\bibfnamefont {M.}~\bibnamefont {Broughton}}, \bibinfo {author} {\bibfnamefont {B.~B.}\ \bibnamefont {Buckley}}, \bibinfo {author} {\bibfnamefont {D.~A.}\ \bibnamefont {Buell}}, \bibinfo {author} {\bibfnamefont {T.}~\bibnamefont {Burger}}, \bibinfo {author} {\bibfnamefont {B.}~\bibnamefont {Burkett}}, \bibinfo {author} {\bibfnamefont {N.}~\bibnamefont {Bushnell}}, \bibinfo {author} {\bibfnamefont {Z.}~\bibnamefont {Chen}}, \bibinfo {author} {\bibfnamefont {B.}~\bibnamefont {Chiaro}}, \bibinfo {author} {\bibfnamefont {D.}~\bibnamefont {Chik}},
  \bibinfo {author} {\bibfnamefont {C.}~\bibnamefont {Chou}}, \bibinfo {author} {\bibfnamefont {J.}~\bibnamefont {Cogan}}, \bibinfo {author} {\bibfnamefont {R.}~\bibnamefont {Collins}}, \bibinfo {author} {\bibfnamefont {P.}~\bibnamefont {Conner}}, \bibinfo {author} {\bibfnamefont {W.}~\bibnamefont {Courtney}}, \bibinfo {author} {\bibfnamefont {A.~L.}\ \bibnamefont {Crook}}, \bibinfo {author} {\bibfnamefont {B.}~\bibnamefont {Curtin}}, \bibinfo {author} {\bibfnamefont {D.~M.}\ \bibnamefont {Debroy}}, \bibinfo {author} {\bibfnamefont {A.}~\bibnamefont {Del Toro~Barba}}, \bibinfo {author} {\bibfnamefont {S.}~\bibnamefont {Demura}}, \bibinfo {author} {\bibfnamefont {A.}~\bibnamefont {Dunsworth}}, \bibinfo {author} {\bibfnamefont {D.}~\bibnamefont {Eppens}}, \bibinfo {author} {\bibfnamefont {C.}~\bibnamefont {Erickson}}, \bibinfo {author} {\bibfnamefont {L.}~\bibnamefont {Faoro}}, \bibinfo {author} {\bibfnamefont {E.}~\bibnamefont {Farhi}}, \bibinfo {author} {\bibfnamefont {R.}~\bibnamefont {Fatemi}}, \bibinfo
  {author} {\bibfnamefont {V.~S.}\ \bibnamefont {Ferreira}}, \bibinfo {author} {\bibfnamefont {L.~F.}\ \bibnamefont {Burgos}}, \bibinfo {author} {\bibfnamefont {E.}~\bibnamefont {Forati}}, \bibinfo {author} {\bibfnamefont {A.~G.}\ \bibnamefont {Fowler}}, \bibinfo {author} {\bibfnamefont {B.}~\bibnamefont {Foxen}}, \bibinfo {author} {\bibfnamefont {W.}~\bibnamefont {Giang}}, \bibinfo {author} {\bibfnamefont {C.}~\bibnamefont {Gidney}}, \bibinfo {author} {\bibfnamefont {D.}~\bibnamefont {Gilboa}}, \bibinfo {author} {\bibfnamefont {M.}~\bibnamefont {Giustina}}, \bibinfo {author} {\bibfnamefont {R.}~\bibnamefont {Gosula}}, \bibinfo {author} {\bibfnamefont {A.~G.}\ \bibnamefont {Dau}}, \bibinfo {author} {\bibfnamefont {J.~A.}\ \bibnamefont {Gross}}, \bibinfo {author} {\bibfnamefont {S.}~\bibnamefont {Habegger}}, \bibinfo {author} {\bibfnamefont {M.~C.}\ \bibnamefont {Hamilton}}, \bibinfo {author} {\bibfnamefont {M.}~\bibnamefont {Hansen}}, \bibinfo {author} {\bibfnamefont {M.~P.}\ \bibnamefont {Harrigan}},
  \bibinfo {author} {\bibfnamefont {S.~D.}\ \bibnamefont {Harrington}}, \bibinfo {author} {\bibfnamefont {P.}~\bibnamefont {Heu}}, \bibinfo {author} {\bibfnamefont {J.}~\bibnamefont {Hilton}}, \bibinfo {author} {\bibfnamefont {M.~R.}\ \bibnamefont {Hoffmann}}, \bibinfo {author} {\bibfnamefont {T.}~\bibnamefont {Huang}}, \bibinfo {author} {\bibfnamefont {A.}~\bibnamefont {Huff}}, \bibinfo {author} {\bibfnamefont {W.~J.}\ \bibnamefont {Huggins}}, \bibinfo {author} {\bibfnamefont {L.~B.}\ \bibnamefont {Ioffe}}, \bibinfo {author} {\bibfnamefont {S.~V.}\ \bibnamefont {Isakov}}, \bibinfo {author} {\bibfnamefont {J.}~\bibnamefont {Iveland}}, \bibinfo {author} {\bibfnamefont {E.}~\bibnamefont {Jeffrey}}, \bibinfo {author} {\bibfnamefont {Z.}~\bibnamefont {Jiang}}, \bibinfo {author} {\bibfnamefont {C.}~\bibnamefont {Jones}}, \bibinfo {author} {\bibfnamefont {P.}~\bibnamefont {Juhas}}, \bibinfo {author} {\bibfnamefont {D.}~\bibnamefont {Kafri}}, \bibinfo {author} {\bibfnamefont {T.}~\bibnamefont {Khattar}}, \bibinfo
  {author} {\bibfnamefont {M.}~\bibnamefont {Khezri}}, \bibinfo {author} {\bibfnamefont {M.}~\bibnamefont {Kieferov{\'a}}}, \bibinfo {author} {\bibfnamefont {S.}~\bibnamefont {Kim}}, \bibinfo {author} {\bibfnamefont {A.}~\bibnamefont {Kitaev}}, \bibinfo {author} {\bibfnamefont {P.~V.}\ \bibnamefont {Klimov}}, \bibinfo {author} {\bibfnamefont {A.~R.}\ \bibnamefont {Klots}}, \bibinfo {author} {\bibfnamefont {A.~N.}\ \bibnamefont {Korotkov}}, \bibinfo {author} {\bibfnamefont {F.}~\bibnamefont {Kostritsa}}, \bibinfo {author} {\bibfnamefont {J.~M.}\ \bibnamefont {Kreikebaum}}, \bibinfo {author} {\bibfnamefont {D.}~\bibnamefont {Landhuis}}, \bibinfo {author} {\bibfnamefont {P.}~\bibnamefont {Laptev}}, \bibinfo {author} {\bibfnamefont {K.~M.}\ \bibnamefont {Lau}}, \bibinfo {author} {\bibfnamefont {L.}~\bibnamefont {Laws}}, \bibinfo {author} {\bibfnamefont {J.}~\bibnamefont {Lee}}, \bibinfo {author} {\bibfnamefont {K.~W.}\ \bibnamefont {Lee}}, \bibinfo {author} {\bibfnamefont {B.~J.}\ \bibnamefont {Lester}}, \bibinfo
  {author} {\bibfnamefont {A.~T.}\ \bibnamefont {Lill}}, \bibinfo {author} {\bibfnamefont {W.}~\bibnamefont {Liu}}, \bibinfo {author} {\bibfnamefont {A.}~\bibnamefont {Locharla}}, \bibinfo {author} {\bibfnamefont {E.}~\bibnamefont {Lucero}}, \bibinfo {author} {\bibfnamefont {F.~D.}\ \bibnamefont {Malone}}, \bibinfo {author} {\bibfnamefont {O.}~\bibnamefont {Martin}}, \bibinfo {author} {\bibfnamefont {J.~R.}\ \bibnamefont {McClean}}, \bibinfo {author} {\bibfnamefont {T.}~\bibnamefont {McCourt}}, \bibinfo {author} {\bibfnamefont {M.}~\bibnamefont {McEwen}}, \bibinfo {author} {\bibfnamefont {K.~C.}\ \bibnamefont {Miao}}, \bibinfo {author} {\bibfnamefont {A.}~\bibnamefont {Mieszala}}, \bibinfo {author} {\bibfnamefont {M.}~\bibnamefont {Mohseni}}, \bibinfo {author} {\bibfnamefont {S.}~\bibnamefont {Montazeri}}, \bibinfo {author} {\bibfnamefont {E.}~\bibnamefont {Mount}}, \bibinfo {author} {\bibfnamefont {R.}~\bibnamefont {Movassagh}}, \bibinfo {author} {\bibfnamefont {W.}~\bibnamefont {Mruczkiewicz}}, \bibinfo
  {author} {\bibfnamefont {O.}~\bibnamefont {Naaman}}, \bibinfo {author} {\bibfnamefont {M.}~\bibnamefont {Neeley}}, \bibinfo {author} {\bibfnamefont {C.}~\bibnamefont {Neill}}, \bibinfo {author} {\bibfnamefont {A.}~\bibnamefont {Nersisyan}}, \bibinfo {author} {\bibfnamefont {M.}~\bibnamefont {Newman}}, \bibinfo {author} {\bibfnamefont {J.~H.}\ \bibnamefont {Ng}}, \bibinfo {author} {\bibfnamefont {A.}~\bibnamefont {Nguyen}}, \bibinfo {author} {\bibfnamefont {M.}~\bibnamefont {Nguyen}}, \bibinfo {author} {\bibfnamefont {M.~Y.}\ \bibnamefont {Niu}}, \bibinfo {author} {\bibfnamefont {T.~E.}\ \bibnamefont {O'Brien}}, \bibinfo {author} {\bibfnamefont {S.}~\bibnamefont {Omonije}}, \bibinfo {author} {\bibfnamefont {A.}~\bibnamefont {Petukhov}}, \bibinfo {author} {\bibfnamefont {R.}~\bibnamefont {Potter}}, \bibinfo {author} {\bibfnamefont {L.~P.}\ \bibnamefont {Pryadko}}, \bibinfo {author} {\bibfnamefont {C.}~\bibnamefont {Quintana}}, \bibinfo {author} {\bibfnamefont {C.}~\bibnamefont {Rocque}}, \bibinfo {author}
  {\bibfnamefont {N.~C.}\ \bibnamefont {Rubin}}, \bibinfo {author} {\bibfnamefont {N.}~\bibnamefont {Saei}}, \bibinfo {author} {\bibfnamefont {D.}~\bibnamefont {Sank}}, \bibinfo {author} {\bibfnamefont {K.}~\bibnamefont {Sankaragomathi}}, \bibinfo {author} {\bibfnamefont {K.~J.}\ \bibnamefont {Satzinger}}, \bibinfo {author} {\bibfnamefont {H.~F.}\ \bibnamefont {Schurkus}}, \bibinfo {author} {\bibfnamefont {C.}~\bibnamefont {Schuster}}, \bibinfo {author} {\bibfnamefont {M.~J.}\ \bibnamefont {Shearn}}, \bibinfo {author} {\bibfnamefont {A.}~\bibnamefont {Shorter}}, \bibinfo {author} {\bibfnamefont {N.}~\bibnamefont {Shutty}}, \bibinfo {author} {\bibfnamefont {V.}~\bibnamefont {Shvarts}}, \bibinfo {author} {\bibfnamefont {J.}~\bibnamefont {Skruzny}}, \bibinfo {author} {\bibfnamefont {W.~C.}\ \bibnamefont {Smith}}, \bibinfo {author} {\bibfnamefont {R.}~\bibnamefont {Somma}}, \bibinfo {author} {\bibfnamefont {G.}~\bibnamefont {Sterling}}, \bibinfo {author} {\bibfnamefont {D.}~\bibnamefont {Strain}}, \bibinfo
  {author} {\bibfnamefont {M.}~\bibnamefont {Szalay}}, \bibinfo {author} {\bibfnamefont {A.}~\bibnamefont {Torres}}, \bibinfo {author} {\bibfnamefont {G.}~\bibnamefont {Vidal}}, \bibinfo {author} {\bibfnamefont {B.}~\bibnamefont {Villalonga}}, \bibinfo {author} {\bibfnamefont {C.~V.}\ \bibnamefont {Heidweiller}}, \bibinfo {author} {\bibfnamefont {T.}~\bibnamefont {White}}, \bibinfo {author} {\bibfnamefont {B.~W.~K.}\ \bibnamefont {Woo}}, \bibinfo {author} {\bibfnamefont {C.}~\bibnamefont {Xing}}, \bibinfo {author} {\bibfnamefont {Z.~J.}\ \bibnamefont {Yao}}, \bibinfo {author} {\bibfnamefont {P.}~\bibnamefont {Yeh}}, \bibinfo {author} {\bibfnamefont {J.}~\bibnamefont {Yoo}}, \bibinfo {author} {\bibfnamefont {G.}~\bibnamefont {Young}}, \bibinfo {author} {\bibfnamefont {A.}~\bibnamefont {Zalcman}}, \bibinfo {author} {\bibfnamefont {Y.}~\bibnamefont {Zhang}}, \bibinfo {author} {\bibfnamefont {N.}~\bibnamefont {Zhu}}, \bibinfo {author} {\bibfnamefont {N.}~\bibnamefont {Zobrist}}, \bibinfo {author} {\bibfnamefont
  {H.}~\bibnamefont {Neven}}, \bibinfo {author} {\bibfnamefont {S.}~\bibnamefont {Boixo}}, \bibinfo {author} {\bibfnamefont {A.}~\bibnamefont {Megrant}}, \bibinfo {author} {\bibfnamefont {J.}~\bibnamefont {Kelly}}, \bibinfo {author} {\bibfnamefont {Y.}~\bibnamefont {Chen}}, \bibinfo {author} {\bibfnamefont {V.}~\bibnamefont {Smelyanskiy}}, \bibinfo {author} {\bibfnamefont {E.~A.}\ \bibnamefont {Kim}}, \bibinfo {author} {\bibfnamefont {I.}~\bibnamefont {Aleiner}}, \bibinfo {author} {\bibfnamefont {P.}~\bibnamefont {Roushan}}, \bibinfo {author} {\bibfnamefont {G.~Q.}\ \bibnamefont {AI}},\ and\ \bibinfo {author} {\bibnamefont {Collaborators}},\ }\href {https://doi.org/10.1038/s41586-023-05954-4} {\bibfield  {journal} {\bibinfo  {journal} {Nature}\ }\textbf {\bibinfo {volume} {618}},\ \bibinfo {pages} {264} (\bibinfo {year} {2023})}\BibitemShut {NoStop}%
\bibitem [{\citenamefont {Chen}\ \emph {et~al.}(2025)\citenamefont {Chen}, \citenamefont {Ren}, \citenamefont {Fan},\ and\ \citenamefont {Jaffe}}]{liyuan}%
  \BibitemOpen
  \bibfield  {author} {\bibinfo {author} {\bibfnamefont {L.}~\bibnamefont {Chen}}, \bibinfo {author} {\bibfnamefont {Y.}~\bibnamefont {Ren}}, \bibinfo {author} {\bibfnamefont {R.}~\bibnamefont {Fan}},\ and\ \bibinfo {author} {\bibfnamefont {A.}~\bibnamefont {Jaffe}},\ }\href {https://doi.org/10.1038/s41534-025-01063-4} {\bibfield  {journal} {\bibinfo  {journal} {npj Quantum Information}\ }\textbf {\bibinfo {volume} {11}},\ \bibinfo {pages} {112} (\bibinfo {year} {2025})}\BibitemShut {NoStop}%
\bibitem [{\citenamefont {Urban}\ \emph {et~al.}(2009)\citenamefont {Urban}, \citenamefont {Johnson}, \citenamefont {Henage}, \citenamefont {Isenhower}, \citenamefont {Yavuz}, \citenamefont {Walker},\ and\ \citenamefont {Saffman}}]{blockade}%
  \BibitemOpen
  \bibfield  {author} {\bibinfo {author} {\bibfnamefont {E.}~\bibnamefont {Urban}}, \bibinfo {author} {\bibfnamefont {T.~A.}\ \bibnamefont {Johnson}}, \bibinfo {author} {\bibfnamefont {T.}~\bibnamefont {Henage}}, \bibinfo {author} {\bibfnamefont {L.}~\bibnamefont {Isenhower}}, \bibinfo {author} {\bibfnamefont {D.~D.}\ \bibnamefont {Yavuz}}, \bibinfo {author} {\bibfnamefont {T.~G.}\ \bibnamefont {Walker}},\ and\ \bibinfo {author} {\bibfnamefont {M.}~\bibnamefont {Saffman}},\ }\href {https://doi.org/10.1038/nphys1178} {\bibfield  {journal} {\bibinfo  {journal} {Nature Physics}\ }\textbf {\bibinfo {volume} {5}},\ \bibinfo {pages} {110} (\bibinfo {year} {2009})}\BibitemShut {NoStop}%
\bibitem [{\citenamefont {Trivedi}\ \emph {et~al.}(2024)\citenamefont {Trivedi}, \citenamefont {Franco~Rubio},\ and\ \citenamefont {Cirac}}]{stability_analog}%
  \BibitemOpen
  \bibfield  {author} {\bibinfo {author} {\bibfnamefont {R.}~\bibnamefont {Trivedi}}, \bibinfo {author} {\bibfnamefont {A.}~\bibnamefont {Franco~Rubio}},\ and\ \bibinfo {author} {\bibfnamefont {J.~I.}\ \bibnamefont {Cirac}},\ }\href {https://doi.org/10.1038/s41467-024-50750-x} {\bibfield  {journal} {\bibinfo  {journal} {Nature Communications}\ }\textbf {\bibinfo {volume} {15}},\ \bibinfo {pages} {6507} (\bibinfo {year} {2024})}\BibitemShut {NoStop}%
\bibitem [{\citenamefont {Daley}\ \emph {et~al.}(2022)\citenamefont {Daley}, \citenamefont {Bloch}, \citenamefont {Kokail}, \citenamefont {Flannigan}, \citenamefont {Pearson}, \citenamefont {Troyer},\ and\ \citenamefont {Zoller}}]{practical_advantage}%
  \BibitemOpen
  \bibfield  {author} {\bibinfo {author} {\bibfnamefont {A.~J.}\ \bibnamefont {Daley}}, \bibinfo {author} {\bibfnamefont {I.}~\bibnamefont {Bloch}}, \bibinfo {author} {\bibfnamefont {C.}~\bibnamefont {Kokail}}, \bibinfo {author} {\bibfnamefont {S.}~\bibnamefont {Flannigan}}, \bibinfo {author} {\bibfnamefont {N.}~\bibnamefont {Pearson}}, \bibinfo {author} {\bibfnamefont {M.}~\bibnamefont {Troyer}},\ and\ \bibinfo {author} {\bibfnamefont {P.}~\bibnamefont {Zoller}},\ }\href {https://doi.org/10.1038/s41586-022-04940-6} {\bibfield  {journal} {\bibinfo  {journal} {Nature}\ }\textbf {\bibinfo {volume} {607}},\ \bibinfo {pages} {667} (\bibinfo {year} {2022})}\BibitemShut {NoStop}%
\bibitem [{\citenamefont {{Kitaev}}(2009)}]{KitaevChain}%
  \BibitemOpen
  \bibfield  {author} {\bibinfo {author} {\bibfnamefont {A.}~\bibnamefont {{Kitaev}}},\ }in\ \href {https://doi.org/10.1063/1.3149495} {\emph {\bibinfo {booktitle} {Advances in Theoretical Physics: Landau Memorial Conference}}},\ \bibinfo {series} {American Institute of Physics Conference Series}, Vol.\ \bibinfo {volume} {1134},\ \bibinfo {editor} {edited by\ \bibinfo {editor} {\bibfnamefont {V.}~\bibnamefont {{Lebedev}}}\ and\ \bibinfo {editor} {\bibfnamefont {M.}~\bibnamefont {{Feigel'Man}}}}\ (\bibinfo  {publisher} {AIP},\ \bibinfo {year} {2009})\ pp.\ \bibinfo {pages} {22--30},\ \Eprint {https://arxiv.org/abs/0901.2686} {arXiv:0901.2686 [cond-mat.mes-hall]} \BibitemShut {NoStop}%
\bibitem [{\citenamefont {{Dag}}\ \emph {et~al.}(2024)\citenamefont {{Dag}}, \citenamefont {{Ma}}, \citenamefont {{Myles Eugenio}}, \citenamefont {{Fang}},\ and\ \citenamefont {{Yelin}}}]{Ceren}%
  \BibitemOpen
  \bibfield  {author} {\bibinfo {author} {\bibfnamefont {C.~B.}\ \bibnamefont {{Dag}}}, \bibinfo {author} {\bibfnamefont {H.}~\bibnamefont {{Ma}}}, \bibinfo {author} {\bibfnamefont {P.}~\bibnamefont {{Myles Eugenio}}}, \bibinfo {author} {\bibfnamefont {F.}~\bibnamefont {{Fang}}},\ and\ \bibinfo {author} {\bibfnamefont {S.~F.}\ \bibnamefont {{Yelin}}},\ }\href {https://doi.org/10.48550/arXiv.2411.13643} {\bibfield  {journal} {\bibinfo  {journal} {arXiv e-prints}\ ,\ \bibinfo {eid} {arXiv:2411.13643}} (\bibinfo {year} {2024})},\ \Eprint {https://arxiv.org/abs/2411.13643} {arXiv:2411.13643 [quant-ph]} \BibitemShut {NoStop}%
\bibitem [{\citenamefont {Jo}\ \emph {et~al.}(2020)\citenamefont {Jo}, \citenamefont {Song}, \citenamefont {Kim},\ and\ \citenamefont {Ahn}}]{weakly_coupling}%
  \BibitemOpen
  \bibfield  {author} {\bibinfo {author} {\bibfnamefont {H.}~\bibnamefont {Jo}}, \bibinfo {author} {\bibfnamefont {Y.}~\bibnamefont {Song}}, \bibinfo {author} {\bibfnamefont {M.}~\bibnamefont {Kim}},\ and\ \bibinfo {author} {\bibfnamefont {J.}~\bibnamefont {Ahn}},\ }\href {https://doi.org/10.1103/PhysRevLett.124.033603} {\bibfield  {journal} {\bibinfo  {journal} {Phys. Rev. Lett.}\ }\textbf {\bibinfo {volume} {124}},\ \bibinfo {pages} {033603} (\bibinfo {year} {2020})}\BibitemShut {NoStop}%
\bibitem [{\citenamefont {Betts}(1998)}]{Betts_survey}%
  \BibitemOpen
  \bibfield  {author} {\bibinfo {author} {\bibfnamefont {J.~T.}\ \bibnamefont {Betts}},\ }\href {https://doi.org/10.2514/2.4231} {\bibfield  {journal} {\bibinfo  {journal} {Journal of Guidance, Control, and Dynamics}\ }\textbf {\bibinfo {volume} {21}},\ \bibinfo {pages} {193} (\bibinfo {year} {1998})},\ \Eprint {https://arxiv.org/abs/https://doi.org/10.2514/2.4231} {https://doi.org/10.2514/2.4231} \BibitemShut {NoStop}%
\bibitem [{\citenamefont {Kelly}(2017)}]{direct_collocation_intro}%
  \BibitemOpen
  \bibfield  {author} {\bibinfo {author} {\bibfnamefont {M.}~\bibnamefont {Kelly}},\ }\href {https://doi.org/10.1137/16M1062569} {\bibfield  {journal} {\bibinfo  {journal} {SIAM Review}\ }\textbf {\bibinfo {volume} {59}},\ \bibinfo {pages} {849} (\bibinfo {year} {2017})},\ \Eprint {https://arxiv.org/abs/https://doi.org/10.1137/16M1062569} {https://doi.org/10.1137/16M1062569} \BibitemShut {NoStop}%
\bibitem [{\citenamefont {Betts}(2020)}]{OC_book}%
  \BibitemOpen
  \bibfield  {author} {\bibinfo {author} {\bibfnamefont {J.~T.}\ \bibnamefont {Betts}},\ }\href {https://doi.org/10.1137/1.9781611976199} {\emph {\bibinfo {title} {Practical Methods for Optimal Control Using Nonlinear Programming, Third Edition}}}\ (\bibinfo  {publisher} {Society for Industrial and Applied Mathematics},\ \bibinfo {address} {Philadelphia, PA},\ \bibinfo {year} {2020})\ \Eprint {https://arxiv.org/abs/https://epubs.siam.org/doi/pdf/10.1137/1.9781611976199} {https://epubs.siam.org/doi/pdf/10.1137/1.9781611976199} \BibitemShut {NoStop}%
\bibitem [{Note2()}]{Note2}%
  \BibitemOpen
  \bibinfo {note} {Here, we emphasize that unitary fidelity may be an overly stringent metric, as analog simulations typically focus on local observables and signals, which are often more robust and perform better than global measures like unitary fidelity \cite {stability_analog}.}\BibitemShut {Stop}%
\bibitem [{\citenamefont {{Wurtz}}\ \emph {et~al.}(2023)\citenamefont {{Wurtz}}, \citenamefont {{Bylinskii}}, \citenamefont {{Braverman}}, \citenamefont {{Amato-Grill}}, \citenamefont {{Cantu}}, \citenamefont {{Huber}}, \citenamefont {{Lukin}}, \citenamefont {{Liu}}, \citenamefont {{Weinberg}}, \citenamefont {{Long}}, \citenamefont {{Wang}}, \citenamefont {{Gemelke}},\ and\ \citenamefont {{Keesling}}}]{aquila}%
  \BibitemOpen
  \bibfield  {author} {\bibinfo {author} {\bibfnamefont {J.}~\bibnamefont {{Wurtz}}}, \bibinfo {author} {\bibfnamefont {A.}~\bibnamefont {{Bylinskii}}}, \bibinfo {author} {\bibfnamefont {B.}~\bibnamefont {{Braverman}}}, \bibinfo {author} {\bibfnamefont {J.}~\bibnamefont {{Amato-Grill}}}, \bibinfo {author} {\bibfnamefont {S.~H.}\ \bibnamefont {{Cantu}}}, \bibinfo {author} {\bibfnamefont {F.}~\bibnamefont {{Huber}}}, \bibinfo {author} {\bibfnamefont {A.}~\bibnamefont {{Lukin}}}, \bibinfo {author} {\bibfnamefont {F.}~\bibnamefont {{Liu}}}, \bibinfo {author} {\bibfnamefont {P.}~\bibnamefont {{Weinberg}}}, \bibinfo {author} {\bibfnamefont {J.}~\bibnamefont {{Long}}}, \bibinfo {author} {\bibfnamefont {S.-T.}\ \bibnamefont {{Wang}}}, \bibinfo {author} {\bibfnamefont {N.}~\bibnamefont {{Gemelke}}},\ and\ \bibinfo {author} {\bibfnamefont {A.}~\bibnamefont {{Keesling}}},\ }\href {https://doi.org/10.48550/arXiv.2306.11727} {\bibfield  {journal} {\bibinfo  {journal} {arXiv e-prints}\ ,\ \bibinfo {eid} {arXiv:2306.11727}}
  (\bibinfo {year} {2023})},\ \Eprint {https://arxiv.org/abs/2306.11727} {arXiv:2306.11727 [quant-ph]} \BibitemShut {NoStop}%
\bibitem [{\citenamefont {{Lu}}\ \emph {et~al.}(2025)\citenamefont {{Lu}}, \citenamefont {{Jiao}}, \citenamefont {{Wolinski}}, \citenamefont {{Kornja{\v{c}}a}}, \citenamefont {{Hu}}, \citenamefont {{Cantu}}, \citenamefont {{Liu}}, \citenamefont {{Yelin}},\ and\ \citenamefont {{Wang}}}]{analog_digital_lu}%
  \BibitemOpen
  \bibfield  {author} {\bibinfo {author} {\bibfnamefont {J.~Z.}\ \bibnamefont {{Lu}}}, \bibinfo {author} {\bibfnamefont {L.}~\bibnamefont {{Jiao}}}, \bibinfo {author} {\bibfnamefont {K.}~\bibnamefont {{Wolinski}}}, \bibinfo {author} {\bibfnamefont {M.}~\bibnamefont {{Kornja{\v{c}}a}}}, \bibinfo {author} {\bibfnamefont {H.-Y.}\ \bibnamefont {{Hu}}}, \bibinfo {author} {\bibfnamefont {S.}~\bibnamefont {{Cantu}}}, \bibinfo {author} {\bibfnamefont {F.}~\bibnamefont {{Liu}}}, \bibinfo {author} {\bibfnamefont {S.~F.}\ \bibnamefont {{Yelin}}},\ and\ \bibinfo {author} {\bibfnamefont {S.-T.}\ \bibnamefont {{Wang}}},\ }\href {https://doi.org/10.1088/2058-9565/ad9177} {\bibfield  {journal} {\bibinfo  {journal} {Quantum Science and Technology}\ }\textbf {\bibinfo {volume} {10}},\ \bibinfo {eid} {015038} (\bibinfo {year} {2025})},\ \Eprint {https://arxiv.org/abs/2401.02940} {arXiv:2401.02940 [quant-ph]} \BibitemShut {NoStop}%
\bibitem [{\citenamefont {Michel}\ \emph {et~al.}(2023)\citenamefont {Michel}, \citenamefont {Grijalva}, \citenamefont {Henriet}, \citenamefont {Domain},\ and\ \citenamefont {Browaeys}}]{PhysRevA.107.042602}%
  \BibitemOpen
  \bibfield  {author} {\bibinfo {author} {\bibfnamefont {A.}~\bibnamefont {Michel}}, \bibinfo {author} {\bibfnamefont {S.}~\bibnamefont {Grijalva}}, \bibinfo {author} {\bibfnamefont {L.}~\bibnamefont {Henriet}}, \bibinfo {author} {\bibfnamefont {C.}~\bibnamefont {Domain}},\ and\ \bibinfo {author} {\bibfnamefont {A.}~\bibnamefont {Browaeys}},\ }\href {https://doi.org/10.1103/PhysRevA.107.042602} {\bibfield  {journal} {\bibinfo  {journal} {Phys. Rev. A}\ }\textbf {\bibinfo {volume} {107}},\ \bibinfo {pages} {042602} (\bibinfo {year} {2023})}\BibitemShut {NoStop}%
\bibitem [{Note3()}]{Note3}%
  \BibitemOpen
  \bibinfo {note} {We include a multiplicative shift term $k\Omega (t)$ in the model for $\Omega (t)$ because calibration experiments, including our own and those described in Ref.~\cite {Ceren}, indicate that the control error varies with the input values.}\BibitemShut {Stop}%
\bibitem [{\citenamefont {Gu}\ \emph {et~al.}(2024)\citenamefont {Gu}, \citenamefont {Cincio},\ and\ \citenamefont {Coles}}]{HL_Gu}%
  \BibitemOpen
  \bibfield  {author} {\bibinfo {author} {\bibfnamefont {A.}~\bibnamefont {Gu}}, \bibinfo {author} {\bibfnamefont {L.}~\bibnamefont {Cincio}},\ and\ \bibinfo {author} {\bibfnamefont {P.~J.}\ \bibnamefont {Coles}},\ }\href {https://doi.org/10.1038/s41467-023-44008-1} {\bibfield  {journal} {\bibinfo  {journal} {Nature Communications}\ }\textbf {\bibinfo {volume} {15}},\ \bibinfo {pages} {312} (\bibinfo {year} {2024})}\BibitemShut {NoStop}%
\bibitem [{\citenamefont {Hangleiter}\ \emph {et~al.}(2024)\citenamefont {Hangleiter}, \citenamefont {Roth}, \citenamefont {Fuksa}, \citenamefont {Eisert},\ and\ \citenamefont {Roushan}}]{robust_learning}%
  \BibitemOpen
  \bibfield  {author} {\bibinfo {author} {\bibfnamefont {D.}~\bibnamefont {Hangleiter}}, \bibinfo {author} {\bibfnamefont {I.}~\bibnamefont {Roth}}, \bibinfo {author} {\bibfnamefont {J.}~\bibnamefont {Fuksa}}, \bibinfo {author} {\bibfnamefont {J.}~\bibnamefont {Eisert}},\ and\ \bibinfo {author} {\bibfnamefont {P.}~\bibnamefont {Roushan}},\ }\href {https://doi.org/10.1038/s41467-024-52629-3} {\bibfield  {journal} {\bibinfo  {journal} {Nature Communications}\ }\textbf {\bibinfo {volume} {15}},\ \bibinfo {pages} {9595} (\bibinfo {year} {2024})}\BibitemShut {NoStop}%
\bibitem [{\citenamefont {Hu}\ \emph {et~al.}(2025{\natexlab{a}})\citenamefont {Hu}, \citenamefont {Gu}, \citenamefont {Majumder}, \citenamefont {Ren}, \citenamefont {Zhang}, \citenamefont {Wang}, \citenamefont {You}, \citenamefont {Minev}, \citenamefont {Yelin},\ and\ \citenamefont {Seif}}]{rss}%
  \BibitemOpen
  \bibfield  {author} {\bibinfo {author} {\bibfnamefont {H.-Y.}\ \bibnamefont {Hu}}, \bibinfo {author} {\bibfnamefont {A.}~\bibnamefont {Gu}}, \bibinfo {author} {\bibfnamefont {S.}~\bibnamefont {Majumder}}, \bibinfo {author} {\bibfnamefont {H.}~\bibnamefont {Ren}}, \bibinfo {author} {\bibfnamefont {Y.}~\bibnamefont {Zhang}}, \bibinfo {author} {\bibfnamefont {D.~S.}\ \bibnamefont {Wang}}, \bibinfo {author} {\bibfnamefont {Y.-Z.}\ \bibnamefont {You}}, \bibinfo {author} {\bibfnamefont {Z.}~\bibnamefont {Minev}}, \bibinfo {author} {\bibfnamefont {S.~F.}\ \bibnamefont {Yelin}},\ and\ \bibinfo {author} {\bibfnamefont {A.}~\bibnamefont {Seif}},\ }\href {https://doi.org/10.1038/s41467-025-57349-w} {\bibfield  {journal} {\bibinfo  {journal} {Nature Communications}\ }\textbf {\bibinfo {volume} {16}},\ \bibinfo {pages} {2943} (\bibinfo {year} {2025}{\natexlab{a}})}\BibitemShut {NoStop}%
\bibitem [{\citenamefont {Hu}\ \emph {et~al.}(2025{\natexlab{b}})\citenamefont {Hu}, \citenamefont {Ma}, \citenamefont {Gong}, \citenamefont {Ye}, \citenamefont {Tong}, \citenamefont {Flammia},\ and\ \citenamefont {Yelin}}]{ansatz_free_learning}%
  \BibitemOpen
  \bibfield  {author} {\bibinfo {author} {\bibfnamefont {H.-Y.}\ \bibnamefont {Hu}}, \bibinfo {author} {\bibfnamefont {M.}~\bibnamefont {Ma}}, \bibinfo {author} {\bibfnamefont {W.}~\bibnamefont {Gong}}, \bibinfo {author} {\bibfnamefont {Q.}~\bibnamefont {Ye}}, \bibinfo {author} {\bibfnamefont {Y.}~\bibnamefont {Tong}}, \bibinfo {author} {\bibfnamefont {S.~T.}\ \bibnamefont {Flammia}},\ and\ \bibinfo {author} {\bibfnamefont {S.~F.}\ \bibnamefont {Yelin}},\ }\href {https://doi.org/10.1103/j7b8-pb77} {\bibfield  {journal} {\bibinfo  {journal} {PRX Quantum}\ }\textbf {\bibinfo {volume} {6}},\ \bibinfo {pages} {040315} (\bibinfo {year} {2025}{\natexlab{b}})}\BibitemShut {NoStop}%
\bibitem [{\citenamefont {Schweigler}\ \emph {et~al.}(2017)\citenamefont {Schweigler}, \citenamefont {Kasper}, \citenamefont {Erne}, \citenamefont {Mazets}, \citenamefont {Rauer}, \citenamefont {Cataldini}, \citenamefont {Langen}, \citenamefont {Gasenzer}, \citenamefont {Berges},\ and\ \citenamefont {Schmiedmayer}}]{higher_correlations}%
  \BibitemOpen
  \bibfield  {author} {\bibinfo {author} {\bibfnamefont {T.}~\bibnamefont {Schweigler}}, \bibinfo {author} {\bibfnamefont {V.}~\bibnamefont {Kasper}}, \bibinfo {author} {\bibfnamefont {S.}~\bibnamefont {Erne}}, \bibinfo {author} {\bibfnamefont {I.}~\bibnamefont {Mazets}}, \bibinfo {author} {\bibfnamefont {B.}~\bibnamefont {Rauer}}, \bibinfo {author} {\bibfnamefont {F.}~\bibnamefont {Cataldini}}, \bibinfo {author} {\bibfnamefont {T.}~\bibnamefont {Langen}}, \bibinfo {author} {\bibfnamefont {T.}~\bibnamefont {Gasenzer}}, \bibinfo {author} {\bibfnamefont {J.}~\bibnamefont {Berges}},\ and\ \bibinfo {author} {\bibfnamefont {J.}~\bibnamefont {Schmiedmayer}},\ }\href {https://doi.org/10.1038/nature22310} {\bibfield  {journal} {\bibinfo  {journal} {Nature}\ }\textbf {\bibinfo {volume} {545}},\ \bibinfo {pages} {323} (\bibinfo {year} {2017})}\BibitemShut {NoStop}%
\bibitem [{\citenamefont {Menta}\ \emph {et~al.}(2025)\citenamefont {Menta}, \citenamefont {Cioni}, \citenamefont {Aiudi}, \citenamefont {Polini},\ and\ \citenamefont {Giovannetti}}]{global_superconducting_architecture}%
  \BibitemOpen
  \bibfield  {author} {\bibinfo {author} {\bibfnamefont {R.}~\bibnamefont {Menta}}, \bibinfo {author} {\bibfnamefont {F.}~\bibnamefont {Cioni}}, \bibinfo {author} {\bibfnamefont {R.}~\bibnamefont {Aiudi}}, \bibinfo {author} {\bibfnamefont {M.}~\bibnamefont {Polini}},\ and\ \bibinfo {author} {\bibfnamefont {V.}~\bibnamefont {Giovannetti}},\ }\href {https://doi.org/10.1103/PhysRevResearch.7.L012065} {\bibfield  {journal} {\bibinfo  {journal} {Phys. Rev. Res.}\ }\textbf {\bibinfo {volume} {7}},\ \bibinfo {pages} {L012065} (\bibinfo {year} {2025})}\BibitemShut {NoStop}%
\bibitem [{\citenamefont {Cioni}\ \emph {et~al.}(2026)\citenamefont {Cioni}, \citenamefont {Menta}, \citenamefont {Aiudi}, \citenamefont {Polini},\ and\ \citenamefont {Giovannetti}}]{belt_superconducting}%
  \BibitemOpen
  \bibfield  {author} {\bibinfo {author} {\bibfnamefont {F.}~\bibnamefont {Cioni}}, \bibinfo {author} {\bibfnamefont {R.}~\bibnamefont {Menta}}, \bibinfo {author} {\bibfnamefont {R.}~\bibnamefont {Aiudi}}, \bibinfo {author} {\bibfnamefont {M.}~\bibnamefont {Polini}},\ and\ \bibinfo {author} {\bibfnamefont {V.}~\bibnamefont {Giovannetti}},\ }\href {https://doi.org/10.1103/6zzp-ctyx} {\bibfield  {journal} {\bibinfo  {journal} {Phys. Rev. A}\ }\textbf {\bibinfo {volume} {113}},\ \bibinfo {pages} {012439} (\bibinfo {year} {2026})}\BibitemShut {NoStop}%
\bibitem [{\citenamefont {{Wang}}\ \emph {et~al.}(2025)\citenamefont {{Wang}}, \citenamefont {{Zhou}}, \citenamefont {{Shi}}, \citenamefont {{Huang}}, \citenamefont {{Yang}}, \citenamefont {{Zhang}}, \citenamefont {{Zhao}}, \citenamefont {{Xu}}, \citenamefont {{Li}}, \citenamefont {{Zhao}}, \citenamefont {{Feng}}, \citenamefont {{Xue}}, \citenamefont {{Liu}}, \citenamefont {{Ma}}, \citenamefont {{Fang}}, \citenamefont {{Liu}}, \citenamefont {{Wang}}, \citenamefont {{Xu}}, \citenamefont {{Yu}}, \citenamefont {{Fan}},\ and\ \citenamefont {{Zhao}}}]{tunable_superconducting}%
  \BibitemOpen
  \bibfield  {author} {\bibinfo {author} {\bibfnamefont {Z.~T.}\ \bibnamefont {{Wang}}}, \bibinfo {author} {\bibfnamefont {S.-Y.}\ \bibnamefont {{Zhou}}}, \bibinfo {author} {\bibfnamefont {Y.-H.}\ \bibnamefont {{Shi}}}, \bibinfo {author} {\bibfnamefont {K.}~\bibnamefont {{Huang}}}, \bibinfo {author} {\bibfnamefont {Z.~H.}\ \bibnamefont {{Yang}}}, \bibinfo {author} {\bibfnamefont {J.}~\bibnamefont {{Zhang}}}, \bibinfo {author} {\bibfnamefont {K.}~\bibnamefont {{Zhao}}}, \bibinfo {author} {\bibfnamefont {Y.}~\bibnamefont {{Xu}}}, \bibinfo {author} {\bibfnamefont {H.}~\bibnamefont {{Li}}}, \bibinfo {author} {\bibfnamefont {S.~K.}\ \bibnamefont {{Zhao}}}, \bibinfo {author} {\bibfnamefont {Y.}~\bibnamefont {{Feng}}}, \bibinfo {author} {\bibfnamefont {G.}~\bibnamefont {{Xue}}}, \bibinfo {author} {\bibfnamefont {Y.}~\bibnamefont {{Liu}}}, \bibinfo {author} {\bibfnamefont {W.-G.}\ \bibnamefont {{Ma}}}, \bibinfo {author} {\bibfnamefont {C.-P.}\ \bibnamefont {{Fang}}}, \bibinfo {author} {\bibfnamefont {H.-T.}\
  \bibnamefont {{Liu}}}, \bibinfo {author} {\bibfnamefont {Y.-Y.}\ \bibnamefont {{Wang}}}, \bibinfo {author} {\bibfnamefont {K.}~\bibnamefont {{Xu}}}, \bibinfo {author} {\bibfnamefont {H.}~\bibnamefont {{Yu}}}, \bibinfo {author} {\bibfnamefont {H.}~\bibnamefont {{Fan}}},\ and\ \bibinfo {author} {\bibfnamefont {S.~P.}\ \bibnamefont {{Zhao}}},\ }\href {https://doi.org/10.48550/arXiv.2509.02180} {\bibfield  {journal} {\bibinfo  {journal} {arXiv e-prints}\ ,\ \bibinfo {eid} {arXiv:2509.02180}} (\bibinfo {year} {2025})},\ \Eprint {https://arxiv.org/abs/2509.02180} {arXiv:2509.02180 [quant-ph]} \BibitemShut {NoStop}%
\bibitem [{\citenamefont {Huang}\ \emph {et~al.}(2020{\natexlab{b}})\citenamefont {Huang}, \citenamefont {Kueng},\ and\ \citenamefont {Preskill}}]{classical_shadow}%
  \BibitemOpen
  \bibfield  {author} {\bibinfo {author} {\bibfnamefont {H.-Y.}\ \bibnamefont {Huang}}, \bibinfo {author} {\bibfnamefont {R.}~\bibnamefont {Kueng}},\ and\ \bibinfo {author} {\bibfnamefont {J.}~\bibnamefont {Preskill}},\ }\href {https://doi.org/10.1038/s41567-020-0932-7} {\bibfield  {journal} {\bibinfo  {journal} {Nature Physics}\ }\textbf {\bibinfo {volume} {16}},\ \bibinfo {pages} {1050} (\bibinfo {year} {2020}{\natexlab{b}})}\BibitemShut {NoStop}%
\bibitem [{\citenamefont {Elben}\ \emph {et~al.}(2023)\citenamefont {Elben}, \citenamefont {Flammia}, \citenamefont {Huang}, \citenamefont {Kueng}, \citenamefont {Preskill}, \citenamefont {Vermersch},\ and\ \citenamefont {Zoller}}]{randomized_measurement}%
  \BibitemOpen
  \bibfield  {author} {\bibinfo {author} {\bibfnamefont {A.}~\bibnamefont {Elben}}, \bibinfo {author} {\bibfnamefont {S.~T.}\ \bibnamefont {Flammia}}, \bibinfo {author} {\bibfnamefont {H.-Y.}\ \bibnamefont {Huang}}, \bibinfo {author} {\bibfnamefont {R.}~\bibnamefont {Kueng}}, \bibinfo {author} {\bibfnamefont {J.}~\bibnamefont {Preskill}}, \bibinfo {author} {\bibfnamefont {B.}~\bibnamefont {Vermersch}},\ and\ \bibinfo {author} {\bibfnamefont {P.}~\bibnamefont {Zoller}},\ }\href {https://doi.org/10.1038/s42254-022-00535-2} {\bibfield  {journal} {\bibinfo  {journal} {Nature Reviews Physics}\ }\textbf {\bibinfo {volume} {5}},\ \bibinfo {pages} {9} (\bibinfo {year} {2023})}\BibitemShut {NoStop}%
\bibitem [{\citenamefont {{Rudolph}}\ \emph {et~al.}(2025)\citenamefont {{Rudolph}}, \citenamefont {{Jones}}, \citenamefont {{Teng}}, \citenamefont {{Angrisani}},\ and\ \citenamefont {{Holmes}}}]{pauli_prop}%
  \BibitemOpen
  \bibfield  {author} {\bibinfo {author} {\bibfnamefont {M.~S.}\ \bibnamefont {{Rudolph}}}, \bibinfo {author} {\bibfnamefont {T.}~\bibnamefont {{Jones}}}, \bibinfo {author} {\bibfnamefont {Y.}~\bibnamefont {{Teng}}}, \bibinfo {author} {\bibfnamefont {A.}~\bibnamefont {{Angrisani}}},\ and\ \bibinfo {author} {\bibfnamefont {Z.}~\bibnamefont {{Holmes}}},\ }\href {https://doi.org/10.48550/arXiv.2505.21606} {\bibfield  {journal} {\bibinfo  {journal} {arXiv e-prints}\ ,\ \bibinfo {eid} {arXiv:2505.21606}} (\bibinfo {year} {2025})},\ \Eprint {https://arxiv.org/abs/2505.21606} {arXiv:2505.21606 [quant-ph]} \BibitemShut {NoStop}%
\bibitem [{\citenamefont {Martinez}\ \emph {et~al.}(2025)\citenamefont {Martinez}, \citenamefont {Angrisani}, \citenamefont {Pankovets}, \citenamefont {Fawzi},\ and\ \citenamefont {Stilck Fran\ifmmode~\mbox{\c{c}}\else \c{c}\fi{}a}}]{j1gg-s6zb}%
  \BibitemOpen
  \bibfield  {author} {\bibinfo {author} {\bibfnamefont {V.}~\bibnamefont {Martinez}}, \bibinfo {author} {\bibfnamefont {A.}~\bibnamefont {Angrisani}}, \bibinfo {author} {\bibfnamefont {E.}~\bibnamefont {Pankovets}}, \bibinfo {author} {\bibfnamefont {O.}~\bibnamefont {Fawzi}},\ and\ \bibinfo {author} {\bibfnamefont {D.}~\bibnamefont {Stilck Fran\ifmmode~\mbox{\c{c}}\else \c{c}\fi{}a}},\ }\href {https://doi.org/10.1103/j1gg-s6zb} {\bibfield  {journal} {\bibinfo  {journal} {Phys. Rev. Lett.}\ }\textbf {\bibinfo {volume} {134}},\ \bibinfo {pages} {250602} (\bibinfo {year} {2025})}\BibitemShut {NoStop}%
\bibitem [{\citenamefont {Shao}\ \emph {et~al.}(2024{\natexlab{b}})\citenamefont {Shao}, \citenamefont {Wei}, \citenamefont {Cheng},\ and\ \citenamefont {Liu}}]{PhysRevLett.133.120603}%
  \BibitemOpen
  \bibfield  {author} {\bibinfo {author} {\bibfnamefont {Y.}~\bibnamefont {Shao}}, \bibinfo {author} {\bibfnamefont {F.}~\bibnamefont {Wei}}, \bibinfo {author} {\bibfnamefont {S.}~\bibnamefont {Cheng}},\ and\ \bibinfo {author} {\bibfnamefont {Z.}~\bibnamefont {Liu}},\ }\href {https://doi.org/10.1103/PhysRevLett.133.120603} {\bibfield  {journal} {\bibinfo  {journal} {Phys. Rev. Lett.}\ }\textbf {\bibinfo {volume} {133}},\ \bibinfo {pages} {120603} (\bibinfo {year} {2024}{\natexlab{b}})}\BibitemShut {NoStop}%
\bibitem [{\citenamefont {Patra}\ \emph {et~al.}(2024)\citenamefont {Patra}, \citenamefont {Jahromi}, \citenamefont {Singh},\ and\ \citenamefont {Or\'us}}]{PhysRevResearch.6.013326}%
  \BibitemOpen
  \bibfield  {author} {\bibinfo {author} {\bibfnamefont {S.}~\bibnamefont {Patra}}, \bibinfo {author} {\bibfnamefont {S.~S.}\ \bibnamefont {Jahromi}}, \bibinfo {author} {\bibfnamefont {S.}~\bibnamefont {Singh}},\ and\ \bibinfo {author} {\bibfnamefont {R.}~\bibnamefont {Or\'us}},\ }\href {https://doi.org/10.1103/PhysRevResearch.6.013326} {\bibfield  {journal} {\bibinfo  {journal} {Phys. Rev. Res.}\ }\textbf {\bibinfo {volume} {6}},\ \bibinfo {pages} {013326} (\bibinfo {year} {2024})}\BibitemShut {NoStop}%
\bibitem [{\citenamefont {Tindall}\ \emph {et~al.}(2024)\citenamefont {Tindall}, \citenamefont {Fishman}, \citenamefont {Stoudenmire},\ and\ \citenamefont {Sels}}]{PRXQuantum.5.010308}%
  \BibitemOpen
  \bibfield  {author} {\bibinfo {author} {\bibfnamefont {J.}~\bibnamefont {Tindall}}, \bibinfo {author} {\bibfnamefont {M.}~\bibnamefont {Fishman}}, \bibinfo {author} {\bibfnamefont {E.~M.}\ \bibnamefont {Stoudenmire}},\ and\ \bibinfo {author} {\bibfnamefont {D.}~\bibnamefont {Sels}},\ }\href {https://doi.org/10.1103/PRXQuantum.5.010308} {\bibfield  {journal} {\bibinfo  {journal} {PRX Quantum}\ }\textbf {\bibinfo {volume} {5}},\ \bibinfo {pages} {010308} (\bibinfo {year} {2024})}\BibitemShut {NoStop}%
\bibitem [{\citenamefont {Tindall}\ and\ \citenamefont {Fishman}(2023)}]{bp_TN}%
  \BibitemOpen
  \bibfield  {author} {\bibinfo {author} {\bibfnamefont {J.}~\bibnamefont {Tindall}}\ and\ \bibinfo {author} {\bibfnamefont {M.}~\bibnamefont {Fishman}},\ }\href {https://doi.org/10.21468/SciPostPhys.15.6.222} {\bibfield  {journal} {\bibinfo  {journal} {SciPost Phys.}\ }\textbf {\bibinfo {volume} {15}},\ \bibinfo {pages} {222} (\bibinfo {year} {2023})}\BibitemShut {NoStop}%
\bibitem [{\citenamefont {{Rudolph}}\ and\ \citenamefont {{Tindall}}(2025)}]{2025arXiv250711424R}%
  \BibitemOpen
  \bibfield  {author} {\bibinfo {author} {\bibfnamefont {M.~S.}\ \bibnamefont {{Rudolph}}}\ and\ \bibinfo {author} {\bibfnamefont {J.}~\bibnamefont {{Tindall}}},\ }\href@noop {} {\bibfield  {journal} {\bibinfo  {journal} {arXiv e-prints}\ ,\ \bibinfo {eid} {arXiv:2507.11424}} (\bibinfo {year} {2025})},\ \Eprint {https://arxiv.org/abs/2507.11424} {arXiv:2507.11424 [quant-ph]} \BibitemShut {NoStop}%
\bibitem [{\citenamefont {Sivak}\ \emph {et~al.}(2022)\citenamefont {Sivak}, \citenamefont {Eickbusch}, \citenamefont {Liu}, \citenamefont {Royer}, \citenamefont {Tsioutsios},\ and\ \citenamefont {Devoret}}]{PhysRevX.12.011059}%
  \BibitemOpen
  \bibfield  {author} {\bibinfo {author} {\bibfnamefont {V.~V.}\ \bibnamefont {Sivak}}, \bibinfo {author} {\bibfnamefont {A.}~\bibnamefont {Eickbusch}}, \bibinfo {author} {\bibfnamefont {H.}~\bibnamefont {Liu}}, \bibinfo {author} {\bibfnamefont {B.}~\bibnamefont {Royer}}, \bibinfo {author} {\bibfnamefont {I.}~\bibnamefont {Tsioutsios}},\ and\ \bibinfo {author} {\bibfnamefont {M.~H.}\ \bibnamefont {Devoret}},\ }\href {https://doi.org/10.1103/PhysRevX.12.011059} {\bibfield  {journal} {\bibinfo  {journal} {Phys. Rev. X}\ }\textbf {\bibinfo {volume} {12}},\ \bibinfo {pages} {011059} (\bibinfo {year} {2022})}\BibitemShut {NoStop}%
\bibitem [{\citenamefont {Huang}\ \emph {et~al.}(2024)\citenamefont {Huang}, \citenamefont {Liu}, \citenamefont {Broughton}, \citenamefont {Kim}, \citenamefont {Anshu}, \citenamefont {Landau},\ and\ \citenamefont {McClean}}]{10.1145/3618260.3649722}%
  \BibitemOpen
  \bibfield  {author} {\bibinfo {author} {\bibfnamefont {H.-Y.}\ \bibnamefont {Huang}}, \bibinfo {author} {\bibfnamefont {Y.}~\bibnamefont {Liu}}, \bibinfo {author} {\bibfnamefont {M.}~\bibnamefont {Broughton}}, \bibinfo {author} {\bibfnamefont {I.}~\bibnamefont {Kim}}, \bibinfo {author} {\bibfnamefont {A.}~\bibnamefont {Anshu}}, \bibinfo {author} {\bibfnamefont {Z.}~\bibnamefont {Landau}},\ and\ \bibinfo {author} {\bibfnamefont {J.~R.}\ \bibnamefont {McClean}},\ }in\ \href {https://doi.org/10.1145/3618260.3649722} {\emph {\bibinfo {booktitle} {Proceedings of the 56th Annual ACM Symposium on Theory of Computing}}},\ \bibinfo {series and number} {STOC 2024}\ (\bibinfo  {publisher} {Association for Computing Machinery},\ \bibinfo {address} {New York, NY, USA},\ \bibinfo {year} {2024})\ p.\ \bibinfo {pages} {1343–1351}\BibitemShut {NoStop}%
\bibitem [{\citenamefont {{Cao}}\ \emph {et~al.}(2024)\citenamefont {{Cao}}, \citenamefont {{Liu}}, \citenamefont {{Deng}}, \citenamefont {{Xia}}, \citenamefont {{Wu}},\ and\ \citenamefont {{Wang}}}]{2024arXiv241207764C}%
  \BibitemOpen
  \bibfield  {author} {\bibinfo {author} {\bibfnamefont {Y.}~\bibnamefont {{Cao}}}, \bibinfo {author} {\bibfnamefont {S.}~\bibnamefont {{Liu}}}, \bibinfo {author} {\bibfnamefont {H.}~\bibnamefont {{Deng}}}, \bibinfo {author} {\bibfnamefont {Z.}~\bibnamefont {{Xia}}}, \bibinfo {author} {\bibfnamefont {X.}~\bibnamefont {{Wu}}},\ and\ \bibinfo {author} {\bibfnamefont {Y.-X.}\ \bibnamefont {{Wang}}},\ }\href {https://doi.org/10.48550/arXiv.2412.07764} {\bibfield  {journal} {\bibinfo  {journal} {arXiv e-prints}\ ,\ \bibinfo {eid} {arXiv:2412.07764}} (\bibinfo {year} {2024})},\ \Eprint {https://arxiv.org/abs/2412.07764} {arXiv:2412.07764 [quant-ph]} \BibitemShut {NoStop}%
\bibitem [{\citenamefont {{Goll}}\ and\ \citenamefont {{Jonsson}}(2025)}]{Geometric_speed}%
  \BibitemOpen
  \bibfield  {author} {\bibinfo {author} {\bibfnamefont {M.}~\bibnamefont {{Goll}}}\ and\ \bibinfo {author} {\bibfnamefont {R.~H.}\ \bibnamefont {{Jonsson}}},\ }\href {https://doi.org/10.48550/arXiv.2504.15175} {\bibfield  {journal} {\bibinfo  {journal} {arXiv e-prints}\ ,\ \bibinfo {eid} {arXiv:2504.15175}} (\bibinfo {year} {2025})},\ \Eprint {https://arxiv.org/abs/2504.15175} {arXiv:2504.15175 [quant-ph]} \BibitemShut {NoStop}%
\bibitem [{\citenamefont {{Wiedmann}}\ and\ \citenamefont {{Burgarth}}(2025)}]{QSL}%
  \BibitemOpen
  \bibfield  {author} {\bibinfo {author} {\bibfnamefont {M.}~\bibnamefont {{Wiedmann}}}\ and\ \bibinfo {author} {\bibfnamefont {D.}~\bibnamefont {{Burgarth}}},\ }\href {https://doi.org/10.48550/arXiv.2506.10069} {\bibfield  {journal} {\bibinfo  {journal} {arXiv e-prints}\ ,\ \bibinfo {eid} {arXiv:2506.10069}} (\bibinfo {year} {2025})},\ \Eprint {https://arxiv.org/abs/2506.10069} {arXiv:2506.10069 [quant-ph]} \BibitemShut {NoStop}%
\bibitem [{\citenamefont {Nielsen}(2006)}]{geometric_lower_bound}%
  \BibitemOpen
  \bibfield  {author} {\bibinfo {author} {\bibfnamefont {M.~A.}\ \bibnamefont {Nielsen}},\ }\href@noop {} {\bibfield  {journal} {\bibinfo  {journal} {Quantum Info. Comput.}\ }\textbf {\bibinfo {volume} {6}},\ \bibinfo {pages} {213–262} (\bibinfo {year} {2006})}\BibitemShut {NoStop}%
\bibitem [{\citenamefont {{D'Alessandro}}(2008)}]{2008arXiv0803.1193D}%
  \BibitemOpen
  \bibfield  {author} {\bibinfo {author} {\bibfnamefont {D.}~\bibnamefont {{D'Alessandro}}},\ }\href {https://doi.org/10.48550/arXiv.0803.1193} {\bibfield  {journal} {\bibinfo  {journal} {arXiv e-prints}\ ,\ \bibinfo {eid} {arXiv:0803.1193}} (\bibinfo {year} {2008})},\ \Eprint {https://arxiv.org/abs/0803.1193} {arXiv:0803.1193 [quant-ph]} \BibitemShut {NoStop}%
\bibitem [{\citenamefont {Bravo}\ \emph {et~al.}(2024)\citenamefont {Bravo}, \citenamefont {Ponce}, \citenamefont {Hu},\ and\ \citenamefont {Yelin}}]{analogQML}%
  \BibitemOpen
  \bibfield  {author} {\bibinfo {author} {\bibfnamefont {R.~A.}\ \bibnamefont {Bravo}}, \bibinfo {author} {\bibfnamefont {J.~G.}\ \bibnamefont {Ponce}}, \bibinfo {author} {\bibfnamefont {H.-Y.}\ \bibnamefont {Hu}},\ and\ \bibinfo {author} {\bibfnamefont {S.~F.}\ \bibnamefont {Yelin}},\ }\href {https://doi.org/10.1063/5.0235279} {\bibfield  {journal} {\bibinfo  {journal} {APL Quantum}\ }\textbf {\bibinfo {volume} {1}},\ \bibinfo {pages} {046121} (\bibinfo {year} {2024})},\ \Eprint {https://arxiv.org/abs/https://pubs.aip.org/aip/apq/article-pdf/doi/10.1063/5.0235279/20289864/046121\_1\_5.0235279.pdf} {https://pubs.aip.org/aip/apq/article-pdf/doi/10.1063/5.0235279/20289864/046121\_1\_5.0235279.pdf} \BibitemShut {NoStop}%
\bibitem [{\citenamefont {Barenco}\ \emph {et~al.}(1995)\citenamefont {Barenco}, \citenamefont {Bennett}, \citenamefont {Cleve}, \citenamefont {DiVincenzo}, \citenamefont {Margolus}, \citenamefont {Shor}, \citenamefont {Sleator}, \citenamefont {Smolin},\ and\ \citenamefont {Weinfurter}}]{Barenco_1995_UQC}%
  \BibitemOpen
  \bibfield  {author} {\bibinfo {author} {\bibfnamefont {A.}~\bibnamefont {Barenco}}, \bibinfo {author} {\bibfnamefont {C.~H.}\ \bibnamefont {Bennett}}, \bibinfo {author} {\bibfnamefont {R.}~\bibnamefont {Cleve}}, \bibinfo {author} {\bibfnamefont {D.~P.}\ \bibnamefont {DiVincenzo}}, \bibinfo {author} {\bibfnamefont {N.}~\bibnamefont {Margolus}}, \bibinfo {author} {\bibfnamefont {P.}~\bibnamefont {Shor}}, \bibinfo {author} {\bibfnamefont {T.}~\bibnamefont {Sleator}}, \bibinfo {author} {\bibfnamefont {J.~A.}\ \bibnamefont {Smolin}},\ and\ \bibinfo {author} {\bibfnamefont {H.}~\bibnamefont {Weinfurter}},\ }\href {https://doi.org/10.1103/PhysRevA.52.3457} {\bibfield  {journal} {\bibinfo  {journal} {Phys. Rev. A}\ }\textbf {\bibinfo {volume} {52}},\ \bibinfo {pages} {3457} (\bibinfo {year} {1995})}\BibitemShut {NoStop}%
\bibitem [{\citenamefont {Oszmaniec}\ and\ \citenamefont {Zimborás}(2017)}]{Oszmaniec_2017}%
  \BibitemOpen
  \bibfield  {author} {\bibinfo {author} {\bibfnamefont {M.}~\bibnamefont {Oszmaniec}}\ and\ \bibinfo {author} {\bibfnamefont {Z.}~\bibnamefont {Zimborás}},\ }\bibfield  {journal} {\bibinfo  {journal} {Physical Review Letters}\ }\textbf {\bibinfo {volume} {119}},\ \href {https://doi.org/10.1103/physrevlett.119.220502} {10.1103/physrevlett.119.220502} (\bibinfo {year} {2017})\BibitemShut {NoStop}%
\bibitem [{\citenamefont {{Mark}}\ \emph {et~al.}(2024)\citenamefont {{Mark}}, \citenamefont {{Hu}}, \citenamefont {{Kwan}}, \citenamefont {{Kokail}}, \citenamefont {{Choi}},\ and\ \citenamefont {{Yelin}}}]{dwave}%
  \BibitemOpen
  \bibfield  {author} {\bibinfo {author} {\bibfnamefont {D.~K.}\ \bibnamefont {{Mark}}}, \bibinfo {author} {\bibfnamefont {H.-Y.}\ \bibnamefont {{Hu}}}, \bibinfo {author} {\bibfnamefont {J.}~\bibnamefont {{Kwan}}}, \bibinfo {author} {\bibfnamefont {C.}~\bibnamefont {{Kokail}}}, \bibinfo {author} {\bibfnamefont {S.}~\bibnamefont {{Choi}}},\ and\ \bibinfo {author} {\bibfnamefont {S.~F.}\ \bibnamefont {{Yelin}}},\ }\href {https://doi.org/10.48550/arXiv.2412.13186} {\bibfield  {journal} {\bibinfo  {journal} {arXiv e-prints}\ ,\ \bibinfo {eid} {arXiv:2412.13186}} (\bibinfo {year} {2024})},\ \Eprint {https://arxiv.org/abs/2412.13186} {arXiv:2412.13186 [cond-mat.quant-gas]} \BibitemShut {NoStop}%
\bibitem [{\citenamefont {Maskara}\ \emph {et~al.}(2025)\citenamefont {Maskara}, \citenamefont {Ostermann}, \citenamefont {Shee}, \citenamefont {Kalinowski}, \citenamefont {McClain~Gomez}, \citenamefont {Araiza~Bravo}, \citenamefont {Wang}, \citenamefont {Krylov}, \citenamefont {Yao}, \citenamefont {Head-Gordon}, \citenamefont {Lukin},\ and\ \citenamefont {Yelin}}]{rydberg_chemistry}%
  \BibitemOpen
  \bibfield  {author} {\bibinfo {author} {\bibfnamefont {N.}~\bibnamefont {Maskara}}, \bibinfo {author} {\bibfnamefont {S.}~\bibnamefont {Ostermann}}, \bibinfo {author} {\bibfnamefont {J.}~\bibnamefont {Shee}}, \bibinfo {author} {\bibfnamefont {M.}~\bibnamefont {Kalinowski}}, \bibinfo {author} {\bibfnamefont {A.}~\bibnamefont {McClain~Gomez}}, \bibinfo {author} {\bibfnamefont {R.}~\bibnamefont {Araiza~Bravo}}, \bibinfo {author} {\bibfnamefont {D.~S.}\ \bibnamefont {Wang}}, \bibinfo {author} {\bibfnamefont {A.~I.}\ \bibnamefont {Krylov}}, \bibinfo {author} {\bibfnamefont {N.~Y.}\ \bibnamefont {Yao}}, \bibinfo {author} {\bibfnamefont {M.}~\bibnamefont {Head-Gordon}}, \bibinfo {author} {\bibfnamefont {M.~D.}\ \bibnamefont {Lukin}},\ and\ \bibinfo {author} {\bibfnamefont {S.~F.}\ \bibnamefont {Yelin}},\ }\href {https://doi.org/10.1038/s41567-024-02738-z} {\bibfield  {journal} {\bibinfo  {journal} {Nature Physics}\ }\textbf {\bibinfo {volume} {21}},\ \bibinfo {pages} {289} (\bibinfo {year} {2025})}\BibitemShut
  {NoStop}%
\bibitem [{\citenamefont {{Evered}}\ \emph {et~al.}(2025)\citenamefont {{Evered}}, \citenamefont {{Kalinowski}}, \citenamefont {{Geim}}, \citenamefont {{Manovitz}}, \citenamefont {{Bluvstein}}, \citenamefont {{Li}}, \citenamefont {{Maskara}}, \citenamefont {{Zhou}}, \citenamefont {{Ebadi}}, \citenamefont {{Xu}}, \citenamefont {{Campo}}, \citenamefont {{Cain}}, \citenamefont {{Ostermann}}, \citenamefont {{Yelin}}, \citenamefont {{Sachdev}}, \citenamefont {{Greiner}}, \citenamefont {{Vuleti{\'c}}},\ and\ \citenamefont {{Lukin}}}]{2025arXiv250118554E}%
  \BibitemOpen
  \bibfield  {author} {\bibinfo {author} {\bibfnamefont {S.~J.}\ \bibnamefont {{Evered}}}, \bibinfo {author} {\bibfnamefont {M.}~\bibnamefont {{Kalinowski}}}, \bibinfo {author} {\bibfnamefont {A.~A.}\ \bibnamefont {{Geim}}}, \bibinfo {author} {\bibfnamefont {T.}~\bibnamefont {{Manovitz}}}, \bibinfo {author} {\bibfnamefont {D.}~\bibnamefont {{Bluvstein}}}, \bibinfo {author} {\bibfnamefont {S.~H.}\ \bibnamefont {{Li}}}, \bibinfo {author} {\bibfnamefont {N.}~\bibnamefont {{Maskara}}}, \bibinfo {author} {\bibfnamefont {H.}~\bibnamefont {{Zhou}}}, \bibinfo {author} {\bibfnamefont {S.}~\bibnamefont {{Ebadi}}}, \bibinfo {author} {\bibfnamefont {M.}~\bibnamefont {{Xu}}}, \bibinfo {author} {\bibfnamefont {J.}~\bibnamefont {{Campo}}}, \bibinfo {author} {\bibfnamefont {M.}~\bibnamefont {{Cain}}}, \bibinfo {author} {\bibfnamefont {S.}~\bibnamefont {{Ostermann}}}, \bibinfo {author} {\bibfnamefont {S.~F.}\ \bibnamefont {{Yelin}}}, \bibinfo {author} {\bibfnamefont {S.}~\bibnamefont {{Sachdev}}}, \bibinfo {author}
  {\bibfnamefont {M.}~\bibnamefont {{Greiner}}}, \bibinfo {author} {\bibfnamefont {V.}~\bibnamefont {{Vuleti{\'c}}}},\ and\ \bibinfo {author} {\bibfnamefont {M.~D.}\ \bibnamefont {{Lukin}}},\ }\href {https://doi.org/10.48550/arXiv.2501.18554} {\bibfield  {journal} {\bibinfo  {journal} {arXiv e-prints}\ ,\ \bibinfo {eid} {arXiv:2501.18554}} (\bibinfo {year} {2025})},\ \Eprint {https://arxiv.org/abs/2501.18554} {arXiv:2501.18554 [quant-ph]} \BibitemShut {NoStop}%
\bibitem [{\citenamefont {Hall}(2003)}]{hall2003liegroup}%
  \BibitemOpen
  \bibfield  {author} {\bibinfo {author} {\bibfnamefont {B.}~\bibnamefont {Hall}},\ }\href {https://books.google.com/books?id=m1VQi8HmEwcC} {\emph {\bibinfo {title} {Lie Groups, Lie Algebras, and Representations: An Elementary Introduction}}},\ Graduate Texts in Mathematics\ (\bibinfo  {publisher} {Springer},\ \bibinfo {year} {2003})\BibitemShut {NoStop}%
\bibitem [{\citenamefont {Rossmann}(2006)}]{rossmann2006liegroup}%
  \BibitemOpen
  \bibfield  {author} {\bibinfo {author} {\bibfnamefont {W.}~\bibnamefont {Rossmann}},\ }\href {https://books.google.com/books?id=bAjulQ65W-UC} {\emph {\bibinfo {title} {Lie Groups: An Introduction Through Linear Groups}}},\ Oxford Mathematics\ (\bibinfo  {publisher} {Oxford University Press},\ \bibinfo {year} {2006})\BibitemShut {NoStop}%
\bibitem [{\citenamefont {Suzuki}(1976)}]{Suzuki1976}%
  \BibitemOpen
  \bibfield  {author} {\bibinfo {author} {\bibfnamefont {M.}~\bibnamefont {Suzuki}},\ }\href {https://doi.org/10.1007/BF01609348} {\bibfield  {journal} {\bibinfo  {journal} {Communications in Mathematical Physics}\ }\textbf {\bibinfo {volume} {51}},\ \bibinfo {pages} {183} (\bibinfo {year} {1976})}\BibitemShut {NoStop}%
\bibitem [{\citenamefont {Suzuki}(1977)}]{Suzuki1977}%
  \BibitemOpen
  \bibfield  {author} {\bibinfo {author} {\bibfnamefont {M.}~\bibnamefont {Suzuki}},\ }\href {https://doi.org/10.1007/BF01614161} {\bibfield  {journal} {\bibinfo  {journal} {Communications in Mathematical Physics}\ }\textbf {\bibinfo {volume} {57}},\ \bibinfo {pages} {193} (\bibinfo {year} {1977})}\BibitemShut {NoStop}%
\bibitem [{\citenamefont {Sagle}\ and\ \citenamefont {Walde}(1986)}]{sagle1986introduction}%
  \BibitemOpen
  \bibfield  {author} {\bibinfo {author} {\bibfnamefont {A.}~\bibnamefont {Sagle}}\ and\ \bibinfo {author} {\bibfnamefont {R.}~\bibnamefont {Walde}},\ }\href {https://books.google.com/books?id=gnzNCgAAQBAJ} {\emph {\bibinfo {title} {Introduction to Lie Groups and Lie Algebra, 51}}}\ (\bibinfo  {publisher} {Academic Press},\ \bibinfo {year} {1986})\BibitemShut {NoStop}%
\bibitem [{\citenamefont {Deutsch}\ \emph {et~al.}(1995)\citenamefont {Deutsch}, \citenamefont {Barenco},\ and\ \citenamefont {Ekert}}]{deutsch}%
  \BibitemOpen
  \bibfield  {author} {\bibinfo {author} {\bibfnamefont {D.~E.}\ \bibnamefont {Deutsch}}, \bibinfo {author} {\bibfnamefont {A.}~\bibnamefont {Barenco}},\ and\ \bibinfo {author} {\bibfnamefont {A.}~\bibnamefont {Ekert}},\ }\href {https://doi.org/10.1098/rspa.1995.0065} {\bibfield  {journal} {\bibinfo  {journal} {Proceedings of the Royal Society of London. Series A: Mathematical and Physical Sciences}\ }\textbf {\bibinfo {volume} {449}},\ \bibinfo {pages} {669} (\bibinfo {year} {1995})},\ \Eprint {https://arxiv.org/abs/https://royalsocietypublishing.org/doi/pdf/10.1098/rspa.1995.0065} {https://royalsocietypublishing.org/doi/pdf/10.1098/rspa.1995.0065} \BibitemShut {NoStop}%
\bibitem [{\citenamefont {Lloyd}\ and\ \citenamefont {Montangero}(2014)}]{Lloyd_2014}%
  \BibitemOpen
  \bibfield  {author} {\bibinfo {author} {\bibfnamefont {S.}~\bibnamefont {Lloyd}}\ and\ \bibinfo {author} {\bibfnamefont {S.}~\bibnamefont {Montangero}},\ }\bibfield  {journal} {\bibinfo  {journal} {Physical Review Letters}\ }\textbf {\bibinfo {volume} {113}},\ \href {https://doi.org/10.1103/physrevlett.113.010502} {10.1103/physrevlett.113.010502} (\bibinfo {year} {2014})\BibitemShut {NoStop}%
\bibitem [{\citenamefont {Chitode}(2021)}]{chitode2021information}%
  \BibitemOpen
  \bibfield  {author} {\bibinfo {author} {\bibfnamefont {J.}~\bibnamefont {Chitode}},\ }\href {https://books.google.com/books?id=tZYdEAAAQBAJ} {\emph {\bibinfo {title} {Information Theory and Coding}}}\ (\bibinfo  {publisher} {Amazon Digital Services LLC - KDP Print US},\ \bibinfo {year} {2021})\BibitemShut {NoStop}%
\bibitem [{\citenamefont {Hernández-Hernández}\ \emph {et~al.}(1996)\citenamefont {Hernández-Hernández}, \citenamefont {Hernandez-Lerma},\ and\ \citenamefont {Taksar}}]{Hernandez1996}%
  \BibitemOpen
  \bibfield  {author} {\bibinfo {author} {\bibfnamefont {D.}~\bibnamefont {Hernández-Hernández}}, \bibinfo {author} {\bibfnamefont {O.}~\bibnamefont {Hernandez-Lerma}},\ and\ \bibinfo {author} {\bibfnamefont {M.}~\bibnamefont {Taksar}},\ }\href {https://doi.org/10.4064/am-24-1-17-33} {\bibfield  {journal} {\bibinfo  {journal} {Applicationes Mathematicae}\ }\textbf {\bibinfo {volume} {24}},\ \bibinfo {pages} {17} (\bibinfo {year} {1996})}\BibitemShut {NoStop}%
\bibitem [{\citenamefont {Vidal}(2003)}]{Vidal_2003_Efficient}%
  \BibitemOpen
  \bibfield  {author} {\bibinfo {author} {\bibfnamefont {G.}~\bibnamefont {Vidal}},\ }\bibfield  {journal} {\bibinfo  {journal} {Physical Review Letters}\ }\textbf {\bibinfo {volume} {91}},\ \href {https://doi.org/10.1103/physrevlett.91.147902} {10.1103/physrevlett.91.147902} (\bibinfo {year} {2003})\BibitemShut {NoStop}%
\bibitem [{\citenamefont {Vidal}(2004)}]{Vidal_2004}%
  \BibitemOpen
  \bibfield  {author} {\bibinfo {author} {\bibfnamefont {G.}~\bibnamefont {Vidal}},\ }\bibfield  {journal} {\bibinfo  {journal} {Physical Review Letters}\ }\textbf {\bibinfo {volume} {93}},\ \href {https://doi.org/10.1103/physrevlett.93.040502} {10.1103/physrevlett.93.040502} (\bibinfo {year} {2004})\BibitemShut {NoStop}%
\bibitem [{\citenamefont {Verstraete}\ \emph {et~al.}(2004{\natexlab{a}})\citenamefont {Verstraete}, \citenamefont {García-Ripoll},\ and\ \citenamefont {Cirac}}]{Verstraete_2004}%
  \BibitemOpen
  \bibfield  {author} {\bibinfo {author} {\bibfnamefont {F.}~\bibnamefont {Verstraete}}, \bibinfo {author} {\bibfnamefont {J.~J.}\ \bibnamefont {García-Ripoll}},\ and\ \bibinfo {author} {\bibfnamefont {J.~I.}\ \bibnamefont {Cirac}},\ }\bibfield  {journal} {\bibinfo  {journal} {Physical Review Letters}\ }\textbf {\bibinfo {volume} {93}},\ \href {https://doi.org/10.1103/physrevlett.93.207204} {10.1103/physrevlett.93.207204} (\bibinfo {year} {2004}{\natexlab{a}})\BibitemShut {NoStop}%
\bibitem [{\citenamefont {Verstraete}\ \emph {et~al.}(2004{\natexlab{b}})\citenamefont {Verstraete}, \citenamefont {Porras},\ and\ \citenamefont {Cirac}}]{Verstraete_2004_DMRG}%
  \BibitemOpen
  \bibfield  {author} {\bibinfo {author} {\bibfnamefont {F.}~\bibnamefont {Verstraete}}, \bibinfo {author} {\bibfnamefont {D.}~\bibnamefont {Porras}},\ and\ \bibinfo {author} {\bibfnamefont {J.~I.}\ \bibnamefont {Cirac}},\ }\bibfield  {journal} {\bibinfo  {journal} {Physical Review Letters}\ }\textbf {\bibinfo {volume} {93}},\ \href {https://doi.org/10.1103/physrevlett.93.227205} {10.1103/physrevlett.93.227205} (\bibinfo {year} {2004}{\natexlab{b}})\BibitemShut {NoStop}%
\bibitem [{\citenamefont {Prosen}\ and\ \citenamefont {Pi\ifmmode~\check{z}\else \v{z}\fi{}orn}(2007)}]{Prosen_2007_Op_EE}%
  \BibitemOpen
  \bibfield  {author} {\bibinfo {author} {\bibfnamefont {T.~c.~v.}\ \bibnamefont {Prosen}}\ and\ \bibinfo {author} {\bibfnamefont {I.}~\bibnamefont {Pi\ifmmode~\check{z}\else \v{z}\fi{}orn}},\ }\href {https://doi.org/10.1103/PhysRevA.76.032316} {\bibfield  {journal} {\bibinfo  {journal} {Phys. Rev. A}\ }\textbf {\bibinfo {volume} {76}},\ \bibinfo {pages} {032316} (\bibinfo {year} {2007})}\BibitemShut {NoStop}%
\bibitem [{\citenamefont {Bravyi}\ and\ \citenamefont {Kitaev}(2002)}]{Bravyi_2002}%
  \BibitemOpen
  \bibfield  {author} {\bibinfo {author} {\bibfnamefont {S.~B.}\ \bibnamefont {Bravyi}}\ and\ \bibinfo {author} {\bibfnamefont {A.~Y.}\ \bibnamefont {Kitaev}},\ }\href {https://doi.org/10.1006/aphy.2002.6254} {\bibfield  {journal} {\bibinfo  {journal} {Annals of Physics}\ }\textbf {\bibinfo {volume} {298}},\ \bibinfo {pages} {210–226} (\bibinfo {year} {2002})}\BibitemShut {NoStop}%
\bibitem [{\citenamefont {Terhal}\ and\ \citenamefont {DiVincenzo}(2002)}]{Terhal_2002}%
  \BibitemOpen
  \bibfield  {author} {\bibinfo {author} {\bibfnamefont {B.~M.}\ \bibnamefont {Terhal}}\ and\ \bibinfo {author} {\bibfnamefont {D.~P.}\ \bibnamefont {DiVincenzo}},\ }\bibfield  {journal} {\bibinfo  {journal} {Physical Review A}\ }\textbf {\bibinfo {volume} {65}},\ \href {https://doi.org/10.1103/physreva.65.032325} {10.1103/physreva.65.032325} (\bibinfo {year} {2002})\BibitemShut {NoStop}%
\bibitem [{\citenamefont {DiVincenzo}\ and\ \citenamefont {Terhal}(2005)}]{DiVincenzo_2005}%
  \BibitemOpen
  \bibfield  {author} {\bibinfo {author} {\bibfnamefont {D.~P.}\ \bibnamefont {DiVincenzo}}\ and\ \bibinfo {author} {\bibfnamefont {B.~M.}\ \bibnamefont {Terhal}},\ }\href {https://doi.org/10.1007/s10701-005-8657-0} {\bibfield  {journal} {\bibinfo  {journal} {Foundations of Physics}\ }\textbf {\bibinfo {volume} {35}},\ \bibinfo {pages} {1967–1984} (\bibinfo {year} {2005})}\BibitemShut {NoStop}%
\bibitem [{\citenamefont {Chalopin}\ \emph {et~al.}(2025{\natexlab{b}})\citenamefont {Chalopin}, \citenamefont {Bojovi\ifmmode~\acute{c}\else \'{c}\fi{}}, \citenamefont {Bourgund}, \citenamefont {Wang}, \citenamefont {Franz}, \citenamefont {Bloch},\ and\ \citenamefont {Hilker}}]{superlattice_experiment}%
  \BibitemOpen
  \bibfield  {author} {\bibinfo {author} {\bibfnamefont {T.}~\bibnamefont {Chalopin}}, \bibinfo {author} {\bibfnamefont {P.}~\bibnamefont {Bojovi\ifmmode~\acute{c}\else \'{c}\fi{}}}, \bibinfo {author} {\bibfnamefont {D.}~\bibnamefont {Bourgund}}, \bibinfo {author} {\bibfnamefont {S.}~\bibnamefont {Wang}}, \bibinfo {author} {\bibfnamefont {T.}~\bibnamefont {Franz}}, \bibinfo {author} {\bibfnamefont {I.}~\bibnamefont {Bloch}},\ and\ \bibinfo {author} {\bibfnamefont {T.}~\bibnamefont {Hilker}},\ }\href {https://doi.org/10.1103/PhysRevLett.134.053402} {\bibfield  {journal} {\bibinfo  {journal} {Phys. Rev. Lett.}\ }\textbf {\bibinfo {volume} {134}},\ \bibinfo {pages} {053402} (\bibinfo {year} {2025}{\natexlab{b}})}\BibitemShut {NoStop}%
\bibitem [{\citenamefont {Turner}\ \emph {et~al.}(2018)\citenamefont {Turner}, \citenamefont {Michailidis}, \citenamefont {Abanin}, \citenamefont {Serbyn},\ and\ \citenamefont {Papić}}]{Turner_2018_PXP}%
  \BibitemOpen
  \bibfield  {author} {\bibinfo {author} {\bibfnamefont {C.~J.}\ \bibnamefont {Turner}}, \bibinfo {author} {\bibfnamefont {A.~A.}\ \bibnamefont {Michailidis}}, \bibinfo {author} {\bibfnamefont {D.~A.}\ \bibnamefont {Abanin}}, \bibinfo {author} {\bibfnamefont {M.}~\bibnamefont {Serbyn}},\ and\ \bibinfo {author} {\bibfnamefont {Z.}~\bibnamefont {Papić}},\ }\bibfield  {journal} {\bibinfo  {journal} {Physical Review B}\ }\textbf {\bibinfo {volume} {98}},\ \href {https://doi.org/10.1103/physrevb.98.155134} {10.1103/physrevb.98.155134} (\bibinfo {year} {2018})\BibitemShut {NoStop}%
\bibitem [{\citenamefont {Khaneja}\ \emph {et~al.}(2005{\natexlab{b}})\citenamefont {Khaneja}, \citenamefont {Reiss}, \citenamefont {Kehlet}, \citenamefont {Schulte-Herbr{\"u}ggen},\ and\ \citenamefont {Glaser}}]{khaneja2005optimal}%
  \BibitemOpen
  \bibfield  {author} {\bibinfo {author} {\bibfnamefont {N.}~\bibnamefont {Khaneja}}, \bibinfo {author} {\bibfnamefont {T.}~\bibnamefont {Reiss}}, \bibinfo {author} {\bibfnamefont {C.}~\bibnamefont {Kehlet}}, \bibinfo {author} {\bibfnamefont {T.}~\bibnamefont {Schulte-Herbr{\"u}ggen}},\ and\ \bibinfo {author} {\bibfnamefont {S.~J.}\ \bibnamefont {Glaser}},\ }\href {https://doi.org/10.1016/j.jmr.2004.11.004} {\bibfield  {journal} {\bibinfo  {journal} {Journal of magnetic resonance}\ }\textbf {\bibinfo {volume} {172}},\ \bibinfo {pages} {296} (\bibinfo {year} {2005}{\natexlab{b}})}\BibitemShut {NoStop}%
\bibitem [{\citenamefont {Rabitz}\ \emph {et~al.}(2004)\citenamefont {Rabitz}, \citenamefont {Hsieh},\ and\ \citenamefont {Rosenthal}}]{rabitz2004quantum}%
  \BibitemOpen
  \bibfield  {author} {\bibinfo {author} {\bibfnamefont {H.~A.}\ \bibnamefont {Rabitz}}, \bibinfo {author} {\bibfnamefont {M.~M.}\ \bibnamefont {Hsieh}},\ and\ \bibinfo {author} {\bibfnamefont {C.~M.}\ \bibnamefont {Rosenthal}},\ }\href {https://doi.org/10.1088/1367-2630/12/7/075008} {\bibfield  {journal} {\bibinfo  {journal} {{Science}}\ }\textbf {\bibinfo {volume} {303}},\ \bibinfo {pages} {1998} (\bibinfo {year} {2004})}\BibitemShut {NoStop}%
\bibitem [{\citenamefont {Day}\ \emph {et~al.}(2019)\citenamefont {Day}, \citenamefont {Bukov}, \citenamefont {Weinberg}, \citenamefont {Mehta},\ and\ \citenamefont {Sels}}]{day2019glassy}%
  \BibitemOpen
  \bibfield  {author} {\bibinfo {author} {\bibfnamefont {A.~G.~R.}\ \bibnamefont {Day}}, \bibinfo {author} {\bibfnamefont {M.}~\bibnamefont {Bukov}}, \bibinfo {author} {\bibfnamefont {P.}~\bibnamefont {Weinberg}}, \bibinfo {author} {\bibfnamefont {P.}~\bibnamefont {Mehta}},\ and\ \bibinfo {author} {\bibfnamefont {D.}~\bibnamefont {Sels}},\ }\href {https://doi.org/10.1103/PhysRevLett.122.020601} {\bibfield  {journal} {\bibinfo  {journal} {{Phys. Rev. Lett.}}\ }\textbf {\bibinfo {volume} {122}},\ \bibinfo {pages} {020601} (\bibinfo {year} {2019})}\BibitemShut {NoStop}%
\bibitem [{\citenamefont {W{\"a}chter}\ and\ \citenamefont {Biegler}(2006)}]{wachter2006implementation}%
  \BibitemOpen
  \bibfield  {author} {\bibinfo {author} {\bibfnamefont {A.}~\bibnamefont {W{\"a}chter}}\ and\ \bibinfo {author} {\bibfnamefont {L.~T.}\ \bibnamefont {Biegler}},\ }\href {https://doi.org/10.1007/s10107-004-0559-y} {\bibfield  {journal} {\bibinfo  {journal} {Mathematical programming}\ }\textbf {\bibinfo {volume} {106}},\ \bibinfo {pages} {25} (\bibinfo {year} {2006})}\BibitemShut {NoStop}%
\bibitem [{\citenamefont {Bhattacharyya}\ \emph {et~al.}(2024)\citenamefont {Bhattacharyya}, \citenamefont {An}, \citenamefont {Kozbiel}, \citenamefont {Goldschmidt},\ and\ \citenamefont {Chong}}]{bhattacharyya2024using}%
  \BibitemOpen
  \bibfield  {author} {\bibinfo {author} {\bibfnamefont {B.}~\bibnamefont {Bhattacharyya}}, \bibinfo {author} {\bibfnamefont {F.}~\bibnamefont {An}}, \bibinfo {author} {\bibfnamefont {D.}~\bibnamefont {Kozbiel}}, \bibinfo {author} {\bibfnamefont {A.~J.}\ \bibnamefont {Goldschmidt}},\ and\ \bibinfo {author} {\bibfnamefont {F.~T.}\ \bibnamefont {Chong}},\ }in\ \href {https://doi.org/10.1109/QCE60285.2024.00159} {\emph {\bibinfo {booktitle} {2024 IEEE International Conference on Quantum Computing and Engineering (QCE)}}},\ Vol.~\bibinfo {volume} {01}\ (\bibinfo {year} {2024})\ pp.\ \bibinfo {pages} {1336--1346}\BibitemShut {NoStop}%
\bibitem [{\citenamefont {Goldschmidt}\ \emph {et~al.}()\citenamefont {Goldschmidt}, \citenamefont {Pel\'{a}ez~Cisneros}, \citenamefont {Sitler}, \citenamefont {Olsson}, \citenamefont {Smith},\ and\ \citenamefont {Quiroz}}]{goldschmidt2025quantum}%
  \BibitemOpen
  \bibfield  {author} {\bibinfo {author} {\bibfnamefont {A.~J.}\ \bibnamefont {Goldschmidt}}, \bibinfo {author} {\bibfnamefont {E.}~\bibnamefont {Pel\'{a}ez~Cisneros}}, \bibinfo {author} {\bibfnamefont {R.}~\bibnamefont {Sitler}}, \bibinfo {author} {\bibfnamefont {K.}~\bibnamefont {Olsson}}, \bibinfo {author} {\bibfnamefont {K.~N.}\ \bibnamefont {Smith}},\ and\ \bibinfo {author} {\bibfnamefont {G.}~\bibnamefont {Quiroz}},\ }\bibinfo {note} {in preparation}\BibitemShut {NoStop}%
\bibitem [{\citenamefont {{Harmoniqs}}(2025)}]{piccolo2025}%
  \BibitemOpen
  \bibfield  {author} {\bibinfo {author} {\bibnamefont {{Harmoniqs}}},\ }\href@noop {} {\bibinfo {title} {{Piccolo.jl}: Fine tuned quantum optimal control}},\ \bibinfo {howpublished} {\url{https://github.com/harmoniqs/Piccolo.jl}} (\bibinfo {year} {2025}),\ \bibinfo {note} {see \url{https://www.harmoniqs.co} for more information}\BibitemShut {NoStop}%
\bibitem [{\citenamefont {Sauvage}\ and\ \citenamefont {Mintert}(2022)}]{sauvage2022optimal}%
  \BibitemOpen
  \bibfield  {author} {\bibinfo {author} {\bibfnamefont {F.}~\bibnamefont {Sauvage}}\ and\ \bibinfo {author} {\bibfnamefont {F.}~\bibnamefont {Mintert}},\ }\href {https://doi.org/https://doi.org/10.1103/PhysRevLett.129.050507} {\bibfield  {journal} {\bibinfo  {journal} {Physical Review Letters}\ }\textbf {\bibinfo {volume} {129}},\ \bibinfo {pages} {050507} (\bibinfo {year} {2022})}\BibitemShut {NoStop}%
\bibitem [{\citenamefont {Chadwick}\ and\ \citenamefont {Chong}(2023)}]{chadwick2023efficient}%
  \BibitemOpen
  \bibfield  {author} {\bibinfo {author} {\bibfnamefont {J.~D.}\ \bibnamefont {Chadwick}}\ and\ \bibinfo {author} {\bibfnamefont {F.~T.}\ \bibnamefont {Chong}},\ }in\ \href {https://doi.org/https://doi.org/10.1109/QCE57702.2023.00145} {\emph {\bibinfo {booktitle} {2023 IEEE International Conference on Quantum Computing and Engineering (QCE)}}},\ Vol.~\bibinfo {volume} {1}\ (\bibinfo {organization} {IEEE},\ \bibinfo {year} {2023})\ pp.\ \bibinfo {pages} {1286--1294}\BibitemShut {NoStop}%
\bibitem [{\citenamefont {Gautier}\ \emph {et~al.}(2025)\citenamefont {Gautier}, \citenamefont {Genois},\ and\ \citenamefont {Blais}}]{gautier2025optimal}%
  \BibitemOpen
  \bibfield  {author} {\bibinfo {author} {\bibfnamefont {R.}~\bibnamefont {Gautier}}, \bibinfo {author} {\bibfnamefont {E.}~\bibnamefont {Genois}},\ and\ \bibinfo {author} {\bibfnamefont {A.}~\bibnamefont {Blais}},\ }\href {https://doi.org/10.1103/PhysRevLett.134.070802} {\bibfield  {journal} {\bibinfo  {journal} {Phys. Rev. Lett.}\ }\textbf {\bibinfo {volume} {134}},\ \bibinfo {pages} {070802} (\bibinfo {year} {2025})}\BibitemShut {NoStop}%
\bibitem [{\citenamefont {Pas}\ \emph {et~al.}(2022)\citenamefont {Pas}, \citenamefont {Schuurmans},\ and\ \citenamefont {Patrinos}}]{alpaqa}%
  \BibitemOpen
  \bibfield  {author} {\bibinfo {author} {\bibfnamefont {P.}~\bibnamefont {Pas}}, \bibinfo {author} {\bibfnamefont {M.}~\bibnamefont {Schuurmans}},\ and\ \bibinfo {author} {\bibfnamefont {P.}~\bibnamefont {Patrinos}},\ }in\ \href {https://doi.org/10.23919/ECC55457.2022.9838172} {\emph {\bibinfo {booktitle} {2022 European Control Conference (ECC)}}}\ (\bibinfo {year} {2022})\ pp.\ \bibinfo {pages} {417--422}\BibitemShut {NoStop}%
\bibitem [{\citenamefont {Betts}(2010)}]{betts2010practical}%
  \BibitemOpen
  \bibfield  {author} {\bibinfo {author} {\bibfnamefont {J.~T.}\ \bibnamefont {Betts}},\ }\href {https://doi.org/10.1137/1.9780898718577} {\emph {\bibinfo {title} {Practical Methods for Optimal Control and Estimation Using Nonlinear Programming, Second Edition}}},\ \bibinfo {edition} {2nd}\ ed.\ (\bibinfo  {publisher} {Society for Industrial and Applied Mathematics},\ \bibinfo {year} {2010})\BibitemShut {NoStop}%
\bibitem [{\citenamefont {Kamen}\ \emph {et~al.}(2026)\citenamefont {Kamen}, \citenamefont {Fine}, \citenamefont {Bhattacharyya}, \citenamefont {Chong},\ and\ \citenamefont {Goldschmidt}}]{kamen2026accurate}%
  \BibitemOpen
  \bibfield  {author} {\bibinfo {author} {\bibfnamefont {A.}~\bibnamefont {Kamen}}, \bibinfo {author} {\bibfnamefont {S.}~\bibnamefont {Fine}}, \bibinfo {author} {\bibfnamefont {B.}~\bibnamefont {Bhattacharyya}}, \bibinfo {author} {\bibfnamefont {F.~T.}\ \bibnamefont {Chong}},\ and\ \bibinfo {author} {\bibfnamefont {A.~J.}\ \bibnamefont {Goldschmidt}},\ }\bibfield  {journal} {\bibinfo  {journal} {arXiv preprint arXiv:2508.19075}\ }\href {https://doi.org/https://doi.org/10.48550/arXiv.2602.10349} {https://doi.org/10.48550/arXiv.2602.10349} (\bibinfo {year} {2026})\BibitemShut {NoStop}%
\end{thebibliography}%

\clearpage
\appendix
\begin{center}
	\noindent\textbf{Supplementary Material}
	\bigskip
		
	\noindent\textbf{\large{}}
\end{center}
\section{Notation and Preliminaries \label{appendix:Notations_Preliminaries}}

In this work, we use $X,Y,Z$ to denote the Pauli matrices. The $N$-fold tensor products of single-qubit Pauli matrices along with the identity matrix are represented by $\mathbb{P}_N$:
\begin{equation}
    \mathbb{P}_N \coloneqq \left\{ \bigotimes_{i=1}^N P_i: P_i = I,X,Y,Z
 \right\}\;.
\end{equation}
The subscript $i$ indicates that the operator is applied to the $i$-th qubit. 
We typically omit the identity matrix when writing a multi-qubit Pauli operator, such as $I_1Z_2Z_3I_4 \to Z_2Z_3$.

We use $c_i^\dagger$ and $c_i$ ($b_i^\dagger$ and $b_i$) to represent the fermionic (bosonic) creation and annihilation operators at site $i$, respectively. They satisfy the following canonical commutation relations:
\begin{equation} \label{eqn:canonical_comm_relations}
    \{c_i,c_i^\dagger\} = \delta_{ij},\quad \{c_i^\dagger,c_j^\dagger\} = \{c_i,c_j\} = 0\;,\quad 
    [b_i,b_i^\dagger] = \delta_{ij},\quad [b_i^\dagger,b_j^\dagger] = [b_i,b_j] = 0\;,
\end{equation}
where $\{A,B\} = AB+BA$ and $[A,B] = AB-BA$ are the anticommutator and commutator, respectively. 
The following commutation relations are frequently used in our proof of universality:
\eqs{ \label{eqn:useful_comm_relations}
&[AB,CD]=A[B,C]D+AC[B,D]+[A,C]DB+C[A,D]B\;,\\
&[AB,CD]=A\{B,C\}D-AC\{B,D\}+\{A,C\}DB-C\{A,D\}B\;.
}

\begin{table}[h]
\centering
\caption{Summary of mathematical notations.}
\label{tab:notations}
\begin{tabular}{cc}
\toprule
\textbf{Symbol} & \textbf{Meaning} \\
\midrule
$X_j, Z_j$      & Pauli $X$ and $Z$ operators acting on qubit $j$ \\
$P_i$           & A single-qubit Pauli operator ($X$, $Y$, or $Z$) \\
$\mathbb{P}_N$ & $N$-fold tensor products of Pauli matrices\\
$H_X, H_Z$      & Global $X$ and $Z$ fields, $\sum_j X_j$, $\sum_j Z_j$ \\
$R$             & Matrix representation of the lattice reflection operator \\
$\mathcal{R}$  & Lattice reflection operation \\
$\mathcal{G}$          & The generating set of the dynamical Lie algebra \\
$\mathfrak{g}$          & Dynamical Lie algebra \\
$\oplus_{v}$ & Direct sum of the vector space (Hilbert space)\\
$\oplus_{m}$ & Direct sum of the matrix space\\
$\mathcal{H}$ & Hilbert space: a vector space with inner product.\\
$\text{End}(\mathcal{H})$ & Endomorphism on the space $\mathcal{H}$: operators that map $\mathcal{H}$ to itself.\\
$\mathrm{Ad}_H$ & Conjugated by $H$, i.e., $\Ad_H(O) = HOH^{-1}$\\
$\mathrm{ad}_H$ & commutated by $H$, i.e., $\ad_H(O) = [H,O]$
\\
\bottomrule
\end{tabular}
\end{table}

We summarize the essential mathematical notations throughout the paper in \Cref{tab:notations}. In the quantum control literature, when studying the expressivity of a given set of control pulses, people explore its attainable effective evolution based on the Baker-Campbell-Hausdorff (BCH), which is defined as follows~\cite{hall2003liegroup,rossmann2006liegroup}:
\begin{definition}[The Baker-Campbell-Hausdoff formula] \label{def:BCH_formula}
    Given two elements $A,B$ in a Lie algebra, the Baker-Campbell-Hausdorff (BCH) formula gives the element $C$ which solves the following equation:
    \begin{equation}
        e^A e^B = e^C\;,\notag
    \end{equation}
    in case this equation has a solution.
    To the first four orders of the commutator, the explicit formula of $Z$ is:
    \begin{equation} \label{eqn:BCH_Z_XY}
        C(A,B) = A+B + \frac{1}{2}[A,B] + \frac{1}{12}\left([A,[A,B]]+ [B,[B,A]]\right) -\frac{1}{24}[B,[A,[A,B]]]+\cdots\;,
    \end{equation}
    where $\cdots$ represents the higher-order terms.
\end{definition}

When specific to quantum unitary evolution, $A,B$ and $C$ take the form of anti-Hermitian operators $iH$ for some Hermitian $H$.
This ensures both sides of \Cref{eqn:BCH_Z_XY} are consistent.
Given a quantum control system with control Hamiltonians $H_1,H_2,\cdots,H_l$, we can define a generating set $\mathcal{G} = \{iH_1,iH_2,\dots,iH_l\}$. 
Then, the physically obtainable evolution is determined by the \textit{Dynamical Lie Algebra} (DLA) generated by $\mathcal{G}$, which is defined by:
\begin{definition}[Dynamical Lie Algebra] \label{def:DLA}
   Given a control system with generators $\mathcal{G} = \{iH_1,iH_2,\dots,iH_l\}$, the
Dynamical Lie Algebra (DLA) $\mathfrak{g}$ is the subalgebra of
$\mathfrak{su}(d)$ spanned by the repeated nested commutators of
the elements in $\mathcal{G}$, i.e.
$$\mathfrak{g}=\mathrm{span}_{\mathbb{R}}\langle iH_1,iH_2,\dots,iH_l\rangle_{\mathrm{Lie}}\subseteq \mathfrak{su}(d)\;,$$
where $\mathrm{span}_{\mathbb{R}}\langle iH_1,iH_2,\dots,iH_l\rangle_{\mathrm{Lie}}$ denotes the Lie closure under nested commutators, and $d$ is the dimension of the Hilbert space.
\end{definition}
In \Cref{def:DLA}, the $k$-th order nested commutator of $\mathcal{G}$ is defined by $[G_1,[G_2,\dots,[G_{k-1},G_k]]]$ for $G_1,G_2,\dots,G_k \in \mathcal{G}$. The DLA $\mathfrak{g}$ contains all linear combinations (with real coefficients for anti-Hermicity) of those nested commutators with arbitrary orders.
For a specific task like quantum computation or quantum simulation, there is a largest subalgebra of $\mathfrak{su}(d)$ one can attain.
If $\mathfrak{g}$ equals this subalgebra, it is termed \textit{universal}, as detailed in \Cref{appendix:Universality,appendix:Universal_fermion}.  

To understand why $\mathfrak{g}$ is the attainable evolution given the control $\mathcal{G}$, we first notice that there are two basic operations generating the algebra $\mathfrak{g}$: (1) linear combination and (2) Lie bracket of arbitrary two elements in $\mathfrak{g}$.
Starting with an initial algebra $\mathfrak{g}_0\equiv \mathcal{G}$, we can construct an enlarged algebra $\mathfrak{g}_1$ by adding linearly independent operators obtained from linear combining and commuting any two elements in $\mathfrak{g}_0$.
Then we can repeat the process for $\mathfrak{g}_1$ to obtain $\mathfrak{g}_2$, until no non-trivial new elements can be found. This gives the final closed algebra $\mathfrak{g}$.
From the above analysis, to conclude that $\mathfrak{g}$ is attainable by the control in $\mathcal{G}$, it is sufficient to prove that the evolutions under (1) and (2) are realizable by the physical system.
This is done in the proofs of \Cref{lemma:Linear_Combo_Generators,lemma:Commutator_Generators} below.
The alternative proofs using the Trotter-Suzuki formula~\cite{Suzuki1976,Suzuki1977} can be found in standard textbooks on Lie algebras, such as Theorem 5.16 in \cite{sagle1986introduction}.

\begin{definition}[Repertoire of unitary dynamics \cite{deutsch}] \label{def:Repertoire}
Given a set of unitaries $\{U_i\}_i$, the repertoire of the unitary dynamics generated by this set is the collection of all unitaries that can be approximated to arbitrary accuracy using sequences of $\{U_i\}_i$.
\end{definition}

\begin{lemma}
\label{lemma:Linear_Combo_Generators}
    Consider a repertoire of unitary dynamics on a finite-dimensional Hilbert space. If the generated dynamics of a pair of Hermitian operators (Hamiltonians) $H_1$ and $H_2$ are in the repertoire, then every operation or unitary dynamics generated by $\alpha H_1+\beta H_2$ is also in the repertoire.
\end{lemma}
\begin{proof}
    For any bounded $A$ and $B$, the Lie-product (Trotter) formula converges in norm. Let $\Sigma^A_2 = \sum_{n=2}^{\infty} A^n/n!   $. By the fundamental theorem of calculus,
\[
    \Sigma^A_2
    = A^2 \int_0^1 ds \int_0^s dt\, e^{tA}\;,
    \quad \text{so}\quad
    \norm{\Sigma^A_2}
    \leq \frac{\norm{A}^2}{2} e^{\norm{A}}\;.
    \quad\text{Thus}\quad
    \norm{\Sigma^{A/n}_2}
    \leq \frac{\norm{A}}{2n^2}
    e^{\norm{A}/n}\;.
\]
Then  
    \[
        \norm{e^{(A+B)/n} - e^{A/n}e^{B/n}} 
        \leq \norm{\Sigma^{(A+B)/n}_2   - \Sigma^{A/n}_2(I+\frac{B}{n}) - (I+ \frac{A}{n})      \Sigma^{B/n}_2  -\Sigma^{A/n}_2 \Sigma^{B/n}_2
        -\frac{AB}{n^2}}
        \leq \frac{5}{n^2} e^{(\norm{A}+\norm{B})/n}\;.
    \]
Thus we can expand the approximation to the formula into $n$ terms, and have the bound
\begin{equation}\label{eqn:inter_Trotter_formula}
        \norm{e^{A+B}- \left(e^{A/n}     e^{B/n}\right)^{n}}=  \norm{\sum_{j=0}^{{n-1}}   e^{j(A+B)/n}  \left(e^{(A+B)/n}   -e^{A/n}e^{B/n}\right)  \left(e^{A/n}e^{B/n} \right)^{n-j-1}}  
        \leq \frac{5}{n}\, e^{\norm{A}+\norm{B}}\;.
    \end{equation}
    By taking $A = i\alpha H_1$ and $B = i\beta H_2$ for some real numbers $\alpha$ and $\beta$, in the asymptotic limit $n\to \infty$:
    \eqs{e^{i(\alpha H_1+\beta H_2) \notag}
    =\lim_{n\rightarrow\infty}\left(e^{i\alpha H_1/n}e^{i\beta H_2/n}\right)^n\;,
    }
    so the evolution under $\alpha H_1 + \beta H_2$ is simulable given the control of $H_1$ and $H_2$.
\end{proof}

\begin{lemma}  \label{lemma:Commutator_Generators}
    Consider a repertoire of unitary dynamics on a finite-dimensional Hilbert space.
    If the generated unitary dynamics of a pair of Hermitian operators (Hamiltonians) $H_1$ and $H_2$ are in the repertoire, then every operation or unitary dynamics generated by the commutator $i[H_1,H_2]$ is also in the repertoire. We prove a modified Trotter formula for bounded, self-adjoint generators $H_j$. 
        \eqs{
    e^{-[H_1,H_2]}=\lim_{n\rightarrow\infty}\left(e^{iH_1/\sqrt{n}}e^{iH_2/\sqrt{n}}e^{-iH_1/\sqrt{n}}e^{-iH_2/\sqrt{n}}\right)^n\;.
    }
    This may be elsewhere in the literature, we have also found it in~\cite{sagle1986introduction}, remark 1g following Theorem 5.16.
\end{lemma}
\begin{proof}
    We consider bounded and skew-Hermitian operators $A$ and $B$. We have the following integral formula:
    \begin{align} \label{eqn:commutator_Trotter_1}
    {e^{A/\sqrt{n}} }{  e^{B/\sqrt{n}}}{e^{-A/\sqrt{n}} }{  e^{-B/\sqrt{n}}} -I &= 
    {e^{sA/\sqrt{n}} }{  e^{B/\sqrt{n}}}{e^{-sA/\sqrt{n}} }{  e^{-B/\sqrt{n}}} |^{s=1}_{s=0}\notag\\
    &= \frac{1}{\sqrt{n}} \int^1_0 ds\, e^{sA/\sqrt{n}}\left( A e^{B/\sqrt{n}} e^{-sA/\sqrt{n}} - e^{B/\sqrt{n}} A e^{-sA/\sqrt{n}}\right) e^{-B/\sqrt{n}} \notag\\
    &= \frac{1}{\sqrt{n}} \int^1_0 ds\, e^{sA/\sqrt{n}} \left[A,e^{B/\sqrt{n}}\right] e^{-sA/\sqrt{n}}e^{-B/\sqrt{n}}\;,
    \end{align}
    where one can verify the second equality by taking a derivative of $s$ on the left-hand side of the equation.
    Using the following identity:
    \begin{equation} \label{eqn:integral_indetity}
    e^{sA/\sqrt{n}} = I + \left( e^{sA/\sqrt{n}} - I\right) = I + \frac{A}{\sqrt{n}} \int^s_0 dt\, e^{tA/\sqrt{n}}\;,
    \end{equation}
    the \Cref{eqn:commutator_Trotter_1} becomes:
    \begin{align}
        &\frac{1}{n}\int^1_0 ds\, \int^1_0 dt\, e^{sA/\sqrt{n}} \left[A,Be^{tB/\sqrt{n}}\right] e^{-sA/\sqrt{n}} e^{-B/\sqrt{n}} \notag\\
        &=\frac{1}{n}\int^1_0 ds\, \int^1_0 dt\, e^{sA/\sqrt{n}} \left([A,B] e^{tB/\sqrt{n}} + B\left[A,e^{tB/\sqrt{n}}\right]\right) e^{-sA/\sqrt{n}} e^{-B/\sqrt{n}} \notag\\
        &= \frac{[A,B]}{n} + R\;,
    \end{align}
    where this defines $R$ as  the remainder.
    We now use the assumption that $A$ and $B$ are skew-Hermitian, so for real $\alpha,\beta,\gamma$ both $e^{\alpha A}$ and $e^{\beta B}$ are unitary, as is $e^{\gamma[A,B]}$.
    Note that in the last equality, we have used the integral identity in \Cref{eqn:integral_indetity} to expand each exponential. 
    Using
    \[
        \norm{\left( e^{sA/\sqrt{n}} - I\right)} \leq \frac{\norm{A}}{\sqrt{n}}\;,
    \]
    and noticing that the remainder $R$ contains at least one and at most four such terms of order $n^{-1/2}$, we have:
    \[
        \norm{R} \leq 2^6\left(\norm{A}+\norm{B}+1\right)^4 \frac{1}{n^{3/2}} = \frac{M}{n^{3/2}}\;,
    \]
    where $M$ is a constant independent of $n$.
    Thus, the constant term and the linear term in the expansion of the exponential $e^{[A,B]/n}$ cancel with the corresponding terms in $e^{A/\sqrt{n}}e^{B/\sqrt{n}}e^{-A/\sqrt{n}}e^{-B/\sqrt{n}}$, which gives:
    \[
        \norm{e^{[A,B]/n} - e^{A/\sqrt{n}} e^{B/\sqrt{n}} e^{-A/\sqrt{n}} e^{-B/\sqrt{n}}} = \norm{R} \leq \frac{M}{n^{3/2}}\;.
    \]
    By a similar expansion  as in \Cref{eqn:inter_Trotter_formula}, we obtain:
    \begin{align}
        &\norm{e^{[A,B]} - \left( e^{A/\sqrt{n}} e^{B/\sqrt{n}} e^{-A/\sqrt{n}}e^{-B/\sqrt{n}}\right)^n} \notag\\
        &= \norm{\sum_{j=0}^{n-1} e^{j[A,B]/n} \left(e^{[A,B]/n} - e^{A/\sqrt{n}} e^{B/\sqrt{n}} e^{-A/\sqrt{n}}e^{-B/\sqrt{n}}\right) \left(e^{A/\sqrt{n}} e^{B/\sqrt{n}} e^{-A/\sqrt{n}}e^{-B/\sqrt{n}}\right)^{n-j-1}}\leq \frac{M}{\sqrt{n}}\;. \notag
    \end{align}
    Thus, by taking $A = iH_1$, $B = iH_2$, in the asymptotic limit $n\to \infty$:
    \eqs{
    e^{-[H_1,H_2]}=\lim_{n\rightarrow\infty}\left(e^{iH_1/\sqrt{n}}e^{iH_2/\sqrt{n}}e^{-iH_1/\sqrt{n}}e^{-iH_2/\sqrt{n}}\right)^n\;,\notag
    }
    which completes the proof.
\end{proof}

By \Cref{lemma:Linear_Combo_Generators,lemma:Commutator_Generators}, we can understand the expressivity of a generating set $\mathcal{G}$ by investigating its DLA $\mathfrak{g}$. 

\section{Universality of analog quantum computation} \label{appendix:Universality}
Here we give  details of the proof of \Cref{qubit:UQC_chain}, which establishes necessary and sufficient conditions for one-dimensional universal quantum computation under global control. 
In a qubit system, it is notable that  realizing universal quantum computation requires that the DLA $\mathfrak{g}$ equals the entire algebra $\mathfrak{su}(2^N)$.  Here, we rewrite this formally.
\begin{reptheorem}{qubit:UQC_chain}
[Minimal requirement for universal quantum computation on a qubit chain] 
Consider a chain of $N\geq 2$ qubits with homogeneous nearest neighbor Ising interactions $H_{ZZ} = \sum_{j}Z_{j} Z_{j+1}$. 
Suppose the system is equipped with global $X$ and $Z$ control fields, given by $H_{X} = \sum_j X_j$, $H_{Z} = \sum_j Z_j$.
Thus, the control Hamiltonian is described by \Cref{eq:uniform_global} as:
\[H_q(t)
= u_X(t) H_X + u_Z(t) H_Z + u_{ZZ}(t) H_{ZZ}
= u_X(t)\sum_{i} X_i + u_Z(t)\sum_{i} Z_i + u_{ZZ}(t)\sum_{\langle i,j\rangle} Z_i Z_j\;,\]
with tunable time-dependent control pulses $u_X(t),u_Z(t),u_{ZZ}(t)$.
Let $R$ denote the reflection on $N$ qubits sending site $i$ to $N+1-i$.
The control system realizes universal quantum computation, if and only if there exists an additional Hamiltonian $H_{\mathrm{break}}$ that breaks the reflection symmetry, i.e., $RH_{\mathrm{break}}R^{-1}\neq H_{\mathrm{break}}$. 
In other words, the dynamical Lie algebra satisfies
\[
    \mathrm{span}_{\mathbb{R}}\langle iH_{X}, iH_{Z}, iH_{ZZ}, iH_{\mathrm{break}} \rangle_{\mathrm{Lie}} = \mathfrak{su}(2^N)\;,
\]
if and only if $R H_{\mathrm{break}}R^{-1}\neq H_{\mathrm{break}}$.
\end{reptheorem}

\begin{remark}
    Compared to the \Cref{qubit:UQC_chain} in the main text, here, we focus on the Ising-type nearest-neighbor interaction $H_{ZZ}$ because the specific single-Pauli interaction is irrelevant.
    This is apparent later in the proof of \Cref{lemma:equiv_DLA_2}.
    More specifically, given any one of the homogeneous single-Pauli interactions (e.g. $H_{ZZ}$), we can generate the other two ($H_{XX}$ and $H_{YY}$) by commuting $H_X,H_Z$ with it.
    Thus, all such interactions are included in the generating set, and it suffices to choose any one of them.
    Moreover, for arbitrary nearest-neighbor interactions,
    as long as the coefficients do not fall into specific regimes that introduce some additional symmetries, the system is still universal,  as detailed in \Cref{corollary:generalization_interaction,corollary:generalization_interaction_2}.

    Here, the symmetry-breaking field $H_{\mathrm{break}}$ is general, not restricted to any specific form. 
    Proving this general condition utilizes techniques in representation theory, as detailed in the proof.
\end{remark}

Before delving into the details of the proof, we provide some examples to enhance understanding of the theorem and explain the definition of lattice reflection symmetry breaking.

\begin{remark} \label{remark:one_diml_examples}
    We can represent the one-dimensional chain as a graph $G(V,E)$ consisting of vertices $V$ and edges $E$. 
    In the graph, we place the $j$-th qubit on the vertex $j \in V$, and we use the edge $( j,j+1 ) \in E$ to represent the nearest-neighbor qubit pair having Ising-type interaction $Z_j Z_{j+1}$.
    Initially, all vertices and edges are uncolored. 
    This implies that the fields applied on the vertices, i.e., $H_X$ and $H_Z$, and edges, i.e., $H_{ZZ}$, are uniform.
    Then, we can use different colorings of the vertices and edges to represent the support of certain additional symmetry-breaking control fields.
    As an example, we consider the Hamiltonians $H_{\mathrm{break},\alpha} = H_{X,\alpha} = \sum_{i\in \alpha} X_i$, for $\alpha = A,B,C,\cdots$, i.e., global $X$ controls on various types of lattices labeled $\alpha$.
    The corresponding lattices are colored differently.
    
    Then, $RH_{X,\alpha}R^{-1} \neq H_{X,\alpha}$ requires the combined pattern formed by $\alpha$ to break the reflection symmetry.
    This is equivalent to demanding the graph representation to have a trivial automorphism group $\mathrm{Aut}(G) = \{\mathbb{I}\}$, as the only nontrivial automorphism of the one-dimensional lattice graph is the reflection. 
    As an example, if we consider $N=8$ and $|\alpha|=0$, i.e. no symmetry-breaking field, the graph representation is:
    \begin{equation} 
    \begin{tikzpicture}[baseline={(current bounding box.center)}][scale=1]
      \def\N{8}
      \def\sep{1.0}
      \foreach \i in {1,...,\N}{
          \pgfmathtruncatemacro{\k}{mod(\i-1,2)}
          \ifnum\k=0
              \node[circle,draw,minimum size=5pt,inner sep=0.5pt,
                    label={[yshift=-2pt]below:$\i$}](site\i) at (\sep*\i,0){};
          \else
              \node[circle,draw,minimum size=5pt,inner sep=0.5pt,
                    label={[yshift=-2pt]below:$\i$}](site\i) at (\sep*\i,0){};
          \fi
      }
      \foreach \i in {1,...,7}{
          \pgfmathtruncatemacro{\j}{\i+1}
          \draw[dashed](site\i)--(site\j);
      }
    \draw[dashdotted] (4.5,0.5) -- (4.5,-0.5);
    \end{tikzpicture}\;, \notag
    \end{equation}
    where we use circles and dashed lines to represent vertices and edges, respectively. The graph is invariant under reflection with respect to the axis labeled by the vertical dashed-dotted line. The reflection operator $R$ can be represented as a product of SWAP gates:
    \begin{equation}
        R = \mathrm{SWAP}_{1,8}\mathrm{SWAP}_{2,7} \mathrm{SWAP}_{3,6} \mathrm{SWAP}_{4,5}\;,\quad \mathrm{where}\quad \mathrm{SWAP}_{i,j} = \frac{1}{2}\left(I + X_iX_j + Y_i Y_j + Z_i Z_j\right)\;,
    \end{equation}
    and it is easy to verify that:
    \begin{equation} \label{eqn:reflection_operation}
        RZ_iR^{-1} = Z_{N-i+1},\quad RX_iR^{-1} = X_{N-i+1}, \quad [R,H] = 0,\quad \mathrm{for}\quad H\in\left\{H_{X},H_Z,H_{ZZ}\right\}\;.
    \end{equation}
    Since the system has global reflection symmetry, it cannot realize universal quantum computation.

    If we apply a symmetry-breaking field $H_{X,A} = \sum_{i=5}^8 X_i$ on the right half of the chain, we can color the last four sites black, yielding the following graph representation:
    \begin{equation}
            \begin{tikzpicture}[baseline={(current bounding box.center)}][scale=1]
      \def\N{8}
      \def\sep{1.0}
      \foreach \i in {1,...,\N}{
          \ifnum\i<5
              \node[circle,draw,minimum size=5pt,inner sep=0.5pt,
                    label={[yshift=-2pt]below:$\i$}](site\i) at (\sep*\i,0){};
          \else
              \node[circle,draw,fill=black,minimum size=5pt,inner sep=0.5pt,
                    label={[yshift=-2pt]below:$\i$}](site\i) at (\sep*\i,0){};
          \fi
      }
      \foreach \i in {1,...,7}{
          \pgfmathtruncatemacro{\j}{\i+1}
          \draw[dashed](site\i)--(site\j);
      }
    \end{tikzpicture}\;.\notag
    \end{equation}
    The above graph has no nontrivial automorphism. According to \Cref{qubit:UQC_chain}, the generating set $\mathcal{G} = \{H_{X},H_Z,H_{ZZ},H_{X,A}\}$ is universal.
\end{remark}

\noindent\textbf{Proof sketch of \Cref{qubit:UQC_chain}.}
We first demonstrate that uniform controls enable universality in the reflection-symmetric subalgebra (\Cref{prop:fix-span}) and analyze its matrix representation (\Cref{lem:block-diag}).
Then, any symmetry breaking control $H_{\mathrm{break}}$ introduces non-zero elements in the reflection-anti-symmetric sector of the matrix (\Cref{prop:symmetric-pair,lem:break-split-correct}).
This element, combined with the established universality, enables us to single out all generators of $\mathrm{su}(2^N)$ by a constructive method (\Cref{prop:irreducible-m}), which completes the proof (\Cref{prop:closure}). 

We outline six steps, proved in detail below.

\medskip
\begin{itemize}
    \item \textbf{Step 1 (Symmetric controls).}
We first study the algebra generated by the uniform controls $H_X,H_Z,H_{ZZ}$. 
Let \(\theta\coloneqq\Ad_R\), i.e., $\theta(O) = ROR^{-1}$, and define the \(R\)-fixed Lie subalgebra
\[
\mathfrak l\ \coloneqq\ \{X\in\mathfrak{su}(2^N):[X,R]=0\}\ =\ \Fix(\theta)\cap\mathfrak{su}(2^N)\;,
\]
where $\mathrm{Fix}(\theta)$ is the fixed point of $\theta$: the set of $2^N$-dimensional complex matrices that commute with $R$.

Using the controls \(H_X, H_Z,\, H_{ZZ}\), we can generate the mirrored local terms (\Cref{lemma:equiv_DLA_2})
$$\widetilde X_j=X_j+\theta(X_j)\;,\  \widetilde Z_j=Z_j+\theta(Z_j)\;,\ \mathrm{and}\ \widetilde{ZZ}_j=Z_jZ_{j+1}+\theta(Z_jZ_{j+1})\;,$$
which builds all reflection-symmetric Pauli strings (\Cref{lem:mirror-generation}).
Furthermore, we demonstrate that these strings generate $\mathfrak{l}$ as (\Cref{prop:fix-span}):
\[
\mathrm{span}_{\mathbb R}\!\left\langle\, i\widetilde X_j,\,i\widetilde Z_j,\,i\widetilde{ZZ}_j\,\right\rangle_{\mathrm{Lie}}
=\mathfrak l
\quad\;.
\]
Thus, it suffices to investigate the structure of $\mathfrak{l}$.

\item \textbf{Step 2 (Block structure of \(\mathfrak l\)).}
As $R^2 = I$, we can decompose the Hilbert space into the $\pm1$ \(R\)–eigenspaces as \(\mathcal H=\mathcal H_+\oplus_v\mathcal H_-\), with dimensions \(d_\pm\).
Any operator in $\mathfrak{l}$ commutes with \(R\), so it is block–diagonal in this basis (\Cref{lem:block-diag}).
Further imposing skew–Hermiticity and tracelessness yields the decomposition of the space $\mathfrak{l}$ as:
\[
\mathfrak l\ \cong\ \mathfrak{su}(d_+)\ \oplus_m\ \mathfrak{su}(d_-)\ \oplus_m\ \mathfrak u(1)_{\mathrm{rel}},
\]
with the relative phase $\mathfrak{u}(1)_{\mathrm{rel}}$ central in \(\mathfrak l\) (\Cref{prop:l-structure}).

\item \textbf{Step 3 (Involution, projections, and decomposition of the algebra).}
Then, we analyze the remaining part of $\mathfrak{su}(2^N)$.
For the involution automorphism \(\theta\), i.e. $\theta^2 = \mathrm{id}$, we can decompose the matrix space into the $\pm 1$ $\theta$-eigenspaces using the complementary projectors
\(E_\pm=\tfrac12(\mathrm{id}\pm\theta)\).
They act on $\mathfrak{su}(2^N)$ with:
\[
\mathrm{Im}\,E_+=\mathfrak l,\qquad \mathrm{Im}\,E_-=\mathfrak m,\qquad
\mathfrak{su}(2^N)=\mathfrak l\oplus_m \mathfrak m,
\]
and the brackets respect parity:
\([\mathfrak l,\mathfrak l]\subseteq\mathfrak l\),
\([\mathfrak l,\mathfrak m]\subseteq\mathfrak m\),
\([\mathfrak m,\mathfrak m]\subseteq\mathfrak l\), and we can focus on their corresponding matrix representations
(\Cref{lem:theta-proj,prop:symmetric-pair}).
\item \textbf{Step 4 (Decomposing the breaking term).}
Any Hermitian \(H_{\mathrm{break}}\) splits as \(H_\pm=E_\pm(H_{\mathrm{break}})\) into the two $\theta$-eigenspaces with
\(iH_+\in\mathfrak l\), \(iH_-\in\mathfrak m\) (\Cref{lem:break-split-correct}).
The assumption \(RH_{\mathrm{break}}R^{-1}\neq H_{\mathrm{break}}\) is equivalent to \(H_-\neq 0\), so we have a nonzero odd element in $\mathfrak{m}$.
\item \textbf{Step 5 (Irreducible $\mathfrak l$–module $\mathfrak m$).}
From the matrix representation, we can identify $\mathfrak m\cong\Hom(\mathcal H_-,\mathcal H_+)\cong\mathcal H_+\otimes\mathcal H_-^{*}$, i.e., elements in $\mathfrak{m}$ are mapping between different $R$-eigenspaces. For a given element \(K=\mathrm{diag}(A,B)\in\mathfrak l\), and an element $M_T \in \mathfrak{m}$, the adjoint action is
\([K,M_T]=M_{AT-TB}\), (outer tensor product action on $\mathcal H_+\otimes\mathcal H_-^{*}$; \Cref{lem:m-model-action}).
Since we have the universality $\mathfrak l \cong\ \mathfrak{su}(d_+)\oplus_m\mathfrak{su}(d_-)\oplus_m\mathfrak u(1)_{\mathrm{rel}}$ (\Cref{prop:l-structure}), we can construct a matrix-basis proof to show the $\mathfrak l$–module $\mathfrak m$ is irreducible (\Cref{prop:irreducible-m}), hence the $\mathfrak l$–orbit of any $0\neq M\in\mathfrak m$ linearly generates all of $\mathfrak m$.
\item \textbf{Step 6 (Closure).}
From Step~1, we obtain \(\mathfrak l\subseteq \mathfrak g\coloneqq\mathrm{span}_{\mathbb R}\langle iH_X,iH_Z,iH_{ZZ},iH_{\mathrm{break}}\rangle_{\mathrm{Lie}}\).
From Step~4, there is one nonzero \(iH_-\in\mathfrak m\cap\mathfrak g\).
By Step~5, the \(\mathfrak l\)–orbit of \(iH_-\) spans \(\mathfrak m\), hence we get \(\mathfrak m\subseteq\mathfrak g\).
Therefore \(\mathfrak g\) contains both summands in \(\mathfrak{su}(2^N)=\mathfrak l\oplus_m\mathfrak m\), so
\(\mathfrak g=\mathfrak{su}(2^N)\) (\Cref{prop:closure}).
The necessity direction is immediate: if \(H_-=0\), the DLA remains inside \(\mathfrak l\) (\Cref{cor:iff}).
\hfill$\square$
\end{itemize}

\begin{step}{1: Symmetric controls generate the fixed-point Lie algebra}
\medskip

We start by showing that the DLA generated by $iH_X,iH_Z$ and $iH_{ZZ}$ spans the whole $R$-symmetric Lie subalgebra.

\begin{lemma} \label{lemma:equiv_DLA_2}
    Let \(\theta\coloneqq\Ad_R\), and define the mirror local terms as \(\widetilde X_j=X_j+\theta(X_j)\), \(\widetilde Z_j=Z_j+\theta(Z_j)\), and \(\widetilde{ZZ}_j=Z_jZ_{j+1}+\theta(Z_jZ_{j+1})\).
    We have the following equivalence between DLAs:
    \begin{equation} \label{eqn:equiv_DLAs_uniform}
        \begin{aligned}
            \mathrm{span}_{\mathbb{R}}\langle iH_{X}, iH_{Z}, iH_{ZZ}\rangle_{\mathrm{Lie}} &\equiv \begin{cases}
            \mathrm{span}_{\mathbb{R}}\langle \cup_j  i\widetilde X_j,\cup_j i\widetilde Z_j,\cup_j i\widetilde{ZZ}_j \rangle_{\mathrm{Lie}}\;, &\mathrm{for\ even\ } N \\
            \mathrm{span}_{\mathbb{R}}\langle \cup_j  i\widetilde X_j,\cup_j i\widetilde Z_j,\cup_j i\widetilde{ZZ}_j , iX_{\lfloor N/2\rfloor +1} , iZ_{\lfloor N/2\rfloor +1}\rangle_{\mathrm{Lie}}\;, &\mathrm{for\ odd\ }N
        \end{cases}\; \\
        &\coloneqq \mathrm{span}_{\mathbb R}\!\left\langle\, i\widetilde X_j,\,i\widetilde Z_j,\,i\widetilde{ZZ}_j\,\right\rangle_{\mathrm{Lie}}\;.
        \end{aligned}
    \end{equation}
\end{lemma}

\begin{proof}
    Given $H_{ZZ},H_Z$ and $H_X$, we first obtain an $H_{YY}$ (the uniform $YY$ interaction) as follows:
    \[ \label{eqn:nearest_neighbor_inter_algebra}
        \begin{aligned}
        &H_{YZ} = \sum_{j=1}^{N-1} Y_jZ_{j+1} + Z_jY_{j+1} \propto [H_{X},H_{ZZ}]\;, \\
        &H_{YYZZ} = \sum_{j=1}^{N-1} Y_jY_{j+1} - Z_j Z_{j+1} \propto [H_X, H_{YZ}]\;,\\
        &H_{YY} = \sum_{j=1}^{N-1} Y_jY_{j+1} = H_{YYZZ} + H_{ZZ}\;.
    \end{aligned}
    \]
    Then, we can compute the commutator between $H_{YY}$ and $H_{ZZ}$, to obtain the following intermediate Hamiltonian:
    \[
        H_{1} = \sum_{j=1}^{N-2} Z_j X_{j+1} Y_{j+2} + Y_j X_{j+1} Z_{j+2} \propto [H_{YY},H_{ZZ}]\;,
    \]
    where we have used:
    \[  [Y_{j}Y_{j+1},Z_jZ_{j+1}] = 0\;,\quad [Y_{j}Y_{j+1},Z_{j+1}Z_{j+2}] = 2i Y_jX_{j+1}Z_{j+2}\;, \quad [Y_{j+1}Y_{j+2},Z_{j}Z_{j+1}] = 2i Z_jX_{j+1}Y_{j+2}\;. 
    \]
    Then, we can compute:
    \[
        \begin{aligned}
        &H_2 = H_{YY} + \sum_{j=2}^{N-2} Y_jY_{j+1} - 2\sum_{j=1}^{N-3}Z_j X_{j+1}X_{j+2}Z_{j+3} \propto [H_{ZZ},H_{1}]\;,\\
        &H_3 = H_2 - H_{YY} = \sum_{j=2}^{N-2} Y_jY_{j+1} - 2\sum_{j=1}^{N-3}Z_j X_{j+1}X_{j+2}Z_{j+3}\;, 
    \end{aligned}
    \]
    where we have used:
    \[
        \begin{aligned}
        &[Z_jZ_{j+1},Z_jX_{j+1}Y_{j+2}] = 2i Y_{j+1}Y_{j+2}\;,\quad [Z_{j+1}Z_{j+2},Y_jX_{j+1}Z_{j+2}] = 2i Y_jY_{j+1}\;,\\
        &[Z_jZ_{j+1},Y_{j+1}X_{j+2}Z_{j+3}] = -2iZ_{j}X_{j+1}X_{j+2}Z_{j+3},\quad [Y_{j+2}Y_{j+3},Z_{j}X_{j+1}Y_{j+2}] = -2iZ_{j}X_{j+1}X_{j+2}Z_{j+3}\;. 
        \end{aligned}
    \]
    Notice that in $H_3$, there are no $Y_{1}Y_2$ and $Y_{N-1}Y_{N}$ terms, because the open boundary condition makes the boundary $YY$ terms distinguishable from others. By doing the commutation again, we obtain:
    \[
        \begin{aligned}
        &H_4 = \sum_{j=1}^{N-3} Z_j X_{j+1} Y_{j+2} + \sum_{j=2}^{N-2} Y_j X_{j+1} Z_{j+2} \propto [H_{ZZ},H_3]\;,\\
        &H_5 = \sum_{j=2}^{N-2} Y_j Y_{j+1} - \sum_{j=1}^{N-3} Z_jX_{j+1}X_{j+2}Z_{j+3} \propto [H_{ZZ},H_4]\;. 
    \end{aligned}
    \]
    Comparing $H_5$ to $H_3$, we notice that the coefficients of the two summations are different, which implies that we can linearly combine $H_3$ and $H_5$ to single out each of the summations. Thus, we obtain:
    \[
         \begin{aligned}
        &\widetilde{H}_{YY} = \sum_{j=2}^{N-2} Y_j Y_{j+1} = 2H_5 - H_3\;, \\
        & H_{YY,1} = Y_1Y_2 + Y_{N-1}Y_N = H_{YY} - \widetilde{H}_{YY}\;.
    \end{aligned}
    \]
    We have singled out the boundary term $H_{YY,1}$, which is crucial in obtaining $\widetilde{X}_1,\widetilde{Z}_1$ and $\widetilde{ZZ}_1$. First, we obtain $\widetilde{ZZ}_1$ by the following algebra:
    \[
    \begin{aligned}
        &H_{YZ,1} = (Y_1Z_2+Z_1Y_2) + (Y_{N-1}Z_N + Z_{N-1}Y_N) \propto [H_X, H_{YY,1}]\;,\\
        &H_{YYZZ,1} = (Z_1Z_2 + Z_{N-1}Z_N) - (Y_1Y_2+Y_{N-1}Y_N) \propto [H_X,H_{YZ,1}]\;,\\
        &\widetilde{ZZ}_1 = Z_1Z_2 + Z_{N-1}Z_N = H_{YYZZ,1} + H_{YY,1}\;.
    \end{aligned}
    \]
    Then, we obtain the $\widetilde{X}_1$ by the following algebra:
    \[
    \begin{aligned}
        &H_{X,12} = \widetilde{X}_1 + \widetilde{X}_2 = (X_1+X_{N})+(X_2 + X_{N-1}) \propto [H_{YZ,1},\widetilde{ZZ}_1]\\
        &\widetilde{X}_2 = X_2 + X_{N-1} \propto [\widetilde{H}_{YY},[\widetilde{H}_{YY},H_{X,12}]] \\
        &\widetilde{X}_1 = X_1 + X_{N} =  H_{X,12} - \widetilde{X}_2\;.
    \end{aligned}
    \]
    Finally, we can single out $\widetilde{Z}_1$ by:
    \[
        \widetilde{Z}_1 = Z_1 + Z_{N} \propto [\widetilde{X}_1,[\widetilde{X}_1,H_Z]]\;.
    \]
    Given the terms $\widetilde{X}_k,\widetilde{Z}_k$ and $\widetilde{ZZ}_k$ on the first $k$ boundaries, we can easily obtain the terms on the $(k+1)$-th boundary by the following algebra:
    \[
        \begin{aligned}
        &\widetilde{H}_{X,k+1} = \sum_{j=k+1}^{N-(k+1)+1} X_j = H_X - \sum_{j=1}^k \widetilde{X}_j\;,\quad \widetilde{H}_{ZZ,k+1} = \sum_{j=k+1}^{N-(k+1)}  Z_jZ_{j+1} = H_{ZZ} - \sum_{j=1}^k \widetilde{ZZ}_j\;,\\
        &\widetilde{X}_{k+1} = X_{k+1} + X_{N-(k+1)+1} \propto [\widetilde{ZZ}_k,[\widetilde{ZZ}_k,\widetilde{H}_{X,k+1}]] \;,\\
        &\widetilde{Z}_{k+1} = Z_{k+1} + Z_{N-(k+1)+1} \propto [\widetilde{X}_{k+1},[\widetilde{X}_{k+1},H_{Z}]]\;,\\
        &\widetilde{ZZ}_{k+1} = Z_{k+1}Z_{k+2} + Z_{N-(k+2)+1} Z_{N-(k+1)+1} \propto [\widetilde{X}_{k+1},[\widetilde{X}_{k+1},\widetilde{H}_{ZZ,k+1}]]\;.
    \end{aligned}
    \]
    Therefore, by induction, we can obtain all $\widetilde{X}_j,\widetilde{Z}_j,\widetilde{ZZ}_j$ for $j = 1,\cdots,\lfloor N/2\rfloor$ and also $X_{\lfloor N/2\rfloor +1}, Z_{\lfloor N/2\rfloor +1}$ if $N$ is odd. Therefore, we have completed the proof.
\end{proof}

\begin{lemma}[Mirrored generators produce all symmetric strings]\label{lem:mirror-generation}
The real Lie algebra generated by $\{\,i\widetilde X_j,\, i\widetilde Z_j,\, i\widetilde{ZZ}_j\,\}$ contains $i$ times every reflection-symmetric Pauli string.
\end{lemma}

\begin{proof}
On each mirror pair $(j,\bar j)$, $[\,i\widetilde Z_j,\,i\widetilde X_j\,]=2i\,\widetilde Y_j$, yielding local $\mathfrak{su}(2)$ on the pair.
Commutators with $i\widetilde{ZZ}_j$ extend support by one site while preserving mirrored form; induction on support length yields all symmetric strings.
\end{proof}

Next, we prove that the reflection-symmetric Pauli strings span the algebra $\mathfrak{l}$.
Remember that we use $\mathbb{P}_N$ to represent the $N$-qubit Paulis.
$\mathbb{P}_N$ is an orthogonal basis of $M_{2^N}(\mathbb{C})$ for the Hilbert-Schmidt inner product $\langle A,B\rangle_{\HS}\coloneqq\Tr(A^\dagger B)$~\cite{Nielsen_Chuang_2010}.
Moreover, the real span of i-multiples of the non-identity Paulis equals the traceless skew-Hermitian matrices $\mathfrak{su}(2^N)$ as~\cite{Nielsen_Chuang_2010}:
\[\mathrm{span}_{\mathbb{R}}\{\, iP : P\in \mathbb{P}_N\setminus\{I^{\otimes N}\}\,\} = \mathfrak{su}(2^N)\;.\]
Thus, one can prove the claim by constructing a projection into the reflection-symmetric subalgebra of $\mathfrak{su}(2^N)$ and show the image is $\mathfrak{l}$, as demonstrated in \Cref{prop:fix-span} below.

\begin{proposition}[Step 1: mirrored $X$, $Z$, and $ZZ$ generate $\mathfrak l$]\label{prop:fix-span}
Let $E\coloneqq\tfrac12(\mathrm{id}+\theta)$, where $\mathrm{id}$ is the identity map. Then $E$ defines a projector into the fixed point $\mathrm{Fix}(\theta)$, i.e. $\mathrm{Im}\,E=\Fix(\theta)\coloneqq\{X\in M_{2^N}(\mathbb C):\theta(X)=X\}$, which implies \(\Fix(\theta)=\mathrm{span}_{\mathbb C}\{\,E(P):P\in\mathbb{P}_N\,\}\). 
By taking the intersection between $\mathrm{Fix}(\theta)$ and $\mathfrak{su}(2^N)$, the $R$-fixed subalgebra $\mathfrak{l}$ is given by 
\[
\mathfrak l\coloneqq\Fix(\theta)\cap \mathfrak{su}(2^N)=\mathrm{span}_{\mathbb R}\{\, iE(P)\ :\ P\in\mathbb{P}_N\setminus\{I^{\otimes N}\}\,\}\;,
\]
which implies that
\[
    \mathrm{span}_{\mathbb R}\!\left\langle\, i\widetilde X_j,\,i\widetilde Z_j,\,i\widetilde{ZZ}_j\,\right\rangle_{\mathrm{Lie}}
=\mathfrak l
\]
\end{proposition}

\begin{proof}
First, to be a projector, $E$ must be \emph{linear and idempotent}. Linearity is immediate. Since $\theta^2=\mathrm{id}$, \(E^2=\tfrac14(\mathrm{id}+\theta)(\mathrm{id}+\theta)=\tfrac14(\mathrm{id}+2\theta+\theta^2)=E\), so it is idempotent.

Then, we should that $\mathrm{Im}\ E = \mathrm{Fix}(\theta)$. For any $X \in M_{2^N}(\mathbb{C})$,
\(
\theta(E(X))=\tfrac12(\theta X+\theta^2 X)=E(X),
\)
so $E(X)\in\Fix(\theta)$ and $\mathrm{Im}\,E\subseteq\Fix(\theta)$.
Conversely, if $Y\in\Fix(\theta)$ then $E(Y)=\tfrac12(Y+\theta Y)=Y$, hence $Y\in\mathrm{Im}\,E$.
Therefore $\mathrm{Im}\,E=\Fix(\theta)$, i.e.\,\ $E$ is the (linear, idempotent) projection onto the fixed-point space.

Since the Paulis $\mathbb{P}_N$ span $M_d(\mathbb C)$, by linearity, applying $E$ gives
\[
\Fix(\theta)=\mathrm{Im}\,E=\mathrm{span}_{\mathbb C}\{E(P):P\in\mathbb{P}_N\}.
\]
Now we take the intersection with $\mathfrak{su}(d)$. For $P\neq I^{\otimes N}$, $\theta$ is a $*$-automorphism that preserves trace, so $E(P)=\tfrac12(P+\theta(P))$ is Hermitian and traceless. Hence
\[
\Fix(\theta)\cap\{ \text{traceless Hermitian} \}
= \mathrm{span}_{\mathbb R}\{\,E(P):P\in\mathbb{P}_N\setminus\{I^{\otimes N}\}\,\}.
\]
Finally, the map $H\mapsto iH$ is a real-linear isomorphism from traceless Hermitians onto $\mathfrak{su}(d)$.
Applying it to the previous line yields
\[
\mathfrak l=\Fix(\theta)\cap \mathfrak{su}(d)
=\mathrm{span}_{\mathbb R}\{\, iE(P)\ :\ P\in\mathbb{P}_N\setminus\{I^{\otimes N}\}\,\}.\]
By \Cref{lemma:equiv_DLA_2,lem:mirror-generation}, the right-hand side is generated by $H_X,H_Z,H_{ZZ}$, yielding
\[
\mathrm{span}_{\mathbb{R}}\!\left\langle\, i\widetilde X_j,\ i\widetilde Z_j,\ i\widetilde{ZZ}_j \, \right\rangle_{\mathrm{Lie}}
\;=\; \mathfrak l\;,
\]
which completes the proof.
\end{proof}
\end{step}

\begin{step}{2: Block structure of the $R$-symmetric algebra}
\medskip

We focus on analyzing the structure of $\mathfrak{l}$, which is block diagonalized by the symmetry $R$.
Let $P_\pm\coloneqq\tfrac12(I\pm R)$ be the projector onto the $\pm1$ eigenspaces of $R$, and define $\mathcal H_\pm\coloneqq P_\pm\mathbb C^{2^N}$ with dimensions $d_\pm$.

\begin{lemma}[Block-diagonalization by the reflection]\label{lem:block-diag}
$[X,R]=0$ iff $X$ is block-diagonal in the $R$-eigenbasis:
\[
X=\begin{bmatrix}A&0\\ 0&B\end{bmatrix}
\quad\text{with}\quad
A=P_+XP_+\in \End(\mathcal H_+),\ \ B=P_-XP_-\in\End(\mathcal H_-).
\]
Conversely, every such block-diagonal $X$ commutes with $R$.
\end{lemma}

\begin{proof}
If $[X,R]=0$, then $XP_\pm=P_\pm X$, hence $P_\mp X P_\pm=P_\mp P_\pm X=0$, so $X$ preserves $\mathcal H_\pm$ and is block-diagonal as stated.
Conversely, a block-diagonal $X$ clearly commutes with $R=\mathrm{diag}(I_{d_+},-I_{d_-})$.
\end{proof}

\begin{proposition}[Step 2: Structure of the fixed-point Lie algebra]\label{prop:l-structure}
Let
\[
\mathfrak l\;\coloneqq\;\{X\in \mathfrak{su}(2^N):[X,R]=0\}.
\]
Then
\[
\mathfrak l \;=\; \big(\,\mathfrak u(d_+)\oplus_m \mathfrak u(d_-)\,\big)\ \cap\ \mathfrak{su}(2^N)
\ \cong\ \mathfrak{su}(d_+)\ \oplus_m\ \mathfrak{su}(d_-)\ \oplus_m\ \mathfrak u(1)_{\mathrm{rel}}.
\]
Explicitly, every $X\in\mathfrak l$ has a unique decomposition
\[
X=\begin{bmatrix}A_0&0\\[2pt]0&B_0\end{bmatrix}
\;+\;
i\varphi\!\left(\frac{1}{d_+}\,I_{d_+}\right)\oplus_m\!\left(-\frac{1}{d_-}\,I_{d_-}\right),
\qquad
A_0\in\mathfrak{su}(d_+),\ B_0\in\mathfrak{su}(d_-),\ \varphi\in\mathbb R,
\]
and the map
\[
\Phi:\ \mathfrak{su}(d_+)\oplus_m\mathfrak{su}(d_-)\oplus_m i\mathbb R
\longrightarrow \mathfrak l,\qquad
\Phi(A_0,B_0,i\varphi)
=
\begin{bmatrix}A_0&0\\[2pt]0&B_0\end{bmatrix}
+i\varphi\!\left(\frac{1}{d_+}\,I_{d_+}\right)\oplus_m\!\left(-\frac{1}{d_-}\,I_{d_-}\right)
\]
is a Lie algebra isomorphism. The $i\mathbb R$ factor is the \emph{relative phase}
$\mathfrak u(1)_{\mathrm{rel}}$, which is central in $\mathfrak l$.
\end{proposition}

\begin{proof}
By Lemma~\ref{lem:block-diag}, $X\in\mathfrak l$ iff $X=\mathrm{diag}(A,B)$ with $A\in\End(\mathcal H_+)$, $B\in\End(\mathcal H_-)$.
Impose $X\in\mathfrak{su}(2^N)$: $X^\dagger=-X$ and $\mathrm{tr}\,X=0$.
From $X^\dagger=-X$ we get $A^\dagger=-A$, $B^\dagger=-B$, i.e.\,\ $A\in\mathfrak u(d_+)$, $B\in\mathfrak u(d_-)$.
From $\mathrm{tr}\,X=0$ we obtain $\mathrm{tr}A+\mathrm{tr}B=0$.
Conversely, any such block–diagonal $X$ lies in $\mathfrak l$.
Write $A=A_0+i\alpha I_{d_+}$ and $B=B_0+i\beta I_{d_-}$ with $A_0\in\mathfrak{su}(d_+)$, $B_0\in\mathfrak{su}(d_-)$ and $d_+\alpha+d_-\beta=0$; set $\varphi\coloneqq d_+\alpha=-\,d_-\beta$ to get the stated form.
The relative phase term is scalar on each block, hence central in $\mathfrak l$.

We carry out the explicit calculation to demonstrate the Lie isomorphism.
Let
\[
D\ \coloneqq\ \Big(\tfrac{1}{d_+}\,I_{d_+}\Big)\ \oplus_m\ \Big(-\tfrac{1}{d_-}\,I_{d_-}\Big),\qquad
\]
\[
X\coloneqq\Phi(A_0,B_0,i\varphi)=\begin{bmatrix}A_0&0\\[2pt]0&B_0\end{bmatrix}+i\varphi\,D,\quad Y\coloneqq\Phi(A_0',B_0',i\varphi')=\begin{bmatrix}A_0'&0\\[2pt]0&B_0'\end{bmatrix}+i\varphi'\,D.
\]
Then
\[
[X,Y]
=\begin{bmatrix}[A_0,A_0']&0\\[2pt]0&[B_0,B_0']\end{bmatrix}
=\Phi\big([A_0,A_0'],\, [B_0,B_0'],\, 0\big).
\]
On the domain \(\mathfrak{su}(d_+)\oplus_m\mathfrak{su}(d_-)\oplus_m i\mathbb R\), the Lie bracket is
\(
[(A_0,B_0,i\varphi),(A_0',B_0',i\varphi')]=([A_0,A_0'],[B_0,B_0'],0)
\)
since the \(i\mathbb R\) component is central. Therefore \(\Phi\) preserves brackets, i.e.\,\ it is a Lie algebra homomorphism. Combined with the uniqueness of the block decomposition (already shown), \(\Phi\) is a bijection; hence \(\Phi\) is a Lie algebra isomorphism.

We can carry out the dimension check to verify the isomorphism.
As real Lie algebras,
\[
\dim \mathfrak l \;=\; (d_+^2-1)+(d_-^2-1)+1 \;=\; d_+^2+d_-^2-1,
\]
which matches the dimension of traceless block–diagonal skew–Hermitian matrices of sizes $d_+$ and $d_-$.
\end{proof}
\end{step}

\begin{step}{3: Involution, projections, and decomposition of the algebra}
\medskip
Here, we investigate the remaining part of $\mathfrak{su}(2^N)$. 
The $\mathfrak{su}(2^N)$ is decomposed into even and odd eigenspaces of $\theta$, and their matrix representations can be obtained explicitly.

\begin{lemma}[Even/odd projections and direct sum]\label{lem:theta-proj}
Let $E_\pm\coloneqq\tfrac12(\mathrm{id}\pm\theta)$, where $\theta=\Ad_R$ satisfies $\theta^2=\mathrm{id}$.
Then $E_\pm$ are complementary projections ($E_\pm^2=E_\pm$, $E_+E_-=0$, $E_++E_-=\mathrm{id}$) with images
\(
\mathrm{Im}\,E_+=\{X:\theta(X)=X\}=\mathfrak l
\)
and
\(
\mathrm{Im}\,E_-=\{X:\theta(X)=-X\}=: \mathfrak m.
\)
Consequently,
\[
\mathfrak{su}(2^N)=\mathfrak l\ \oplus_m\ \mathfrak m
\]
\end{lemma}

\begin{proof}
Compute $E_\pm^2=\tfrac14(\mathrm{id}\pm\theta)^2=\tfrac14(\mathrm{id}\pm2\theta+\theta^2)=E_\pm$ and $E_+E_-=\tfrac14(\mathrm{id}-\theta^2)=0$, $E_++E_-=\mathrm{id}$.
For any $X$, $\theta(E_+(X))=E_+(X)$ and $\theta(E_-(X))=-E_-(X)$, so $\mathrm{Im}\,E_\pm\subseteq\Fix(\pm\theta)$. Conversely, if $Y$ satisfies $\theta(Y)=\pm Y$, then $E_\pm(Y)=Y$, hence $Y\in\mathrm{Im}\,E_\pm$. Therefore $\mathrm{Im}\,E_\pm=\Fix(\pm\theta)$. Finally $E_++E_-=\mathrm{id}$ yields the direct sum and trivial intersection.
\end{proof}

\begin{proposition}[Step 3: Explicit odd space and brackets]\label{prop:symmetric-pair}
In the $R$-eigenbasis $\mathcal H=\mathcal H_+\oplus_v\mathcal H_-$,
\begin{align*}
\mathfrak l&=\left\{\begin{psmallmatrix}A&0\\ 0&B\end{psmallmatrix}: A\in\mathfrak u(d_+),\, B\in\mathfrak u(d_-),\, \Tr A+\Tr B=0\right\},\\
\mathfrak m&=\left\{\begin{psmallmatrix}0&T\\ -T^\dagger&0\end{psmallmatrix}: T\in\mathbb C^{d_+\times d_-}\right\},
\end{align*}
and the brackets obey $[\mathfrak l,\mathfrak l]\subseteq\mathfrak l$, $[\mathfrak l,\mathfrak m]\subseteq\mathfrak m$, $[\mathfrak m,\mathfrak m]\subseteq\mathfrak l$.
\end{proposition}

\begin{proof}
\emph{Block form of $\theta$ and of $E_\pm$.}
Choose an orthonormal basis adapted to $R$, so
\(
R=\mathrm{diag}(I_{d_+},-I_{d_-})
\)
and
\(
\theta=\Ad_R: X\mapsto RXR^{-1}.
\)
Write a general matrix $X$ in block form relative to $\mathcal H_+\oplus_v\mathcal H_-$:
\[
X=\begin{psmallmatrix}A&C\\ D&B\end{psmallmatrix}.
\]
A direct multiplication gives
\[
\theta(X)=RXR=\begin{psmallmatrix}A&-C\\ -D&B\end{psmallmatrix},
\quad \mathrm{thus}\quad 
E_+(X)=\begin{psmallmatrix}A&0\\ 0&B\end{psmallmatrix}\;,\qquad
E_-(X)=\begin{psmallmatrix}0&C\\ D&0\end{psmallmatrix}.
\]

\emph{Description of $\mathfrak l$.} The descriptions of $\mathfrak l$ from Lemma~\ref{lem:block-diag}.

\emph{Description of $\mathfrak m$.}
By Lemma~\ref{lem:theta-proj}, $\mathfrak m=\mathrm{Im}\,E_-\cap\mathfrak{su}(2^N)$ consists of block off–diagonal matrices
\(
X=\begin{psmallmatrix}0&C\\ D&0\end{psmallmatrix}.
\)
Imposing $X^\dagger=-X$ (skew-Hermitian) yields
\(
\begin{psmallmatrix}0&D^\dagger\\ C^\dagger&0\end{psmallmatrix}
= -\begin{psmallmatrix}0&C\\ D&0\end{psmallmatrix}
\),
so $D=-C^\dagger$. Writing $T\coloneqq C$ gives the claimed parametrization
\(
\mathfrak m=\big\{\begin{psmallmatrix}0&T\\ -T^\dagger&0\end{psmallmatrix}:T\in\mathbb C^{d_+\times d_-}\big\}.
\)
Tracelessness holds automatically because such $X$ are skew–Hermitian.

Bracket relations are checked by direct block multiplication.
\end{proof}

\begin{remark}[Orthogonality and dimensions]
With the Hilbert–Schmidt inner product $\langle X,Y\rangle=\mathrm{Tr}(X^\dagger Y)$, $\mathfrak l$ and
$\mathfrak m$ are orthogonal subspaces: $\mathrm{Tr}(\mathrm{diag}(A,B)\,M_T)=0$.
The real dimensions add correctly:
\[
\dim\mathfrak l = d_+^2+d_-^2-1,\qquad
\dim\mathfrak m = 2\,d_+d_-,\qquad
\dim\mathfrak{su}(2^N) = (d_++d_-)^2-1.
\]
\end{remark}
\end{step}

\begin{step}{4: Decomposing the reflection-breaking term}
\medskip
Any reflection-breaking term will introduce a non-zero element in $\mathfrak{m}$.

\begin{lemma}[Step 4: Reflection $\pm$ components of a Hamiltonian]\label{lem:break-split-correct}
For a Hermitian $H_{\mathrm{break}}$, define $H_\pm\coloneqq E_\pm(H_{\mathrm{break}})$, so $H_{\mathrm{break}}=H_++H_-$ and $\theta(H_\pm)=\pm H_\pm$.
Then $iH_+\in\mathfrak l$, $iH_-\in\mathfrak m$, and $RH_{\mathrm{break}}R\neq H_{\mathrm{break}}$ iff $H_-\neq 0$.
\end{lemma}

\begin{proof}
$E_\pm$ preserve Hermiticity, so $H_\pm$ are Hermitian with the stated parity.
Multiplying by $i$ yields skew-Hermitian elements in $\mathfrak l$ and $\mathfrak m$.
The breaking equivalence is $RH_{\mathrm{break}}R=H_{\mathrm{break}}\iff E_-(H_{\mathrm{break}})=0$.
\end{proof}

\begin{remark}[Hermitian odd vs.\,\ Lie odd]\label{rem:odd-blocks}
In the $R$-eigenbasis, $\theta$-odd \emph{Hermitian} matrices are $\begin{psmallmatrix}0&T\\ T^\dagger&0\end{psmallmatrix}$; elements of the \emph{Lie} odd space are $\begin{psmallmatrix}0&S\\ -S^\dagger&0\end{psmallmatrix}$. Multiplication by $i$ maps one to the other via $S=iT$.
\end{remark}

\begin{remark}[Tracelessness and the identity]
If $H_{\mathrm{break}}$ has a trace component, it contributes a multiple of $I$, which is $\theta$-even and central.
Since $[I,\cdot]=0$ and $\mathfrak{su}(2^N)$ consists of traceless operators, replacing $H_{\mathrm{break}}$ by
$H_{\mathrm{break}}-\frac{\mathrm{Tr}\,H_{\mathrm{break}}}{2^N}I$ does not change the generated DLA.
Hence we may assume $H_{\mathrm{break}}$ is traceless without loss of generality.
\end{remark}
\end{step}

\begin{step}{5: $\mathfrak m$ is an irreducible $\mathfrak l$-module}
\medskip
We prove that, $\mathfrak{m}$ is irreducible under the action of $\mathfrak{l}$, implying given the full $\mathfrak{l}$ and any non-zero element in $\mathfrak{m}$, we can generate the full $\mathfrak{m}$. Set $V\coloneqq\mathcal H_+$ and $W\coloneqq\mathcal H_-$. Identify
\[
\mathfrak m\;\cong\;\Hom(W,V)\;\cong\;V\otimes W^*,\qquad
T\longleftrightarrow M_T\coloneqq\begin{psmallmatrix}0&T\\ -T^\dagger&0\end{psmallmatrix}.
\]

\begin{lemma}[Model and action]\label{lem:m-model-action}
For $K=\mathrm{diag}(A,B)\in\mathfrak l$ one has
\[
[\,K,M_T\,]=M_{AT-TB}.
\]
Equivalently, on $V\otimes W^*$,
\[
(A,B)\cdot(v\otimes\psi)=Av\otimes\psi+v\otimes(B\cdot\psi),\qquad
(B\cdot\psi)(w)\coloneqq-\psi(Bw).
\]
Let $H_{\mathrm{rel}}\coloneqq i\big(\tfrac1{d_+}I_{d_+}\big)\oplus_m i\big(-\tfrac1{d_-}I_{d_-}\big)$. Then
$\ad_{H_{\mathrm{rel}}}(M_T)=M_{icT}$ with $c\coloneqq d_+^{-1}+d_-^{-1}$.
\end{lemma}

\begin{proof}[Proof of Lemma~\ref{lem:m-model-action}]
Write $K=\mathrm{diag}(A,B)$ and $M_T=\begin{psmallmatrix}0&T\\ -T^\dagger&0\end{psmallmatrix}$. Then $[K,M_T]=\begin{psmallmatrix}0&AT-TB\\ -(AT-TB)^\dagger&0\end{psmallmatrix}=M_{AT-TB}$ from direct matrix calculation.
Transport to $V\otimes W^*$ via $T\leftrightarrow \Phi(v\otimes\psi)$ to get
$(A,B)\cdot(v\otimes\psi)=Av\otimes\psi+v\otimes(B\cdot\psi)$ with $(B\cdot\psi)(w)=-\psi(Bw)$.
For $H_{\mathrm{rel}}=i(\tfrac1{d_+}I_{d_+})\oplus i(-\tfrac1{d_-}I_{d_-})$, we can compute$\ad_{H_{\mathrm{rel}}}(M_T)
=M_{\,ic\,T}, c=d_+^{-1}+d_-^{-1}.$
\end{proof}

\begin{lemma}[Spectral separation by one Cartan]\label{lem:spectral-sep}
Choose $H=\mathrm{diag}(H^+,H^-)\in\mathfrak l$ with
\[
H^+=i\,\mathrm{diag}(\alpha_1,\dots,\alpha_{d_+}),\qquad
H^-=i\,\mathrm{diag}(\beta_1,\dots,\beta_{d_-}),
\]
and set $\Delta_{ij}\coloneqq \alpha_i-\beta_j$. For $E_{ij}\in\Hom(W,V)$ with $E_{ij}(w_j)=v_i$ and $E_{ij}(w_{j'})=0$ for $j'\neq j$,
\[
\ad_H(M_{E_{ij}})=\Delta_{ij}\,M_{iE_{ij}},\qquad
\ad_H(M_{iE_{ij}})=-\Delta_{ij}\,M_{E_{ij}}.
\]
Hence on each real $2$-plane $\mathrm{span}_{\mathbb R}\{M_{E_{ij}},M_{iE_{ij}}\}$ we have $\ad_H^2=-\Delta_{ij}^2I$. If the values $\{\Delta_{ij}^2\}$ are pairwise distinct (e.g. $\alpha_i=i-\bar\alpha$, $\beta_j=j\sqrt2-\bar\beta$ with $\tr H^\pm=0$), then for any $0\neq M_T\in\mathfrak m$ there exists a real polynomial $p$ with
\[
p(\ad_H^2)\,M_T\in \mathrm{span}_{\mathbb R}\{M_{E_{i_0j_0}},M_{iE_{i_0j_0}}\}\setminus\{0\}
\]
for some $(i_0,j_0)$.
\end{lemma}

\begin{proof}
Let $H=\mathrm{diag}(H^+,H^-)$ with $H^+=i\,\mathrm{diag}(\alpha_1,\dots,\alpha_{d_+})$ and
$H^-=i\,\mathrm{diag}(\beta_1,\dots,\beta_{d_-})$ where all differences $\alpha_i-\beta_j$ are pairwise distinct.
For $E_{ij}\in\Hom(W,V)$,
\[
[H,M_{E_{ij}}]= \begin{psmallmatrix}0&H^+E_{ij}-E_{ij}H^-\\ -(H^+E_{ij}-E_{ij}H^-)^\dagger&0\end{psmallmatrix}  = \begin{psmallmatrix}0&(\alpha_i - \beta_j)iE_{ij}\\ (\alpha_i - \beta_j)iE_{ji}&0\end{psmallmatrix}
= (\alpha_i-\beta_j)\,M_{iE_{ij}},
\]
and similarly $[H,M_{iE_{ij}}]=-(\alpha_i-\beta_j)\,M_{E_{ij}}$.
Thus each real $2$-plane $\mathrm{span}_{\mathbb R}\{M_{E_{ij}},M_{iE_{ij}}\}$ is an eigenspace of $\ad_H^2$ with
eigenvalue $\lambda_{ij}\coloneqq -\Delta_{ij}^2 = -(\alpha_i-\beta_j)^2$, and these eigenvalues are all distinct by construction.

For any $0\neq M_T=\sum_{i,j}(t_{ij}M_{E_{ij}}+u_{ij}M_{iE_{ij}})$ choose an index $(i_0,j_0)$ with $(t_{i_0j_0},u_{i_0j_0})\neq(0,0)$ and set
\[
p(x)\coloneqq \prod_{(i,j)\neq(i_0,j_0)} \frac{x-\lambda_{ij}}{\lambda_{i_0j_0}-\lambda_{ij}} \ \in\ \mathbb R[x].
\]
Then $p(\ad_H^2)$ is the projector onto $\mathrm{span}_{\mathbb R}\{M_{E_{i_0j_0}},M_{iE_{i_0j_0}}\}$, so $p(\ad_H^2)M_T\in
\mathrm{span}_{\mathbb R}\{M_{E_{i_0j_0}},M_{iE_{i_0j_0}}\}\setminus\{0\}$.
\end{proof}

\begin{lemma}[Splitting a $2$-plane]\label{lem:plane-split}
Let $U\coloneqq M_{E_{ij}}$ and $V\coloneqq M_{iE_{ij}}$. Then
$\ad_{H_{\mathrm{rel}}}U=cV$ and $\ad_{H_{\mathrm{rel}}}V=-cU$ with $c=d_+^{-1}+d_-^{-1}$.
Consequently, for any $W=aU+bV\neq0$,
\[
\Big(aI-\tfrac{b}{c}\,\ad_{H_{\mathrm{rel}}}\Big)W\propto U,\qquad
\Big(bI+\tfrac{a}{c}\,\ad_{H_{\mathrm{rel}}}\Big)W\propto V.
\]
\end{lemma}

\begin{proof}
With $U=M_{E_{ij}}$ and $V=M_{iE_{ij}}$, Lemma~\ref{lem:m-model-action} gives
$\ad_{H_{\mathrm{rel}}}U=M_{ic\,E_{ij}}=c\,V$ and $\ad_{H_{\mathrm{rel}}}V=M_{ic\,iE_{ij}}=-c\,U$.
For $W=aU+bV$, compute $\ad_{H_{\mathrm{rel}}}W=ac\,V-bc\,U$, and then
\[
\Big(aI-\tfrac{b}{c}\ad_{H_{\mathrm{rel}}}\Big)W
= a(aU+bV)-\tfrac{b}{c}(ac\,V-bc\,U)
= (a^2+b^2)\,U,
\]
and similarly $\big(bI+\tfrac{a}{c}\ad_{H_{\mathrm{rel}}}\big)W=(a^2+b^2)\,V$.
\end{proof}

\begin{lemma}[Index moves]\label{lem:index-moves-concise}
Let $E^{(+)}_{pq}$ and $E^{(-)}_{rs}$ be the matrix units on $V$ and $W$, for $p,q = 1,\dots, d_+$ and $r,s = 1,\dots, d_-$.
Let $X^+_{pi}\coloneqq E^{(+)}_{pi}-E^{(+)}_{ip}\in\lie{su}(d_+)$ and
$X^-_{jq}\coloneqq E^{(-)}_{jq}-E^{(-)}_{qj}\in\lie{su}(d_-)$. Then
\[
[\,\mathrm{diag}(X^+_{pi},0),\,M_{E_{ij}}\,]=M_{E_{pj}},\qquad
[\,\mathrm{diag}(0,X^-_{jq}),\,M_{E_{ij}}\,]=-\,M_{E_{iq}},
\]
and the same holds with $M_{E_{ij}}$ replaced by $M_{iE_{ij}}$.
Thus commutators with such off-diagonal elements move $(i,j)$ to any $(p,q)$.
\end{lemma}

\begin{proof}
For $K=\mathrm{diag}(A,B)$ and $T\in\Hom(W,V)$, Lemma~\ref{lem:m-model-action} yields $[K,M_T]=M_{AT-TB}$.
Set $A=X^+_{pi}=E^{(+)}_{pi}-E^{(+)}_{ip}$, $B=0$. Then
$AT=(E^{(+)}_{pi}-E^{(+)}_{ip})E_{ij}=E_{pj}$ and $TB=0$, hence
$[\,\mathrm{diag}(X^+_{pi},0),M_{E_{ij}}\,]=M_{E_{pj}}$.
With $A=0$ and $B=X^-_{jq}=E^{(-)}_{jq}-E^{(-)}_{qj}$ we have
$AT=0$ and $TB=E_{ij}(E^{(-)}_{jq}-E^{(-)}_{qj})=E_{iq}$, giving
$[\,\mathrm{diag}(0,X^-_{jq}),M_{E_{ij}}\,]=-M_{E_{iq}}$.
Linearity gives the identical relations for $M_{iE_{ij}}$.
\end{proof}

\begin{proposition}[Step 5: Irreducibility of $\mathfrak m$]\label{prop:irreducible-m}
For any nonzero $M\in\mathfrak m$, the real Lie algebra generated by its $\mathfrak l$-orbit equals $\mathfrak m$:
\[
\mathrm{span}_{\mathbb R}\Big\langle\,\ad_{K_t}\cdots\ad_{K_1}(M):\ t\ge0,\ K_s\in\mathfrak l\,\Big\rangle=\mathfrak m.
\]
Thus, $M$ is cyclic on $\mathfrak{m}$.
\end{proposition}

\begin{proof}
Given $0\neq M\in\mathfrak m$, Lemma~\ref{lem:spectral-sep} provides $p\in\mathbb R[x]$ with
$W\coloneqq p(\ad_H^2)M\in\mathrm{span}_{\mathbb R}\{M_{E_{i_0j_0}},M_{iE_{i_0j_0}}\}\setminus\{0\}$.
By Lemma~\ref{lem:plane-split}, a polynomial in $\ad_{H_{\mathrm{rel}}}$ extracts $M_{E_{i_0j_0}}$ (up to scale).
Lemma~\ref{lem:index-moves-concise} then reaches any $M_{E_{pq}}$, and applying $\ad_{H_{\mathrm{rel}}}$ yields the partner $M_{iE_{pq}}$.
Since $\{M_{E_{pq}},M_{iE_{pq}}\}$ is a real basis of $\mathfrak m$, the real Lie algebra generated by the $\mathfrak l$-orbit of $M$ equals $\mathfrak m$.
\end{proof}
\end{step}

\begin{step}{6: Closure and conclusion}

\begin{proposition}[Step 6: Closure once a single odd element is present, i.e., \Cref{qubit:UQC_chain}.]\label{prop:closure}
Let
\[
\mathfrak g\ \coloneqq\ \mathrm{span}_{\mathbb R}\!\left\langle\, iH_X,\ iH_Z,\ iH_{ZZ},\ iH_{\mathrm{break}}\,\right\rangle_{\mathrm{Lie}}.
\]
If $(H_{\mathrm{break}})_-\neq 0$, then $\mathfrak g=\mathfrak{su}(2^N)$.
\end{proposition}

\begin{proof}
By Step~1, we have $\mathfrak l\subseteq\mathfrak g$. By Step~4, $iH_-\in\mathfrak m$ is nonzero.
By Step~5, the $\mathfrak l$-orbit spans $\mathfrak m$, hence we obtain $\mathfrak m\subseteq\mathfrak g$.
Since $\mathfrak{su}(2^N)=\mathfrak l\oplus_m\mathfrak m$ (Lemma~\ref{lem:theta-proj}), we conclude $\mathfrak g=\mathfrak{su}(2^N)$.
\end{proof}

\begin{corollary}[Necessary and sufficient criterion]\label{cor:iff}
With uniform controls $H_X,H_Z,H_{ZZ}$,
\[
\mathrm{span}_{\mathbb R}\!\left\langle\, iH_X,\ iH_Z,\ iH_{ZZ},\, iH_{\mathrm{break}}\,\right\rangle_{\mathrm{Lie}}
=\begin{cases}
\mathfrak{su}(2^N), & (H_{\mathrm{break}})_-\neq 0,\\[3pt]
\mathfrak l, & (H_{\mathrm{break}})_-=0.
\end{cases}
\]
\end{corollary}

\begin{remark}[Edge case $N=1$]
If $N=1$, then $R=I$, $\mathfrak m=\{0\}$ and $\mathfrak l=\mathfrak{su}(2)$; the criterion is vacuous. For $N\ge 2$ both parity sectors are nontrivial.
\end{remark}
By combining the above six steps, we complete the proof of \Cref{qubit:UQC_chain}. \hfill $\square$
\medskip
\end{step}

As noted in \Cref{remark:generalization_thm_1}, the assumption of homogeneous single-Pauli interactions is not essential. 
The framework naturally extends to more general cases, including multiple-Pauli interactions (linear combinations of single-Pauli terms) and mixed-Pauli interactions (reflection-symmetric bilinear couplings such as $XY+YX$). 
In these broader settings, universality is preserved except in certain degenerate cases where the interaction coefficients enforce additional symmetries. 
We first treat the multiple-Pauli case.

\begin{corollary}\label{corollary:generalization_interaction}
Consider a homogeneous nearest-neighbor interaction of the form
\begin{equation}
    H_{\mathrm{int}} = c_X H_{XX} + c_Y H_{YY} + c_Z H_{ZZ}\;,
\end{equation}
where $H_{PP} = \sum_{\langle i,j\rangle} P_i P_j$ for $P \in \{X,Y,Z\}$. 
If the coefficients satisfy either
\begin{equation} \label{eqn:condition_coeff_multiple_Pauli_1}
    c_X \neq c_Y,\quad c_X \neq c_Z,\quad c_X + c_Y + c_Z \neq 0\;,
\end{equation}
or
\begin{equation}\label{eqn:condition_coeff_multiple_Pauli_2}
    c_X \neq c_Y,\quad c_X + c_Y - c_Z \neq 0,\quad c_X + c_Y + c_Z \neq 0\;,
\end{equation}
together with cyclic permutations $X \to Y \to Z$, then the universality result of \Cref{qubit:UQC_chain} still holds. In particular, adding any reflection symmetry-breaking control Hamiltonian suffices to achieve universal quantum computation.
\end{corollary}

\begin{proof}
Assume the condition in \Cref{eqn:condition_coeff_multiple_Pauli_1}. A straightforward calculation gives
\begin{equation}
    [H_Z,[H_Z,H_{\mathrm{int}}]] \propto (c_X - c_Y)(H_{XX}-H_{YY})\;.
\end{equation}
When $c_X \neq c_Y$, this yields $H_1 = H_{XX}-H_{YY}$. Substituting back,
\begin{equation}
    H_2 = H_{\mathrm{int}} + c_Y H_1 = (c_X+c_Y)H_{XX} + c_Z H_{ZZ}\;.
\end{equation}
Similarly,
\begin{equation}
    [H_Y,[H_Y,H_{\mathrm{int}}]] \propto (c_X - c_Z)(H_{ZZ}-H_{XX})\;,
\end{equation}
which isolates $H_3 = H_{ZZ}-H_{XX}$ when $c_X \neq c_Z$. Combining with $H_2$,
\begin{equation}
    H_4 = H_2 + (c_X+c_Y)H_3 = (c_X+c_Y+c_Z) H_{ZZ}\;.
\end{equation}
Thus, if $c_X+c_Y+c_Z \neq 0$, we recover a pure $H_{ZZ}$, reducing to the setting of \Cref{qubit:UQC_chain}.  

If instead the condition in \Cref{eqn:condition_coeff_multiple_Pauli_2} holds, a similar calculation with
\begin{equation}
    [H_Y,[H_Y,H_{2}]] \propto (c_X+c_Y-c_Z)(H_{ZZ}-H_{XX})
\end{equation}
yields $H_{ZZ}$ directly.  

Since the argument is symmetric under cyclic permutations of $(X,Y,Z)$, the result applies generally. 
\end{proof}

We now turn to the mixed-Pauli case.

\begin{corollary}\label{corollary:generalization_interaction_2}
Consider a homogeneous nearest-neighbor interaction of the form
\begin{equation}
    H_{\mathrm{int}} = c_XH_{XX} + c_YH_{YY}+c_ZH_{ZZ} + c_{XY}H_{XY} + c_{YZ}H_{YZ} + c_{ZX}H_{ZX}\;,
\end{equation}
where $H_{P_1P_2} = \sum_{\langle i,j \rangle} (P_{1,i}P_{2,j}+P_{2,i}P_{1,j})$ for distinct $P_1,P_2 \in \{X,Y,Z\}$ are reflection-symmetric mixed-Pauli interactions. 
If the coefficients satisfy the conditions of \Cref{corollary:generalization_interaction}, then the universality result of \Cref{qubit:UQC_chain} continues to hold.
\end{corollary}

\begin{proof}
It suffices to show that the mixed terms $H_{XY},H_{YZ},H_{ZX}$ can be isolated whenever their coefficients are nonzero. 
Without loss of generality, assume $c_{XY} \neq 0$. Then
\begin{equation}
    [H_Z,[H_Z,H_{\mathrm{int}}]] \propto 2(c_X-c_Y)(H_{YY}-H_{XX}) - 4 c_{XY}H_{XY} - c_{YZ}H_{YZ} - c_{ZX}H_{ZX}\;.
\end{equation}
By taking suitable linear combinations with $H_{\mathrm{int}}$, the contributions from $H_{YZ}$ and $H_{ZX}$ can be canceled, leaving
\begin{equation}
    H_2 = H_{\mathrm{int}} + H_1 = (2c_Y - c_X) H_{XX} + (2c_X - c_Y)H_{YY} + c_ZH_{ZZ} - 3c_{XY}H_{XY}\;.
\end{equation}
Next, commutators with $H_Z$ generate
\begin{equation}
    H_3  = (c_Y - c_X) H_{XY} - 2c_{XY} (H_{YY}-H_{XX}) \propto [H_Z,H_2]\;,
\end{equation}
and
\begin{equation}
    H_4  = 2c_{XY} H_{XY} + (c_Y - c_X) (H_{YY}-H_{XX}) \propto [H_Z,H_3]\;.
\end{equation}
Finally, forming a linear combination of $H_3$ and $H_4$ yields
\begin{equation}
    (4c_{XY}^2+(c_Y-c_X)^2) H_{XY} \propto (c_Y-c_X)H_3 + 2c_{XY}H_4\;.
\end{equation}
Since $c_{XY} \neq 0$ by assumption, this isolates $H_{XY}$. 
Analogous arguments apply to $H_{YZ}$ and $H_{ZX}$. 
Thus, $H_{\mathrm{int}}$ can always be reduced to the multiple-Pauli form of \Cref{corollary:generalization_interaction}, and the same universality condition applies.
\end{proof}

\begin{example}
    We consider three examples with different $H_{X,\alpha}$ patterns that are universal: The first one is the one-dimensional dual-species Ising chain with alternating $A$ and $B$ sites for $H_{X,\alpha}$. This control structure is feasible to the current dual-species neutral atom platforms, with the following graph representation: \begin{equation} 
    \begin{tikzpicture}[baseline={(current bounding box.center)}][scale=1]
      \def\N{12}
      \def\sep{1.0}
      \foreach \i in {1,...,\N}{
          \pgfmathtruncatemacro{\k}{mod(\i-1,2)}
          \ifnum\k=0
              \node[circle,draw,minimum size=5pt,inner sep=0.5pt,
                    label={[yshift=-2pt]below:$A$}](site\i) at (\sep*\i,0){};
          \else
              \node[circle,draw,fill=black,minimum size=5pt,inner sep=0.5pt,
                    label={[yshift=-2pt]below:$B$}](site\i) at (\sep*\i,0){};
          \fi
      }
      \foreach \i in {1,...,11}{
          \pgfmathtruncatemacro{\j}{\i+1}
          \draw[dashed](site\i)--(site\j);
      }
      \foreach \i in {7}{\node at (\i-0.5,0.3){$ZZ$};}
      \draw(site1)node[above,yshift=4pt]{$X_A$};
      \draw(site2)node[above,yshift=4pt]{$X_B$};
      \draw(site3)node[above,yshift=4pt]{$X_A$};
      \draw(site4)node[above,yshift=4pt]{$X_B$};
    \end{tikzpicture}
    \end{equation}
    and the second one is a naive breaking of the lattice inversion symmetry in one dimension by two halves of different types of sites
    \begin{equation} 
    \begin{tikzpicture}[baseline={(current bounding box.center)}][scale=1]
      \def\N{12}
      \def\sep{1.0}
      \foreach \i in {1,...,\N}{
        \ifnum\i<7
        \node[circle,draw,minimum size=5pt,inner sep=0.5pt,label={[yshift=-2pt]below:$A$}](site\i) at (\sep*\i,0){};
          \else
              \node[circle,draw,fill=black,minimum size=5pt,inner sep=0.5pt,
                    label={[yshift=-2pt]below:$B$}](site\i) at (\sep*\i,0){};
          \fi
      }
      \foreach \i in {1,...,11}{
          \pgfmathtruncatemacro{\j}{\i+1}
          \draw[dashed](site\i)--(site\j);
      }
      \foreach \i in {5}{\node at (\i-0.5,0.3){$ZZ$};}
      \draw(site1)node[above,yshift=4pt]{$X_A$};
      \draw(site2)node[above,yshift=4pt]{$X_A$};
      \draw(site7)node[above,yshift=4pt]{$X_B$};
      \draw(site8)node[above,yshift=4pt]{$X_B$};
    \end{tikzpicture}\;.
    \end{equation}

As a two-dimensional example, we consider a Lieb lattice composed of three species of sites $A,B$ and $C$, with controllable global $X$-fields applied separately to each species. The pattern formed by the decorated sites and nearest-neighbor Ising bonds breaks all spatial symmetries. One can also prove its universality using a 2-dimensional version of the algorithm described in proving \Cref{qubit:UQC_chain}.
\begin{equation} \label{eqn:Lieb_lattice}
\begin{tikzpicture}[scale=1.5, baseline={(current bounding box.center)}]
    \def\Nx{2} 
    \def\Ny{2} 

    \def\sep{1.2}

    \foreach \x in {0,...,\Nx}{
        \foreach \y in {0,...,\Ny}{
            \node[circle,draw,minimum size=5pt,inner sep=0.5pt,label={[xshift = 6pt, yshift=0pt]below:$A$}]
                (A\x\y) at (\x*\sep,\y*\sep) {};
        }
    }

    \foreach \x in {0,...,\Nx}{
        \foreach \y in {0,...,\Ny}{
                \node[circle,draw,fill=black,minimum size=5pt,inner sep=0.5pt,label={[xshift=6pt, yshift=0pt]below:$B$}]
                    (B\x\y) at (\x*\sep+0.5*\sep,\y*\sep) {};
        }
    }

    \foreach \x in {0,...,\Nx}{
        \foreach \y in {0,...,\Ny}{
                \node[circle,draw,fill=gray!50,minimum size=5pt,inner sep=0.5pt,label={[xshift = 6pt,yshift=0pt]below:$C$}]
                    (C\x\y) at (\x*\sep,\y*\sep+0.5*\sep) {};
        }
    }

    \foreach \x in {0,...,\Nx}{
        \foreach \y in {0,...,\Ny}{
            \ifnum\x<\Nx
                \draw[dashed] (A\x\y) -- (B\x\y);
                \draw[dashed] (B\x\y) -- (A\the\numexpr\x+1\relax\y);
            \else
            \draw[dashed] (A\x\y) -- (B\x\y);
            \fi
        }
    }

    \foreach \x in {0,...,\Nx}{
        \foreach \y in {0,...,\Ny}{
            \ifnum\y<\Ny
                \draw[dashed] (A\x\y) -- (C\x\y);
                \draw[dashed] (C\x\y) -- (A\x\the\numexpr\y+1\relax);
            \else
            \draw[dashed] (A\x\y) -- (C\x\y);
            \fi
        }
    }

    \draw (A00) node[above,xshift=-12pt, yshift=4pt] {$X_A$};
    \draw (B00) node[above,yshift=4pt] {$X_B$};
    \draw (C00) node[above,xshift=-12pt, yshift=4pt] {$X_C$};

\end{tikzpicture}
\end{equation}
\end{example}

Although \Cref{qubit:UQC_chain} establishes universality for a one-dimensional qubit chain, its extension to arbitrary graphs or higher-dimensional lattices is straightforward. 
For instance, one may directly apply the same framework to the Lieb lattice (see \Cref{eqn:Lieb_lattice}). More generally, by placing qubits on the vertices of any graph and coloring vertices and edges to represent the global controls (cf. \Cref{remark:one_diml_examples}), one can formulate an analogous universality criterion.

We conjecture that breaking all symmetries of the underlying graph via an additional Hamiltonian term $H_{\mathrm{break}}$ suffices to promote the control set $\{H_X,H_Z,H_{ZZ},H_{\mathrm{break}}\}$ to universal quantum computation on any connected graph, as formalized in \Cref{qubit:UQC_graph}.
The proof proceeds in direct analogy with \Cref{qubit:UQC_chain}: one first derives the minimal Lie generators invariant under the graph automorphisms using only the uniform controls $\{H_X, H_Z, H_{ZZ}\}$, as in \Cref{lemma:equiv_DLA_2}, and shows that they are universal in the automorphism-symmetric subalgebra.
Then one shows that $H_{\mathrm{break}}$ is an irreducible module in the algebra, as in \Cref{prop:irreducible-m}, allowing one to achieve the full $\lie{su}(2^N)$. We leave the full technical proof to future work.

\begin{conjecture}[Minimal requirement for universal quantum computation on a qubit graph\label{qubit:UQC_graph}]
    Consider a connected graph $G=(V,E)$ hosting qubits on its vertices $V$, with homogeneous Ising-type nearest-neighbor interactions along edges $E$. Suppose the system admits tunable global control fields $H_{X} = \sum_{j\in V} X_j$ and $H_{Z} = \sum_{j\in V} Z_j$.
    Then, this system realizes universal quantum computation if and only if there exists at least one additional control Hamiltonian $H_{\mathrm{break}}$ breaks all nontrivial automorphisms of the control graph, rendering its automorphism group trivial.
\end{conjecture}

\section{Compilation efficiency of analog quantum computation\label{app:compilation_efficiency}}

In this appendix, we discuss the complexity of pulse optimization from a theoretical point of view.
In analog quantum computation or simulation, one realizes a target unitary by optimizing the control pulses.
In this setting, the complexity of the compilation is determined by the size of the corresponding quantum optimal control problem.
Here, under physical assumptions, we generalize the result in Ref.~\cite{Lloyd_2014} and provide a bound (for a given precision $\varepsilon$) on the problem size, which is polynomial with respect to the set's size of unitaries that are reachable in polynomial time.
We then explore the implications of the result.
We summarize the essential mathematical notations in this section in Table~\ref{tab:notations_compilation}.

\begin{table}[h]
\centering
\caption{Summary of mathematical notations in this section.}
\label{tab:notations_compilation}
\begin{tabular}{cc}
\toprule
\textbf{Symbol} & \textbf{Meaning} \\
\midrule
$d$      & Dimension of the physical Hilbert space \\
$\mathcal{D}$           & Size of the QOC problem \\
$\mathbf{u}(t) = \{u_{\alpha}(t)\}$ & Control fields (pulses)\\
$\{H_\alpha\}$ & Control Hamiltonians \\
$U_{\mathrm{target}}$ & Target unitary of the QOC problem\\
$\overline{\mathbf{u}}(t)$ & Solution of the QOC problem\\
$\overline{U}(t)$ & Resulting unitary of the QOC problem\\
$b_\alpha$ & maximal information in the pulse $u_\alpha(t)$\\
$\mathcal{D}_\alpha$ & Size of the optimization problem for pulse $u_\alpha(t)$\\
$C_\alpha$ & Channel capacity of pulse $u_\alpha(t)$\\
$\mathcal{U}$ & The set of reachable unitaries of a control system \\
$D_{\mathcal{U}}(d)$      & Dimension of $\mathcal{U}$ \\
$\mathcal{U}^+$ & The set of polynomial-time reachable unitaries \\
$D_{\mathcal{U}^+}(d)$      & Dimension of $\mathcal{U}^+$ \\
\bottomrule
\end{tabular}
\end{table}

We start by defining the size $\mathcal{D}$ of the quantum optimal control (QOC) problem.
As shown in \Cref{eq:time_dependent_Hamiltonian}, we consider the unitary dynamics $U(t)$ generated by a time-dependent Hamiltonian $H(t) = \sum_\alpha u_\alpha(t)H_\alpha$, with control fields and control Hamiltonians specified by $\mathbf{u}(t)$ and $\{H_\alpha\}$, respectively.
Here, $U(t)$ is governed by the following dynamical equation:
\begin{equation}
    i\partial_t{U}(t) = H(t)U(t)\;,
\end{equation}
with the initial condition $U(t=0) = I_d$ for $d$-dimensional Hilbert space.
The manifold generated by arbitrary $\mathbf{u}(t)$'s defines a set of reachable unitaries $\mathcal{U}$ with dimension $D_{\mathcal{U}}(d)$.
For universal dynamics generated by $\{H_\alpha\}$, the generated manifold satisfies $\mathcal{U} = SU(d)$, with $D_{\mathcal{U}}(d) = d^2 - 1$.
For a target unitary $U_{\mathrm{target}}$, the QOC problem aims to find a solution $\overline{\mathbf{u}}(t)$ to generate a unitary $\overline{U}(t)$ that is $\varepsilon$-close to $U_{\mathrm{target}}$ in the operator norm, i.e. $\norm{\overline{U}(t) - U_{\mathrm{target}}}<\varepsilon$.
Formally, we have the following definition of a QOC problem:
\begin{definition}
    A QOC problem can be expressed as the minimization of some functional $\mathcal{F}$:
\begin{equation} \label{eq:QOC_functional}
    \min_{\mathrm{u}(t)} \mathcal{F} (U(t),U_{\mathrm{target}},\{\lambda_i\})\;,
\end{equation}
where $\mathcal{F}$ may be a function of the infidelity between $U(t)$ and $U_{\mathrm{target}}$ (as in \Cref{eq:rollout}), supplementing machine constraints specified by a set of Lagrange multipliers $\{\lambda_i\}$. 
\end{definition}
Given a QOC problem, we can define its size as follows:
\begin{definition}
    The size $\mathcal{D}$ of a QOC problem is defined as the minimal number of independent degrees of freedom in the field $\overline{\mathbf{u}}(t)$ to solve \Cref{eq:QOC_functional} up to precision $\varepsilon$. For example, the minimal number of sampling points to describe $\mathbf{u}(t)$.
\end{definition}

Then, we consider the information content in the QOC problem.
In the remainder of this section, we choose base $2$ for the function $\log$.
By treating a single control pulse $u_\alpha(t)$ as a classical noiseless channel, the channel capacity is given by Hartley's law~\cite{chitode2021information}:
\begin{equation}
    C_\alpha = t \Omega_\alpha \log\left( 1+ \frac{\Delta u_\alpha}{\delta u_\alpha} \right) = t\Omega_\alpha b_\alpha\;,
\end{equation}
where $\Omega_\alpha$ is the pulse frequency (in pulses per second), with $\Delta u_\alpha = u_{\alpha,\max}-u_{\alpha,\min}$ and $\delta u_{\alpha}$ the maximal and minimal variations of the pulse, respectively.
Here, the bit depth $b_\alpha = \log(1+\Delta u_\alpha/\delta u_\alpha)$ characterizes the maximal information (in bits) encoded into a single pulse through the variations.
Since $t\Omega_\alpha$ is the number of sampling points of the pulse (assuming uniform sampling rate), we can take the size $\mathcal{D}_\alpha = t\Omega_\alpha$ for optimizating $u_\alpha(t)$.
Then, the total size of the QOC problem is given by $\mathcal{D} = \sum_\alpha D_\alpha$.
Similarly, the total channel capacity is $C=\sum_\alpha C_\alpha$.
By defining the maximal and minimal bit depths of all $u_\alpha(t)$ by $b_{\max} \coloneqq \max_{\alpha} b_\alpha$ and $b_{\min} \coloneqq \min_{\alpha} b_\alpha$, we have the following upper and lower bounds for the total information content:
\begin{equation}
   \mathcal{D}b_{\min}\equiv C_{\min} \leq C \leq C_{\max} \equiv \mathcal{D}b_{\max}\;.
\end{equation}

In the remainder of this section, we focus on $N$-qubit systems with dimension $d = 2^N$.
In these systems, we study unitaries that are of physical interests.
First, we assume that the control system satisfies physical constraints, so it has bounded control Hamiltonians and pulses, i.e., $\norm{H_{\alpha}} = 1,u_\alpha(t) \in [u_{\alpha,\min},u_{\alpha,\max}],\forall \alpha,\forall t$.
Second, we introduce a set of unitaries $\mathcal{U}^+ \subseteq \mathcal{U}$, with dimension denoted by $D_{\mathcal{U^+}}(d)$.
$\mathcal{U}^+$ contains unitaries which are reachable (with precision $\varepsilon$) in pulse time $T$ that is polynomial to the set size $D_{\mathcal{U}^+}(d)$, namely:
\begin{equation}
    \mathcal{U}^+ = \left\{U_{\mathrm{target}}\left| \norm{U_{\mathrm{target}} - U(t)} \leq \epsilon, \forall \,U(t) = \mathcal{T}\exp \left(-i\int_{0}^{T} dt \sum_\alpha u_\alpha(t)H_\alpha\right)\right.\right\}\;,
\end{equation}
where $T = \mathrm{poly}(D_{\mathcal{U}^+})$.
Notably, for any universal control system $\{H_\alpha\}$, the whole space of unitaries $\mathcal{U} = SU(d)$ is reachable in \emph{exponential} time.
Moreover, if the physical constraints are relaxed, we always have $D_{\mathcal{U}^+}(d) = D_{\mathcal{U}}(d)$.

Here, we provide an information-theoretic counting argument to bound the total size $\mathcal{D}$ of the physically interesting QOC problems.
Our main statement is the following proposition:
\begin{proposition}[Generalization of the main theorem in~\cite{Lloyd_2014}] \label{prop:size_QOC_problem}
    The size $\mathcal{D}$ of a QOC problem, up to precision $\varepsilon$, is a function polynomial of the dimension $D_{\mathcal{U}^+}(d)$ of the polynomial-time reachable unitaries.
\end{proposition}

\begin{proof}
    We first provide a lower bound of $\mathcal{D}$ by a polynomial function of $D_{\mathcal{U}^+}$. 
    We partition the time-polynomial reachable unitaries $\mathcal{U}^+$ by $\varepsilon$-balls with size $\varepsilon^{D_{\mathcal{U}^+}}$.
    By setting the size of $\mathcal{U}^+$ as the unit, it is necessary to have $\varepsilon^{-D_{\mathcal{U^+}}}$ balls to cover $\mathcal{U}^+$, and the target unitary $U_{\mathrm{target}}$ is enveloped by one $\varepsilon$-ball.
    The information necessary to specify the target ball $\varepsilon$ is $ \log \varepsilon^{-D_{\mathcal{U}^+}}$, so we require $C_{\min} \geq \log \varepsilon^{-D_{\mathcal{U}^+}}$:
    \begin{equation}
        \varepsilon \geq 2^{- \frac{\mathcal{D} b_{\min}}{D_{\mathcal{U^+}}}}\;.
    \end{equation}
    Thus, given a constant $\varepsilon$ and $b_{\min}$, we have:
    \begin{equation}
        \mathcal{D} \geq CD_{\mathcal{U}^+}\;,\quad C = -\frac{\log\varepsilon}{b_{\min}}\;,
    \end{equation}
    which implies that $\mathcal{D}$ is lower bounded by a polynomial function of $D_{\mathcal{U}^+}$.

    Next, we prove the upper bound.
    The QOC problem solves a path $\gamma$ with length $L$ that connects the initial unitary $I_{d}$ to the target unitary $U_{\mathrm{target}}$.
    The maximum of (non-redundant) information is the necessary information content to describe the path.
    Specifying each $\varepsilon$-ball along the path requires $\log \varepsilon^{-D_{\mathcal{U}^+}}$ bits of information, while the path has $n_{\varepsilon} = L/\varepsilon$ balls.
    Since the Hamiltonians and pulses are bounded, when traversing along the path (from $I_d$ to $U_{\mathrm{target}}$), the velocity is bounded by some constant $c$.
    In other words, the implementation of $U_{\mathrm{target}}$ (up to precision $\varepsilon$) is not instantaneous.
    Thus, we have the following bound on the maximum information:
    \begin{equation}
    n_{\varepsilon} \log \varepsilon^{-D_{\mathcal{U}^+}} \leq \frac{cT}{\varepsilon} \log \varepsilon^{-D_{\mathcal{U}^+}} = \frac{c\cdot\mathrm{poly}(D_{\mathcal{U}_+})}{\varepsilon} \log \varepsilon^{-D_{\mathcal{U}_+}}\;.
    \end{equation}
    By requiring $C_{\max}\leq n_\varepsilon \log \varepsilon^{-D_{\mathcal{U}^+}}$, we have (given constant $\varepsilon,b_{\max}$ and $c$):
    \begin{equation}
        \mathcal{D} \leq C^\prime \mathrm{poly}(D_{\mathcal{U}_+}),\quad C^\prime = -\frac{c\log \varepsilon}{\varepsilon b_{\max}}\;.
    \end{equation}
    By combining the upper and lower bounds, we have:
    \begin{equation}
        \mathcal{D} = \Theta(\mathrm{poly}(D_{\mathcal{U}^+}))\;.
    \end{equation}
Hence $\mathcal{D}$ has the property we claim.
\end{proof}

In \Cref{prop:size_QOC_problem} we demonstrate the correspondence between the number of parameters $\mathcal{D}$ in a QOC problem and the size of physically interesting unitaries $D_{\mathcal{U}^+}$. Let us assume that an efficient, classical representation of $\mathcal{U}^+$ exists. 
This means that one can use a polynomial number of parameters to describe unitaries in $\mathcal{U}^+$, so the size $\mathcal{D}$ of the QOC problem is also polynomial.
Then, as many polynomial-size QOC problems can be mapped to linear programming problems, which can be solved efficiently with overall polynomial smoothed complexity~\cite{Hernandez1996}, we expect that finding the pulses for these classically representable problems is efficient.

There are many unitaries possessing this property, among which the representative one is the matrix-product operator (MPO) with polynomial bond dimension.
Any MPO on $n$ physical units with $m$-dimensional local Hilbert space can be represented by $O(n m^2 \chi^2)$ parameters, where $\chi$ is the bond dimension~\cite{Vidal_2003_Efficient,Vidal_2004,Verstraete_2004,Verstraete_2004_DMRG}.
As $\chi = O (2^{S_{\mathrm{op}}})$, where $S_{\mathrm{op}}$ is the operator entanglement entropy~\cite{Prosen_2007_Op_EE}, if $S_{\mathrm{op}}\sim \log(n)$, i.e. the unitary is slightly entangled, representing the MPO requires approximately $O(m^2\mathrm{poly}(n))$ parameters.
Therefore, if we restrict to the dynamical process whose operator entanglement entropy is always logarithmic, for a fixed evolution time $T$, we must have
\begin{equation}
    D_{\mathcal{U}^+} \leq D_{\mathcal{U}} = O(Tm^2\mathrm{poly}(n))\;.
\end{equation}
Therefore, those are unitaries that have polynomial complexity in the quantum optimal controls.

As shown previously~\cite{Benjamin_PRL}, if one can turn on and off the nearest-neighbor interactions in an alternating fashion, the globally-controlled system is polynomially equivalent to a typical qubit universal gate set.
Since MPOs with logarithmic entanglement entropy can be mapped to local unitary circuits with logarithmic depth, polynomial time is sufficient to simulate the MPOs.
Indeed, by using dynamical decoupling, one can still realize the qubit universal gate set with a polynomial overhead when the interactions are always on.
These facts support the polynomial smoothed complexity in solving the QOC problems when $\mathcal{U}^+$ has an efficient classical representation.
We leave the more comprehensive analysis for future study.

\section{Universality of fermionic and bosonic simulation} \label{appendix:Universal_fermion}
We investigate the universality of quantum simulation in fermionic and bosonic systems under global control. 
A set of physical operations is termed universal if it can generate the entire unitary group acting on the corresponding Hilbert space. 
Given that all relevant physical operations preserve particle number, we restrict our analysis to fixed-particle-number Hilbert spaces. 
Specifically, we focus on systems composed of $n$ fermions or bosons distributed over $d$ modes. 
These Hilbert spaces are represented by the antisymmetric and symmetric subspaces of $(\mathbb{C}^d)^{\otimes n}$, respectively. 
We denote these spaces as $\mathcal{H}_f = \wedge ^{n}(\mathbb{C}^d)$ for fermions and $\mathcal{H}_b = \mathrm{Sym}^{n}(\mathbb{C}^d)$ for bosons.

We begin by considering spinless particles on lattices containing $N$ sites. 
In this scenario, each lattice site corresponds to a single orbital, thus equating the number of modes $d$ to the number of sites $N$. 
In second quantization, each mode is associated with creation and annihilation operators $c_i^\dagger, c_i$ for fermions and $b_i^\dagger, b_i$ for bosons. 
The Hilbert space structures of $\mathcal{H}_f$ and $\mathcal{H}_b$ are embedded within the canonical commutation relations detailed in \Cref{eqn:canonical_comm_relations}. 
We denote passive fermionic and bosonic linear optics by $\mathrm{LO}_f$ and $\mathrm{LO}_b$, respectively. 
These represent free evolutions characterized by quadratic Hamiltonians~\cite{Bravyi_2002,Terhal_2002,DiVincenzo_2005}:
\begin{equation} \label{eqn:free_fermion_boson_Hamiltonian}
H_{\mathrm{free},f} = \sum_{i,j=1}^N \alpha_{ij} c_i^\dagger c_j, \quad H_{\mathrm{free},b} = \sum_{i,j=1}^N \beta_{ij} b_i^\dagger b_j,
\end{equation}
with arbitrary coefficients which satisfy $\alpha_{ij} = \bar{\alpha}_{ji}$ and $\beta_{ij} = \bar{\beta}_{ji}$ to ensure Hermiticity.

We first investigate universal simulation in spinless fermionic systems. The following lemma provides a sufficient condition for achieving universal control over the Hilbert space $\mathcal{H}_f$ for spinless fermions.

\begin{lemma}[Example 3 of \cite{Oszmaniec_2017}] \label{lemma:Universality_spinless_fermion}
Passive fermionic linear optics $\mathrm{LO}_f$, described by $H_{\mathrm{free},f}$, supplemented by any non-quadratic interacting Hamiltonian $H_{\mathrm{int}}$ containing only two-mode terms, generates the entire unitary group $U(\mathcal{H}_f)$, where $\mathcal{H}_f = \wedge^n(\mathbb{C}^d)$ is the fermionic Hilbert space with fixed particle number $n$ and $d$ modes.
\end{lemma}

By \Cref{lemma:Universality_spinless_fermion}, the combination of the free fermionic Hamiltonian $H_{\mathrm{free},f}$ and a uniform nearest-neighbor Hubbard interaction
\begin{equation} \label{eqn:nearest_Hubbard}
H_{U} = \sum_{i=1}^{N-1} n_i n_{i+1}
\end{equation}
achieves universal control, where $n_i = c_i^\dagger c_i$ is the particle number operator at site $i$. Experimentally, fermions are typically controlled globally using optical superlattices, making the Hubbard interaction in \Cref{eqn:nearest_Hubbard} straightforward to implement. On such experimental platforms, hopping amplitudes and on-site chemical potentials are only globally addressable, resulting in periodic spatial patterns.
We expect these global free Hamiltonians to generate $H_{\mathrm{free},f}$ for universal simulation.
The simplest periodic pattern, corresponding to uniform hopping amplitudes and chemical potentials, as given by
$$H^{(\mathrm{hop})} = \sum_{i=1}^{N-1} \left( c_i^\dagger c_{i+1} + \text{h.c.} \right),\quad \text{and}\quad H^{(\mu)} = \sum_{i=1}^{N} n_i  \;,$$
cannot achieve this goal since $[H^{(\mathrm{hop})},H^{(\mu)}] = 0$.
Thus, inspired by \Cref{qubit:UQC_chain}, we seek the minimal periodic pattern that breaks the lattice reflection symmetry, which is the dual-site alternating pattern as illustrated in the figure below:
\begin{equation}
\begin{gathered}
\includegraphics[width = 0.5\textwidth]{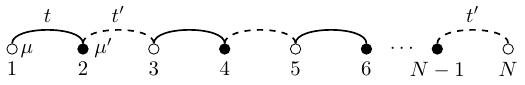}
\end{gathered}
\end{equation}
This alternating arrangement enables the following set of global control Hamiltonians:
\begin{equation} \label{eqn:spinless_fermion_alternating_set}
\begin{aligned}
&H_{\mathrm{odd}}^{(\mathrm{hop})} = \sum_{i \text{ odd}, i\neq N} \left(c_i^\dagger c_{i+1} + \text{h.c.}\right),\quad H_{\mathrm{even}}^{(\mathrm{hop})} = \sum_{i \text{ even}} \left(c_i^\dagger c_{i+1} + \text{h.c.}\right),\ \\
&H_{\mathrm{odd}}^{(\mu)} = \sum_{i \text{ odd}} n_i,\quad H_{\mathrm{even}}^{(\mu)} = \sum_{i \text{ even}} n_i,\quad H_{U} = \sum_{i=1}^{N-1} n_i n_{i+1},
\end{aligned}
\end{equation}
where $H_{\mathrm{odd}/\mathrm{even}}^{(\mathrm{hop})}$ and $H_{\mathrm{odd}/\mathrm{even}}^{(\mu)}$ represent hopping amplitudes and chemical potentials on odd/even bonds and sites, respectively. As demonstrated in \Cref{thm:universal_1d_fermion}, the dual-site controls generated by the Hamiltonians in \Cref{eqn:spinless_fermion_alternating_set} enable universal simulation of spinless fermions.
This again illustrates the connection between geometrical symmetry breaking and universality.

\begin{theorem}[Universality of the spinless Fermi-Hubbard chain] \label{thm:universal_1d_fermion}
Consider an open-boundary spinless Fermi-Hubbard chain with an odd number $N$ of sites. The dual-site control given by \Cref{eqn:spinless_fermion_alternating_set} achieves universal simulation of spinless fermions. Specifically, the unitary evolution generated by the dynamical Lie algebra (DLA) $\mathfrak{g} =\mathrm{span}_{\mathbb{R}}\langle iH_{\mathrm{odd}}^{(\mathrm{hop})},iH_{\mathrm{even}}^{(\mathrm{hop})},iH_{\mathrm{odd}}^{(\mu)},iH_{\mathrm{even}}^{(\mu)},iH_{U}\rangle_{\mathrm{Lie}}$ equals $U(\mathcal{H}_f)$, where $\mathcal{H}_f = \Lambda^n(\mathbb{C}^N)$.
\end{theorem}

Since the global interaction $H_{U}$ is already included in \Cref{eqn:spinless_fermion_alternating_set}, to prove \Cref{thm:universal_1d_fermion}, it suffices to demonstrate that the full passive linear optics Hamiltonian $H_{\mathrm{free},f}$ is generated by the Hamiltonians $H_{\mathrm{odd}/\mathrm{even}}^{(\mathrm{hop})}$ and $H_{\mathrm{odd}/\mathrm{even}}^{(\mu)}$. To accomplish this, we introduce an intermediate generating set, as stated in the following lemma:

\begin{lemma} \label{lemma:fermion_H_free_generators}
The Hamiltonian $H_{\mathrm{free},f}$ is generated by the following set of nearest-neighbor hopping and on-site chemical potential Hamiltonians:
\begin{equation} \label{eqn:local_Hamiltonian_fermions}
H_{i}^{(\mathrm{hop})} = c_i^\dagger c_{i+1} + \text{h.c.},\quad i = 1,\dots,N-1,\quad\text{and}\quad H_{i}^{(\mu)} = n_i = c_i^\dagger c_i,\quad i = 1,\dots, N.
\end{equation}
\end{lemma}

\begin{proof}
    The $H_{\mathrm{free},f}$ in \Cref{eqn:free_fermion_boson_Hamiltonian} is a linear combination of hopping and chemical potential terms, where the nearest-neighbor hoppings are given by $H_{i}^{(\mathrm{hop})}$, and the on-site chemical potential terms are $H_{i}^{(\mu)}$, respectively.
    Thus, it suffices to show that the hopping term $H_{(i,j)}^{(\mathrm{hop})} = c_{i}^\dagger c_j + \text{h.c.}$ between two arbitrary sites $i$ and $j$ can be generated by $H_{i}^{(\mathrm{hop})}$ and $H_{i}^{(\mu)}$. 
    The following algebra gives the hopping term between next-nearest-neighbor sites:
    \begin{equation}
        H_{(i,i+2)}^{(\mathrm{hop})} = [H_{i}^{(\mu)},[H_{i}^{(\mathrm{hop})},H_{{i+1}}^{(\mathrm{hop})}]] = c_{i}^\dagger c_{i+2} + \text{h.c.}\;.
    \end{equation}
    Then, given the hopping between $i$ and $i+k$ sites for an arbitrary $k$ as $H_{t,(i,i+k)} = c_i^\dagger c_{i+k} + \text{h.c.}$, we have
    \begin{equation}
        H_{(i,i+k+1)}^{(\mathrm{hop})} = [H_{i}^{(\mu)},[H_{(i,i+k)}^{(\mathrm{hop})}, H_{i+k}^{(\mathrm{hop})}]] = c_i^\dagger c_{i+k+1} + \text{h.c.} \;.
    \end{equation}
    By induction, any hopping term $H_{(i,j)}^{(\mathrm{hop})}$ can be generated by repeating the procedure, which completes the proof.
\end{proof}
Now we complete the proof of \Cref{thm:universal_1d_fermion} by generating the Hamiltonians in \Cref{eqn:local_Hamiltonian_fermions}.

\begin{proof}[Proof of \Cref{thm:universal_1d_fermion}]
    By Lemmas~\ref{lemma:Universality_spinless_fermion} and \ref{lemma:fermion_H_free_generators}, it suffices to prove that $H_{\mathrm{odd}/\mathrm{even}}^{(\mathrm{hop})}$ and $H_{\mathrm{odd}/\mathrm{even}}^{(\mu)}$ can generate all hopping and chemical potential Hamiltonians $H_{i}^{(\mathrm{hop})}$ and $H_{i}^{(\mu)}$ in \eqref{eqn:local_Hamiltonian_fermions}. This is realized by repeatedly calculating the commutator between $H_{\mathrm{odd}/\mathrm{even}}^{(\mathrm{hop})}$ and $H_{\mathrm{odd}/\mathrm{even}}^{(\mu)}$ and hierarchically singling out the Hamiltonians $H_{i}^{(\mathrm{hop})}$ and $H_{i}^{(\mu)}$ from both ends of the one-dimensional chain.
    In calculating the commutators, we repeatedly use the second relation in \Cref{eqn:useful_comm_relations} and the fermionic canonical commutation relation (\Cref{eqn:canonical_comm_relations}).

    We start by obtaining $H_{1}^{(\mu)} = n_1$. The following algebras are helpful:
    \begin{equation}
        \begin{aligned}
        & [H_{\mathrm{even}}^{(\mathrm{hop})},H_{\mathrm{odd}}^{(\mu)}] = \sum_{i\in \mathrm{even}}[c_i^\dagger c_{i+1} + \text{h.c.}, n_{i+1} ] = \sum_{i\in\mathrm{even}}\left(c_i^\dagger c_{i+1} - c_{i+1}^\dagger c_i\right)\\
        & [H_{\mathrm{odd}}^{(\mu)},[H_{\mathrm{even}}^{(\mathrm{hop})},H_{\mathrm{odd}}^{(\mu)}]] \propto \sum_{i\in \mathrm{even}} \left(c_i^\dagger c_{i+1} + c_{i+1}^\dagger c_i \right)\\
        &H_{a,1} = \sum_{i \in \mathrm{even}} \left(n_i - n_{i+1}\right) \propto [[H_{\mathrm{odd}}^{(\mu)},[H_{\mathrm{even}}^{(\mathrm{hop})},H_{\mathrm{odd}}^{(\mu)}]],[H_{\mathrm{even}}^{(\mathrm{hop})},H_{\mathrm{odd}}^{(\mu)}]] \;.
    \end{aligned}
    \end{equation}
    Notice that the ancillary Hamiltonian $H_{a,1}$ is a linear combination of number operators on different sites, except for $n_1$. Thus, we can single out $n_1$ by:
    \begin{equation}
        H_{1}^{(\mu)} = H_{\mathrm{odd}}^{(\mu)} - H_{\mathrm{even}}^{(\mu)} + H_{a,1} = n_1\;.
    \end{equation}
    With $H_{1}^{(\mu)}$, one can single out $H_{1}^{(\mathrm{hop})}$ via:
    \begin{equation}
        H_{1}^{(\mathrm{hop})}= c_1^\dagger c_2 + c_2^\dagger c_1 \propto [H_{1}^{(\mu)},[H_{\mathrm{odd}}^{(\mathrm{hop})},H_{1}^{(\mu)}]] \;.
    \end{equation}
    Similarly, one can single out $H_{N}^{(\mu)}$ and $H_{N-1}^{(\mathrm{hop})}$ by the following steps:
    \begin{align}
        H_{a,N} &=\sum_{i\in \mathrm{odd},i\neq N} (n_{i+1} - n_i) \propto [[H_{\mathrm{odd}}^{(\mu)},[H_{\mathrm{odd}}^{(\mathrm{hop})},H_{\mathrm{odd}}^{(\mu)}]],[H_{\mathrm{odd}}^{(\mathrm{hop})},H_{\mathrm{odd}}^{(\mu)}]]  \\
        H_{N}^{(\mu)} &= H_{\mathrm{odd}}^{(\mu)} - H_{\mathrm{even}}^{(\mu)} + H_{a,N} = n_N\;,\quad H_{N-1}^{(\mathrm{hop})} = [H_{N}^{(\mu)},[H_{\mathrm{even}}^{(\mathrm{hop})},H_{N}^{(\mu)}]] = c_{N-1}^\dagger c_N + c_N^\dagger c_{N-1}\;.
    \end{align}
    Then, one can define a set of new generators as:
    \begin{align}
        &\widetilde{H}_{\mathrm{odd}}^{(\mathrm{hop})} = H_{\mathrm{odd}}^{(\mathrm{hop})} - H_{1}^{(\mathrm{hop})} = \sum_{i\in \mathrm{odd}, i\neq 1, N} c_i^\dagger c_{i+1} + h.c.\;,\\
        &\widetilde{H}_{\mathrm{even}}^{(\mathrm{hop})} = H_{\mathrm{even}}^{(\mathrm{hop})} - H_{N-1}^{(\mathrm{hop})} = \sum_{i \in \mathrm{even}, i\neq N-1} c_i^\dagger c_{i+1} + h.c.\;,\\
        &\widetilde{H}_{\mathrm{odd}}^{(\mu)} = H_{\mathrm{odd}}^{(\mu)} - H_{1}^{(\mu)} - H_{N}^{(\mu)} = \sum_{i\in \mathrm{odd}, i \neq 1,N} n_i\;,\\
        &\widetilde{H}_{\mathrm{even}}^{(\mu)} = H_{\mathrm{even}}^{(\mu)} = \sum_{i\in \mathrm{even}} n_i \;.
    \end{align}
    These generators define a new one-dimensional chain whose two ends are the $2$-nd and $(N-1)$-th sites in the original chain. Equivalently, this new chain can be obtained by removing the $1$-st and $N$-th sites of the old chain. In the new chain, $\widetilde{H}_{\mathrm{odd}}^{(\mathrm{hop})},\widetilde{H}_{\mathrm{even}}^{(\mathrm{hop})},\widetilde{H}_{\mathrm{odd}}^{(\mu)}$ and $\widetilde{H}_{\mathrm{even}}^{(\mu)}$ play the same role as $H_{\mathrm{even}}^{(\mathrm{hop})}, H_{\mathrm{odd}}^{(\mathrm{hop})}, H_{\mathrm{even}}^{(\mu)}, H_{\mathrm{odd}}^{(\mu)}$ in the previous chain, i.e., the odd and even sites are exchanged. Thus, one can replace the Hamiltonians in the previous steps correspondingly to single out $H_{2}^{(\mathrm{hop})},H_{2}^{(\mu)},H_{N-2}^{(\mathrm{hop})},H_{N-1}^{(\mu)}$ as:
    \begin{align}
        &H_{2}^{(\mu)} = n_2 \propto\widetilde{H}_{\mathrm{even}}^{(\mu)} - \widetilde{H}_{\mathrm{odd}}^{(\mu)} + \frac{1}{2}[[\widetilde{H}_{\mathrm{even}}^{(\mu)},[\widetilde{H}_{\mathrm{odd}}^{(\mathrm{hop})},\widetilde{H}_{\mathrm{even}}^{(\mu)}]],[\widetilde{H}_{\mathrm{odd}}^{(\mathrm{hop})},\widetilde{H}_{\mathrm{even}}^{(\mu)}]] \;,\\
        &H_{2}^{(\mathrm{hop})} = c_2^\dagger c_3 + c_3^\dagger c_2 \propto [H_{2}^{(\mu)},[\widetilde{H}_{\mathrm{even}}^{(\mathrm{hop})},H_{2}^{(\mu)}]] \;,\\
        &H_{N-1}^{(\mu)} = n_{N-1} \propto \widetilde{H}_{\mathrm{even}}^{(\mu)} - \widetilde{H}_{\mathrm{odd}}^{(\mu)} + \frac{1}{2}[[\widetilde{H}_{\mathrm{even}}^{(\mu)}, [\widetilde{H}_{\mathrm{even}}^{(\mathrm{hop})},\widetilde{H}_{\mathrm{even}}^{(\mu)}]],[\widetilde{H}_{\mathrm{even}}^{(\mathrm{hop})},\widetilde{H}_{\mathrm{even}}^{(\mu)}]] \;,\\
        &H_{N-2}^{(\mathrm{hop})} = c_{N-2}^\dagger c_{N-1} + c_{N-1}c_{N-2}^\dagger \propto [H_{N-1}^{(\mu)},[\widetilde{H}_{\mathrm{odd}}^{(\mathrm{hop})},H_{N-1}^{(\mu)}]] \;. 
    \end{align}
    By repeating the above procedure to the middle of the chain, we can single out all $H_{i}^{(\mathrm{hop})}$ and $H_{i}^{(\mu)}$ in Lemma~\ref{lemma:fermion_H_free_generators}. Therefore, we have completed the proof of Theorem~\ref{thm:universal_1d_fermion}.
\end{proof}

The proof for the universality of spinless bosonic systems is analogous.  
Similar to \Cref{lemma:Universality_spinless_fermion}, we have the following sufficient condition for achieving universal control over $\mathcal{H}_b$.
\begin{lemma}[Adapted from Theorem 2 of \cite{Oszmaniec_2017}] \label{lemma:Universality_spinless_boson}
For bosonic systems with $d>2$ modes, passive bosonic linear optics described by $H_{\mathrm{free},b}$, supplemented by any Hamiltonian $H$ that is not contained within the linear optics set $\mathrm{LO}_b$, generates the entire unitary group $U(\mathcal{H}_b)$, where $\mathcal{H}_b = \mathrm{Sym}^n(\mathbb{C}^d)$ is the bosonic Hilbert space with fixed particle number $n$.
\end{lemma}

Since $d = N$, and for arbitrary many-body systems we always have $N > 2$, the assumption of \Cref{lemma:Universality_spinless_boson} is always satisfied. In bosonic systems, the Bose-Hubbard chain with dual-site control can be readily implemented using optical superlattices, providing the following control Hamiltonians analogous to their fermionic counterparts in \Cref{eqn:spinless_fermion_alternating_set}:
\begin{equation} \label{eqn:spinless_boson_alternating_set}
\begin{aligned}
&H_{\mathrm{odd}}^{(\mathrm{hop})} = \sum_{i \text{ odd}, i\neq N} \left(b_i^\dagger b_{i+1} + \text{h.c.}\right),\quad H_{\mathrm{even}}^{(\mathrm{hop})} = \sum_{i \text{ even}} \left(b_i^\dagger b_{i+1} + \text{h.c.}\right),\ \\
&H_{\mathrm{odd}}^{(\mu)} = \sum_{i \text{ odd}} n_i,\quad H_{\mathrm{even}}^{(\mu)} = \sum_{i \text{ even}} n_i,\quad H_{U} = \sum_{i=1}^{N-1} n_i (n_{i}-1).
\end{aligned}
\end{equation}
As established in \Cref{thm:universal_1d_boson}, the dual-site controls defined by \Cref{eqn:spinless_boson_alternating_set} enable universal simulation of spinless bosons.

\begin{theorem}[Universality of the spinless Bose-Hubbard chain] \label{thm:universal_1d_boson}
Consider an open-boundary spinless Bose-Hubbard chain with an odd number of sites $N>2$. The dual-site control provided by \Cref{eqn:spinless_boson_alternating_set} achieves universal simulation of spinless bosons. Specifically, the DLA $\mathfrak{g} =\mathrm{span}_{\mathbb{R}}\langle iH_{\mathrm{odd}}^{(\mathrm{hop})},iH_{\mathrm{even}}^{(\mathrm{hop})},iH_{\mathrm{odd}}^{(\mu)},iH_{\mathrm{even}}^{(\mu)},iH_{U}\rangle_{\mathrm{Lie}}$ equals $U(\mathcal{H}_b)$ generates the full unitary group $U(\mathcal{H}_b)$, where $\mathcal{H}_b = \mathrm{Sym}^n(\mathbb{C}^N)$.
\end{theorem}
\begin{proof}
The proof follows closely the approach used in \Cref{thm:universal_1d_fermion}. Since the Hubbard Hamiltonian $H_{U}$ is not contained within the set of passive linear optics $\mathrm{LO}_b$, \Cref{lemma:Universality_spinless_boson} implies that it suffices to demonstrate the generation of $H_{\mathrm{free},b}$ from the Hamiltonians $H_{\mathrm{odd/even}}^{(\mathrm{hop})}$ and $H_{\mathrm{odd/even}}^{(\mu)}$. By conducting a calculation analogous to the one used in the proof of \Cref{thm:universal_1d_fermion}, but substituting the anticommutation relations of fermionic operators $c^\dagger,c$ with the commutation relations of bosonic operators $b^\dagger,b$, and employing the second identity in \Cref{eqn:useful_comm_relations}, we obtain the following intermediate set of nearest-neighbor hopping and on-site chemical potential Hamiltonians:
\begin{equation} \label{eqn:local_Hamiltonian_bosons}
H_{i}^{(\mathrm{hop})} = b_i^\dagger b_{i+1} + \text{h.c.},\quad i = 1,\dots,N-1,\quad\text{and}\quad H_{i}^{(\mu)} = n_i = b_i^\dagger b_i,\quad i = 1,\dots, N.
\end{equation}
Finally, following a similar argument as in \Cref{lemma:fermion_H_free_generators}, one can verify that this set indeed generates $H_{\mathrm{free},b}$, thereby completing the proof.
\end{proof}

The physically relevant fermionic systems are spinful, requiring a generalization of \Cref{thm:universal_1d_fermion} to the spinful Fermi-Hubbard chain. The generators for the one-dimensional spinful Fermi-Hubbard model on $N$ sites (with $N$ odd), under open boundary conditions and dual-site superlattice control, are given by: 
\begin{equation} \label{eqn:spinful_fermion_alternating_set}
\begin{aligned}
    &H_{\mathrm{odd}}^{(\mathrm{hop})} = \sum_{i \in \mathrm{odd}, i\neq N,\sigma} c_{i,\sigma}^\dagger c_{i+1,\sigma} + \text{h.c.},\quad H_{\mathrm{even}}^{(\mathrm{hop})} = \sum_{i \in \mathrm{even},\sigma} c_{i,\sigma}^\dagger c_{i+1,\sigma} + \text{h.c.},\quad\\ &H_{\mathrm{odd}}^{(\mu)} = \sum_{i\in \mathrm{odd}} n_{i,\uparrow}+n_{i,\downarrow},\quad H_{\mathrm{even}}^{(\mu)} = \sum_{i\in\mathrm{even}} n_{i,\uparrow}+n_{i,\downarrow},\quad\\
    &H_{B_X} = \sum_{i=1}^N c_{i,\uparrow}^{\dagger}c_{i,\downarrow}+h.c.,\quad H_{B_Z} = \sum_{i=1}^N(a*i+b)(n_{i,\uparrow}-n_{i,\downarrow}),\quad H_{U} = \sum_{i=1}^N n_{i,\uparrow} n_{i,\downarrow}\;.
    \end{aligned}
\end{equation}
Here, $H_{\mathrm{odd}}^{(\mathrm{hop})}$ and $H_{\mathrm{even}}^{(\mathrm{hop})}$ describe spin-preserving hopping on odd and even bonds for $\sigma=\uparrow,\downarrow$. The terms $H_{\mathrm{odd}}^{(\mu)}$ and $H_{\mathrm{even}}^{(\mu)}$ correspond to chemical potentials on odd and even sites, respectively. $H_{B_X}$ represents a global magnetic field along the $x$-axis, while $H_{B_Z}$ corresponds to a spatially tilted magnetic field along the $z$-axis. Finally, $H_{U}$ denotes the on-site Hubbard interaction.

As stated in \Cref{thm:universal_1d_spinful_fermion}, the dual-site controls defined by the Hamiltonians in \Cref{eqn:spinful_fermion_alternating_set} enable universal simulation of spinful fermions. A key distinction in the spinful case is that the number of modes becomes $d = 2N$, as the spin degree of freedom at each site doubles the number of modes.

\begin{theorem}[Universality of the spinful Fermi-Hubbard chain] \label{thm:universal_1d_spinful_fermion}
Consider an open-boundary spinful Fermi-Hubbard chain with an odd number $N$ of sites. The dual-site control given by \Cref{eqn:spinful_fermion_alternating_set} achieves universal simulation of spinful fermions. Specifically, the unitary evolution generated by the dynamical Lie algebra (DLA) $\mathfrak{g} =\mathrm{span}_{\mathbb{R}}\langle iH_{\mathrm{odd}}^{(\mathrm{hop})},iH_{\mathrm{even}}^{(\mathrm{hop})},iH_{\mathrm{odd}}^{(\mu)},iH_{\mathrm{even}}^{(\mu)},iH_{B_X},iH_{B_Z},iH_{U}\rangle_{\mathrm{Lie}}$ equals $U(\mathcal{H}_f)$, where $\mathcal{H}_f = \Lambda^n(\mathbb{C}^{2N})$.
\end{theorem}

\begin{proof}
By \Cref{lemma:Universality_spinless_fermion,lemma:fermion_H_free_generators}, it is sufficient to show that the set $\{iH_{\mathrm{odd}}^{(\mathrm{hop})},iH_{\mathrm{even}}^{(\mathrm{hop})},iH_{\mathrm{odd}}^{(\mu)},iH_{\mathrm{even}}^{(\mu)},iH_{B_X},iH_{B_Z}\}$ can generate arbitrary free fermion gates. Without specification, we will choose $a=1,b=0$ for $H_{B_Z}$ to realize the tilted $z$-field.

Starting with
\eqs{
[H_{\mathrm{even}}^{(\mathrm{hop})},i[H_{\mathrm{even}}^{(\mathrm{hop})},H_{\mathrm{odd}}^{(\mu)}]]=\sum_{i\in \mathrm{even},\sigma}n_{i,\sigma}-n_{i+1,\sigma},
}
we can single out $\sum_{\sigma}n_{1,\sigma}$ by the algebra:
\eqs{ \label{eqn:n_1_sigma}
\sum_{\sigma}n_{1,\sigma}=H_{\mathrm{odd}}^{(\mu)}-H_{\mathrm{even}}^{(\mu)}+[H_{\mathrm{even}}^{(\mathrm{hop})},[H_{\mathrm{even}}^{(\mathrm{hop})},H_{\mathrm{odd}}^{(\mu)}]].
}
By defining the following ancillary Hamiltonians:
\begin{equation}
    \begin{aligned}
        H_{a,1} &= H_{B_Z}+H_{\mathrm{odd}}^{(\mu)}+2H_{\mathrm{even}}^{(\mathrm{hop})}\\
    &=2n_{1,\uparrow}+4n_{2,\uparrow}+\sum_{i\in\mathrm{odd},i\neq 1}(i+1)n_{i,\uparrow}-(i-1)n_{i,\downarrow}+\sum_{i\in \mathrm{even},i\neq 2}(i+2)n_{i,\uparrow}-(i-2)n_{i,\downarrow},\\
    H_{a,2}  &= [H_{\mathrm{odd}}^{(\mathrm{hop})},[H_{\mathrm{odd}}^{(\mathrm{hop})},\sum_{\sigma}n_{1,\sigma}]]=\sum_{\sigma}(c^{\dagger}_{1,\sigma}c_{2,\sigma}+c^{\dagger}_{2,\sigma}c_{1,\sigma})\;,
    \end{aligned}
\end{equation}
we can obtain $n_{2,\uparrow}-n_{1,\uparrow}$ by
\begin{equation}
    [H_{a,2},[H_{a,2},H_{a,1}]]=2(n_{2,\uparrow}-n_{1,\uparrow})\;.
\end{equation}
Then, we can use $n_{2,\uparrow}-n_{1,\uparrow}$ and the uniform $B_X$ field to single out the $B_X$ field on sites $1$ and $2$ as:
\eqs{
H_{a,3}=[n_{2,\uparrow}-n_{1,\uparrow},[H_{B_X},n_{2,\uparrow}-n_{1,\uparrow}]]=c^{\dagger}_{1,\uparrow}c_{1,\downarrow}+c^{\dagger}_{1,\downarrow}c_{1,\uparrow}+c^{\dagger}_{2,\uparrow}c_{2,\downarrow}+c^{\dagger}_{2,\downarrow}c_{2,\uparrow}.
}
Using the uniform $B_Z$ field by selecting $a =0, b = 1$, i.e., $\sum_{i}n_{i,\uparrow}-n_{i,\downarrow}$, we construct another two ancillary Hamiltonians as:
\begin{equation}
    \begin{aligned}
        H_{a,4}&=\sum_{i}n_{i,\uparrow}-n_{i,\downarrow}+\sum_{\sigma}n_{1,\sigma}+\sum_{\sigma}n_{2,\sigma}=2n_{1,\uparrow}+2n_{2,\uparrow}+\sum_{i\neq 1,2}n_{i,\uparrow}-n_{i,\downarrow},\\
        H_{a,5}&=2H_{a,4}-H_{a,1} = 2n_{1,\uparrow}-\sum_{i\in\mathrm{odd},i\neq 1}((i+3)n_{i,\uparrow}-(i-3)n_{i,\downarrow})-\sum_{i\in \mathrm{even},i\neq 2}(in_{i,\uparrow}-in_{i,\downarrow}).
    \end{aligned}
\end{equation}
Then we can single out the spin flip on the first site as:
\eqs{
H_{a,6}=[H_{a,5},[H_{a,3},H_{a,5}]]=c^{\dagger}_{1,\uparrow}c_{1,\downarrow}+c^{\dagger}_{1,\downarrow}c_{1,\uparrow}\;,
}
which gives: 
\eqs{
H_{a,7}=[[H_{a,3},H_{a,5}],H_{a,6}]=n_{1,\uparrow}-n_{1,\downarrow}\;.
}
With $n_{1,\uparrow}+n_{1,\downarrow}$ (see \Cref{eqn:n_1_sigma}) and $n_{1,\uparrow}-n_{1,\downarrow}$, we can control $n_{1,\uparrow}$ and $n_{1,\downarrow}$ separately. 
Then, we can get the hopping between sites 1 and 2 for spin-up and spin-down individually as:
\eqs{
&[[H_{\mathrm{odd}}^{(\mathrm{hop})},n_{1,\uparrow}],n_{1,\uparrow}]=c^{\dagger}_{1,\uparrow}c_{2,\uparrow}+c^{\dagger}_{2,\uparrow}c_{1,\uparrow},\\
&[[H_{\mathrm{odd}}^{(\mathrm{hop})},n_{1,\downarrow}],n_{1,\downarrow}]=c^{\dagger}_{1,\downarrow}c_{2,\downarrow}+c^{\dagger}_{2,\downarrow}c_{1,\downarrow}.
}
Therefore, we have all individual controls on site 1 being constructed: $n_{1,\uparrow}$, $n_{1,\downarrow}$, $c^{\dagger}_{1,\uparrow}c_{2,\uparrow}+c^{\dagger}_{2,\uparrow}c_{1,\uparrow}$, $c^{\dagger}_{1,\downarrow}c_{2,\downarrow}+c^{\dagger}_{2,\downarrow}c_{1,\downarrow}$ and $c^{\dagger}_{1,\uparrow}c_{1,\downarrow}+c^{\dagger}_{1,\downarrow}c_{1,\uparrow}$. Following the same constructive method, one can get all individual $n_{i,\uparrow/\downarrow}$ and hopping terms. By \Cref{lemma:fermion_H_free_generators}, we show $\{iH_{\mathrm{odd}}^{(\mathrm{hop})},iH_{\mathrm{even}}^{(\mathrm{hop})},iH_{\mathrm{odd}}^{(\mu)},iH_{\mathrm{even}}^{(\mu)},iH_{B_X},iH_{B_Z}\}$ can generate all free fermion operations. Together with the onsite Hubbard interaction $H_{U}$ and \Cref{lemma:Universality_spinless_fermion}, we show that this globally controlled spinful fermionic chain is universal.

\end{proof}

\section{Next-Nearest-Neighbor Hoppings in the Fermi-Hubbard model} \label{appendix:NNN_hopping}
\begin{figure}[t]
    \centering
    \includegraphics[width=\linewidth]{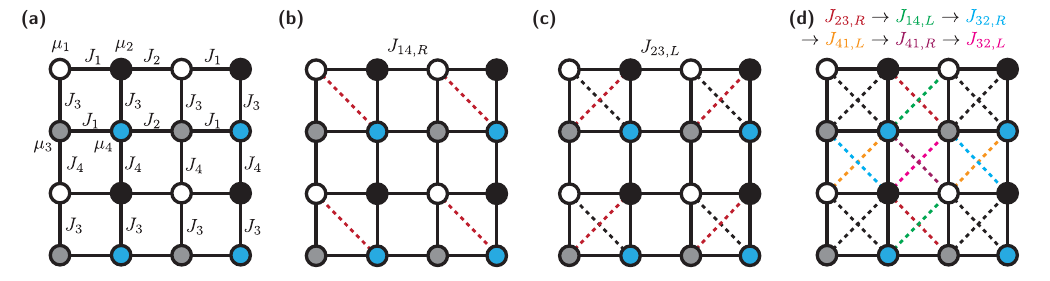}
    \caption{Implementation of next-nearest-neighbor (NNN) hoppings in the two-dimensional Fermi-Hubbard model using global control.  
    (a) A two-dimensional superlattice with four atomic species, labeled \(1\)–\(4\) and color-coded as white, black, gray, and blue, respectively. The control Hamiltonians consist of species-dependent chemical potentials \(\mu_1 \sim \mu_4\) and nearest-neighbor (NN) hopping terms \(J_1 \sim J_4\), indicated on the corresponding bonds.  
    (b) Realization of the rightward NNN hopping \(J_{14,R}\) between type-1 and type-4 sites, shown by red dashed lines.  
    (c) Given \(J_{14,R}\), implementation of the leftward NNN hopping \(J_{23,L}\) between type-2 and type-3 sites (also in red dashed lines).  
    (d) Sequential construction of all NNN hoppings \(J_{23,R}, J_{14,L}, J_{32,R}, J_{41,L}, J_{41,R}, J_{32,L}\), indicated by colored dashed lines corresponding to the hopping types.
    }
    \label{fig:NNN_FH_model}
\end{figure}

The two-dimensional Fermi-Hubbard model with next-nearest-neighbor (NNN) hoppings is widely regarded as a promising platform for capturing key features of high-temperature superconductivity. Here, we propose a protocol to implement such NNN hoppings in an analog fermionic quantum simulator based on four atomic species, as illustrated in Figure~\ref{fig:NNN_FH_model}.  

Figure~\ref{fig:NNN_FH_model}(a) presents the superlattice geometry and the corresponding global control scheme, using a \(3\times3\) lattice as a representative example. The superlattice consists of four distinct site types, indicated by white, black, gray, and blue nodes. Each site type \(i\) is subjected to a homogeneous, species-specific chemical potential \(\mu_i\). The nearest-neighbor hopping pattern is periodic: horizontal hoppings from odd (even) columns to the adjacent rightward even (odd) columns are controlled by \(J_1\) (\(J_2\)), while vertical hoppings from odd (even) rows to neighboring downward even (odd) rows are governed by \(J_3\) (\(J_4\)). This structured pattern of hopping and on-site potentials is particularly natural for implementation in fermionic optical superlattices~\cite{superlattice_experiment}.

We focus on the spinless fermion case, as the absence of a magnetic field implies that spin-up and spin-down components are dynamically decoupled. The spinful case can be straightforwardly generalized by including spin degrees of freedom. To implement the next-nearest-neighbor (NNN) hopping terms, we employ a set of global control Hamiltonians corresponding to site-dependent chemical potentials (\( H_{i}^{(\mu)},\, i = 1, \dots, 4 \)) and nearest-neighbor (NN) hopping operators (\( H_{i}^{(\mathrm{hop})},\, i = 1, \dots, 4 \)), defined as:
\begin{equation}
\begin{aligned}
    &H_{1}^{(\mu)} = \sum_{i \in \mathrm{odd}} \sum_{j \in \mathrm{odd}} n_{i,j},\quad 
    H_{2}^{(\mu)} = \sum_{i \in \mathrm{odd}} \sum_{j \in \mathrm{even}} n_{i,j},\quad 
    H_{3}^{(\mu)} = \sum_{i \in \mathrm{even}} \sum_{j \in \mathrm{odd}} n_{i,j},\quad 
    H_{4}^{(\mu)} = \sum_{i \in \mathrm{even}} \sum_{j \in \mathrm{even}} n_{i,j},\\
    &H_{1}^{(\mathrm{hop})} = \sum_i \sum_{j \in \mathrm{odd}} \left( c_{i,j+1}^\dagger c_{i,j} + \mathrm{h.c.} \right),\quad 
    H_{2}^{(\mathrm{hop})} = \sum_i \sum_{j \in \mathrm{even}} \left( c_{i,j+1}^\dagger c_{i,j} + \mathrm{h.c.} \right),\\
    &H_{3}^{(\mathrm{hop})} = \sum_{i \in \mathrm{odd}} \sum_j \left( c_{i+1,j}^\dagger c_{i,j} + \mathrm{h.c.} \right),\quad 
    H_{4}^{(\mathrm{hop})} = \sum_{i \in \mathrm{even}} \sum_j \left( c_{i+1,j}^\dagger c_{i,j} + \mathrm{h.c.} \right).
\end{aligned}
\end{equation}

To engineer the rightward NNN hopping between type-1 and type-4 sites, as shown in Figure~\ref{fig:NNN_FH_model}(b), we begin by isolating the appropriate NN hopping terms using commutators with the chemical potential Hamiltonians:
\begin{align}
    H_{1}^{(\mathrm{hop})\prime} &= i[H_{2}^{(\mu)}, H_{1}^{(\mathrm{hop})}] 
    = i \sum_{i \in \mathrm{odd}} \sum_{j \in \mathrm{odd}} \left( c_{i,j+1}^\dagger c_{i,j} - c_{i,j}^\dagger c_{i,j+1} \right),\\
    H_{3}^{(\mathrm{hop})\prime} &= i[H_{4}^{(\mu)}, H_{3}^{(\mathrm{hop})}] 
    = i \sum_{i \in \mathrm{odd}} \sum_{j \in \mathrm{even}} \left( c_{i+1,j}^\dagger c_{i,j} - c_{i,j}^\dagger c_{i+1,j} \right).
\end{align}
Their nested commutator yields an imaginary NNN hopping:
\begin{equation}
    H_{14,R}^{(\mathrm{hop})\prime} = i[H_{1}^{(\mathrm{hop})\prime}, H_{3}^{(\mathrm{hop})\prime}] = i \sum_{i \in \mathrm{odd}} \sum_{j \in \mathrm{odd}} \left( c_{i+1,j+1}^\dagger c_{i,j} - c_{i,j}^\dagger c_{i+1,j+1} \right).
\end{equation}
To obtain a Hermitian NNN hopping, we eliminate the imaginary prefactor using a final commutator with a chemical potential term:
\begin{equation}
    H_{14,R}^{(\mathrm{NNN})} = i[H_{1}^{(\mu)}, H_{14,R}^{(\mathrm{hop})\prime}] = \sum_{i \in \mathrm{odd}} \sum_{j \in \mathrm{odd}} \left( c_{i+1,j+1}^\dagger c_{i,j} + c_{i,j}^\dagger c_{i+1,j+1} \right).
\end{equation}
Thus, the rightward NNN hopping term \( H_{14R}^{(\mathrm{NNN})} \) is obtained via the triple nested commutator:
\begin{equation}
    H_{14,R}^{(\mathrm{NNN})} = [H_{1}^{(\mu)}, [[H_{2}^{(\mu)}, H_{1}^{(\mathrm{hop})}], [H_{4}^{(\mu)}, H_{3}^{(\mathrm{hop})}]]].
\end{equation}

Similarly, the leftward NNN hopping between type-2 and type-3 sites, depicted in Figure~\ref{fig:NNN_FH_model}(c), is given by:
\begin{equation}
    H_{23,L}^{(\mathrm{NNN})} = [H_{2}^{(\mu)}, [[H_{1}^{(\mu)}, H_{1}^{(\mathrm{hop})}], [H_{3}^{(\mu)}, H_{3}^{(\mathrm{hop})}]]].
\end{equation}
Together, these two terms realize the NNN hoppings on a single plaquette, and the remaining six NNN hoppings can be obtained analogously using the following expressions:
\begin{align}
    H_{23,R}^{(\mathrm{NNN})} &= [H_{2}^{(\mu)}, [[H_{1}^{(\mu)}, H_{2}^{(\mathrm{hop})}], [H_{3}^{(\mu)}, H_{3}^{(\mathrm{hop})}]]],\\
    H_{14,L}^{(\mathrm{NNN})} &= [H_{1}^{(\mu)}, [[H_{2}^{(\mu)}, H_{2}^{(\mathrm{hop})}], [H_{4}^{(\mu)}, H_{3}^{(\mathrm{hop})}]]],\\
    H_{32,R}^{(\mathrm{NNN})} &= [H_{3}^{(\mu)}, [[H_{4}^{(\mu)}, H_{1}^{(\mathrm{hop})}], [H_{2}^{(\mu)}, H_{4}^{(\mathrm{hop})}]]],\\
    H_{41,L}^{(\mathrm{NNN})} &= [H_{4}^{(\mu)}, [[H_{3}^{(\mu)}, H_{1}^{(\mathrm{hop})}], [H_{1}^{(\mu)}, H_{4}^{(\mathrm{hop})}]]],\\
    H_{41,R}^{(\mathrm{NNN})} &= [H_{4}^{(\mu)}, [[H_{3}^{(\mu)}, H_{2}^{(\mathrm{hop})}], [H_{1}^{(\mu)}, H_{4}^{(\mathrm{hop})}]]],\\
    H_{32,L}^{(\mathrm{NNN})} &= [H_{3}^{(\mu)}, [[H_{4}^{(\mu)}, H_{2}^{(\mathrm{hop})}], [H_{2}^{(\mu)}, H_{4}^{(\mathrm{hop})}]]].
\end{align}

These expressions follow a clear structural rule: for a given NNN hopping \( H_{ij,L/R}^{(\mathrm{NNN})} \), we identify a two-step path from site \(i\) to \(j\) via an intermediate site \(k\), such that the three chemical potential Hamiltonians correspond to \(H_{i}^{(\mu)}, H_{k}^{(\mu)}, H_{j}^{(\mu)}\), and the NN hopping operators correspond to the two bonds along the path. The order of commutators reflects this path structure.

With all NNN hopping terms constructed, we can implement the full NNN Fermi-Hubbard Hamiltonian using Trotterized evolution, as described in \Cref{lemma:Linear_Combo_Generators,lemma:Commutator_Generators}. This approach enables the efficient analogue simulation of strongly correlated electron systems, including candidate models for high-temperature superconductivity.

\section{Symmetry-Protected-Topological Phases\label{app:spt}}
The effective Hamiltonian we focus on is the cluster-Ising model:
\eqs{
H_{\text{ZXZ}}=J_{\text{eff}}\sum_{i}Z_{i-1}X_{i}Z_{i+1},\label{eq:zxz_app}
}
which serves as a pedagogical example of a system realizing a symmetry-protected topological (SPT) phase. Through the Jordan-Wigner transformation, this model can be mapped to a pair of decoupled Kitaev chains \cite{KitaevChain, ruben_spt}, supporting a topological edge mode, as illustrated in \Cref{fig:fig2}(a). When the total number of sites $N$ is even, the Hamiltonian in \Cref{eq:zxz_app} respects a $\mathbb{Z}_2 \times \mathbb{Z}_2$ symmetry generated by the operators $P_1 = X_1 X_3 X_5 \cdots X_{N-1}$ and $P_2 = X_2 X_4 Z_6 \cdots X_N$. Although $P_1$ and $P_2$ are non-local, their action in the ground-state subspace is effectively localized near the boundaries. Specifically, they reduce to boundary operators: $P_{1}^{(L)}=X_{1}Z_{2}$, $P_{1}^{(R)}=Z_{N}$ and $P_{2}^{(L)}=Z_1$, $P_2^{(R)}=Z_{N-1}X_{N}$. Notably, the boundary operators $P_1^{(L)}$ and $P_2^{(L)}$ are anti-commuting symmetries on the boundary, which confirms the presence of degenerate topological edge modes in the cluster-Ising model. In the following, we show that using quantum optimal control, it is possible to design a sequence of global control pulses $u_{\alpha}^{(\tau)}(t)$ such that 
\eqs{
\mathcal{T}\left[e^{-i\int dt \sum_{\alpha}u_{\alpha}^{(\tau)}(t)H_{\alpha} }\right]=e^{-i \tau H_{\mathrm{ZXZ}}}
}
realizing the target unitary evolution under the cluster-Ising 

\subsection{The limitation of quantum simulation in the blockade regime} \label{appendix:limitation_blockade}
Most current quantum simulation experiments using Rydberg atoms operate in the blockade regime, where the $V_{jl}$ in \Cref{eq:Rydberg_Hamiltonian} is sufficiently large for adjacent sites $j$ and $l$, so the neighboring atoms cannot be excited simultaneously. 
In this regime, the system is effectively governed by the PXP model Hamiltonian~\cite{Turner_2018_PXP, quantum_scar}:
\begin{equation} \label{eqn:PXP_Hamiltonian}
    H_{PXP} = \frac{\Omega(t)}{2} \sum_{i=2}^{N-1} P_{i-1} X_i P_{i+1} - \Delta(t)\sum_{i=1}^N n_i\;,
\end{equation}
where $\Omega(t)$ and $\Delta(t)$ are the global Rabi frequency and detuning as in \Cref{eq:Rydberg_Hamiltonian}, respectively.
In \Cref{eqn:PXP_Hamiltonian}, $P_{i} = \frac{1}{2}(I_i+Z_i)$ is the projector to the $\ket{0}$ subspace at the site $i$, and $n_i = I_i - P_i = \frac{1}{2}(I_i - Z_i)$ is the number operator.

Here, we prove that the PXP model cannot generate the dynamics of the ZXZ model described by \Cref{eq:zxz_app}.
If we define the control Hamiltonians of $H_{PXP}$ as 
\begin{equation} \label{eq:PXP_control}
    H_{\Omega} = \sum_{i=2}^{N-1} P_{i-1}X_i P_{i+1}\;,\quad H_{\Delta} = \sum_{i=1}^N n_i\;,
\end{equation}
we will prove the following proposition:
\begin{proposition} \label{prop:limitaion_PXP}
    Consider a chain of qubits.
    Given the quantum control described by the PXP model in \Cref{eqn:PXP_Hamiltonian}, one cannot simulate the dynamics of the ZXZ Hamiltonian in \Cref{eq:zxz_app}.
\end{proposition}

\begin{proof}
    It suffices to show that there is one evolution by $H_{ZXZ}$ that is not attainable by $H_{PXP}$. 
    For the chain with $N$ qubits, we define the computational basis states as $\{\ket{\mathbf{s}}\}$ for all binary strings $\mathbf{s} \in \mathbb{F}_2^N$.
    We define a doubly-excited state as the computational basis states $\ket{\mathbf{s}}$ with at least two consecutive $1$'s in the string $\mathbf{s}$, such as $\ket{1100\dots00}$ and $\ket{11100\dots 0}$.

   If the qubit chain is initialized to the all $0$ state $\ket{0}^{\otimes N}$, using the PXP control Hamiltonians in \Cref{eq:PXP_control}, we cannot evolve the system to any doubly-excited state, a defining property of the PXP model, as illustrated initially in~\cite{Turner_2018_PXP}.
   To see this, we first notice that $H_{\Delta}$ is diagonal in the computational basis, so we focus on the evolution generated by $H_{\Omega}$.
   For the $i$-th qubit, there is only one term, $P_{i-1}X_iP_{i+1}$, that flips it along the $X$-axis.
   The flipping happens when the neighboring qubits at sites $i-1$ and $i+1$ are both at the $\ket{0}$ state, a constraint imposed by the projectors $P_{i-1}$ and $P_{i+1}$.
   This prohibits the evolution from $\ket{0}^{\otimes N}$ to any doubly-excited state.

   However, one can use $H_{ZXZ}$ to evolve $\ket{0}^{\otimes N}$ to a  doubly excited state $\ket{\phi} \equiv \ket{011\dots 10}$ where there are $(N-2)$ consecutive $\ket{1}$'s between the two $\ket{0}$'s at both ends.
    To see this, we notice that:
    \begin{equation}
        H_{ZXZ} = \sum_{i=2}^{N-1}Z_{i-1}X_i Z_{i+1} = U_{CZ} \left(\sum_{i=2}^{N-1} X_i \right)U_{CZ}\;,\quad \mathrm{with}\quad U_{CZ} = \prod_{i=1}^{N-1} CZ_{i,i+1}\;,
    \end{equation}
    where $CZ_{i,i+1}$ is the controlled-Z gate between the $i$-th and $(i+1)$-th qubits, which has the following algebra with Pauli operators:
    \begin{equation}
        CZ_{i,i+1}\cdot X_i\cdot  CZ_{i,i+1} = X_i Z_{i+1}\;,\quad CZ_{i,i+1}\cdot  X_{i+1}\cdot  CZ_{i,i+1} = Z_i X_{i+1}\;.
    \end{equation}
    Utilizing the fact that $U_{CZ} \ket{0}^{\otimes N} = \ket{0}^{\otimes N}$ and $U_{CZ}\ket{\phi} = (-1)^{N-3}\ket{\phi}$, we find that a $\frac{\pi}{2}$-evolution by $H_{ZXZ}$ can evolve $\ket{0}^{\otimes N}$ to $\ket{\phi}$ as (using $U_{CZ}^2 = I$)
    \begin{equation}
    \begin{aligned}
        \bra{\phi} e^{-i\frac{\pi}{2}H_{ZXZ} } \ket{0}^{\otimes N} &= \bra{\phi}U_{CZ} e^{-i\frac{\pi}{2} \sum_{i=2}^{N-1} X_i} U_{CZ} \ket{0}^{\otimes N} = (-1)^{N-3}\bra{\phi} \prod_{i=2}^{N-1} e^{-i\frac{\pi}{2} X_i}\ket{0}^{\otimes N} \\
        &= (-1)^{N-3} \bra{\phi} \prod_{i=2}^{N-1} (-i)X_i \ket{0}^{\otimes N} = (-1)^{N-3} (-i)^{N-2}\;.
    \end{aligned}
    \end{equation}
    Up to a physically irrelevant phase factor $(-1)^{N-3} (-i)^{N-2}$, we evolve $\ket{0}^{\otimes N}$ to $\ket{\phi}$, a task impossible for the $PXP$ model. Therefore, we reach our conclusion.
\end{proof}

\section{Direct Trajectory Optimization\label{app:DC}}

\textit{Quantum optimal control} describes an optimization problem over the space of time-dependent state and control \textit{trajectories} subject to the Schr\"odinger equation. The goal of a quantum optimal control problem is to return controls $u(t)$ such that the generator of the dynamics---the Hamiltonian $H[\textbf{u}(t)]$---drives the quantum state to a desired form: for a state vector, that form might be a specific state preparation, while for a unitary propagator, it might be a high-fidelity gate. In all but the simplest scenarios, $H[\textbf{u}(t)]$ induces dynamics involving many non-commuting Hamiltonian terms. Analytic controls $u(t)$ are not easily found, and numerical approaches are necessary.

\textit{Trajectory optimization problems} are an umbrella term for this category of optimization problem. The core structure is
\begin{align}
    \min_{x(t), u(t)} &\quad \int_0^T dt \ \ell(x(t), u(t)) + \ell_T(x(T)) \label{eq:to-prob} \\
    \text{s.t.} &\quad \dot{x} = f(x(t), u(t), t) \label{eq:to-dynamics} \\
    &\quad x(0) = x_{\text{init}} 
\end{align}
Here, $x(t)$ is the state trajectory, $u(t)$ is the control trajectory, $t \in [0, T]$, $\ell(\cdot, \cdot)$ is an arbitrary objective function along the entire trajectory, and $\ell_T(\cdot)$ is a terminal objective. A typical terminal objective measures the distance between $x(T)$ and a goal state, $x_{\text{goal}}$. In practice, trajectory optimization problems can be solved numerically by first discretizing time into $N$ intervals, with the trajectory data retained as $N$ \textit{knot points} $z_k = \qty(x_k, u_k)$, $k \in \qty{1\dots N}$. The dynamics, Eq.~\eqref{eq:to-dynamics}, are enforced between knot points over each interval $[t_k, t_k + \Delta t_k]$, such that $x_{k+1} =  F(x_k, u_k, t_k, \Delta t_k) = x_k + \int_{t_k}^{t_k + \Delta t_k}\text{d}t \, f(x(t), u(t), t) $.

A common approach in quantum optimal control is to solve \ref{eq:to-prob} using \textit{indirect} trajectory optimization, i.e., optimizing solely over the control variables $u_{1:N-1}$:
\begin{align}
    \minimize{u_{1:N-1}} &\quad \sum_{k=1}^{N-1} \ell(x_k(u_{1:k-1}), u_k) + \ell_T\qty(x_N(u_{1:N-1})) \label{eq:to-indirect} \\
    \text{subject to} &\quad \abs{u(t)} \le u_\text{max} \label{eq:to-constraint-indirect}
\end{align}
where it is natural to add bound constraints on the controls, Eq.~\eqref{eq:to-constraint-indirect}. The GRAPE algorithm~\cite{khaneja2005optimal} is the most well-known example. Advantageously, indirect approaches involve a relatively small number of optimization variables because controls are fewer in dimension than states. However, the cost of this advantage is paid in function evaluations, as each $x_k$ must be retrieved by evolving the initial state according to the dynamics. This evolution becomes a part of any gradients involving $x_k$, which must propagate its dependence on $u_{1:{k-1}}$, making constraints on intermediate states difficult to enforce. Moreover, the cost landscape over the space of control trajectories becomes highly nonlinear and challenging to navigate~\cite{rabitz2004quantum, day2019glassy}.

In robotics and aerospace engineering, offline control problems like Eq.\eqref{eq:to-prob} are often solved using \textit{direct} trajectory optimization methods. The direct approach treats the state variables $x_{1:N}$ as optimization variables alongside the controls $u_{1:N-1}$, e.g., Ref.~\cite{goldschmidt2022model, trowbridge2023direct}:
\begin{align}
    \minimize{x_{1:N}, u_{1:N-1}} &\quad \sum_{k=1}^{N-1} \ell(x_k, u_k) + \ell_T(x_N) \label{to-direct} \\
    \text{subject to} &\quad x_{k+1} = F(x_k, u_k, t_k, \Delta t_k) \label{eq:to-dynamics-direct} \\
    &\quad c(x_k, u_k) \le 0 \label{eq:to-constraints-direct}\\ 
    &\quad x_1 = x_{\text{init}}
\end{align}
where $F$ in Eq.~\eqref{eq:to-dynamics-direct} denotes the dynamics, which is enforced as a nonlinear constraint. We include this dynamics constraint alongside arbitrary non-linear constraints, represented by Eq.~\eqref{eq:to-constraints-direct} and encompassing the bounds constraints on controls. Both constraints are handled by solving Eq.~\eqref{to-direct} using a modern nonlinear optimization tool like IPOPT~\cite{wachter2006implementation}, which is expressly designed for large-scale nonlinear programming with constraints. Direct methods have a significant structural advantage over indirect methods when functions of the state variables are required because the gradients will depend directly on $x_k$. Indeed, direct methods offer much finer control over properties like smoothness, time-optimality, and robustness that inform the design of the final controls~\cite{bhattacharyya2024using, goldschmidt2025quantum}. For example, by directly optimizing timesteps, direct methods can seek minimum time solutions or provide control over otherwise non-actuated dynamics (like drift Hamiltonians). Moreover, the cost landscape of $z_{1:N}$ is fundamentally different from the landscape over $u_{1:N}$, critical for navigating highly constrained spaces~\cite{trowbridge2023direct}. For the case of unitary optimal control, we set up the following direct trajectory optimization problem:
\begin{align}
    \minimize{z_{1:N}} &\quad \sum_k J(z_k) + Q \qty(1 - \mathcal{F}\qty(\vec{\tilde{U}}_N)) \label{eq:problem-template} \\
    \text{subject to} &\quad \vec{\tilde{U}}_{k+1} = \Phi \qty(\vec{\tilde{U}}_k, u_k, \dot{u}_k, t_k, \Delta t_k) \\
    &\quad u_{k+1} = u_k + \dot{u}_k \Delta t_k \\
    &\quad \dot{u}_{k+1} = \dot{u}_k + \ddot{u}_k \Delta t_k \\
    &\quad \abs{\ddot{u}_k} < \ddot{u}_{\max} \\
    &\quad t_{k+1} = t_k + \Delta t_k \\
    &\quad \vec{\tilde{U}}_1 = \text{isovec}(I)
\end{align}
where
\begin{equation}
   z_k = \begin{pmatrix} \vec{\tilde{U}}_k \\ u_k \\ \dot{u}_k \\ \ddot{u}_k \\ t_k \\ \Delta t_k \end{pmatrix} \qquad 
   \text{and} 
   \qquad \vec{\tilde{U}}_k = \text{isovec}(U_k) = \text{vec} \begin{pmatrix} \Re(U_k) \\ \Im(U_k) \end{pmatrix}
\end{equation}
To set up and solve Eq.~\eqref{eq:problem-template}, we use \texttt{Piccolo.jl}~\cite{piccolo2025}, a state-of-the-art software ecosystem addressing quantum optimal control using direct trajectory optimization.

Control pulses that vary smoothly over a parameterized gate family can be obtained by solving a coordinated quantum optimal control problem and interpolating the results~\cite{sauvage2022optimal,chadwick2023efficient,bhattacharyya2024using}. If the interpolated representation is a neural network, efficient calibration of the entire control manifold implementing the gate family is possible~\cite{bhattacharyya2024using}.

\section{Quantum Optimal Control Performance\label{app:direct_and_indirect_data}}
We compare memory and runtime costs for different approaches to optimal control of a quantum state in Table~\ref{tab:complexity_comparison}. In the table, $D$ is the dimension of the Hilbert space and $N$ is the number of knot points. The runtime of indirect solvers is dominated by numerical integration. The per-knot cost is $\mathcal{O}(\nu D^2)$ for $\nu$ matrix-vector operations between knot points. We include two categories of indirect solvers in the table, automatic differentiation and the adjoint state method. The difference between the two is the way the control gradients are computed and the resulting memory overhead. Automatic differentiation stores the computational graph across knots, then backtracks to compute the gradient. The memory requirement is increased to $\mathcal{O}(N D)$, but the backward pass can be parallelized (the forward pass must still be computed serially). Alternatively, adjoint state methods rely on a separate system of equations for the gradient that evolve backward in time, lowering the memory overhead to $\mathcal{O}(D)$ by adding a constant factor increase to runtime~\cite{gautier2025optimal}.

In a direct method solved using interior-point methods, like IPOPT, the dominant runtime cost comes from solving a large, sparse linear equation called the Karush-Kuhn-Tucker (KKT) system~\cite{wachter2006implementation}, which encodes the first-order optimality condition for the constrained problem. Due to the causal structure of optimal control, this linear equation is block-banded with block sizes of $D$. The banded linear system can be solved with $N$ separate block factorizations, for a total cost of $\mathcal{O}(N D^3)$. Direct trajectory optimization stores the state variables at each knot point, similar to automatic differentiation; however, a direct method must also store the dynamics Jacobians, which are used as part of the KKT system, naively requiring $\mathcal{O}(N D^2)$ in memory. 

Other classes of solvers can be adapted to achieve similar performance without paying this memory overhead. A notable class is matrix-free Augmented Lagrangian (AL) methods, à la Alpaqa \cite{alpaqa}. Matrix-free solver methods offer significant potential for scaling direct trajectory optimization methods in the quantum domain because the dynamics' Jacobians need not be explicitly formed; only Jacobian-vector products (JVPs) are required. The memory cost reduces back to the indirect case, where only states must be stored. AL also replaces the expensive KKT solve with an inner loop of easier problems that restores the $\mathcal{O}(N D^2)$ runtime. If computational resources allow, dividing the quantum state trajectory into $N$ knot points offers additional parallelization and convergence advantages unique to a direct approach. This is included as an additional column in Table~\ref{tab:complexity_comparison}.


\begin{table}
    \centering
    \renewcommand{\arraystretch}{1.5}
    \begin{tabular}{llcccc}
        \hline
        & Method & Memory & Runtime & Parallelizable? \\ \hline

        \multirow{2}{*}{\parbox{3.3cm}{\textbf{Indirect trajectory optimization}}}
        & Automatic differentiation
        & $\mathcal{O}(N D)$
        & $\mathcal{O}(\nu N D^2)$
        & $\times$ \\
        & Adjoint state method
        & $\mathcal{O}( D)$
        & $\mathcal{O}(\nu N D^2)$
        & $\times$ \\
        \hline

        \multirow{2}{*}{\parbox{3.3cm}{\textbf{Direct trajectory optimization}}}
        & IPOPT
        & $\mathcal{O}(N D^2)$
        & $\mathcal{O}(N D^3)$
        & $\checkmark$ \\
        & Augmented Lagrangian (AL)
        & $\mathcal{O}(ND)$
        & $\mathcal{O}(\nu N D^2)$
        & $\checkmark$ \\
        \hline
    \end{tabular}
    \caption{Per iteration memory and runtime complexity comparison of trajectory optimization methods for a quantum state in a Hilbert space of dimension $D$ with $N$ knot points. Between knot points, $\nu$ quantifies the number of computations per numerical integration.}
    \label{tab:complexity_comparison}
\end{table}


A well-conditioned optimizer converges better than an ill-conditioned optimizer. Direct methods are time-local and almost block diagonal, so sensitivity (and thereby, conditioning) scales weakly with the number of knot points, usually linearly. An indirect method compresses all time slices into one endpoint sensitivity, so even for unitary dynamics, conditioning worsens when many time slices push the gradient in the same direction. As knot points increase, the sensitivity eventually starts to diverge as the impact of controls becomes unidentifiable. In addition, IPOPT is particularly effective at achieving feasibility under difficult constraint sets. Alongside conditioning, this helps explain its ubiquity in fields like aerospace engineering, robotics, and chemical process engineering~\cite{betts2010practical}. Indeed, even the increased per-iteration runtime of IPOPT (which can be obviated by matrix-free AL methods) is often balanced out by its fast convergence properties and practical versatility~\cite{wachter2006implementation}. Reference~\cite{kamen2026accurate} has demonstrated that direct trajectory optimization (IPOPT) can significantly outperform indirect trajectory optimization (adjoint method) on constrained quantum control problems.
 
To directly evaluate the performance of direct quantum optimal control on our problems, we prepared pulses using GRadient Ascent Pulse Engineering (GRAPE), an \textit{indirect} quantum optimal control method frequently implemented using automatic differentiation (although a lower memory overhead is possible using hardcoded gradients or adjoint state methods)~\cite{khaneja2005optimal,gautier2025optimal}. In what follows, we outline our implementation of GRAPE. 

Our pulse parameters consist of the intermediate values of $\Omega(t)$ and $\Delta(t)$ (as the values at $t = 0$ and $t = T$ are fixed to zero). Between these intermediate points, the controls are linearly interpolated to produce piecewise linear pulses. To respect the time resolution available on Aquila, in total, there are 48 parameters for a pulse of length $T = 1.25\mu \text{s}$, with 24 parameters each for the detuning and Rabi frequency. We fix the interatomic spacing to be $8.9 \mu \text{m}$, the same value used for the direct method of pulse optimization. 
We utilized automatic differentiation and the Adam optimizer available in JAX to minimize the following loss function: 

\begin{equation}
    \ell = 1 - \langle F_i \rangle_i + \lambda \sum_k g_k + r\sum_\alpha \left\langle \frac{d^2 u_\alpha}{dt^2}\right\rangle_t
\end{equation}
Here, $F_i$ refers to the fidelity between a Haar random state $|\phi_i\rangle$ evolved under the target $H_{ZXZ}$ Hamiltonian for $\tau = 0.8$ and the same Haar random state evolved under the global control pulse for the native Hamiltonian. We average over fifty states $|\phi_i\rangle$ in total. Physical constraints on the value and first derivative of the control are encoded via the quadratic penalty functions $g_k$, which remain zero unless the upper or lower bounds are violated. If any violations occur, the difference between the violated bound and the value is squared and added to $g_k$, which is then weighted by $\lambda = 100$ to encourage respecting the physical constraints. Finally, $r$ is a regularization weighting the average second derivative of the control to encourage smoothness in the controls $u_\alpha$. We vary $r$ between 0.0 and 1e-6. 

We perform 100 optimization trials with the initial values for $u_\alpha$ randomly selected from a uniform distribution, using bounds $(0, b)$ for $\Omega(t)$ and $(-b, b)$ for $\Delta(t)$. After generating 100 optimization trials each for $b = 1, 2, $ and $5$, we determined that there is little dependence on optimization range. These results are provided in Fig. \ref{app_fig:grape} The optimized pulses are more heavily dependent on the regularization $r$; this dependence is visualized in Fig. \ref{fig:fig3} of the main text.

\begin{figure*}
    \includegraphics[width=6.0in]{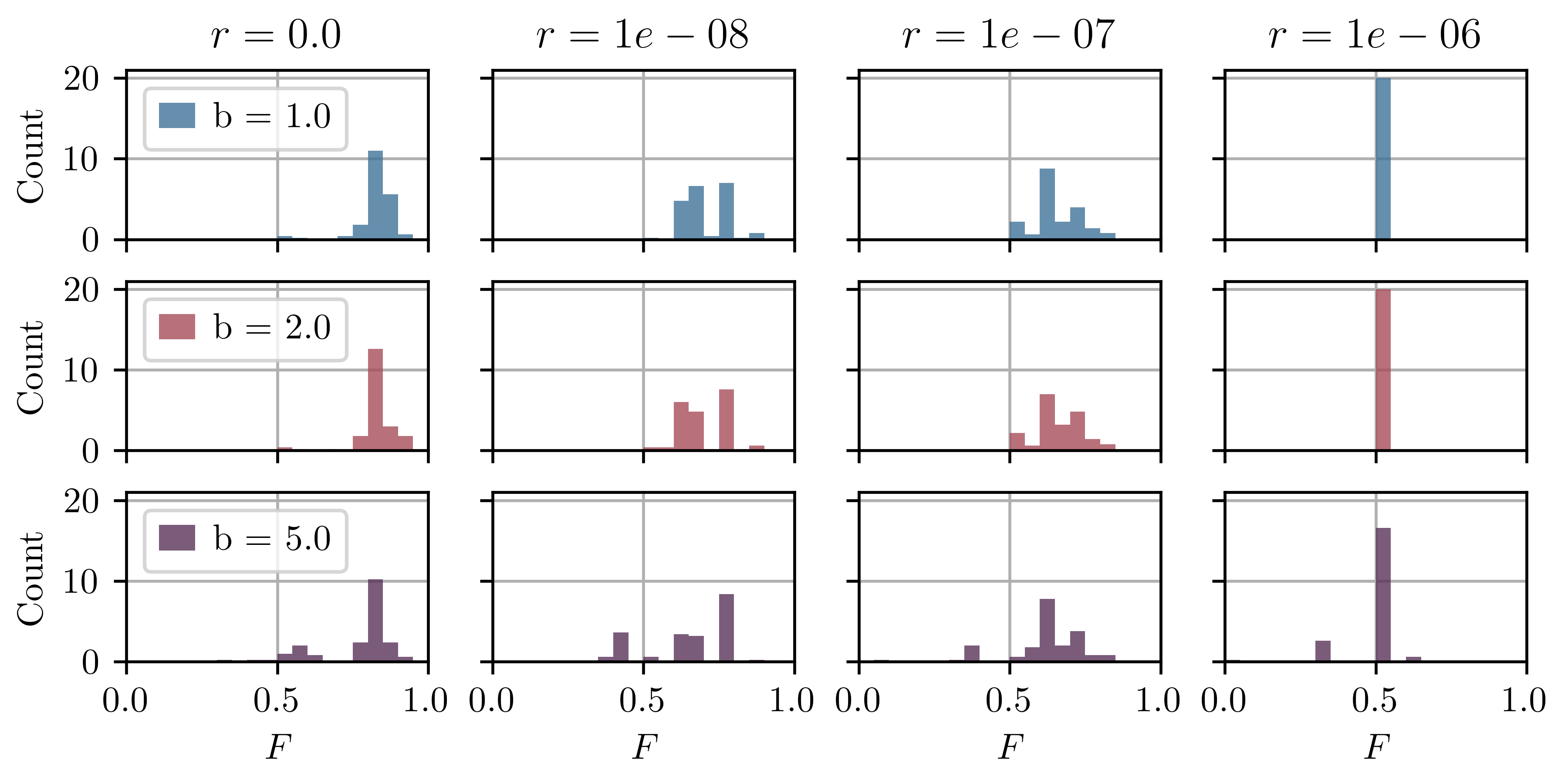}
    \caption{\textbf{GRAPE simulation data}. The second derivative regularization ($r$) is varied across columns, and the spread of the random initialization ($b$) is varied across rows. No appreciable dependence on $b$ is found in simulation, but the value of $r$ can highly influence the optimized pulse. Smaller values of $r$ tend to produce higher fidelity pulses; however, these solutions are increasingly non-smooth. \label{app_fig:grape}}
\end{figure*}


\section{Additional experimental results\label{app:additional_experimental}}

\begin{figure*}[htbp]
    \includegraphics[width=1\linewidth]{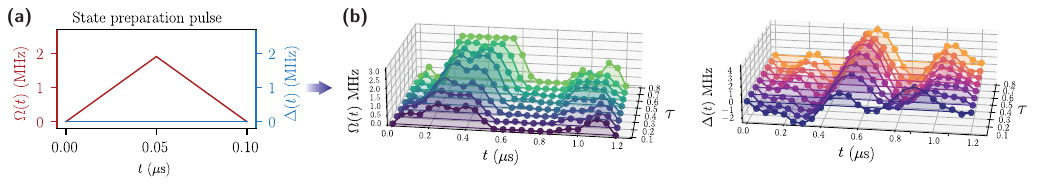}
    \caption{\textbf{Control pulses for additional experiments}. (a) State preparation pulse used to initialize the $\Lambda$-state, a different initial state for benchmarking. (b) Replotted control pulses from \Cref{fig:fig2}(c) used to engineer the ZXZ Hamiltonian dynamics, shown as a function of effective time $\tau$. \label{app_fig:state_prep_pulse}}
\end{figure*}

\begin{figure*}[htbp]
    \includegraphics[width=1\linewidth]{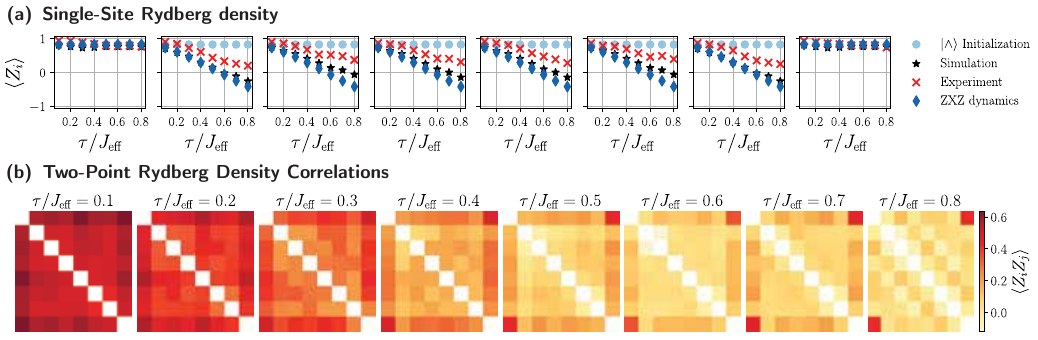}
    \caption{\textbf{Additional experimental signatures of topological dynamics}.
(a) Single-site Rydberg densities after applying the optimized global pulse [\Cref{fig:fig2}(c), \Cref{app_fig:state_prep_pulse}(b)] to an 8-atom chain initialized in the $\Lambda$-state using the state preparation pulse in \Cref{app_fig:state_prep_pulse}(a). Boundary atoms (1 and 8) retain high populations, while bulk atoms exhibit decay with increasing $\tau/J_{\mathrm{eff}}$, consistent with edge mode dynamics of the ZXZ Hamiltonian (black stars). Experimental data (red crosses) agrees with ideal simulations (blue diamonds).
(b) Connected two-point correlations $\langle Z_i Z_j \rangle$ measured at various $\tau/J_{\mathrm{eff}}$, showing persistent edge-edge correlations in experiment, in agreement with theoretical predictions. The size of the experimental error bars (standard deviation) is smaller than that of the red cross markers.\label{app_fig:additional_data}}
\end{figure*}

In this subsection, we provide additional experimental details. The Rydberg atom array platform used in our experiments is provided by QuEra Computing. More information can be found in the Aquila white paper \cite{aquila}. \Cref{tab:control_constraints} summarizes the key control parameters and constraints relevant to our work. Note that in the white paper the unit for the Rabi frequency is given in $\text{rad}/\mu\text{s}$, with the conversion $2\pi~(\text{rad}/\mu\text{s}) = 1~\text{MHz}$.

Due to hardware limitations discussed in the main text, we are restricted to measurements in the Pauli-$Z$ basis by detecting the Rydberg density of each atom. This constraint prevents access to information in other bases, limiting our ability to directly probe other topological edge signatures. To address this, we applied our control protocol to different initial states by prepending an additional short-duration global control pulse prior to the effective $ZXZ$ dynamics. Specifically, we used the global pulse shown in \Cref{app_fig:state_prep_pulse}(a) to prepare a distinct initial state, which we refer to as the $\Lambda$-state due to the shape of the pulse.

Prepending a state preparation pulse increases the total duration of the experiment and can introduce additional decoherence. To minimize this effect, we designed the pulse in \Cref{app_fig:state_prep_pulse}(a) to ramp the Rabi frequency up to near its maximum and then quickly back down. The corresponding experimental results are shown in \Cref{app_fig:additional_data}. In subplot (a), we observe good agreement between experiment and ideal simulation for the edge atoms, which maintain high Rydberg density, while the bulk atoms show slight deviations, possibly due to control errors. In subplot (b), persistent edge-edge correlations between the two boundary atoms are also observed, consistent with the expected topological behavior.

\begin{table}[htbp]
\centering
\begin{tabular}{|l|c|}
\hline
\textbf{Description} & \textbf{Constraint} \\
\hline
Rabi frequency bounds & $0 < \Omega < 2.41~\text{MHz}$ \\
Detuning range & $-19.9 < \Delta < 19.9~\text{MHz}$ \\
Maximum Rabi frequency slew rate & $\left|\delta \Omega/\delta t\right| < 39.7~\text{MHz}/\mu\text{s}$ \\
Maximum detuning slew rate & $\left|\delta \Delta/\delta t\right| < 397~\text{MHz}/\mu\text{s}$ \\
Minimum time resolution & $\delta t \geq 0.05~\mu\text{s}$ \\
\hline
\end{tabular}
\caption{System control constraints for laser pulse shaping.}
\label{tab:control_constraints}
\end{table}

In \Cref{app_fig:positions}(a), we show the spatial configuration of the atom sets used in \Cref{fig:fig4}. The separation between different sets is at least three times larger than the intra-set atomic spacing, ensuring that residual interactions between sets are negligible. In \Cref{app_fig:positions}(b), we present single-atom Rabi oscillation measurements for the central atom of each set with $\Delta = 0$, for which no significant spatial inhomogeneity is observed. We emphasize, however, that spatial inhomogeneities may arise under time-dependent control protocols.

\begin{figure}[htbp]
    \centering
    \includegraphics[width=1\linewidth]{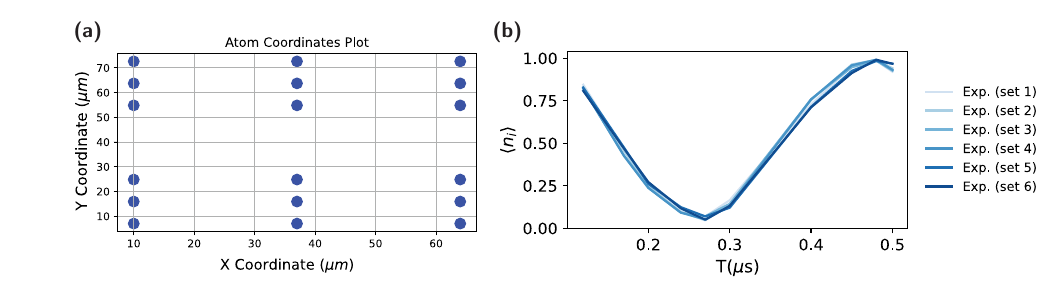}
    \caption{\textbf{Calibration experiments.} (a) Atom positions and the six three-atom sets used in \Cref{fig:fig4}. (b) Single-atom Rabi oscillation measurements for the central atom of each set with $\Delta = 0$.}
    \label{app_fig:positions}
\end{figure}

\section{Anti-concentration and information scrambling\label{app:scrambling}}

In this section, we review the concept of anti-concentration in quantum information scrambling and the simulation details related to \Cref{fig:scrambling}. In quantum information theory, anti-concentration measures the extent of information scrambling in a system by quantifying how widely the many-body wave function spreads across the computational basis $\{\ket{z}\}$
. Highly anti-concentrated states are broadly distributed among basis elements, making them challenging to simulate or learn with classical computers, a property that forms a cornerstone of proposed quantum supremacy experiments using random circuit sampling.

The uniform distribution is the ultimate anti-concentrated distribution because the probability mass is equally spread, with each outcome having probability $1/D$. A distribution is considered anti-concentrated as long as the average fluctuations from the uniform distribution are no larger than $O(1/D)$. One might expect that for random Haar states, the outcome probabilities would resemble a uniform distribution with each outcome having exactly probability $1/D$. However, due to coherent interference, this is not the case; instead, the probabilities exhibit speckle patterns. The statistical properties of $p(z)$ are universal and follow the so-called Porter-Thomas (PT) distribution. Specifically, when $D\gg 1$, the fraction of outcomes with $p(z)\in [x,x+dx]$ follows the Porter-Thomas distribution: $P[p(z)=x]=\mu^{-1}\exp(-x/\mu)dx$, where $\mu=1/D$. The anti-concentration statistics can be captured precisely by the \emph{collision probability}: $\mathcal{Z}=\sum_{z}p(z)^2$. The collision probability measures the probability that measurement outcomes from two independent copies of the quantum state or circuit are identical. For a uniform distribution $p(z)$, the collision probability reaches its minimum: $\mathcal{Z}_{\mathrm{uniform}}=1/D$. In contrast, for random Haar states where $p(z)$ follows the PT distribution, $\mathcal{Z}_{\mathrm{Haar}}=2/(D+1)\sim 2/D$, which is a constant factor larger than $\mathcal{Z}_{\mathrm{uniform}}$. In the literature, an output distribution $p(z)$ is considered anti-concentrated if its collision probability is a constant factor larger than the uniform distribution, i.e., $\mathcal{Z}=c\mathcal{Z}_{\mathrm{uniform}}$ with $c=O(1)$.

In this paper, we are interested in the time required for a globally driven dual-species system to anti-concentrate as a function of system size $N$. The system Hamiltonian is

\eqs{
H(t)=\sum_{j<k}V_{jk}n_i n_k+\sum_{\alpha\in \{A,B\}}\frac{\Omega_\alpha(t)}{2}(\sum_{i\in\{A,B\}} \ket{g_i}\bra{r_i}+\ket{r_i}\bra{g_i})-\sum_{\alpha\in \{A,B\}}\Delta_{\alpha}(t)(\sum_{i\in\{A,B\}}n_i),
}
where $V_{ij}$ is van der Waals interactions between atoms, $\Omega_{\alpha}(t)$ and $\Delta_{\alpha}(t)$ are global Rabi and detuning controls for the $A,B$ species atoms.

To quantify this, we focus on the rescaled output probability $w=Dp_U(z)=D|\langle z\ket{\psi}|^2=D|\langle z|U\ket{0}|^2$, where $U$ is a random unitary generated by the global pulses specified below. For random Haar states, we have
\eqs{
P[w=x]=\frac{D-1}{D}\left(1-\frac{x}{D}\right)^{D-2}\underset{\underset{D\gg 1}{\lim}}{=\!=}e^{-x}.
}

For random states generated by globally driven dual-species neutral atom arrays, we use the following protocol: The initial state is the ground state of all atoms, $\ket{0}^{\otimes N}$, which we abbreviate as $\ket{0}$ when there is no confusion. We place the two atom species in an alternating $ABAB...$ pattern with an even number of atoms in a one-dimensional chain, where $A$ denotes one atomic species and $B$ denotes the other. The interatomic separation is set to $d=8.9\mu \mathrm{m}$. We use $C_6=862690\times 2\pi \,\mathrm{MHz}\cdot \mu \mathrm{m}^6$ to describe the interactions between atoms, $V_{ij}=C_6/|r_i-r_j|^6$, although we emphasize that the particular value here is not critical and depends on the atomic species and levels used. For the random global pulses $\Omega_{A/B}(T)$, we use random piecewise constant functions with a constant time resolution $\Delta t=0.05\mu \mathrm{s}$ and choose values uniformly from $\Omega_{A/B}(t)\in[0,2.387]\,\mathrm{MHz}$. Similarly, the global detuning $\Delta_{A/B}(t)\in[-2.387,2.387]\, \mathrm{MHz}$ is chosen with the same time resolution. The number of piecewise constant segments determines the total evolution time. All these values are chosen to match a near-future dual-species atom array experimental setup.

Each random global pulse generates a unitary evolution $U$ that drives the state to $\ket{\psi_{U}}=U\ket{0}$. We then evaluate the collision probability
\eqs{
\mathcal{Z}=\mathbb{E}_{U}\left[\sum_{z}|\langle z|\psi_{U}\rangle|^4\right].
}
In practice, we sample 1000 random states and used empirical mean to estimate the ensemble average. In \Cref{fig:scrambling}(c), we plot the relative error defined as $\varepsilon_{r}=|\mathcal{Z}-\mathcal{Z}_{\mathrm{Haar}}|/\mathcal{Z}_{\mathrm{Haar}}$. 

Setting the characteristic time $T^*$ to reach anti-concentration as the time when $\varepsilon_{r}\leq 0.05$, we find that this globally driven quantum system with only temporal randomness anti-concentrates at $T^*\sim \log N$, similar to RUC.

\section{Additional results\label{app:additional_large_scale}}

In this section, we present supplementary numerical analyses to contextualize our experimental findings. We find that the discrepancy between the simulated and ideal ZXZ dynamics shown in \Cref{fig:fig5} (a) arises primarily from the finite fidelity of the short-duration pulses used in the experiment. To mitigate the relatively high noise levels of the current hardware, we constrained the evolution to a maximum pulse duration of $T_{\mathrm{max}}=1.2\,\mu\mathrm{s}$. As demonstrated in \Cref{fig:fig4} (c), decoherence remains negligible within this regime; however, this temporal constraint limits the system's expressivity and, consequently, the maximum achievable pulse fidelity. Despite these limitations, we identified high-fidelity pulse sequences using direct quantum collocation. If hardware coherence times improve to support pulses of $T_{\mathrm{max}}\approx 5\,\mu\mathrm{s}$ (\Cref{app_fig:performance_compare} (a)), performance is expected to enhance. To validate this, we compared the expectation values $\langle Z_i\rangle$ for each atom under four conditions: ideal ZXZ dynamics, the optimized $5\,\mu\mathrm{s}$ pulse in \Cref{app_fig:performance_compare} (a) (long-time pulse in the figure), the experimental pulse from \Cref{fig:fig2}(c) (short-time pulse in the figure), and the experimental pulse with Lindblad noise. The results confirm a clear performance gain with longer durations when there is no decoherence. Consequently, we conclude that the minor deviations from ideal ZXZ dynamics observed for atoms 3 and 6 in \Cref{fig:fig5} (a) are predominantly a consequence of pulse imperfections necessitated by the short-time constraint.
Beyond single-site expectation values, we evaluate the spatial correlations by plotting the two-point correlators $\langle Z_i Z_j \rangle$ in \Cref{app_fig:performance_compare} (c). Consistent with our previous observations, the extended pulse duration yields superior performance. In the experimental setting, the interplay between the limited expressivity of short-time pulses and Lindbladian decoherence causes the correlations between boundary atoms to be lower than those predicted by ideal ZXZ dynamics. Future implementations could benefit from hardware advancements, such as the integration of additional trapping lasers to maintain confinement for atoms in Rydberg states. Such state-insensitive trapping would mitigate decoherence, enabling the execution of the higher-fidelity, long-duration pulses identified in our simulations.


\begin{figure*}[htbp]
    \includegraphics[width=1\linewidth]{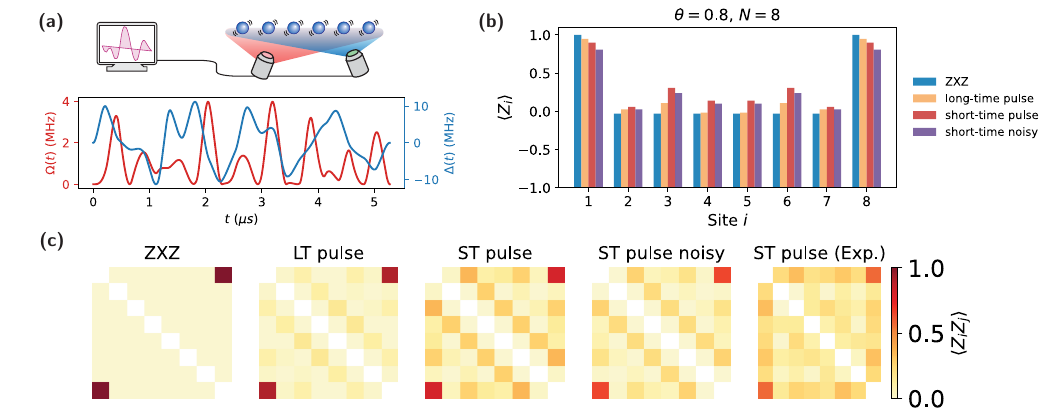}
    \caption{\textbf{Additional results on topological edge modes:} (a) Optimized long-duration control pulse, characterized by the Rabi frequency $\Omega(t)$ and detuning $\Delta(t)$, achieving better fidelity in implementing $\exp(-i\theta H_{\mathrm{ZXZ}})$ with $\theta = 0.8$.
(b) Simulated expectation values $\langle Z_i \rangle$ for each site of an eight-atom chain with interatomic spacing $d = 8.9~\mu\mathrm{m}$. The long-time pulse is shown in panel (a), while the short-time pulse corresponds to the experimental protocol used in \Cref{fig:fig2}(c). We compare ideal ZXZ dynamics with simulations using the long-time pulse, the short-time pulse, and the short-time pulse with Lindblad noise, all based on the full Rydberg Hamiltonian.
(c) Comparison of two-point correlations $\langle Z_i Z_j \rangle$ obtained from ideal ZXZ dynamics, long-time (LT) pulses, short-time (ST) pulses, ST pulses with Lindblad noise, and experimental ST-pulse data. The dominant discrepancy between experiment and ideal ZXZ dynamics arises from imperfections in the short-time pulse. Due to hardware constraints, this pulse represents the best achievable approximation to the ideal ZXZ dynamics for the current hardware. \label{app_fig:performance_compare} }
\end{figure*}

A significant advantage of global control is the ability to optimize pulses on smaller systems and extrapolate them to larger scales. Since the control fields are spatially uniform, pulses optimized for small-to-medium-sized chains can be directly applied to large system regimes. To demonstrate this scalability, we utilized tensor network methods to simulate the performance of the experimental pulses, visualized in \Cref{fig:fig2} (c), on a 50-atom chain.The time evolution of the expectation values $\langle Z_i\rangle$ as a function of the effective dynamical time $\tau/J_{\mathrm{eff}}$ is presented in \Cref{app_fig:50atoms}. Notably, the edge spins $\langle Z_{1}\rangle$ and $\langle Z_{50}\rangle$ exhibit a robust constant, in stark contrast to the decay observed in the bulk atoms. These simulations were performed using a second-order Trotterization scheme with a bond dimension $\chi=256$ and a time step of $\delta t=0.01\,\mu\mathrm{s}$. Given that the third-order long-range interactions are approximately $0.0002$ relative to the leading terms, we truncated the Hamiltonian to include only nearest-neighbor and next-nearest-neighbor interactions without loss of physical accuracy.

\begin{figure*}[htbp]
    \includegraphics[width=1\linewidth]{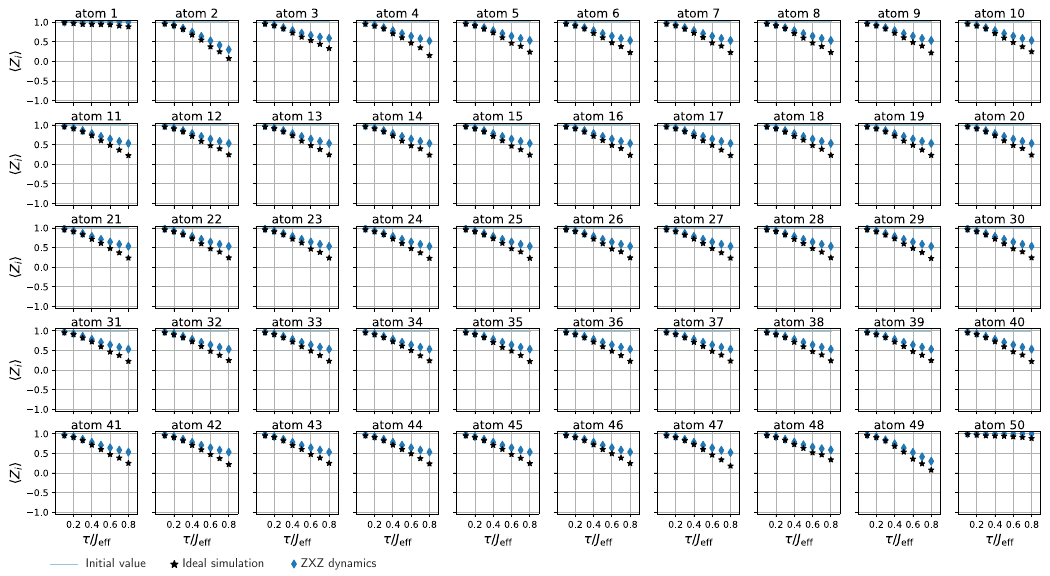}
    \caption{\textbf{Large-scale simulation of a 50-atom chain}. Tensor-network simulation results for a 50-atom chain driven by the same experimental control pulse used in \Cref{fig:fig2}(c). The dynamics of $\langle Z_i \rangle$ are shown for each site, and are compared with the corresponding ideal ZXZ Hamiltonian evolution.  
\label{app_fig:50atoms}}
\end{figure*}

Lastly, we demonstrate another application of the globally controlled universal system by realizing a subregion entangling gate. Specifically, we implement a controlled-NOT (CNOT) gate between the first two qubits, as shown in \Cref{app_fig:CNOT}(a). This gate generates entanglement within a targeted subregion while leaving the remaining qubits unentangled. The experimental setup is inspired by dual-species atom arrays with global control fields. We emphasize that this capability is not limited to a specific platform: any globally driven universal analog system can, in principle, achieve this functionality. Our setup is largely motivated by recent experimental advances in this direction. The corresponding global control Hamiltonian is shown in \Cref{app_fig:CNOT}(b).

Although the target unitary appears simple, it is important to emphasize that the underlying analog system features always-on interactions and therefore naturally tends to generate entanglement across the entire system, whereas the desired unitary entangles only a specific subregion. Using direct quantum optimal control, we identify global control pulses, shown in \Cref{app_fig:CNOT}(c), that achieve this selective entangling operation. In \Cref{app_fig:CNOT}(d), we visualize the density matrices of the ideal target state and the simulated state under the full Rydberg Hamiltonian, demonstrating a high state fidelity of 99.8\%.

\begin{figure*}[htbp]
    \centering
    \includegraphics[width=1\linewidth]{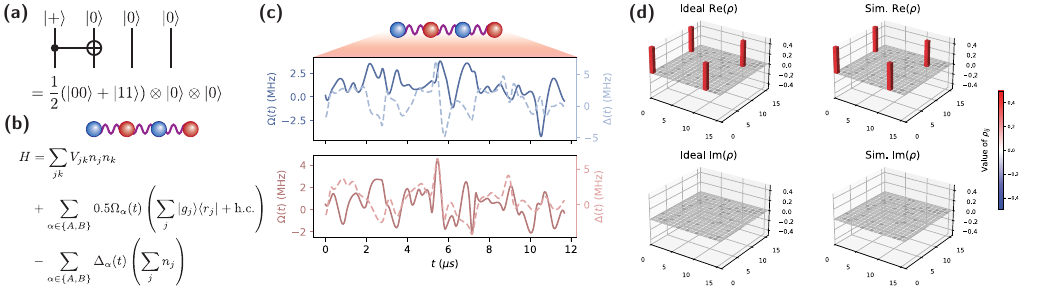}
    \caption{\textbf{Realizing a subregion entangling gate with globally driven dual-species atom arrays.} (a) Target quantum circuit implementing a CNOT gate on the first two qubits. Starting from the initial state $\ket{+,0,0,0}$, the circuit generates the maximally entangled state $\tfrac{1}{\sqrt{2}}(\ket{00}+\ket{11}) \otimes \ket{0} \otimes \ket{0}$, with entanglement localized to the first two qubits.
(b) Proposed experimental setup using a four-atom dual-species array with global control fields. Due to always-on interactions, the system naturally tends to generate multipartite entanglement, posing a challenge for subregion control.
(c) Optimized global control pulses for the two atomic species that implement the desired CNOT gate on the first two qubits.
(d) Real and imaginary parts of the density matrices for the ideal target state and the simulated state under the full Rydberg Hamiltonian. The resulting state fidelity reaches 99.8\%.}
    \label{app_fig:CNOT}
\end{figure*}

\end{document}